\documentclass[11pt]{amsart}
\pdfoutput=1
\usepackage[english]{babel}
\usepackage{graphicx,amsmath}
\usepackage{enumitem}
\usepackage{MnSymbol}
\usepackage{subfig}
\usepackage{float}
\usepackage{floatflt,young,youngtab}
\usepackage{tikz}
\theoremstyle{plain}
\numberwithin{equation}{section}

\newtheorem{theorem}{Theorem}[subsection]
\newtheorem{corollary}[theorem]{Corollary}
\newtheorem{lemma}[theorem]{Lemma}
\newtheorem{proposition}[theorem]{Proposition}
\newtheorem{remark}{Remark}[subsection]

\newtheorem{definition}{Definition}[subsection]
\newtheorem{example}{Example}[subsection]
\newtheorem{conjecture}{Conjecture}[subsection]
\newtheorem{construction}{Construction}[subsection]

\def \GTNN {{Gr^{\mbox{\tiny TNN}} (k,n)}}
\def \GTP {{Gr^{\mbox{\tiny TP}} (k,n)}}
\newcommand\mycom[2]{\genfrac{}{}{0pt}{}{#1}{#2}}
\def \DKP {{\mathcal D}_{\textup{\scriptsize KP},\Gamma}}
\def \DVG {{\mathcal D}_{\textup{\scriptsize vac},\Gamma}}
\def \DS {{\mathcal D}_{\textup{\scriptsize S},\Gamma_0}}
\def \DVN {{\mathcal D}_{\textup{\scriptsize vac},{\mathcal N}^{\prime}}}
\def \DDN {{\mathcal D}_{\textup{\scriptsize dr},{\mathcal N}^{\prime}}}
\def \S {{\mathcal S}_{\mathcal M}^{\mbox{\tiny TNN}}}
\def \gvac {\gamma_{\textup{\scriptsize vac},V_l}}
\def \gdr {\gamma_{\textup{\scriptsize dr},V_l}}

\def \Pdr {P^{\textup{\scriptsize (dr)}}}

\def \i {\mbox{\scriptsize{i}}}
\def \w {\mbox{\scriptsize{w}}}
\title{KP theory, plabic networks in the disk 
and rational degenerations of $\mathtt M$--curves}
\author{Simonetta Abenda}
\address{Dipartimento di Matematica, Universit\`a di Bologna, P.zza di Porta San Donato 5, I-40126 Bologna BO, ITALY
}
\email{simonetta.abenda@unibo.it
}
\author{Petr G. Grinevich}
\address{L.D.Landau Institute for Theoretical Physics,
pr. Ak Semenova 1a, Chernogolovka, 142432, Russia,
{\footnotesize pgg@landau.ac.ru}\\
Lomonosov Moscow State University,
Faculty of Mechanics and Mathematics, 
Russia, 119991, Moscow, GSP-1, 1 Leninskiye Gory, Main Building,\\
Moscow Institute of Physics and Technology, 
9 Institutskiy per., Dolgoprudny,
Moscow Region, 141700, Russia.}

\thanks{
This research has been partially supported by GNFM-INDAM and RFO University of Bologna, by the Russian Foundation for Basic Research, grant 17-01-00366, 
by the program ``Fundamental problems of nonlinear dynamics'', Presidium of RAS. Partially this research was fulfilled during the visit of the second author (P.G.) to IHES, 
Université Paris-Saclay, France in November 2017.}

\begin{document}

\begin{abstract}
{In this paper we extend and complete the program started in \cite{AG1,AG3} of connecting totally non--negative Grassmannians to the reality problem in KP finite--gap theory via the assignment of real regular divisors on rational degenerations of $\mathtt M$--curves for the class of real regular multi--line soliton solutions of Kadomtsev-Petviashvili II (KP) equation whose asymptotic behavior in space-time has been combinatorially characterized in a series of papers by S. Chakravarthy, Y. Kodama and L. Williams \cite{CK, KW1,KW2}. 

At this aim, we use the planar bicolored trivalent networks in the disk which were introduced by A. Postnikov \cite{Pos} to parametrize positroid cells in totally nonnegative Grassmannians $\GTNN$. In our construction the boundary of the disk corresponds to the rational curve associated to the soliton data in the direct spectral problem, and the bicolored graph is the dual of a reducible curve $\Gamma$ which is the rational degeneration of a regular $\mathtt M$--curve whose genus $g$ equals the number of faces of the network diminished by one. 

We then assign systems of edge vectors to such networks. The system of relations satisfied by these vectors has maximal rank and may be reformulated in the form of edge signatures as proposed by T. Lam \cite{Lam2}. 
Adapting remarkable results by A. Postnikov \cite{Pos} and K. Talaska \cite{Tal2} to our setting, we prove that the components of the edge vectors are rational in the edge weights with subtraction free denominators and provide their explicit expressions in terms of conservative and edge flows. 

The edge vectors rule the value of the KP wave function at the double points of $\Gamma$, whereas the signatures at the vertices rule the position of the divisor points in the ovals. In particular,
we provide a combinatorial proof that the degree $g$ divisor satisfies the conditions settled in B. Dubrovin and S. Natanzon \cite{DN} for real finite--gap solutions, {\sl i.e.} there is exactly one divisor point in each finite oval and no divisor point in the oval containing the essential singularity of the wave function. The divisor points may be explicitly computed using the linear relations satisfied by the wave function at the internal vertices of the chosen network.

Finally we explain the role of moves and reductions in the transformation of both the curve and the divisor for given soliton data, and we apply our construction to some examples.} 


\end{abstract}
\maketitle

\tableofcontents
\section{Introduction}

Totally non--negative Grassmannians $\GTNN$ historically first appeared as a special case of the generalization to reductive Lie groups by Lusztig  \cite{Lus1,Lus2} of the classical notion of total positivity \cite{GK,GK2,Sch,Kar}. As for  classical total positivity, $\GTNN$ naturally arise in relevant problems in different areas of mathematics and physics. The combinatorial objects introduced by Postnikov \cite{Pos}, see also \cite{Rie}, to characterize $\GTNN$ have been linked to cluster algebras in \cite{Sc,OPS}. In particular the plabic (planar bicolored) graphs introduced in \cite{Pos} have appeared in many contexts, such as the topological classification of real forms for isolated singularities of plane curves \cite{FPS}, they are on--shell diagrams (twistor diagrams) in scattering amplitudes in $N=4$ supersymmetric Yang--Mills theory \cite{AGP1,AGP2,ADM} and have a statistical mechanical interpretation as dimer models in the disk \cite{Lam1}. Totally non-negative Grassmannians naturally appear in many other areas, including the theory of Josephson junctions \cite{BG}, statistical mechanical models such as the asymmetric exclusion process \cite{CW}.  
In particular, the deep connection of the combinatorial structure of $\GTNN$ with KP real soliton theory was unveiled in a series of papers by Chakravarthy, Kodama and Williams (see \cite{CK,KW1,KW2} and references therein). In \cite{KW1}
it was proven that multi-line soliton solutions of the Kadomtsev-Petviashvili 2 (KP) equation are real and regular in space--time if and only if their soliton data correspond to points in the irreducible part of totally non--negative Grassmannians, whereas the combinatorial structure of the latter was used in \cite{CK,KW2} to classify the asymptotic behavior in space-time of such solutions.

In \cite{AG1,AG3} we started to investigate a connection of different nature between this family of KP solutions and total positivity in the framework of the finite-gap approach, using the fact that any such solution may also be interpreted as a potential in a degenerate spectral problem for the KP hierarchy.

Before continuing, let us briefly recall that the finite-gap approach to soliton systems was first suggested by Novikov \cite{Nov} for the Korteveg-de Vries equation, and extended to the 2+1 KP equation by Krichever in \cite{Kr1,Kr2}, where it was shown that finite-gap KP solutions correspond to non special divisors on arbitrary algebraic curves. Dubrovin and Natanzon \cite{DN} then proved that real regular KP finite gap solutions correspond to divisors on smooth $\mathtt M$--curves satisfying natural reality and regularity conditions. In \cite{Kr4} Krichever developed, in particular, the direct scattering transform for the real regular parabolic operators associated with KP and proved that the corresponding spectral curves are always $M$-curves. In \cite{Kr3, KV} finite gap theory was extended to reducible curves in the case of degenerate solutions. Applications of singular curves to the finite-gap integration are reviewed in \cite{Taim}.

In our setting the degenerate solutions are the real regular multiline KP solitons studied in \cite{BPPP,CK,KW1,KW2}: the real regular KP soliton data correspond to a well defined reduction of the Sato Grassmannian \cite{S}, and they are parametrized by pairs $(\mathcal K, [A])$, {\sl i.e.} $n$ ordered phases $\mathcal K =\{ \kappa_1<\kappa_2 <\cdots <\kappa_n\}$ and a point in an irreducible positroid cell $[A]\in \S \subset Gr^{\mbox{\tiny TNN}}(k,n)$. We recall that the irreducible part of $Gr^{\mbox{\tiny TNN}}(k,n)$ is the natural setting for the minimal parametrization of such solitons \cite{CK,KW2}.

Following \cite{Mal}, to the soliton data $(\mathcal K, [A])$ there is associated a rational spectral curve $\Gamma_0$ (Sato component), with a marked point $P_0$ (essential singularity of the wave function), and $k$ simple real poles $\DS =\{ \gamma_{S,r},\  r\in [k] \}$, such that $\gamma_{S,r}\in [\kappa_1,\kappa_n]$ (Sato divisor). However, due to a mismatch between the dimension of $Gr^{\mbox{\tiny TNN}}(k,n)$ and that of the variety of Sato divisors, generically the Sato divisor is not sufficient to determine the corresponding KP solution. 

In \cite{AG1, AG3} we proposed a completion of the Sato algebraic--geometric data based on the degenerate finite gap theory of \cite{Kr3} and constructed divisors on reducible curves for the real regular multiline KP solitons. In our setting, the data $(\Gamma,P_0,\mathcal D)$, where $\Gamma$ is a reducible curve with a marked point $P_0$, and $\mathcal D\subset\Gamma$ is a divisor, correspond to the soliton data $(\mathcal K, [A])$ if
\begin{enumerate}
\item $\Gamma$ contains $\Gamma_0$ as a rational component and $\DS$ coincides with the restriction of $\mathcal D$ to $\Gamma_0$. We assume here that different rational components of $\Gamma$ are connected at double points;  
\item The data $(\Gamma,P_0,\mathcal D)$ uniquely define the wave function $\hat\psi$ as a meromorphic function on $\Gamma\backslash P_0$ with divisor $\mathcal D$, having an essential singularity at $P_0$. Moreover, at double points the values of the wave function coincide on both components for all times.
\end{enumerate}
In degenerate cases, the construction of the components of the curve and of the divisor is obviously not unique and,
as pointed out by S.P. Novikov, an untrivial question is whether \textbf{real regular} soliton solutions can be obtained as rational degenerations of \textbf{real regular} finite-gap solutions. In the case of the real regular KP multisolitons this imposes the following additional requirements:
\begin{enumerate}
\item $\Gamma$ is the rational degeneration of an $\mathtt M$--curve;
\item The divisor is contained in the union of the ovals.
\end{enumerate}

In \cite{AG1} we provided an optimal answer to the above problem for the real regular soliton data in the totally positive part of the Grassmannian, $\GTP$. We proved that $\Gamma_0$ is a component of a reducible curve $\Gamma(\xi)$ arising as a rational degeneration of some smooth $\mathtt M$--curve of genus equal to the dimension of the positive Grassmannian, $k(n-k)$. We also proved that this class of real regular KP multisoliton solutions may be obtained from real regular finite-gap KP solutions, since soliton data in $\GTP$ can be parametrized by real regular divisors on $\Gamma(\xi)$, {\sl i.e.} one divisor point in each oval of $\Gamma(\xi)$ but the one containing the essential singularity of the wave function.  In \cite{AG1}, we used classical total positivity for the algebraic part of the construction and computed explicitly the divisor positions in the ovals at leading order in $\xi$.

In \cite{AG3} we extended the construction of \cite{AG1} to the whole totally non-negative Grassmannian $\GTNN$. Moreover, we made explicit the relation between the degenerate spectral problem associated to such family of solutions and the stratification of $\GTNN$, by proving that the Le-graph associated to the soliton data $[A]$ is the dual graph of the corresponding reducible spectral curve, and that the linear relations at the vertices of the Le-network uniquely identify the divisor satisfying the reality and regularity conditions established in \cite{DN}. Again our approach was constructive and in \cite{AG2} we applied it to obtain real regular finite gap solutions parametrized by real regular non special divisors on a genus 4 $\mathtt M$--curve obtained from the desingularization of spectral problem for the soliton solutions in $Gr^{\mbox{\tiny TP}}(2,4)$.
	
In \cite{AG3}, we used the canonical acyclic orientation on the Le-graph to prove
that any real regular KP multi-line soliton solution can be obtained as a degeneration of a real regular finite-gap solution of the same equation. Therefore the question of the invariance of the construction with respect to changes of orientation and to the several gauge freedoms in the construction was open. Moreover, any positroid cell in a totally non-negative Grassmannian is represented by an equivalence class of plabic graphs in the disc with respect to a well-defined set of moves and reductions \cite{Pos}. Therefore another set of natural questions left open in \cite{AG3} was how to associate a curve and a divisor to each plabic network in the disk and to explain the transformation of such algebraic geometric data with respect to Postnikov moves and reductions. 

In this paper we answer positively all such questions under some genericity assumptions and we implement an algebraic construction of systems of edge vectors on plabic networks which we think is of more general interest than the present application to KP finite gap theory. Indeed, our construction gives a sufficient set of rules to glue any given positroid cell $\S$ out of the little positive Grassmannians $Gr^{TP}(1,3)$ and $Gr^{TP}(2,3)$, represented by the trivalent white and black vertices of its graph, in the framework of KP theory. This gluing problem has been originally set up in theoretical physics \cite{AGP1,AGP2} and 
is equivalent to subdividing polytopes into smaller polytopes (see \cite{Pos2} and references therein). In particular,
our approach provides sufficient conditions to the formulation in Lam \cite{Lam2} of the problem of assigning signatures to edges of plabic graphs to characterize the totally non-negative part in the space of relations, and looks compatible with the binary codification of total non--negativity in \cite{ATT}.

Below we outline the construction and the main results of this paper.

\smallskip

\paragraph{\textbf{Main results}}
Let the soliton data $(\mathcal K, [A])$ be fixed, with $[A]\in \S \subset Gr^{\mbox{\tiny TNN}}(k,n)$ and $\S$ an irreducible positroid cell of dimension $|D|$, and let $\mathcal G$ be a planar bicolored directed trivalent perfect (PBDTP) graph in the disk representing $\S$ (Definition \ref{def:graph}). In our setting boundary vertices are all univalent, internal sources or sinks are not allowed, internal vertices may be either bivalent or trivalent and $\mathcal G$ may be either reducible or irreducible in Postnikov sense \cite{Pos}. $\mathcal G$ has $g+1$ faces where $g=|D|$ if the graph is reduced, otherwise $g>|D|$.

The construction of the curve $\Gamma$ (Section \ref{sec:gamma}) is analogous to that in \cite{AG3} where we treated the case of Le--graphs. $\mathcal G$ is the dual graph of a reducible curve $\Gamma$ which is the connected union of rational components: the boundary of the disk and all internal vertices of $\mathcal G$ are copies of $\mathbb{CP}^1$, the edges represent the double points where two such components are glued to each other and the faces are the ovals of the real part of the resulting reducible curve. 
We identify the boundary of the disk with the Sato component $\Gamma_0$, and the $n$ boundary vertices $b_1,\dots, b_n$ counted clockwise correspond to the ordered marked points, $\mathcal K = \{ \kappa_1 < \kappa_2 < \cdots < \kappa_n \}$. It is easy to check that  $\Gamma$ is a rational reduction of a smooth genus $g$ $\mathtt M$--curve.
 
Then, we fix an orientation $\mathcal O$ on $\mathcal G$ and assign positive weights to the edges so that the resulting oriented network $(\mathcal N, \mathcal O)$ represents $[A]$. $\mathcal O$ induces natural coordinates on each component of $\Gamma$ associated to a vertex (see Definition \ref{def:loccoor}). On $\mathcal N$ we also fix a reference direction $\mathfrak l$ (gauge ray direction, see Definition \ref{def:gauge_ray}) to measure the winding and count the number of boundary sources encountered along a walk starting at an internal edge and reaching the boundary of the disk. 

In Section \ref{sec:def_edge_vectors}, for any given edge $e$ in $(\mathcal N,\mathcal O,\mathfrak l)$, we consider all directed walks from $e$ to the boundary and to each such walk we assign three numbers: weight, winding and number of intersections with gauge rays starting at boundary sources. Then the $j$--th component of the edge vector $E_e$ is just the (finite or infinite) sum of such signed contributions over all directed walks from $e$ to the boundary vertex $b_j$. Adapting remarkable results in \cite{Pos,Tal2} to our setting, in Theorem \ref{theo:null} we prove that all components of such edge vectors are rational expressions in the edge weights with subtraction free denominators and we provide their explicit expressions in terms of the conservative and the edge flows defined in Section \ref{sec:flows}. Moreover, we prove that the linear system at internal vertices satisfied by edge vectors has maximal rank for any given boundary conditions at the boundary sinks (Theorem \ref{theo:consist}). The linear system has a purely geometrical interpretation: on each $(\mathcal G,\mathcal O,\mathfrak l)$ it induces a unique assignment of signatures to pair of half edges at internal vertices (Section \ref{sec:lam1}). We also provide explicit formulas for the dependence of the edge vectors on the orientation and the weight, vertex and ray direction gauge freedoms of planar networks (Sections \ref{sec:gauge_ray}, \ref{sec:orient} and \ref{sec:different_gauge}). Finally we explain the dependence of edge vectors on Postnikov's moves and reductions (Section \ref{sec:moves_reduc}). 

If $\mathcal G$ is reducible ($g>\mbox{dim} (\S)$), null edge vectors may appear in the solution to the linear system even if there do exist paths starting at the given edge and reaching the boundary (Section \ref{sec:null_vectors}). In this case the Postnikov map is surjective, but not injective, therefore there is an extra freedom in the assignment of the edge weights, which we call the unreduced graph gauge freedom (Remark \ref{rem:gauge_freedom}). We conjecture that, using such extra gauge freedom, it is possible to choose weights on reducible graphs so that all edge vectors are not null (Conjecture \ref{conj:null}). 
Since both the meromorphic extension of the wave function and the construction of the divisor require extra care in the case of null edge vectors (see Section \ref{sec:constr_null}), we postpone technical details concerning this case to a future publication.  

\smallskip  

For the construction of the wave function and the divisor, we assume that the network $\mathcal N$ of graph $\mathcal G$ representing $[A]$ does not possess null vectors. 
We then construct a dressed edge wave function $\Psi (\vec t)$ on the network using the edge vectors (Construction \ref{def:vvw_gen}) and we assign a real number $\gamma_V$ to each trivalent white vertex $V$ (Definition \ref{def:vac_div_gen}). The normalized edge wave function $\hat \Psi(\vec t) = \frac{\Psi (\vec t)}{\Psi (\vec t_0)}$ coincides with Sato wave function at the edges at the boundary of the disk and is independent on the orientation and on the ray direction, weight and vertex gauges of the network (Proposition \ref{prop:change_orient}). Therefore it is fully justified to choose $\hat \Psi$ to rule the value of the KP wave function at the double points of $\Gamma$.

In Section \ref{sec:inv} we use the normalized dressed edge wave function to rule the value of the wave function at the double points and then extend it to $\Gamma$ as follows. Let $\Gamma_V$ be the component in $\Gamma$ corresponding to the internal vertex $V$.
Since $\hat \Psi$ takes the same value at all edges at any given bivalent or black trivalent vertex $V$, on $\Gamma_V$ we extend it to a regular function $\hat \psi(P,\vec t)$ independent on the spectral parameter $P$. On the contrary, $\hat \Psi (\vec t)$ generically takes distinct values at the edges at a trivalent white vertex $V$; then we extend it to a degree one meromorphic function  $\hat \psi(P,\vec t)$, $P\in\Gamma_V$ such that its pole $P_V$ has coordinate $\gamma_V$. In particular, the condition that the wave function takes the same value at each fixed pair of double points for all times is automatically satisfied.
 
The KP divisor on $\Gamma$ is then $\DKP =\DS \cup \{ P_V \, , \, V \mbox{ white trivalent vertex in } \mathcal G \, \}$, where $\DS$ is the Sato divisor. Since the number of trivalent white vertices of a PBDTP graph is $g-k$, $\DKP$ has degree $g$ and by construction it is contained in the union of the ovals of $\Gamma$. $\hat \psi(P,\vec t)$ is the unique meromorphic function on $\Gamma\backslash \{ P_0\}$ and divisor $\DKP$, that is $(\hat \psi(\cdot,\vec t)) + \DKP \ge 0$, for all $\vec t$, therefore $\hat \psi$ is the wave function on $\Gamma$ for the soliton data $(\mathcal K, [A])$ (Theorem \ref{lemma:KPeffvac}).

By construction $\DKP$ is invariant with respect to changes of orientation and of the choice of gauge ray, weight and vertex gauges (Theorem \ref{theo:inv}). Therefore, if $\mathcal G$ is reduced so that $g=\mbox{dim} (\S)$, as a byproduct, we get a local parametrization of positroid cells in terms of non--special divisors. We remark that the map from the weights, parametrizing the cell, to the divisor looses maximal rank and injectivity along certain subvarieties of the positroid cell where the divisor becomes special (see Section \ref{sec:global} for the case $Gr^{\mbox{\tiny TP}} (1,3)$). In the case of reducible graphs ($g>\mbox{dim} (\S)$), the divisor depends on the extra freedom in the assignment of the edge weights (see Remark \ref{rem:div_unred} for an example). 

In Section \ref{sec:position}, we combinatorially detect the oval to which each divisor point $P_V$ belongs to (Theorem \ref{theo:pos_div} and Corollary \ref{cor:eps_tot}) and prove that each finite oval contains exactly one divisor point (Theorem \ref{prop:comb_oval}), {\sl i.e.} $\DKP$ satisfies the reality and regularity conditions established in \cite{DN}. At this aim, we introduce a set of indices
(Definition \ref{def:index_pair}) and characterize their properties (Lemma \ref{lemma:count_eps}). 

In Section \ref{sec:lam1}, we restrict ourselves to the case of reduced networks and we use such indices to formulate the linear relations at internal vertices eliminating the reference to the orientation and the gauge ray direction  (Definition \ref{def:sign_rela}) and we introduce admissible signatures for half edge vectors at vertices (Definition \ref{def:vertex_sign}). Then we reformulate the principal results (Theorems~\ref{theo:consist}, \ref{theo:pos_div} and \ref{prop:comb_oval} and Corollary \ref{cor:DKP}) in invariant form, {\sl i.e.} without explicit reference to the orientation and the gauge ray direction (Theorem \ref{theo:sig_rela}). 
In particular, (\ref{eq:form_3}) and  (\ref{eq:form_4}) make explicit the relation between the geometrical formulation of signatures on $\mathcal G$, which represents the positroid cell $\S$ to which $[A]$ belongs to, and its discrete differential counterpart (Definition \ref{def:sign_rela}), which encodes the possible positions of the divisor. As a consequence, we obtain a direct
relation between the total non--negativity property encoded in the geometrical setting and the reality and regularity condition on the divisor. The choice of a perfect orientation and of a gauge ray direction induces a well defined signature at the vertices of $\mathcal G$ which encordes a finite set of possible positions of real and regular divisors. Among these possible solutions, only one is picked up by fixing the soliton data $(\mathcal K, [A])$ and the normalization time $\vec t_0$.
Vice versa, if we fix the soliton data $(\mathcal K, [A])$ and the normalization time $\vec t_0$, changes of orientation and of gauge ray directions act on signatures at vertices as the addition of a discrete exact differential form which leaves invariant the position of the divisor. 

In Section \ref{sec:moves_reduc} we give the explicit transformation rules of the curve, the edge vectors and the divisor with respect to Postnikov moves and reductions. 

In the last two Sections, we present some examples. In Section \ref{sec:example} we apply our construction to soliton data in ${\mathcal S}_{34}^{\mbox{\tiny TNN}}$, the 3--dimensional positroid cell in $Gr^{\mbox{\tiny TNN}} (2,4)$ corresponding to the matroid
${\mathcal M} = \{ \ 12 \ , \ 13 \ , \ 14 \ ,\ 23 \ , \ 24\ \}$.
We construct both the reducible rational curve and its desingularization to a genus $3$ $\mathtt M$--curve and the KP divisor for generic soliton data $\mathcal K =\{ \kappa_1<\kappa_2<\kappa_3<\kappa_4\}$ and $[A]\in {\mathcal S}_{34}^{\mbox{\tiny TNN}}$. We then apply a parallel edge unreduction and a flip move and compute the divisor on the transformed curve.
We also show the effect of the square move on the divisor for soliton data $(\mathcal K,[A])$ with $[A] \in Gr^{\mbox{\tiny TP}} (2,4)$ in Section \ref{sec:ex_Gr24top}.

\smallskip

\paragraph{\textbf{Remarks and open questions}}
Our construction may be considered as a tropicalization of the spectral problem (smooth $\mathtt M$--curves and divisors) associated to real regular finite--gap KP solutions (potentials) in the rational degeneration of such curves. The tropical limit studied in \cite{KW2} (see also \cite{DMH} for a special case) has a different nature: reconstruct the soliton data from the asymptotic contour plots. In our setting, that would be equivalent to tropicalize the reducible rational spectral problem connecting the asymptotic behavior of the potential (KP solution) to the asymptotic behavior in $\vec t$ of the zero divisor of the KP wave function (see \cite{A2} for some preliminary results concerning soliton data in $Gr^{\mbox{\tiny TP}}(2,4)$). Relations between integrability and cluster algebras were demonstrated in \cite{FG,KG}, and the cluster algebras were essentially motivated by total positivity \cite{FZ1,FZ2}. In \cite{KW2} cluster algebras have appeared in connection with KP solitons asymptotic behavior. We expect that they should also appear in our construction in connection with the tropicalization of the zero divisor.

Let us remark that for a fixed reducible curve the Jacobian may contain more than one connected component associated to real regular solutions. Therefore, in contrast with the smooth case, different connected components may correspond to different Grassmannians. Some of these components may correspond not to full positroid cells, but to special subvarieties. For generic curves the problem of describing these subvarieties is completely open. For a rational degeneration of genus $(n-1)$ hyperelliptic $\mathtt M$-curves this problem was studied in \cite{A1} and it was shown that the corresponding soliton data in $Gr^{\mbox{\tiny TP}}(k,n)$ formed $(n-1)$--dimensional varieties known in literature \cite{BK} to be related to the finite Toda system. The same KP soliton family has been recently re-obtained in \cite{Nak} in the framework of the Sato Grassmannian, whereas the spectral data for the finite Toda was studied earlier in \cite{KV}.

Our construction of the divisor provides a local parametrization of the positroid cell to which the soliton data belong, depending on the normalization time. Indeed, when a pair of divisor points comes simultaneously to a double point, the parametrization becomes singular and requires a resolution of singularities. We plan to discuss how to resolve these singularities in a future work.  

In \cite{AG2} we studied in details the transition from multiline soliton solutions to finite-gap solutions associated to almost degenerate $\mathtt M$-curves in the first non-trivial case, and in \cite{AG3} we provided a generic construction. We expect that the coordinates on the moduli space, compatible with $\mathtt M$-structure, introduced in \cite{Kr5}, may be useful in this study. 

We have noticed an analogy between the momentum--helicity conservation relations in the trivalent planar networks in the approach of \cite{AGP1,AGP2}, and the relations satisfied by the vacuum and dressed edge wave functions in our approach. It is unclear to us whether our approach for KP may be interpreted as a scalar analog of a field theoretic model. In the 
on--shell diagram approach, internal trivalent white and black vertices represent little Grassmannians $Gr(1,3)$ and $Gr(2,3)$, whereas edges correspond to gluings. 
In his mathematical description of the gluing phenomenon,
Lam \cite{Lam2} introduces a relation space analogous to the linear relations satisfied by the edge vectors and the edge wave function in our setting. In his framework, it is essential to choose proper signatures of edges to obtain totally non-negative Grassmannians. In our text we provide rules for signs at edges in terms of local winding and intersection numbers defined using the gauge ray directions. In Section \ref{sec:lam1} we have re-expressed these conditions in invariant form as conservation laws of half-edge quantities, so that our construction gives sufficient conditions for total non-negativity.
It is an open problem whether all admissible signatures at internal vertices 
correspond to a choice of gauge ray direction and vertex gauge freedom, {\sl i.e.} a reference direction with respect to which measure winding of pair of edges and $k$ simple curves from the boundary sources having zero pairwise intersections inside the disk. If the conjecture would be true, the choice of signatures at boundary vertices would single out the total non--negativity property. Another open problem is whether all real and regular divisor positions in the ovals obtained solving equations (\ref{eq:form_3}) and (\ref{eq:form_4}) are realizable as we vary the soliton data in $\S$ and the normalizing time $\vec t_0$. The latter problem is naturally connected to the classification of realizable asymptotic soliton graphs studied in \cite{KW2}.

For non-reduced networks the solution of the linear system may contain null-vectors at internal edges. An open question is whether the null-vectors can be eliminated by using the extra freedom in the assignment of the weights. In Section~\ref{sec:constr_null} we have outlined the modified construction of the curve and the divisor in presence of null-vectors, but it requires a more serious study.

In \cite{Pos2} connections between planar bicolored graphs and more general Grassmannian graphs with geometry of polyhedral subdivisions was discussed, see also \cite{PSW}. An open question is whether and how our construction is connected to this relation.

\smallskip

\paragraph{\textbf{Plan of the paper}:} We did our best to make the paper self--contained. In Section \ref{sec:soliton_theory}, we briefly present some results of KP soliton theory necessary in the rest of the paper. Section \ref{sec:3} contains the main construction and the statements of the principal theorems.
In Section \ref{sec:vectors} we construct and characterize systems of edge vectors on any PBDTP network representing a given point in an irreducible positorid cell $\S$. Sections \ref{sec:anycurve} and \ref{sec:comb} contain the proofs of the main theorems together with the explicit construction of the wave function and of the real regular divisor on the reducible regular curve. In particular, in Section~\ref{sec:lam1} we reformulate our main results in terms of signatures of pair of half edges at vertices.
In Section \ref{sec:moves_reduc} we explain how edge vectors depend on Postnikov moves and reductions and characterize the dependence of the divisors on moves and reductions. Sections \ref{sec:example} and \ref{sec:ex_Gr24top} contain several examples and applications from the previous sections.

{\bf Notations:} We use the following notations throughout the paper:
\begin{enumerate}
\item $k$ and $n$ are positive integers such that $k<n$;
\item  For $s\in {\mathbb N}$ let $[s] =\{ 1,2,\dots, s\}$; if $s,j \in {\mathbb N}$, $s<j$, then
$[s,j] =\{ s, s+1, s+2,\dots, j-1,j\}$;
\item  ${\vec t} = (t_1,t_2,t_3,\dots)$ is the infinite vector of real KP times where $t_1=x$, $t_2=y$, $t_3=t$, and we assume that only a finite number of components are different from zero;
\item We denote $\theta(\zeta,\vec t)= \sum\limits_{s=1}^{\infty} \zeta^s t_s$, due to the previous remark  $\theta(\zeta,\vec t)$ is well-defined for any complex $\zeta$;
\item We denote the real KP phases 
$\kappa_1< \kappa_2 < \cdots < \kappa_n$ and
$\theta_j \equiv \theta (\kappa_j, \vec t)$.
\end{enumerate}

\section{KP multi-line solitons in the Sato Grassmannian and in finite-gap theory}\label{sec:soliton_theory}

Kadomtsev-Petviashvili-II (KP) equation is one of most famous integrable equations, and it is a member of an integrable hierarchy (see \cite{D,DKN,H,MJD,S} for more details).  

The multiline soliton solutions are a special class of solutions to the KP equation \cite{KP}
\begin{equation}\label{eq:KP}
(-4u_t+6uu_x+u_{xxx})_x+3u_{yy}=0,
\end{equation}
and are realized starting from the soliton data $({\mathcal K}, [A])$, where 
${\mathcal K}$ is a set of real ordered phases $\kappa_1<\cdots<\kappa_n$, $[A]$ denotes
a point in the finite dimensional real Grassmannian $Gr (k,n)$ represented by a $k\times n$ real matrix  $A =( A^i_j )$ ($i\in [k], j\in [n]$), of maximal rank $k$.
Following \cite{Mat}, see also \cite{FN}, to such data we associate $k$ linear independent solutions
$f^{(i)}(\vec t) = \sum_{j=1}^n A^i_j e^{\theta_j}$, $i\in [k]$, to the heat hierarchy
$\partial_{t_l} f = \partial_x^l f$, $l=2,3,\dots$.
Then
\begin{equation}\label{eq:KPsol}
u( {\vec t} ) = 2\partial_{x}^2 \log(\tau ( {\vec t}))
\end{equation}
is a multiline soliton solution to (\ref{eq:KP}) with
\[
\tau (\vec t) = Wr_{t_1} (f^{(1)},\dots, f^{(k)})= \sum\limits_{I} \Delta_I (A)\prod_{\mycom{i_1<i_2}{ i_1,i_2 \in I}} (\kappa_{i_2}-\kappa_{i_1} ) e^{ \sum\limits_{i\in I} \theta_i },
\]
where the sum is other all $k$--element ordered subsets $I$ in $[n]$, {\it i.e.} $I=\{ 1\le i_1<i_2 < \cdots < i_k \le n\}$ and $\Delta_I (A)$ are the maximal minors of the matrix $A$ with respect to the columns $I$, {\it i.e.} the Pl\"ucker coordinates for the corresponding point in the finite dimensional Grassmannian $Gr (k,n)$.

$u( {\vec t} ) = 2\partial_{x}^2 \log(\tau ( {\vec t}))$ 
is a real regular multi--line soliton solution to the KP equation (\ref{eq:KP}) bounded for all real $x,y,t$ if and only if $\Delta_I (A) \ge 0$, for all $I$ \cite{KW2}. We remark that the weaker statement that the solution of the KP hierarchy is bounded for all real times if and only if all Pl\"ucker coordinates are non-negative was earlier proven in \cite{Mal}.

Before continuing, let us recall some useful definitions.

\begin{definition}\textbf{Totally non-negative Grassmannian \cite{Pos}.}
Let $Mat^{\mbox{\tiny TNN}}_{k,n}$ denote the set of real $k\times n$ matrices of maximal rank $k$ with non--negative maximal minors $\Delta_I (A)$. Let $GL_k^+$ be the group of $k\times k$ matrices with positive determinants. We define a totally non-negative Grassmannian as 
\[
\GTNN = GL_k^+ \backslash Mat^{\mbox{\tiny TNN}}_{k,n}.
\]
\end{definition}

\begin{remark}
The left multiplication by an element in  $GL_k^+$ does not affect the signs of the minors and values of their ratios. Since left multiplication by $GL_k^+$ matrices preserves the KP real regular multisoliton solution $u({\vec t})$ in (\ref{eq:KPsol}), the soliton data is $[A]$ -- the equivalence class of $A$,  {\sl i.e.} a point in the totally non--negative Grassmannian 
\end{remark}

In the theory of totally non-negative Grassmannians an important role is played by the positroid stratfication. Each cell in this stratification is defined as the intersection of a Gelfand-Serganova stratum \cite{GS,GGMS} with the totally non-negative part of the Grassmannian. More precisely:
\begin{definition}\textbf{Positroid stratification \cite{Pos}.} Let $\mathcal M$ be a matroid i.e. a collection of $k$-element ordered subsets $I$ in $[n]$, satisfying the exchange axiom (see, for example \cite{GS,GGMS}). Then the positroid cell $\S$ is defined as
$$
\S=\{[A]\in \GTNN\ | \ \Delta_{I}(A) >0 \ \mbox{if}\ I\in{\mathcal M} \ \mbox{and} \  \Delta_{I}(A) = 0 \ \mbox{if} \ I\not\in{\mathcal M}  \}
$$
\end{definition}
Combinatorial classification of all non-empty positroid cells and their rational parametrizations were obtained in \cite{Pos}, \cite{Tal2}. We anticipate that in our construction we shall use the classification of positroid cells via directed planar networks in the disk (see Section~\ref{sec:gamma}).

Any given soliton solution is associated to an infinite set of soliton data $({\mathcal K}, [A])$. However there exists an unique \textbf{minimal} pair $(k,n)$ such that the soliton solution can be realized with $n$ phases $\kappa_1<\cdots<\kappa_n$, $[A]\in \GTNN$ but not with $n-1$ phases and $[A^\prime]\in Gr^{\mbox{\tiny TNN}} (k^{\prime},n^{\prime})$ and either $(k^{\prime}, n^{\prime}) =(k, n-1)$ or $(k^{\prime}, n^{\prime}) =(k-1, n-1)$.
In the following, to avoid excessive technicalities we consider only regular and irreducible soliton data.

\begin{definition}\label{def:regsol}{\bf Regular and irreducible soliton data} \cite{CK}
We call $({\mathcal K}, [A])$ regular soliton data if ${\mathcal K} = \{ \kappa_1 < \cdots < \kappa_n \}$ and $[A]\in \GTNN$, that is if the KP soliton solution as in (\ref{eq:KPsol}) is real regular and bounded for all $(x,y,t)\in \mathbb{R}^3$.

Moreover we call the regular soliton data $({\mathcal K}, [A])$ irreducible  if $[A]$ is a point in the irreducible part of the real Grassmannian, {\sl i.e.} if the reduced row echelon matrix $A$ has the following properties:
\begin{enumerate}
\item\label{it:col} Each column of $A$ contains at least a non--zero element;
\item\label{it:row} Each row of $A$ contains at least one nonzero element in addition to the pivot.
\end{enumerate}
If either (\ref{it:col}) or (\ref{it:row}) doesn't occur, we call the soliton data $({\mathcal K}, [A])$ reducible.
\end{definition}

The class of solutions associated to irreducible regular soliton data has remarkable asymptotic properties both in the $(x,y)$ plane at fixed time $t$ and in the tropical limit ($t\to \pm \infty)$, which have been successfully related to the combinatorial classification of the irreducible part $\GTNN$ for generic choice of the phases ${\mathcal K}$ in a series of papers (see \cite{BPPP,CK,DMH,KW1,KW2} and references therein).

According to Sato theory \cite{S}, the wave function associated to regular soliton data $({\mathcal K},[A])$, can be obtained from the dressing (inverse gauge) transformation of the vacuum (zero--potential) eigenfunction $\displaystyle \phi^{(0)} (\zeta, \vec t) =\exp ( \theta(\zeta, {\vec t}))$, which solves
$\partial_x \phi^{(0)} (\zeta, \vec t)=\zeta \phi^{(0)} (\zeta, \vec t)$, 
$\partial_{t_l}\phi^{(0)} (\zeta, \vec t) = \zeta^l \phi^{(0)} (\zeta, \vec t)$, $l\ge 2$,
via the dressing ({\it i.e.} gauge) operator $W = 1 -{\mathfrak w}_1({\vec t})\partial_x^{-1} -\cdots - {\mathfrak w}_k({\vec t})\partial_x^{-k}$,
where ${\mathfrak w}_1({\vec t}),\dots,{\mathfrak w}_k({\vec t})$ are the
solutions to the following linear system of
equations
$\partial_x^k f^{(i)} = {\mathfrak w}_1 \partial_x^{k-1} f^{(i)}+\cdots + {\mathfrak w}_k f^{(i)}$, $i\in [k]$.
Then,
\[
L= W \partial_x W^{-1} = \partial_x + \frac{u(\vec t)}{2}\partial_x^{-1} +\cdots,  \ \
u(\vec t) = 2\partial_x {\mathfrak w}_1 (\vec t), \ \   
\psi^{(0)} (\zeta; \vec t)= W\phi^{(0)} (\zeta; \vec t),
\]
respectively are the KP-Lax operator, the KP--potential (KP solution) and the KP-eigenfunction, {\sl i.e.}
\begin{equation}\label{eq:dress_hier}
L \psi^{(0)} (\zeta; \vec t) =\zeta \psi^{(0)} (\zeta; \vec t), 	\quad\quad
\partial_{t_l}\psi^{(0)} (\zeta; \vec t)= B_l \psi^{(0)} (\zeta; \vec t), \ \ l\ge 2,
\end{equation}
where $B_l = (W \partial_x^l W^{-1} )_+ =(L^l)_+ $ (here and in the following the symbol $(\cdot )_+$ denotes the differential part of the operator).

Let
\begin{equation}\label{eq:D}
{\mathfrak D}^{(k)} = W\partial_x^k = \partial_x^k - {\mathfrak w}_1 (\vec t)\partial_x^{k-1} -\cdots - {\mathfrak w}_k(\vec t)
\end{equation}
The KP-eigenfunction associated to this class of solutions may be equivalently expressed as
\begin{equation}\label{eq:Satowf} 
{\mathfrak D}^{(k)}\phi^{(0)} (\zeta; \vec t)  = W\partial_x^k \phi^{(0)} (\zeta; \vec t)
= \left(\zeta^k -{\mathfrak w}_1 (\vec t)\zeta^{k-1} -\cdots - {\mathfrak w}_k(\vec t)\right)
\phi^{(0)} (\zeta; \vec t) = \zeta^k \psi^{(0)} (\zeta; \vec t).
\end{equation}

\begin{definition}\label{def:Satodiv}{\bf Sato divisor coordinates}
Let the regular soliton data be $({\mathcal K}, [A])$, ${\mathcal K} = \{ \kappa_1 < \cdots < \kappa_n \}$, $[A]\in \GTNN$. We call Sato divisor coordinates at time $\vec t$, the set of roots $\zeta_j (\vec t)$, $j\in [k]$, of the characteristic equation associated to the Dressing transformation
\begin{equation}\label{eq:Dressing_roots}
\zeta_j^k(\vec t) - {\mathfrak w}_1 (\vec t)\zeta_j^{k-1}(\vec t)-\cdots  - {\mathfrak w}_{k-1} (\vec t)\zeta_j(\vec t)- {\mathfrak w}_k(\vec t) = 0. 
\end{equation}
\end{definition}

In \cite{Mal} it is proven the following proposition
\begin{proposition}\textbf{Reality and simplicity of the KP soliton divisor}\label{prop:malanyuk} \cite{Mal}.
Let the regular soliton data be $({\mathcal K}, [A])$, ${\mathcal K} = \{ \kappa_1 < \cdots < \kappa_n \}$, $[A]\in \GTNN$. Then for all $\vec t$, $\zeta_j^k(\vec t)$ are real and satisfy $\zeta_j(\vec t)\in [\kappa_1,\kappa_n]$, $j\in [k]$. 
Moreover for almost every $\vec t$ the roots of (\ref{eq:Dressing_roots}) are simple.
\end{proposition}

The following definition is then fully justified.

\begin{definition}\label{def:Sato_data}\textbf{Sato algebraic--geometric data} Let $({\mathcal K}, [A])$ be given regular soliton data with $[A]$ belonging to a $|D|$ dimensional positroid cell in $\GTNN$.  Let $\Gamma_0$ be a copy of $\mathbb{CP}^1$ with marked points $P_0$, local coordinate $\zeta$ such that $\zeta^{-1} (P_0)=0$ and $\zeta(\kappa_1)<\zeta(\kappa_2)<\cdots<\zeta(\kappa_n)$. Let $\vec t_0$ be real and such that the real roots $\zeta_j (\vec t_0)$ in (\ref{eq:Dressing_roots}) are simple.

Then to the data $({\mathcal K}, [A], \Gamma_0\backslash \{ P_0	\} ,\vec t_0)$ we associate the \textbf{Sato divisor} $\DS=\DS(\vec t_0)$
\begin{equation}\label{eq:Satodiv}
\DS=\{ \gamma_{S,j} \in \Gamma_0 \, : \, \zeta(\gamma_{S,j})= \zeta_j (\vec t_0), \quad j\in [k]\},
\end{equation}
and the normalized \textbf{Sato wave function}
\begin{equation}\label{eq:SatoDN}
{\hat \psi } (P, \vec t) = \frac{{\mathfrak D}\phi^{(0)} (P; \vec t)}{{\mathfrak D}\phi^{(0)} (P; \vec t_0)} = \frac{\psi^{(0)} (P; \vec t)}{\psi^{(0)} (P; \vec t_0)}, \quad\quad \forall P\in \Gamma_0\backslash \{ P_0\},
\end{equation}
with ${\mathfrak D}\phi^{(0)} (\zeta; \vec t)$ as in (\ref{eq:Satowf}). 

By definition $({\hat \psi}_0 (P,\vec t)) + \DS \ge 0$, for all $\vec t$.
\end{definition}

In the following, we use the same symbol for the points in $\Gamma_0$ and their local coordinates to simplify notations. In particular, we use the symbol $\gamma_{S,j}$ both for the Sato divisor points and Sato divisor coordinates.

\begin{remark}\label{rem:fundam}\textbf{Incompleteness of Sato algebraic--geometric data} 
Let $1\le k<n$ and let $\vec t_0$ be fixed. Given the phases $\kappa_1<\cdots <\kappa_n$ and the spectral data $( \Gamma_0\backslash \{ P_0	\} , \DS) $, where $\DS=\DS(\vec t_0) $ is a $k$ point divisor satisfying Proposition \ref{prop:malanyuk}, it is, in general, impossible to identify uniquely the point $[A]\in \GTNN$ corresponding to such spectral data. Indeed, assume that $[A]$ belongs to an irreducible positroid cell of dimension $|D|$. Then the degree of $\DS$ equals $k$, but $\max\{k, n-k\} \le |D| \le k(n-k)$.
\end{remark}

In our construction we propose a completion of the Sato algebraic--geometric data based on singular finite--gap theory on reducible algebraic curves \cite{Kr3, AG1,AG3} and we use the representation of totally non--negative Grassmannians via directed planar networks \cite{Pos} to preserve the reality and regularity of the KP divisor in the solitonic limit.

Indeed, soliton KP solutions can be obtained from regular finite--gap solutions of (\ref{eq:KP}) by proper degenerations of the spectral curve \cite{Kr2}, \cite{DKN}.  The spectral data for KP finite--gap solutions are introduced and described  in \cite{Kr1,Kr2}. The spectral data for this
construction are: a finite genus $g$ compact Riemann surface $\Gamma$ with a marked point $P_0$, a local parameter $1/\zeta$ near $P_0$ and a non-special 
divisor $\mathcal D=\gamma_1+\ldots+\gamma_g$ of degree $g$ in $\Gamma$.

The Baker-Akhiezer function $\hat\psi (P, \vec t)$, $P\in\Gamma$, is defined by the following analytic properties:
\begin{enumerate}
\item For any fixed $\vec t$ the function $\hat\psi (P, \vec t)$ is meromorphic in $P$ on $\Gamma\backslash P_0$;
\item On  $\Gamma\backslash P_0$ the function $\hat\psi (P, \vec t)$ is regular outside the divisor points $\gamma_j$ and has at most first order poles 
at the divisor points. Equivalently, if we consider the line bundle $\mathcal L(\mathcal D)$  associated to $\mathcal D$, then
for each fixed $\vec t$ the function $\hat\psi (P, \vec t)$ is a holomorphic section of $\mathcal L(\mathcal D)$ outside $P_0$.
\item $\hat\psi (P, \vec t)$ has an essential singularity at the point $P_0$ with the following asymptotic:
\[
{\hat \psi} (\zeta, \vec t) = e^{ \zeta x +\zeta^2 y +\zeta^3 t +\cdots} \left( 1 - \chi_1({\vec t})\zeta^{-1} - \cdots
-\chi_k({\vec t})\zeta^{-k}  - \cdots\right). 
\]
\end{enumerate}
For generic data these properties define an unique function, which is a common eigenfunction to all KP hierarchy auxiliary linear operators 
$-\partial_{t_j} + B_j$, where $B_j =(L^j)_+$, and the Lax operator $L=\partial_x+\frac{u(\vec t)}{2}\partial_x^{-1}+ u_2(\vec t)\partial_x^{-2}+\ldots.$
Therefore all these operators commute and the potential $u(\vec t)$ satisfies the KP hierarchy. In particular, the KP equation arises in the 
Zakharov-Shabat commutation representation \cite{ZS} as the compatibility for the second and the third operator:
$[ -\partial_y + B_2, -\partial_t +B_3] =0$, with $B_2 \equiv (L^2)_+ = \partial_x^2 + u$, $B_3 = (L^3)_+ = \partial_x^3 +\frac{3}{4} (u\partial_x +\partial_x u) + \tilde u$
and $\partial_x\tilde u =\frac{3}{4} \partial_y u$.
The Its-Matveev formula represents the KP hierarchy solution $u(\vec t)$ in terms of the Riemann theta-functions associated with $\Gamma$ (see, for example, \cite{Dub}). 

In \cite{DN} there were established the necessary and sufficient conditions on spectral data to generate real regular KP hierarchy solutions for all real $\vec t$, under the assumption that $\Gamma$ is smooth and has genus $g$: 
\begin{enumerate}
\item $\Gamma$ possesses an antiholomorphic involution ${\sigma}:\Gamma\rightarrow\Gamma$, ${\sigma}^2=\mbox{id}$, which has the maximal possible number of fixed components (real ovals), $g+1$, therefore $(\Gamma,\sigma)$ is an $\mathtt M$-curve \cite{Har,Nat,Vi}.
\item $P_0$ lies in one of the ovals, and each other oval contains exactly one divisor point. The oval containing $P_0$ is called ``infinite'' and all 
other ovals are called ``finite''.
\end{enumerate}

The set of real ovals divides $\Gamma$ into two connected components. Each of these components is homeomorphic to a sphere with $g+1$ holes. 
The sufficient condition of the Theorem in \cite{DN} still holds true if the spectral curve $\Gamma$ degenerates in such a way that the divisor remains in the finite ovals 
at a finite distance from the essential singularity \cite{DN}. Of course, this condition is not necessary for degenerate curves. 
Moreover, the algebraic-geometric data for a given soliton data $({\mathcal K},[A])$ are not unique since we can construct 
infinitely many reducible curves generating the same soliton solutions. 

\smallskip

In \cite{AG1}, for any soliton data in $\GTP$ and any fixed value of the  parameter $\xi\gg 1$, we have constructed a curve $\Gamma_{\xi}$, which is the rational degeneration of 
a smooth $\mathtt M$--curve of minimal genus $k(n-k)$ and a degree $k(n-k)$ divisor satisfying the reality conditions of Dubrovin and Natanzon's theorem. 
In \cite{AG3} we have then extended this construction to any soliton data in $\GTNN$ by modeling the spectral curve $\Gamma$ on Postnikov Le--graph so that components, marked points and ovals of the curve correspond to vertices, edges and ovals in the graph. In particular for any given positroid cell, such $\Gamma$ is a rational degeneration of an $\mathtt M$--curve of minimal genus equal to its dimension.

In the following Sections, we generalize the construction in \cite{AG3} to the trivalent plabic (= planar bicolored) networks in the disk in Postnikov equivalence class for $[A]$ and prove the invariance of the KP divisor.
Following \cite{AG3}, we define the desired properties of Baker-Akhiezer functions on reducible curves associated with a given soliton data. 

\begin{definition}
\label{def:rrss}
\textbf{Real regular algebraic-geometric data associated with a given soliton solution.}
Let the soliton data $({\mathcal K},[A])$ be fixed, where ${\mathcal K}$ is a 
collection of real phases $\kappa_1<\kappa_2<\ldots <\kappa_n$, $[A]\in \GTNN$. Let $|D|$ be the dimension of the positroid cell to which $[A]$ belongs. Let $(\Gamma_0, P_0, \DS)$ be the Sato algebraic--geometric data for $({\mathcal K},[A])$ as in Definition \ref{def:Sato_data} for a given $\vec t_0$.

Let $\Gamma$ be a reducible connected curve with a marked point $P_0$, a local parameter $1/\zeta$ near $P_0$ such that 
\begin{enumerate}
\item $\Gamma_0$ is the irreducible component of $\Gamma$ containing $P_0$;
\item $\Gamma$  may be obtained from a rational degeneration of a smooth ${\mathtt M}$-curve of genus $g$, with $g\ge |D|$ and the antiholomorphic involution preserves the maximum number of the ovals in the limit, so that $\Gamma$ possesses $g+1$ real ovals. 
\end{enumerate}

Assume that ${\mathcal D}$ is a degree $g$ non-special divisor on $\Gamma\backslash P_0$, and that $\hat\psi$ is the normalized Baker-Ahkiezer function associated to such data, i.e. for any $\vec t$ its pole divisor is contained in ${\mathcal D}$: $({\hat \psi} (P , \vec t))+{\mathcal D} \ge 0$ on $\Gamma\backslash P_0$, where $(f)$ denotes the divisor of $f$.

We say that  \textbf{the algebraic-geometric data $(\Gamma, P_0,{\mathcal D})$ are associated to the soliton data $({\mathcal K},[A])$,} if the restriction of ${\mathcal D}$ to $\Gamma_0$ coincides with the Sato divisor $\DS$ and
the restriction of $\hat\psi$ to $\Gamma_0$ coincides with the Sato normalized dressed wave function for the soliton data $({\mathcal K},[A])$. 
 
We say that the \textbf{divisor 
${\mathcal D}$ satisfies the reality and regularity conditions} if $P_0$ belongs to one of the fixed ovals and the boundary of each other oval contains exactly one divisor point. 
\end{definition}

\section{Algebraic-geometric approach for irreducible KP soliton data in $\GTNN$}\label{sec:3}

In the following we fix the regular irreducible soliton data $({\mathcal K}, [A])$ and we present a \textbf{direct} construction of algebraic geometric data associated to points in irreducible positroid cells of $\GTNN$.
$\Gamma_0$ is the rational curve associated to Sato dressing and is equipped with a finite number of marked points: the ordered real phases ${\mathcal K} = \{ \kappa_1<\cdots<\kappa_n\}$, the essential singularity 
$P_0$ of the wave function and the Sato divisor $\DS$ as in Definition \ref{def:Sato_data}. The normalized wave function $\hat \psi$ on $\Gamma_0 \backslash \{ P_0\}$ is the normalized Sato wave function (\ref{eq:SatoDN}).
In the present paper, we do the following 

\textbf{Main construction} {\sl Assume we are given a real regular bounded multiline KP soliton solution generated by the following soliton data:
\begin{enumerate}
\item A set of $n$ real ordered phases ${\mathcal K} =\{ \kappa_1<\kappa_2<\dots<\kappa_n\}$;
\item A point $[A]\in \S \subset \GTNN$, where $\S$ is an irreducible positroid cell of dimension $|D|$. 
\end{enumerate}
Let ${\mathcal N}$ be a connected planar bicolored trivalent perfectly orientable network in the disk in Postnikov equivalence class representing $[A]$ and let ${\mathcal G}$ be the graph of ${\mathcal N}$. If the network is reduced, there are no extra conditions, otherwise we assume the data to be generic.

Then, we associate the following algebraic-geometric objects to each triple $({\mathcal K}, [A]; {\mathcal N})$:
\begin{enumerate}
\item A reducible curve $\Gamma=\Gamma(\mathcal G)$ which is the rational degeneration of a smooth $\mathtt M$--curve of genus $g\ge |D|$, where $g+1$ is the number of faces of $\mathcal G$. In our approach, the curve $\Gamma_0$ is one of the irreducible components of $\Gamma$. The marked point $P_0$ belongs to the intersection of $\Gamma_0$ with an oval (infinite oval); 
\item An unique real and regular degree $g$ non--special KP divisor $\DKP({\mathcal K}, [A])$ such that any finite oval contains exactly one divisor point and $\DKP ({\mathcal K}, [A])\cap \Gamma_0$ coincides with Sato divisor;
\item An unique KP wave--function $\hat \psi$ as in Definition \ref{def:rrss} such that
\begin{enumerate}
\item Its restriction to $\Gamma_0\backslash \{P_0\}$ coincides with the normalized Sato wave function (\ref{eq:SatoDN});
\item Its pole divisor has degree $\mathfrak d \le g$ and is contained in $\DKP ({\mathcal K}, [A])$.
\end{enumerate}
\end{enumerate}
}

In particular, if $\mathcal G= \mathcal G_T$ is the trivalent bicolored Le--graph \cite{Pos}, then $\Gamma(\mathcal G_T)$ is the rational degeneration of on $\mathtt M$--curve of minimal genus $|D|$, it has exactly $|D|+1$ ovals, and ${\mathfrak d}=g=|D|$ \cite{AG3}.

\subsection{The reducible rational curve $\Gamma=\Gamma(\mathcal G)$}
\label{sec:gamma}

The construction of $\Gamma(\mathcal G)$ is a straightforward modification of a special case in the classical construction of nodal curves by dual graphs \cite{ACG}.

Following \cite{Pos} we consider the following class of graphs ${\mathcal G}$:
Following \cite{Pos} we consider the following class of graphs ${\mathcal G}$:
\begin{definition}\label{def:graph} \textbf{Planar bicolored directed trivalent perfect graphs in the disk (PBDTP graphs).} A graph ${\mathcal G}$ is called PBDTP if:
\begin{enumerate}
\item  ${\mathcal G}$ is planar, directed and lies inside a disk. Moreover ${\mathcal G}$ is connected, i.e. it does not possess components isolated from the boundary.
\item It has finitely many vertices and edges;
\item It has $n$ boundary vertices on the boundary of the disk labeled $b_1,\cdots,b_n$ clockwise. Each boundary vertex has degree 1. We call a boundary vertex $b_i$ a source (respectively sink) if its edge is outgoing (respectively incoming);
\item The remaining vertices are called internal and are located strictly inside the disk. They are either bivalent or trivalent; 
\item ${\mathcal G}$ is a perfect graph, that is each internal vertex in  ${\mathcal G}$ is incident to exactly one incoming edge or to one outgoing edge;
\item Each vertex is colored black or white. If a trivalent vertex has only one incoming edge, it is colored white, otherwise, it is colored black. Bivalent vertices are assigned either white or black color;
\end{enumerate}
Moreover, to simplify the overall construction we further assume that the boundary vertices $b_j$, $j\in [n]$ lie on a common interval in the boundary of the disk and that each boundary vertex $b_{i}$ is joined by its edge to an internal bivalent white vertex which we denote $V_{i}$, $i\in [n]$. 
\end{definition}

\begin{remark}
The assumption that the boundary vertices $b_j$, $j\in [n]$ lie on a common interval in the boundary of the disk considerably simplifies the use of gauge ray directions to assign winding numbers to walks starting at internal edges and to count the number of boundary source points passed by such walks.  
Instead the requirement that each boundary vertex $b_{i}$ is joined by its edge to an internal bivalent white vertex is completely unnecessary, but useful to prove that the KP divisor does not depend on the orientation of the network used in the construction. 
Indeed the normalized dressed wave function is constant with respect to the spectral parameter on each component of $\Gamma$ corresponding to a bivalent vertex. Therefore, using Move (M3) in \cite{Pos}, we may eliminate all bivalent vertices and the corresponding components in $\Gamma$ without affecting the KP divisor and the properties of the wave function. 
\end{remark}
In Figure \ref{fig:net_curve} [left] we present an example of a PBDTP graph satisfying Definition~\ref{def:graph} and representing a 10-dimensional positroid cell in $Gr^{\mbox{\tiny TNN}}(4,9)$. 

The class of perfect orientations of the PBDTP graph ${\mathcal G}$ are those which are compatible with the coloring of the vertices. The graph is of type $(k,n)$ if it has $n$ boundary vertices and $k$ of them are boundary sources. Any choice of perfect orientation preserves the type of ${\mathcal G}$. To any perfect orientation $\mathcal O$ of ${\mathcal G}$ we assign the base $I_{\mathcal O}\subset [n]$ of the $k$-element source set for $\mathcal O$. Following \cite{Pos} the matroid of ${\mathcal G}$ is the set of $k$-subsets  $I_{\mathcal O}$ for all perfect orientations:
$$
\mathcal M_{\mathcal G}:=\{I_{\mathcal O}|{\mathcal O}\ \mbox{is a perfect orientation of}\ \mathcal G \}
$$
In \cite{Pos} it is proven that $\mathcal M_{\mathcal G}$ is a totally non-negative matroid $\mathcal S^{\mbox{\tiny TNN}}_{\mathcal M_{\mathcal G}}\subset \GTNN$. The following statements are straightforward adaptations of more general statements of \cite{Pos} to the case of  PBDTP graphs:
\begin{theorem}
A PBDTP graph $\mathcal G$ can be transformed into a PBDTP graph $\mathcal G'$ via a finite sequence of Postnikov moves and reductions if and only if $\mathcal M_{\mathcal G}=\mathcal M_{\mathcal G'}$.
\end{theorem}

A graph  $\mathcal G$ is reduced if there is no other graph in its move reduction equivalence class which can be obtained from $\mathcal G$ applying a sequence of transformations containing at least one reduction.

\begin{theorem}
Each Le-graph may be transformed into a PBDTP Le-graph, is reduced, and each positroid cell in the totally non-negative Grassmannian is represented by a Le-graph. 

If $\mathcal G$ is a reduced  PBDTP graph, then the dimension of  $\mathcal S^{\mbox{\tiny TNN}}_{\mathcal M_{\mathcal G}}$ is equal to the number of faces of $\mathcal G$ minus 1.
\end{theorem}

The PBDTP graph in Figure \ref{fig:net_curve} [left] is a PBDTP Le-graph.

\begin{figure}
  \centering
  {\includegraphics[height=0.23\textwidth]{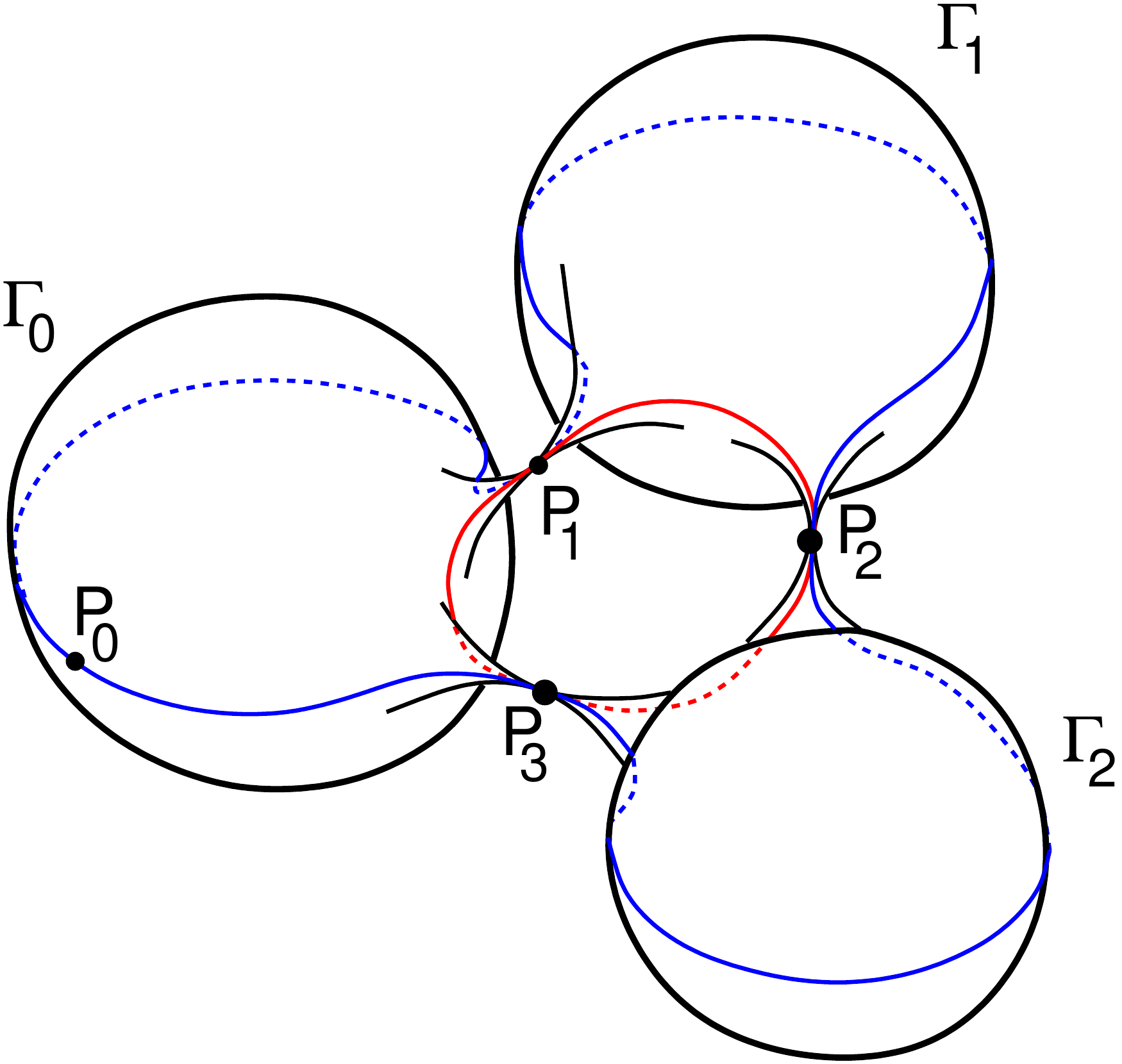}}\hspace{1cm}
  {\includegraphics[height=0.23\textwidth]{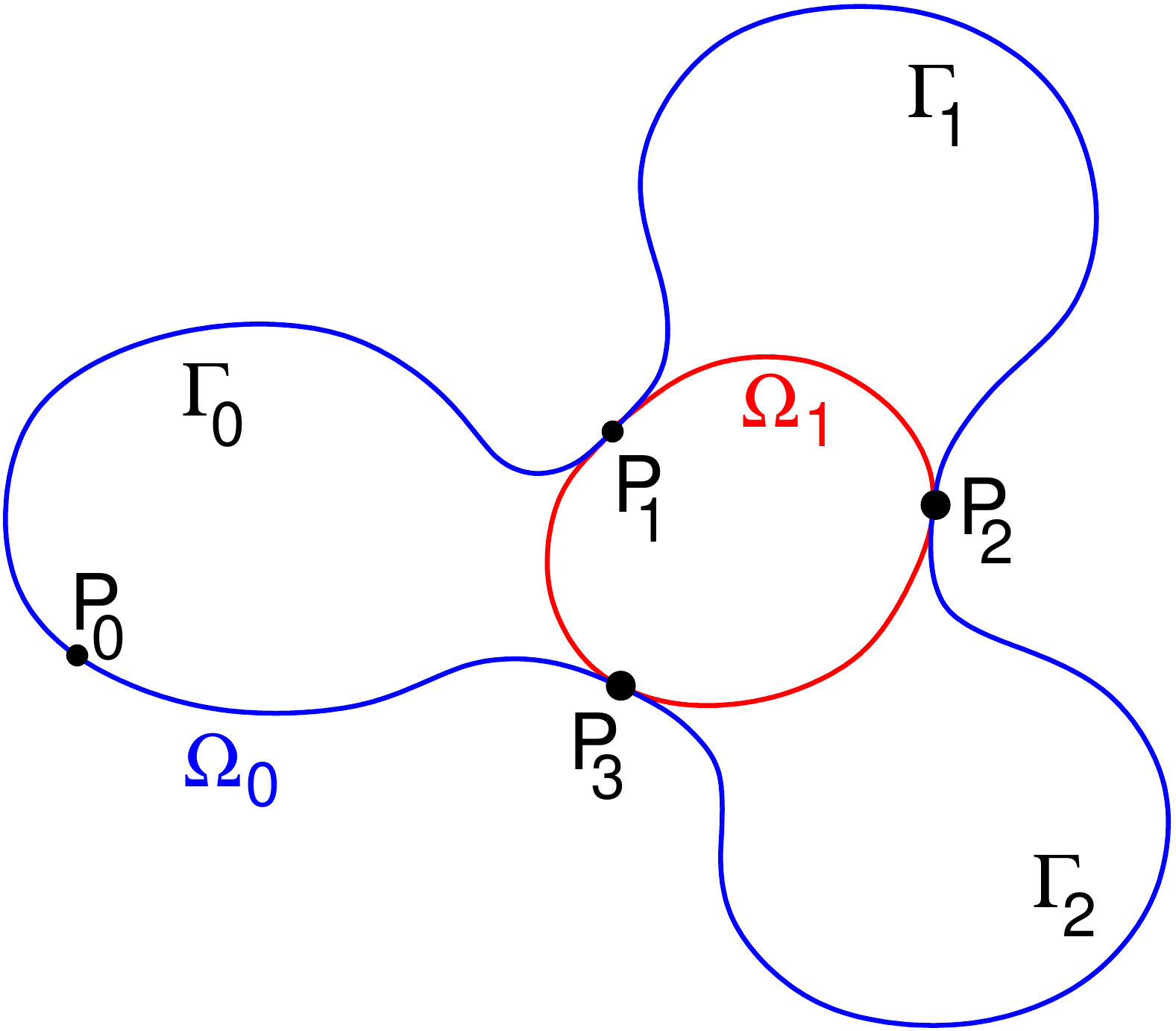}}\hspace{1cm}
  {\includegraphics[height=0.23\textwidth]{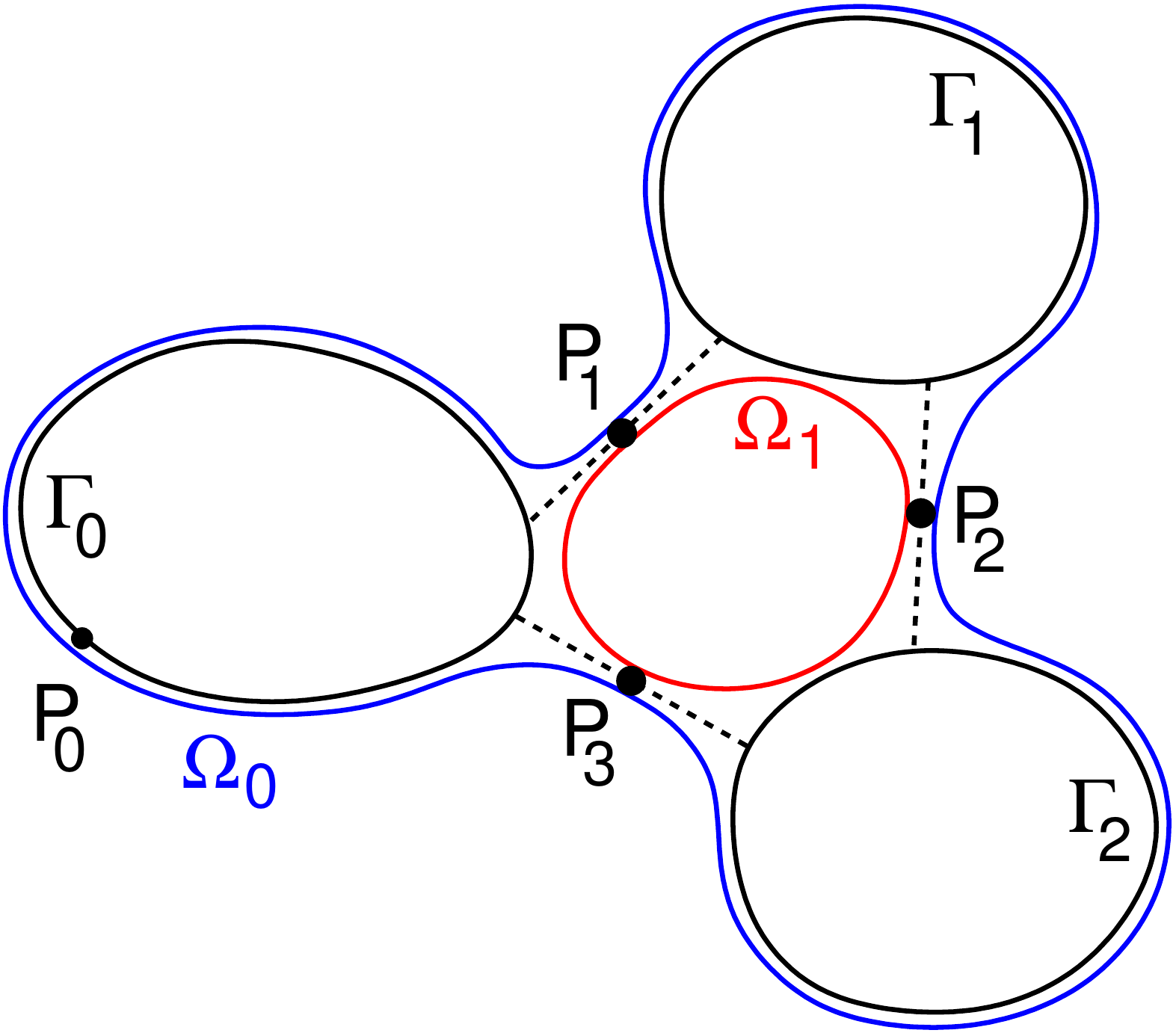}}
  \caption{\label{fig:curve_model}\small{\sl The model of reducible rational curve. On the left we have three Riemann spheres $\Gamma_0$, $\Gamma_1$ and $\Gamma_2$ glued at the points $P_1$, $P_2$, $P_3$. The real part of the curve is drawn in blue and red. In the middle we draw just the real part of the curve and evidence the two real ovals $\Omega_0$ (drawn blue) and $\Omega_1$ (drawn red). On the right we present the representation of the real topological model of the curve used throughout the paper: The real parts of each rational component is a circle and the double points are represented by dashed segments.}}
\end{figure}

\begin{remark}\label{rem:labedges}\textbf{Labeling of edges at vertices}
Let ${\mathcal G}$ be a PBDTP graph. We number the edges at an internal $V$ anticlockwise in increasing order with the following rule: the unique edge starting at a black vertex is numbered 1 and the unique edge ending at a white vertex is numbered 3 (see also Figure \ref{fig:markedpoints}). 
\end{remark}

We construct the curve $\Gamma=\Gamma({\mathcal G})$ gluing a finite number of copies of $\mathbb{CP}^1$, each corresponding to an internal vertex in ${\mathcal G}$, and one copy of $\mathbb{CP}^1=\Gamma_0$, corresponding to the boundary of the disk at pairs of points corresponding to the edges of ${\mathcal G}$. 
On each component, we fix a local affine coordinate $\zeta$ (see Definition~\ref{def:loccoor}) so that  the coordinates at each pair of glued points are real. The points with real $\zeta$ form the real part of the given component. We represent the real part of $\Gamma$ as the union of the ovals (circles) corresponding to the faces of ${\mathcal G}$.
For the case in which $\mathcal G$ is the Le--network see \cite{AG3}.

We use the following representation for real rational curves (see Figure~\ref{fig:curve_model} and \cite{AG1, AG3}): we only draw the real part of the curve, i.e. we replace each  $\Gamma_j=\mathbb{CP}^1$ by a circle. Then we schematically represent the real part of the curve by drawing these circles separately and connecting the glued points by dashed segments. The planarity of the graph implies that $\Gamma$ is a reducible rational $\mathtt M$--curve. 

\begin{figure}
  \centering
  \includegraphics[width=0.7\textwidth]{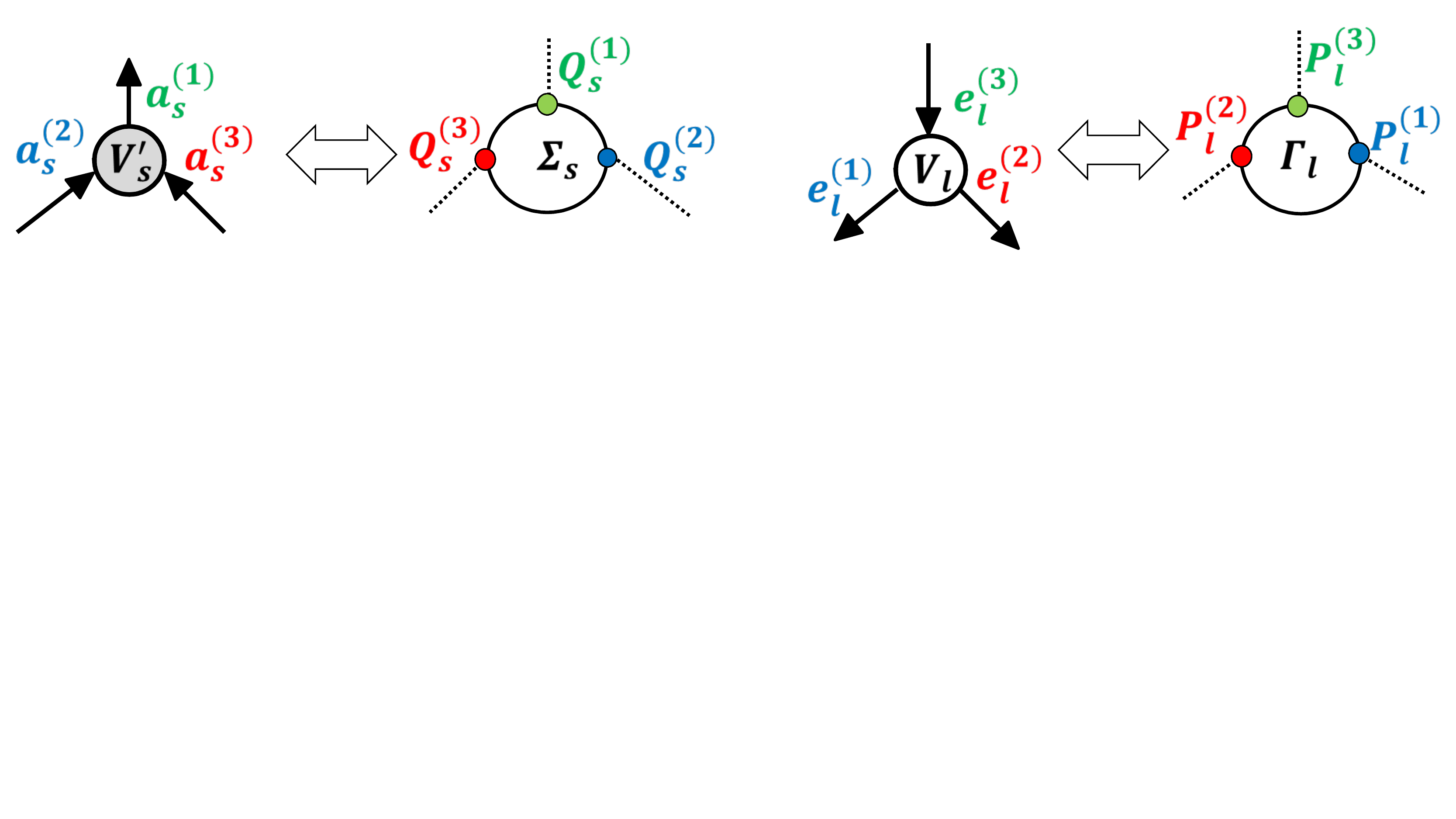}
  \vspace{-4 truecm}
  \caption{\small{\sl The rule for the marked points on the copies $\Sigma_{j}$ and $\Gamma_{l}$ corresponding to the edges of trivalent black and white vertices under the assumption that the boundary vertices lie on a horizontal line.}}
	\label{fig:markedpoints}
\end{figure}

\begin{construction}\label{def:gamma}\textbf{The curve $\Gamma(\mathcal G)$.}
Let ${\mathcal K} =\{\kappa_1 < \cdots < \kappa_n\}$ and let ${\mathcal S}^{\mbox {\tiny TNN}}_{{\mathcal M}}\subset \GTNN$ be a fixed irreducible positroid cell of dimension $|D|$. Let ${\mathcal G}$ be a PBDTP graph representing $\S$ with $g+1$ faces, $g\ge|D|$. Then the curve $\Gamma = \Gamma ({\mathcal G})$ is associated to ${\mathcal G}$ using the correspondence in Table \ref{table:LeG}, after reflecting the graph w.r.t. a line orthogonal to the one containing the boundary vertices (we reflect the graph to have the natural increasing order of the marked points $\kappa_j$ on $\Gamma_0\subset\Gamma(\mathcal G)$).
\begin{table}
\caption{The graph ${\mathcal G}$  vs the reducible rational curve $\Gamma$} 
\centering
\begin{tabular}{|c|c|}
\hline\hline
$\mathcal G$ & $\Gamma$ \\[0.5ex]
\hline
Boundary of disk & Copy of $\mathbb{CP}^1$ denoted $\Gamma_0$ \\
Boundary vertex $b_l$             & Marked point $\kappa_l$ on  $\Gamma_0$\\
Black vertex   $V^{\prime}_{s}$   & Copy of $\mathbb{CP}^1$ denoted $\Sigma_{s}$\\
White vertex   $V_{l}$            & Copy of $\mathbb{CP}^1$ denoted $\Gamma_{l}$\\
Internal Edge                     & Double point\\
Face                              & Oval\\ [1ex]
\hline
\end{tabular}
\label{table:LeG}
\end{table}
More precisely:
\begin{enumerate}
\item We denote $\Gamma_0$ the copy of $\mathbb{CP}^1$ corresponding to the boundary of the disk and mark on it the points $\kappa_1<\cdots <\kappa_n$ corresponding to the boundary vertices $b_1,\dots, b_n$ on $\mathcal G$. We assume that $P_0=\infty$; 
\item A copy of $\mathbb{CP}^1$ corresponds to any internal vertex of $\mathcal G$. We use the symbol $\Gamma_l$ (respectively $\Sigma_s$) for the copy of $\mathbb{CP}^1$ corresponding to the white vertex $V_l$ (respectively the black vertex $V^{\prime}_s$);
\item On each copy of $\mathbb{CP}^1$ corresponding to an internal vertex $V$, we mark as many points as edges at $V$. In Remark~\ref{rem:labedges} we number the edges at $V$ anticlockwise in increasing order, so that, on the corresponding copy of $\mathbb{CP}^1$, the marked points are numbered clockwise because of the mirror rule (see Figure \ref{fig:markedpoints});
\item On each copy $\Gamma_{i}$ corresponding to a bivalent white vertex $V_{i}$ joined to the boundary vertex $b_{i}$, $i\in [n]$, we mark a third point  $P^{(3)}_i$ (Darboux point); 
\item Gluing rules between copies of $\mathbb{CP}^1$ are ruled by edges: we glue copies of $\mathbb{CP}^1$ in pairs at the marked points corresponding to the end points of the edge;
\item The faces of $\mathcal G$ correspond to the ovals of $\Gamma$.  
\end{enumerate}
\end{construction}

In Figure \ref{fig:net_curve} we present an example of curve corresponding to a network representing an irreducible positroid cell
in $Gr^{\mbox{\tiny TNN}} (4,9)$. Other examples are studied  in Sections \ref{sec:example} and \ref{sec:ex_Gr24top}.

\begin{figure}
  \centering
  {\includegraphics[width=0.47\textwidth]{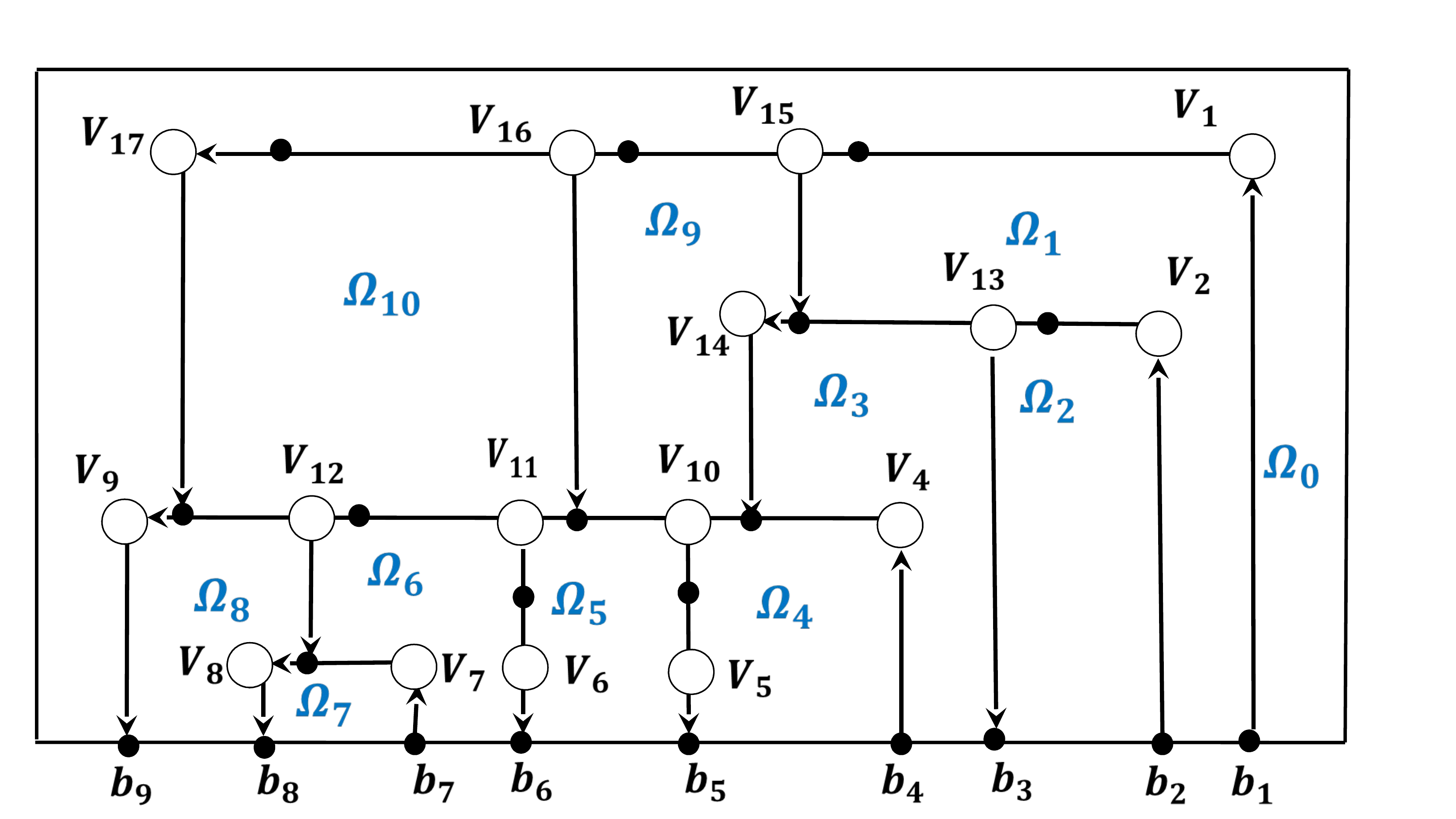}}
  \hfill
	{\includegraphics[width=0.47\textwidth]{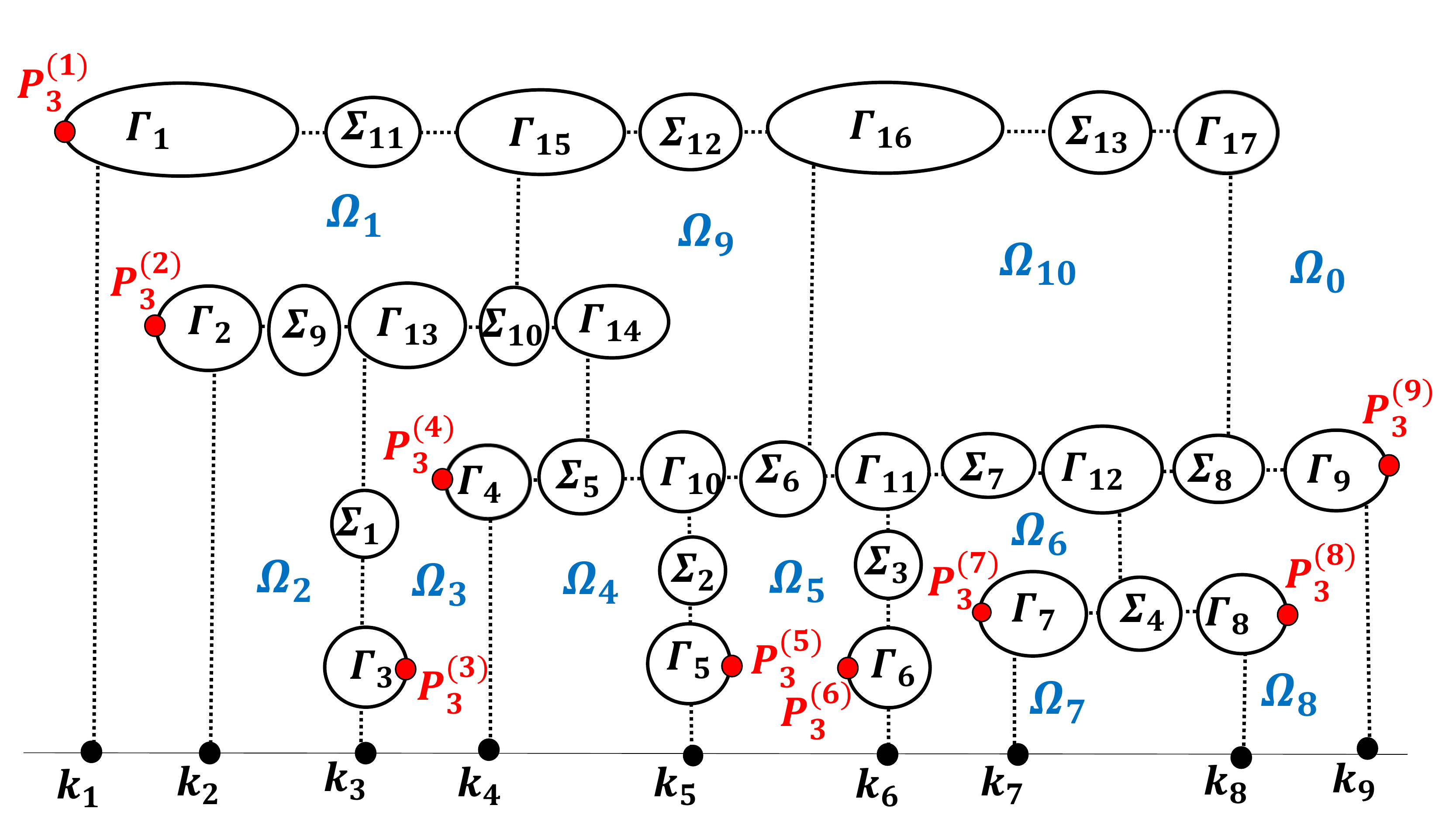}}
  \caption{\small{\sl The correspondence between the graph $\mathcal G$ and the curve $\Gamma (\mathcal G)$ for an irreducible positroid cell in $Gr^{\mbox{\tiny TNN}} (4,9)$. We mark the Darboux points $P^{(3)}_j$, $j\in [9]$.}}
	\label{fig:net_curve}
\end{figure}

\begin{remark} \textbf{Universality of the reducible rational curve $\Gamma(\mathcal G)$.}
If $\mathcal G$ is a trivalent graph representing $\mathcal S$, the construction of $\Gamma=\Gamma(\mathcal G)$ does \textbf{not} require the introduction of parameters. Therefore it provides  an \textbf{universal} curve 
$\Gamma=\Gamma({\mathcal S};\mathcal G)$ for the whole positroid cell ${\mathcal S}$. In Section \ref{sec:anycurve}, we show that the points of ${\mathcal S}$ are parametrized (in the sense of birational equivalence) by the divisor positions at the finite ovals. 

The number of copies of $\mathbb{CP}^1$ used to construct $\Gamma(\mathcal G)$ is \textbf{excessive} in the sense that the number of ovals and the KP divisor is invariant if we eliminate 
all copies of $\mathbb{CP}^1$ corresponding to bivalent vertices (see Section \ref{sec:moves_reduc} and \cite{AG3}). 
\end{remark}

\begin{remark}\label{rem:gauge_vertices}\textbf{Vertex gauge freedom of the graph} The curve $\Gamma$ is the same if we move vertices in $\mathcal G$ without changing their relative positions in the graph. Such transformation acts on edges via rotations, translations and contractions/dilations of their lenghts. We prove that such transformations do not affect the KP divisor.
\end{remark}

The curve $\Gamma(\mathcal G)$ is a partial normalization \cite{ACG} of a connected reducible nodal plane curve with $g+1$ ovals 
and is a rational degeneration of a genus $g$ smooth $\mathtt M$--curve. 

\begin{proposition}\textbf{$\Gamma(\mathcal G)$  is the rational degeneration of a smooth $\mathtt M$-curve of genus $g$.}
Let ${\mathcal K} = \{ \kappa_1 < \cdots < \kappa_n\}$ and ${\mathcal S}^{\mbox {\tiny TNN}}_{{\mathcal M}}$ be an irreducible positroid cell in $\GTNN$ corresponding to the matroid ${\mathcal M}$. Let $\Gamma=\Gamma(\mathcal G)$ be as in Construction \ref{def:gamma}. Then
\begin{enumerate}
\item $\Gamma$ possesses $g+1$ ovals which we label $\Omega_s$, $s\in [0,g]$; 
\item $\Gamma$ is the rational degeneration of a regular $\mathtt M$--curve of genus $g$.
\end{enumerate}
\end{proposition}

\begin{proof}
The proof follows along similar lines as in \cite{AG3}, where we prove the analogous statement in the case of the Le--graph.
Let $t_W, t_B, d_W$ and $d_B$ respectively be the number of trivalent white, trivalent black, bivalent white and bivalent black internal vertices of ${\mathcal G}$. Let $n_I$ be the number of internal edges ({\sl i.e.} edges not connected to a boundary vertex) of ${\mathcal G}$. By Euler formula we have $g = n_I +n -(t_W + t_B+d_W+d_B)$.
The number of edges and the number of vertices, moreover satisfies the following identities
$3(t_W+t_B)+2(d_W+d_B) = 2n_I +n$, $2t_B+ t_W+d_W+d_B =n_I+k$.
Therefore 
\begin{equation}\label{eq:vertex_type}
t_W = g-k, \qquad t_B = g-n+k, \qquad d_W+d_B= n_I+ 2n - 3g.
\end{equation}
By definition, $\Gamma$ is represented by $d= 1+t_W + t_B+d_W+d_B =n_I+n-g+1$ complex projective lines which intersect generically in $d(d-1)/2$ double points. The regular curve is obtained keeping  the number of ovals fixed while perturbing the double points corresponding to the edges in $\mathcal G$ creating regular gaps (see \cite{AG3} for explicit formulas for the perturbations). The total number of desingularized double points, $N_d$ equals the total number of edges in $\mathcal G$: $N_d =n_I+n$. 
Then the genus of the smooth curve is given by the following formula $N_d - d+1 = g $.
\end{proof}

The regular ${\mathtt M}$-curve is the normalization of the nodal plane curve whose explicit construction follows along similar lines as in \cite{AG3}, see also \cite{Kr4}. In Figure \ref{fig:mcurveex} we show such continuous deformation for the example in Figure \ref{fig:net_curve} after the elimination of all copies of $\mathbb{CP}^1$ corresponding to bivalent vertices in the network. In this case $d=14$ ($\Gamma$ is representable by the union of two quadrics and 10 lines). Under genericity assumptions, the reducible rational $\Gamma$ has $\Delta=89$ singular points before partial normalization. The perturbed regular curve is obtained by perturbing $N_d=21$ ordinary intersection points (for each of them $\delta=1$), corresponding to the double points in the topological representation in Figure~\ref{fig:mcurveex}[left]. $\Delta-N_d = 68$ points remain intersections after this deformation and are resolved during normalization. Finally, the normalized perturbed regular curve has then genus $g=\frac{(d-1)(d-2)}{2}-\Delta+N_d=10$.

\begin{figure}
  \centering
  {\includegraphics[width=0.53\textwidth]{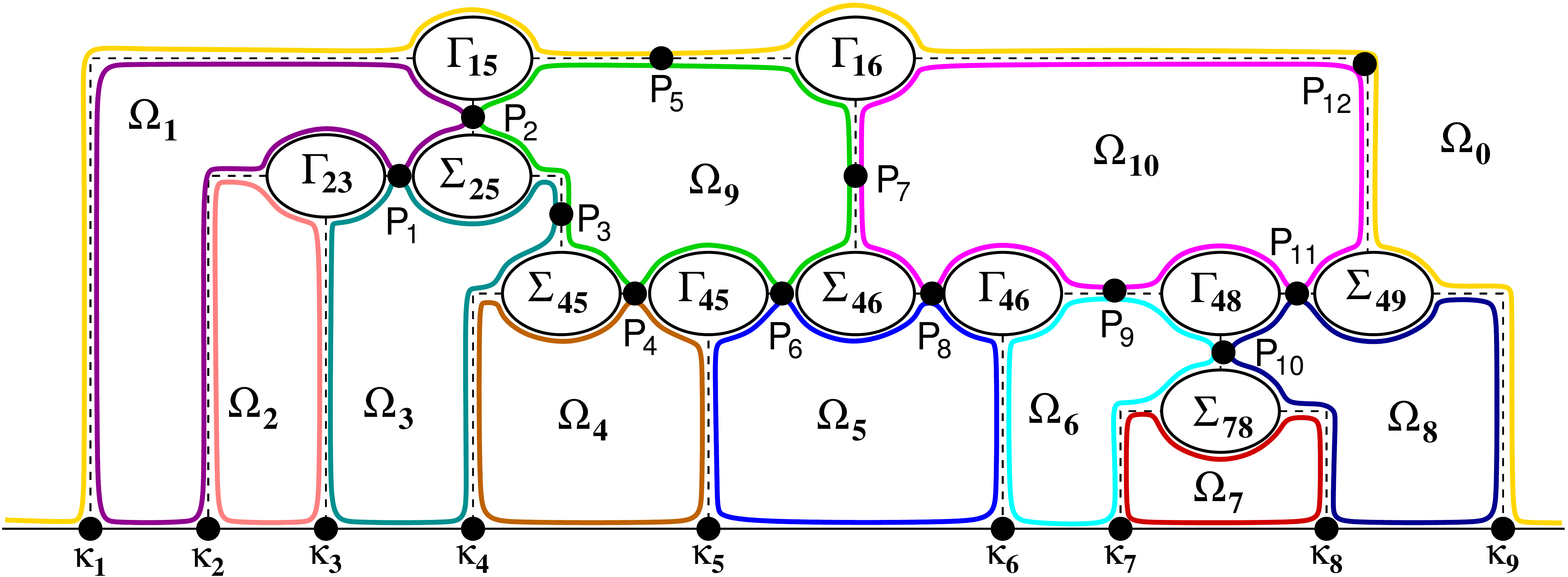}}
  \hfill
  {\includegraphics[width=0.46\textwidth]{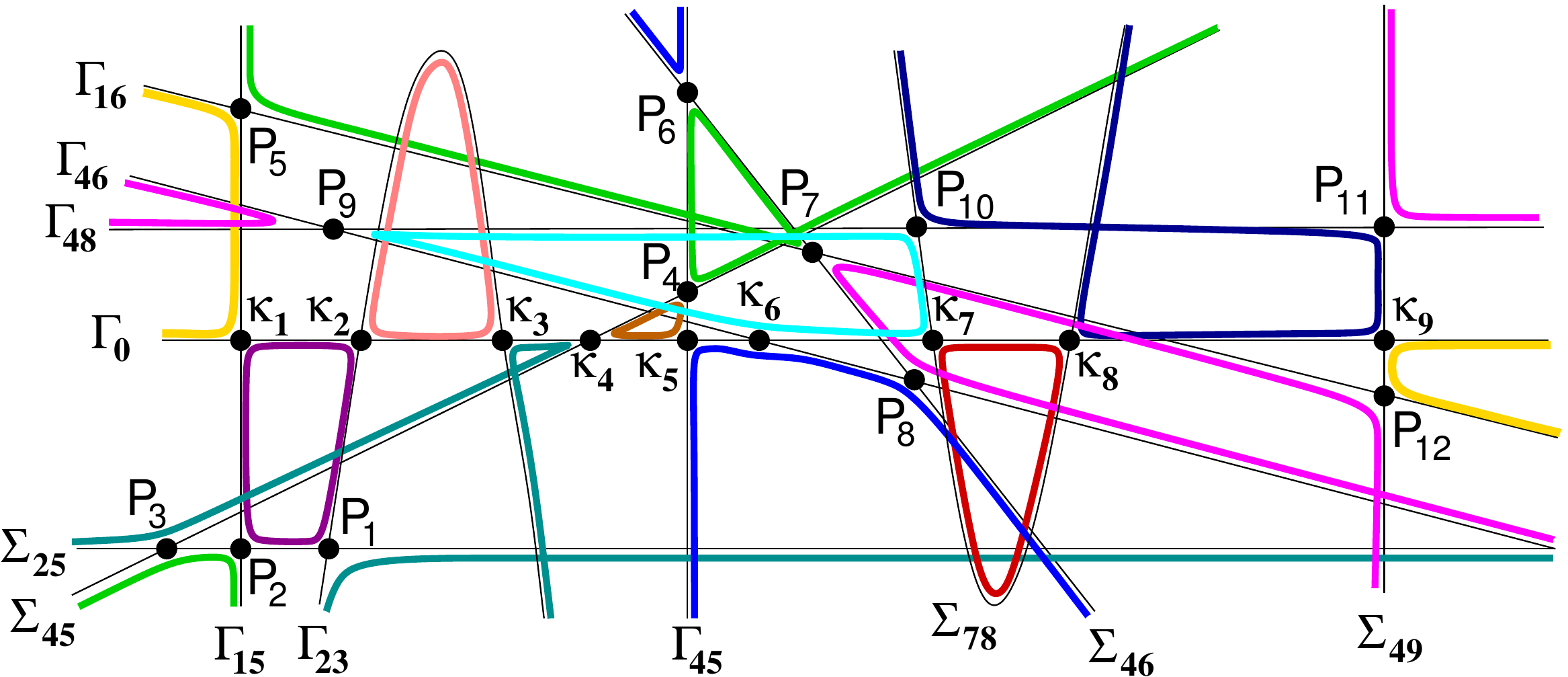}}
  \caption{\small{\sl The topological model of the oval structure of $\Gamma$ [left] is the partial normalization of the reducible plane 
nodal curve [right].  The nodal points, $P_s$, $s\in [12]$, and $\kappa_j$, $j\in [9]$, surviving the partial normalization are marked on the nodal curve [right] and are represented by dashed segments as usual on topological model [left].}\label{fig:mcurveex} }
\end{figure}

\subsection{The KP divisor on $\Gamma(\mathcal G)$ for the soliton data $({\mathcal K}, [A])$}
\label{sec:KPdiv}

Throughout this Section we fix both the soliton stratum $({\mathcal K}, \S)$, with ${\mathcal K} =\{ \kappa_1 < \cdots <\kappa_n\}$ and $\S \subset \GTNN$ an irreducible positroid cell of dimension $|D|$, and the PBDTP graph ${\mathcal G}$ in the disk representing $\S$ in Postnikov equivalence class. Let $g+1$ be the number of faces of $\mathcal G$ with $g\ge |D|$, and let $\Gamma=\Gamma(\mathcal G)$ be the curve corresponding to $\mathcal G$ as in Construction \ref{def:gamma}. 
In this Section we state the main results of our paper: 
\begin{enumerate}
\item \textbf{Construction of the KP divisor $\DKP$ on a given curve $\Gamma$ for soliton data $(\mathcal K, [A])$:}
Given the data $({\mathcal K}, \S;{\mathcal G})$
\begin{enumerate}
\item If ${\mathcal G}$ is reduced i.e. move-equivalent to the Le-graph representing $\S$, then for every $[A] \in \S$ we prove that the curve $\Gamma=\Gamma(\mathcal G)$ can be used as the spectral curve for the soliton data $(\mathcal K, [A])$ by constructing an unique degree $g$ effective real and regular KP divisor $\DKP$ and an unique real and regular KP wave function $\hat \psi(P,\vec t)$ on $\Gamma$.
\item In Section \ref{sec:vectors} we associate to any network an unique system of edge vectors satisfying appropriate boundary conditions. Then, if ${\mathcal G}$ is reducible via Postnikov moves and reductions, and there exists a network $\mathcal N$ of graph ${\mathcal G}$ representing $[A] \in \S$ such that all edge vectors on $\mathcal N$ are not-null, again the curve $\Gamma=\Gamma(\mathcal G)$ can be used as the spectral curve for the soliton data $(\mathcal K, [A])$ by constructing an unique degree $g$ effective real and regular KP divisor $\DKP$ and an unique real and regular KP wave function $\hat \psi(P,\vec t)$ on $\Gamma$.
\item If the network $\mathcal N$ of reducible graph ${\mathcal G}$ representing $[A] \in \S$ possesses null-vectors, we  distinguish two cases in Definition \ref{def:null_type_1_2}. Then, if the edges carrying null vectors are all of type 1, the above construction is still valid with minor modifications, otherwise it is necessary to modify conveniently both the graph and the curve in such a way that $\hat \psi$ be meromorphic on it. We present an example for this case in Section \ref{sec:constr_null} and plan to discuss thoroughly the problem of networks admitting null vectors in a future publication.
\end{enumerate}
\item \textbf{Invariance of $\DKP$} The construction of $\DKP$ is carried out using a 
directed network ${\mathcal N}$ representing $[A]$ of graph ${\mathcal G}$, fixing a gauge ray direction and marking Darboux points on $\Gamma$. We prove that $\DKP$ does not depend neither on the weight gauge, the vertex gauge, on the gauge ray direction and the orientation of the network nor on the position of the Darboux points. In particular, if $\mathcal G$ is a reduced graph move--equivalent to the Le--graph, we get a parametrization of $\S$ via KP divisors on $\Gamma(\mathcal G)$.
\item \textbf{Transformation laws between curves and divisors induced by Postnikov moves and reductions on networks:} Postnikov \cite{Pos} introduces moves and reductions which transform networks preserving the boundary measurement, thus classifying networks representing a given point $[A] \in \S$. In our construction, any such transformation induces a change of the spectral curve and of the KP divisor. In Section \ref{sec:moves_reduc}, we provide the explicit transformation of the divisor under the action of such moves and reductions.
\end{enumerate} 

For any given $[A]\in \S$, $\mathcal N$ denotes a network representing $[A]$ with graph $\mathcal G$ and edge weights $w_e$ \cite{Pos}. There is a fundamental difference in the gauge freedom of assigning weights depending whether or not the graph $\mathcal G$ is reduced.

\begin{remark}\label{rem:gauge_weight}\textbf{The weight gauge freedom \cite{Pos}.} Given a point $[A]\in\S$
and a planar directed graph ${\mathcal G}$ in the disk representing $\S$, then $[A]$
 is represented by infinitely many gauge equivalent systems of weights $w_e$ on the edges $e$ of ${\mathcal G}$. Indeed, if a positive number $t_V$ is assigned to each internal vertex $V$, whereas $t_{b_i}=1$ for each boundary vertex $b_i$, then the transformation on each directed edge $e=(U,V)$
\begin{equation}
\label{eq:gauge}
w_e\rightarrow w_e t_U \left(t_V\right)^{-1},
\end{equation}
transforms the given directed network into an equivalent one representing $[A]$. Using the weight gauge freedom, in the following we assume that edges at boundary vertices carry unit weights.
\end{remark}

\begin{remark}\label{rem:gauge_freedom}\textbf{The unreduced graph gauge freedom.} As it was pointed out in \cite{Pos}, in case of unreduced graphs there is no one-to-one correspondence between the orbits of the gauge weight action (\ref{eq:gauge}) for a fixed directed graph and the points in the corresponding positroid cell. Since we do not consider graphs with components isolated from the boundary, this extra gauge freedom arises if we apply the creation of parallel edges and leafs (see Section~\ref{sec:moves_reduc}). In contrast with gauge transformations of the weights (\ref{eq:gauge}) the unreduced graph gauge freedom affects the KP divisor.
\end{remark}

Throughout the paper, we assign affine coordinates to each component of $\Gamma({\mathcal G})$ using the orientation of the graph $\mathcal G$, and we use the same symbol $\zeta$  for any such affine coordinate to simplify notations.

\begin{figure}
  \centering
  {\includegraphics[width=0.60\textwidth]{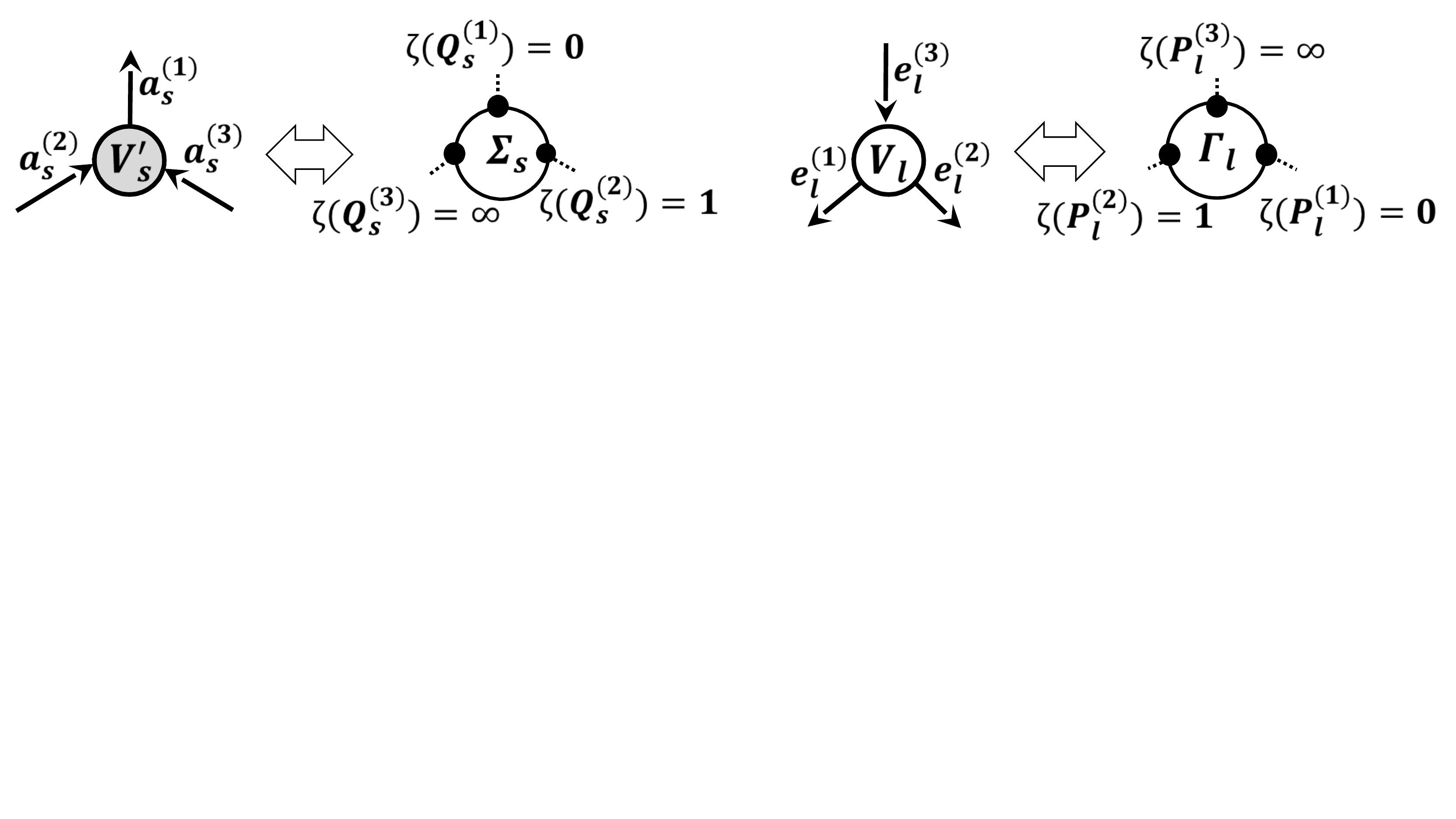}}
	\vspace{-3.8 truecm}
  \caption{\small{\sl Local coordinates on the components $\Gamma_{l}$ and $\Sigma_{s}$ (reflection w.r.t. the vertical line): the marked point $P^{(m)}_{l}\in \Gamma_{l}$ corresponds to the edge $e^{(m)}_{l}$ at the white vertex $V_{l}$ and the marked point $Q^{(m)}_{s}\in \Sigma_{r}$ corresponds to the edge $a^{(m)}_{s}$ at the black vertex $V_{s}^{\prime}$.}}
	\label{fig:lcoord}
\end{figure}

\begin{definition}\label{def:loccoor}{\bf Affine coordinates on $\Gamma(\mathcal G)$}
On each copy of $\mathbb{ CP}^{1}$ the local coordinate $\zeta$ is uniquely identified by the following properties: 
\begin{enumerate}
\item On $\Gamma_0$, $\zeta^{-1} (P_0)=0$ and $\zeta(\kappa_1)< \cdots < \zeta(\kappa_n)$. To abridge notations, we identify the $\zeta$--coordinate with the marked points $\kappa_j=\zeta(\kappa_j)$, $j\in [n]$;
\item On the component $\Gamma_{l}$ corresponding to the internal white vertex $V_{l}$:
\[
\zeta (P^{(1)}_{l}) =0, \;\;\zeta (P^{(2)}_{l}) =1, \;\;\zeta (P^{(3)}_{l}) =\infty,
\]
while on the component $\Sigma_{s}$ corresponding to the internal black vertex $V^{\prime}_{s}$:
\[
\zeta (Q^{(1)}_{s}) =0, \;\;\zeta (Q^{(2)}_s) =1, \;\;\zeta (Q^{(3)}_{s}) =\infty.
\] 
\end{enumerate}
\end{definition}

We illustrate Definition \ref{def:loccoor} in Figure~\ref{fig:lcoord} (see also Figure~\ref{fig:markedpoints}).

In view of Definition \ref{def:rrss}, the desired properties of the KP divisor and of the KP wave function on $\Gamma(\mathcal G)$ for given soliton data $(\mathcal K, [A])$ are the following.

\begin{definition}\label{def:real_KP_div}\textbf{Real regular KP divisor compatible with $[A] \in {\mathcal S}_{\mathcal M}^{\mbox{\tiny TNN}}$.}  
Let $\Omega_0$ be the infinite oval containing the marked point $P_0\in\Gamma_0$ and let $\Omega_s$, $s\in[g]$ be the finite ovals of $\Gamma$. Let $\DS=\DS (\vec t_0)$ be the Sato divisor for the soliton data $(\mathcal K, [A])$.
We call a degree $g$ divisor $\DKP\in\Gamma\backslash \{ P_0\}$ a real and 
regular KP divisor compatible with $(\mathcal K, [A])$ if:
\begin{enumerate}
\item $\DKP\cap \Gamma_0 = \DS$; 
\item There is exactly one divisor point on each component of $\Gamma$ corresponding to a trivalent white vertex; 
\item\label{item:defodd_KP} In any finite oval $\Omega_s$, $s\in [g]$, there is exactly one divisor point;
\item\label{item:defeven_KP} In the infinite oval $\Omega_0$, there is no divisor point.
\end{enumerate}
\end{definition}

\begin{definition}\label{def:KPwave}{\bf A real regular KP wave function on $\Gamma$ corresponding to $\DKP$:}  
Let $\DKP$ be a degree $g$ real regular divisor on $\Gamma$ satisfying Definition~\ref{def:real_KP_div}.
A function ${\hat \psi} (P, \vec t)$, where $P\in\Gamma\backslash \{ P_0\}$ and $\vec t$ are the KP times, is called a real and regular KP wave function on $\Gamma$ corresponding to $\DKP$ if:
\begin{enumerate}
\item $\hat \psi (P, \vec t_0)=1$ at all points $P\in \Gamma\backslash \{P_0\} $;
\item The restriction of $\hat \psi$ to $\Gamma_0\backslash \{P_0\}$ coincides with the normalized Sato wave function defined in (\ref{eq:SatoDN}): ${\hat \psi } (P, \vec t)  = \frac{\psi^{(0)} (P; \vec t)}{\psi^{(0)} (P; \vec t_0)}$; 
\item For all $P\in\Gamma\backslash\Gamma_0$ the function ${\hat \psi} (P, \vec t)$ satisfies all equations (\ref{eq:dress_hier}) of the dressed hierarchy;
\item If both $\vec t$ and $\zeta(P)$ are real, then ${\hat \psi} (\zeta(P), \vec t)$ is real. Here $\zeta(P)$ is the local affine coordinate on the corresponding component of $\Gamma$ as in Definition~\ref{def:loccoor};
\item\label{it:comp} $\hat \psi$ takes equal values at pairs of glued points $P,Q\in \Gamma$, for all $\vec t$:  $\hat \psi(P, \vec t) = \hat \psi(Q, \vec t)$;
\item For each fixed $\vec t$ the function $\hat \psi(P, \vec t)$ is meromorphic of degree $\le g$ in $P$ on $\Gamma\backslash \{ P_0\}$: for any fixed $\vec t$ we have $({\hat \psi} (P, \vec t))+\DKP\ge 0$ on $\Gamma\backslash P_0$, where $(f)$ denotes the divisor 
of $f$. Equivalently, for any fixed $\vec t$ on $\Gamma\backslash \{ P_0\}$ the function $\hat \psi(\zeta, \vec t)$ is regular outside the points of 
$\DKP$ and at each of these points either it has a first order pole or is regular;
\item For each $P\in\Gamma\backslash \{ P_0\}$ outside  $\DKP$  the function $\hat \psi(P, \vec t)$ is regular in $\vec t$ for all times.
\end{enumerate}
\end{definition}

\begin{theorem}\label{theo:exist}\textbf{Existence and uniqueness of a real and regular KP divisor and KP wave function 
on $\Gamma$.}
Let the phases ${\mathcal K}$, the irreducible positroid cell $\S\subset \GTNN$, the PBDTP graph $\mathcal G$ representing $\S$ be fixed. Let $\Gamma=\Gamma(\mathcal G)$ with marked point $P_0\in \Gamma_0$ be as in Construction \ref{def:gamma}. 

If $\mathcal G$ is reduced and equivalent to the Le--network via a finite sequence of moves of type (M1), (M3) and flip moves (M2), there are no extra conditions and, for any $[A]\in \S$, let $\mathcal N$ of graph $\mathcal G$ be a network representing $[A]$. 
Otherwise if $\mathcal G$ is reducible, let $[A]\in \S$ such that there exists a network $\mathcal N$ of graph $\mathcal G$ representing $[A]$ not possessing null edge vectors.

Then, there exists an initial time $\vec t_0$ such that, to the following data $({\mathcal K}, [A]; \Gamma, P_0; \mathcal N; \vec t_0 )$, we associate
\begin{enumerate}
\item An \textbf{unique} real regular degree $g$ KP divisor $\DKP$ as in Definition \ref{def:real_KP_div};
\item An \textbf{unique} real regular KP wave function $\hat \psi(P, \vec t)$ corresponding to this divisor satisfying Definition~\ref{def:KPwave}.
\end{enumerate}
Moreover, $\DKP$ and $\hat \psi(P, \vec t)$ are both invariant with respect to changes of the vertex, the weight and the ray direction gauges and of the orientation of the graph $\mathcal G$.
\end{theorem}

In Section \ref{sec:anycurve} we explicitly construct both the divisor and the wave function for the soliton data $(\mathcal K, [A])$ and we prove that the divisor is contained in the union of the ovals. The construction is carried out explicitly fixing the orientation of the network and the ray direction, the vertex and the weight gauges. In Section \ref{sec:inv} we prove that the KP divisor is independent both on the chosen orientation and such gauges. In Section \ref{sec:comb}, we combinatorially complete the proof of Theorem \ref{theo:exist} by showing that there is exactly one KP divisor point in each finite oval. 

In \cite{AG3} we have constructed a KP divisor and a KP wave function in the special case where $\mathcal G$ is the Le--graph acyclically oriented w.r.t. the lexicographically minimal base of $\mathcal M$. In that case, we use a recursive procedure to define a system of edge vectors on the Le--network which rules the value of both the vacuum and the dressed wave functions at the double points of $\Gamma$. 
Here we extend such construction to any directed network ${\mathcal N}$ representing $[A]$. In Section \ref{sec:vectors} we prove that, for any given gauge ray direction associated to the oriented network, we obtain an \textbf{unique} system of edge vectors satisfying appropriate boundary conditions. In such Section, we give explicit expressions for their components and
explain their dependence on changes of orientation and of ray direction, weight and vertex gauges. 

If all edge vectors on the PBDTP network $\mathcal N$ are not null, then the \textbf{normalized} dressed edge wave function is well defined at the double points of $\Gamma$ (Section \ref{sec:vertex_wavefn_general_case}).  
Then, using the linear relations at the vertices, we assign a dressed network divisor number to any trivalent white vertex $V_l$ in Definition \ref{def:vac_div_gen}. 
The dressed network divisor number is the local coordinate of the KP divisor point $\Pdr_l \in \Gamma_l$ and the position of $\Pdr_l$ is independent on the orientation of the network, on the gauge ray direction, on the weight gauge and on the vertex gauge (Section \ref{sec:inv}). This set of divisor points has degree $g-k$ equal to the number of trivalent white vertices in the PBDTP graph $\mathcal G$.

$\DKP$ is defined as the sum of the degree $k$ Sato divisor $\DS$ 
and of the degree $g-k$ divisor contained in the union of the components $\Gamma_l$.
By definition, $\DKP$ is contained in the union of the ovals of $\Gamma$ and the normalized wave function may be meromorphically extended so that, for any fixed $\vec t$, we have $({\hat \psi} (P, \vec t))+\DKP\ge 0$ on $\Gamma\backslash P_0$.

\smallskip

Finally, in Section \ref{sec:moves_reduc} we explain how edge vectors are transformed under the action of Postnikov moves and reductions. As a consequence of Theorems \ref{theo:exist} and \ref{thm:null_acyclic} and Proposition \ref{prop:Le_net} we get the first statement in the following Corollary. The second part follows from the explicit characterization of the effect of moves and reductions on the wave function and the divisor carried in Section \ref{sec:moves_reduc}.

\begin{corollary}\label{cor:param}
Under the hypotheses of Theorem \ref{theo:exist}:
\begin{enumerate}
\item \textbf{Parametrization of $\S$ via KP divisors:}
If the PBDTP graph $\mathcal G$ representing $\S$ is reduced and equivalent to the Le--graph via a finite sequence of moves (M1), (M3) and flip moves (M2), then for any fixed
$\vec t_0$, there is a local one-to-one correspondence between KP divisors on $\Gamma (\mathcal G)$ and points $[A]\in \S$. 
\item \textbf{Discrete transformation between curves and divisors induced by moves and reductions:}
Let ${\mathcal G}$ and ${\mathcal G}^{\prime}$ be two PBDTP graphs equivalent by a finite sequence of moves and reductions for which Theorem \ref{theo:exist} holds true for the same $\vec t_0$, then there is an explicit transformation of the KP divisor on
$\Gamma(\mathcal G)$ to the KP divisor on $\Gamma({\mathcal G}^{\prime})$. 
\end{enumerate}
\end{corollary}

\begin{remark}\textbf{Global parametrization ia KP divisor or positroid cells} We claim that if, for a given graph $\mathcal G$ representing a given positroid cell $\S$, and any point $[A]\in S$ there is a network $\mathcal N$ of graph $\mathcal G$ not possessing null vectors, it is possible to show that the KP divisors provide a global parametrization of $\S$ after applying some blow-ups in all cases where some of the divisor points occur at double points. We claim that this occurs also when the network only carries null edge vectors of type 1, whereas in the case of null edge vectors of type 2 it is necessary to modify both the network and the divisor. We plan to discuss thoroughly this issue in a future publication. For some examples see Sections \ref{sec:global} and \ref{sec:constr_null}, and also Figures \ref{fig:Gr24_top_D2_null} and \ref{fig:Gr24_top_D3_null}.
\end{remark}

\smallskip

In contrast to \cite{AG1,AG3}, the proof of Theorem~\ref{theo:exist} does not require the technical step of constructing and characterizing a vacuum divisor on $\Gamma$. Nevertheless, as a byproduct of the construction carried in the following Sections, we also get a vacuum divisor with properties similar to those established in \cite{AG3} in the case of the Le-network. We remark that the vacuum and the KP divisors have different properties: the vacuum divisor depends both on the base $I$ which rules the orientation of the network and on the choice of the position of the Darboux points, whereas the KP divisor is invariant with respect to both.

\begin{theorem}\label{theo:vac_div}\textbf{Characterization of the vacuum divisor $\DVG$.} 
Let ${\mathcal K}$, $[A]\in \S\subset \GTNN$, $\mathcal G$, $\mathcal N$, $\vec t_0$, $\Gamma=\Gamma(\mathcal G)$ with marked point $P_0\in \Gamma_0$, be as in Theorem \ref{theo:exist}. 

Let $I$ be the base in $\mathcal M$ associated to the orientation of $\mathcal G$ and $\mathcal N^{\prime}$ be the modified network as in Construction \ref{con:Nprime} with Darboux points $P^{(3)}_j\in \Gamma_j$, $j\in [n]$, where $P^{(3)}_j$ is a Darboux source point if $j\in I$.

Then, there exists an unique real and regular degree $g$ vacuum divisor $\DVG$ associated to the data $({\mathcal K}, [A]; I, \Gamma, P_0, P^{(3)}_{j}, j\in I ; {\mathcal N}; \vec t_0 )$ with the following properties:
\begin{enumerate}
\item $\DVG$ is contained in the union of all the ovals of $\Gamma$; 
\item There is exactly one divisor point on each component of $\Gamma$ corresponding to a trivalent white vertex or a bivalent white vertex $\Gamma_j$, $j\in I$, containing a Darboux source point; 
\item\label{item:defodd} In any finite oval $\Omega_s$, $s\in [g]$, the total number of vacuum divisor poles plus the number of Darboux source points is $1 \,\, (\!\!\!\!\mod 2)$;
\item\label{item:defeven} In the infinite oval $\Omega_0$, the total number of vacuum divisor poles plus the number of Darboux source points plus $k$ is $0 \,\, (\!\!\!\!\mod 2)$.
\end{enumerate}
\end{theorem}

\section{Systems of vectors on PBDTP networks}\label{sec:vectors}

In this Section we construct systems of vectors on the edges $e$ of any given PBDTP network ${\mathcal N}$ in the disk (in the following we call it just network), representing 
a given point $[A]\in{\mathcal S}^{\mbox{\tiny{TNN}}}_{\mathcal M}\subset \GTNN$ in Postnikov classification \cite{Pos}. $\mathcal G$ denotes the graph of $\mathcal N$. 
We associate vectors to edges because the latter correspond to the double points of $\Gamma(\mathcal G)$. Using such vectors, in Section \ref{sec:anycurve} we construct a normalized dressed edge wave function $\hat \Psi_{{\mathcal N}} (\vec t)$ and extend it to a meromorphic function on $\Gamma(\mathcal G)\backslash \{P_0\}$ in Section \ref{sec:inv}. Therefore, in our construction, requirement (\ref{it:comp}) in Definition \ref{def:KPwave} is automatically satisfied. 

We assign the $j$-th vector of the canonical basis to any boundary vertex $b_j$ and define the $j$-th component of the vector $E_e$ at the edge $e\in \mathcal N$ as the sum of the contributions of all directed paths from $e$ to $b_j$. The absolute value of the contribution of one such directed path is equal to the product of the weights on this path. 

\begin{figure}
  \centering{\includegraphics[width=0.65\textwidth]{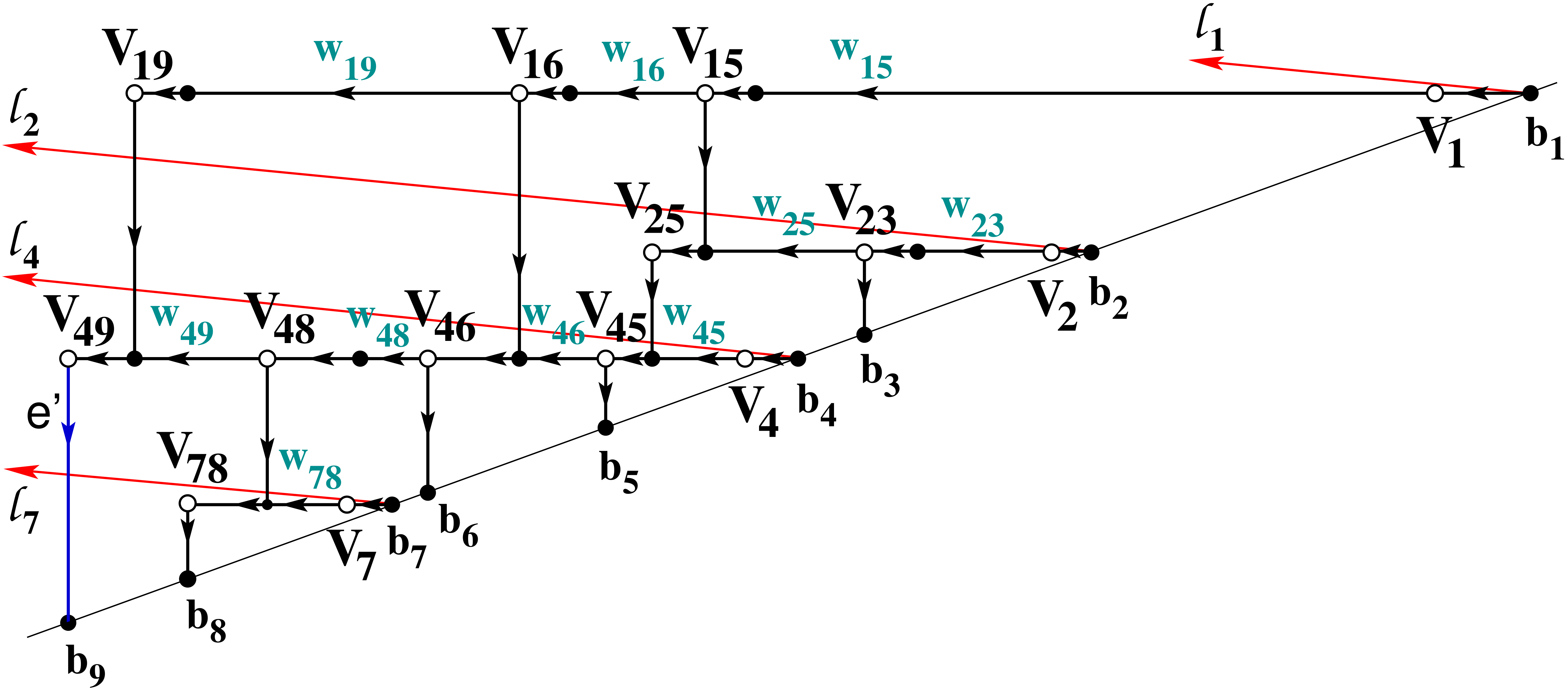}
      \caption{\small{\sl The rays starting at the boundary sources for a given orientation of the 
          network uniquely fix the edge vectors.
					}}\label{fig:Rules0}}
\end{figure}

To rule out the sign of a given path contribution, we introduce an additional structure on $({\mathcal N}, \mathcal O)$: the gauge ray direction $\mathfrak{l}$. 

\begin{definition}\label{def:gauge_ray}\textbf{The gauge ray direction $\mathfrak{l}$.}
We choose an oriented direction ${\mathfrak l}$ with the following properties:
\begin{enumerate}
\item The ray with the direction ${\mathfrak l}$ starting at a boundary vertex points inside the disk; 
\item No internal edge is parallel to this direction;
\item All rays starting at boundary vertices do not contain internal vertices.
\end{enumerate}
\end{definition}

We remark that the first propery may always be satisfied since all boundary vertices lie at a common straight interval in the boundary of $\mathcal N$.
\begin{remark}
In \cite{GSV}, a gauge ray direction was introduced to compute the winding number of a path joining boundary vertices. Here we use it also to generalize the notion of number of boundary source points crossed by a path when the path starts at an internal edge $e$. 
\end{remark}

The sign of the contribution of one such directed path depends on the generalized winding number and the generalized intersection number with gauge rays starting at the boundary sources. 

We then prove that the components of the edge vectors are explicit rational expressions in the edge weights with subtraction free denominators (Section \ref{sec:rational}). The proof of the latter statement follows adapting the calculation of the boundary measurement map in \cite{Pos} and \cite{Tal2}. The edge vectors are also the unique solution of a linear system of relations on $({\mathcal N},\mathcal O, \mathfrak l)$ (Section \ref{sec:linear}).

We remark that null edge vectors are possible even if there exist paths starting at the given internal edge and reaching boundary sinks. However, we prove that null edge vectors are forbidden in all networks admitting an acyclic orientation 
(Section \ref{sec:null_vectors}). Since unreduced graphs have the extra gauge freedom introduced in Remark \ref{rem:gauge_freedom}, we also conjecture that such freedom may be used to avoid null edge vectors also
in the case of reducible networks not admitting an acyclic orientation. 

If the orientation $\mathcal O$ is fixed and the direction $\mathfrak l$ changes, then the vector assigned to the edge $e$ may only change sign (Section \ref{sec:gauge_ray}). 
This property ensures that we uniquely construct a KP divisor on $\Gamma$ for each orientation and that such divisor does not depend on the gauge ray direction (Corollary \ref{cor:indep_gauge}).

If we change the orientation $\mathcal O$ and keep the direction $\mathfrak l$ invariant, each new vector is a linear combination of the old vector and of the rows of a chosen representative matrix of $[A]$ (Section \ref{sec:orient}).
This property implies that the KP divisor on $\Gamma$ does not depend on the orientation of the network (Corollary \ref{cor:indep_orient} and Theorem \ref{theo:inv}). 

In Section \ref{sec:different_gauge} we discuss the effect of both the weight and the vertex gauges introduced in Remarks \ref{rem:gauge_vertices} and \ref{rem:gauge_weight}. In both cases these gauges affect the edge vectors only locally and do not affect the KP divisor on $\Gamma$.

\subsection{Construction of systems of edge vectors and their basic properties}\label{sec:def_edge_vectors}
Let $b_{i_r}$, $r\in[k]$, $b_{j_l}$, $l\in[n-k]$, respectively  be the set of boundary sources and boundary sinks 
associated to the given orientation. Then draw the rays ${\mathfrak l}_{i_r}$, $r\in[k]$, starting at $b_{i_r}$ associated with the pivot columns of the given orientation. In Figure \ref{fig:Rules0} we show an example.

\begin{figure}
  \centering{\includegraphics[width=0.49\textwidth]{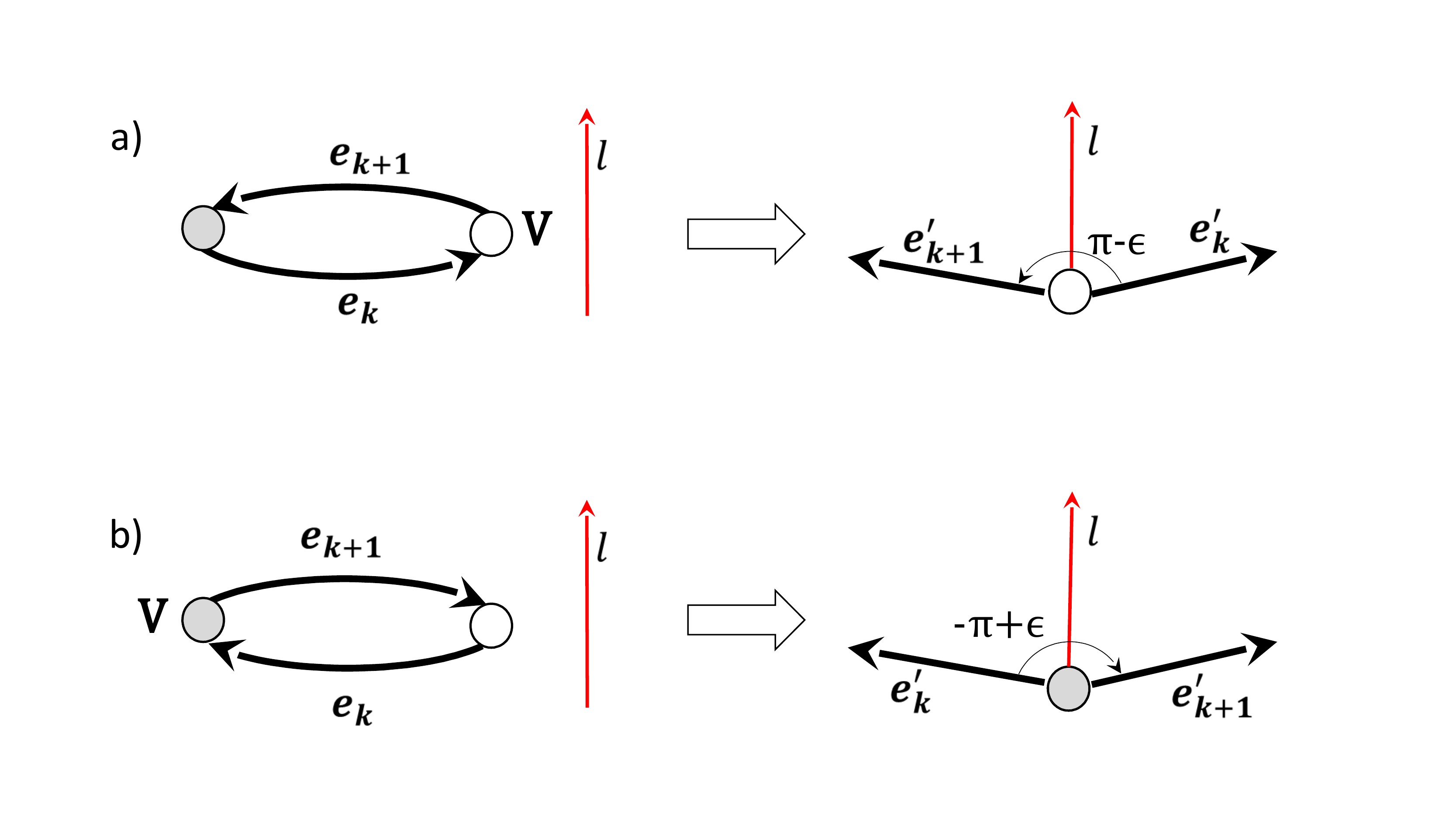}
	\includegraphics[width=0.49\textwidth]{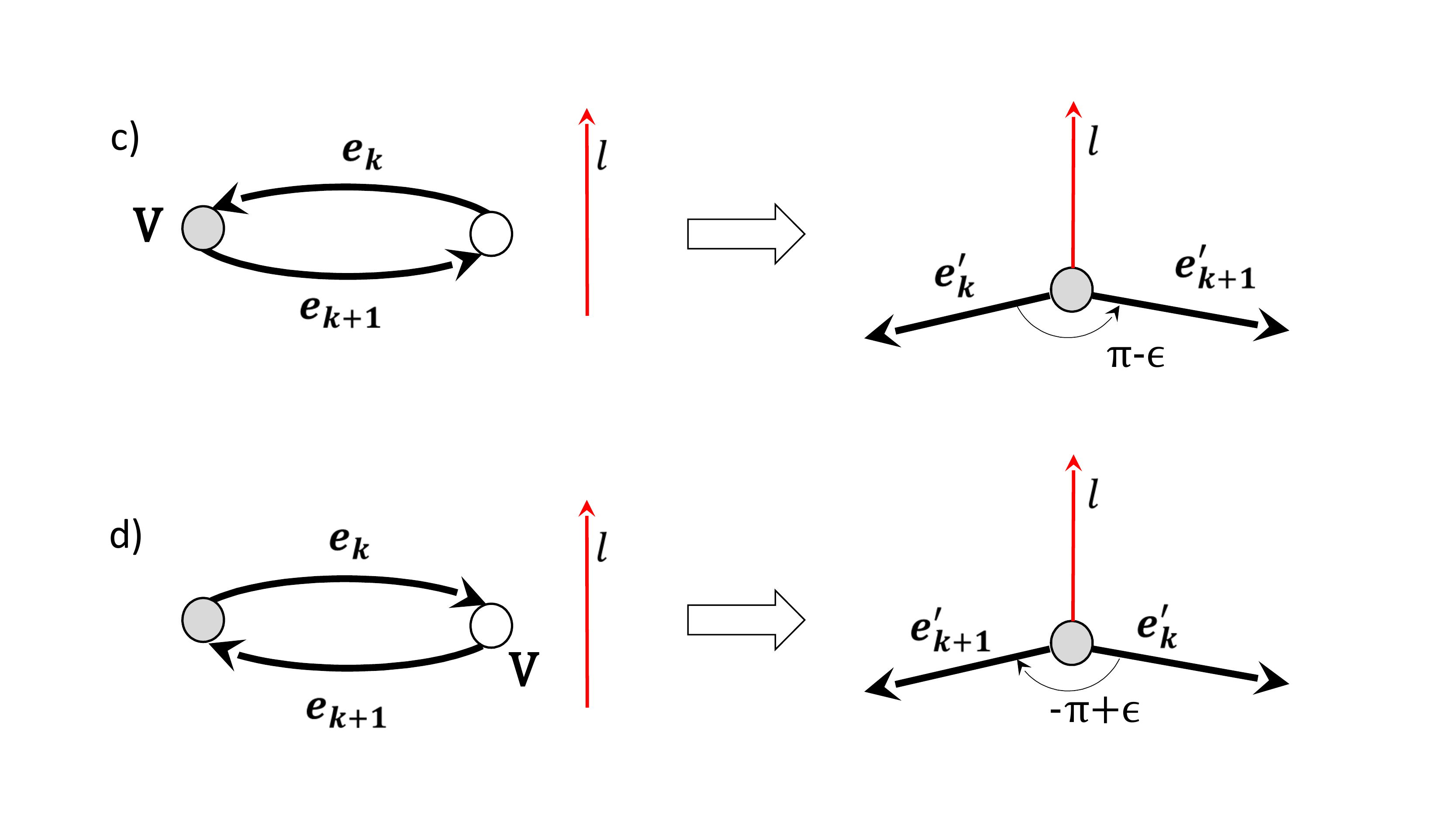}
      \caption{\small{\sl If the ordered pair $(e_k, e_{k+1})$ is antiparallel at $V$, we slightly rotate the two edge vectors at $V$ to compute $s(e_k,e_{k+1})$. Using (\ref{eq:s_antipar}) and (\ref{eq:def_wind}), we get $a)$: $s(e_k,e_{k+1})=1= \mbox{wind}(e_k,e_{k+1})$; $b)$: $s(e_k,e_{k+1})=-1= \mbox{wind}(e_k,e_{k+1})$; $c)$: $s(e_k,e_{k+1})=1$,  $\mbox{wind}(e_k,e_{k+1})=0$; $d)$: $s(e_k,e_{k+1})=-1$,  $\mbox{wind}(e_k,e_{k+1})=0$.
					}}\label{fig:antipar}}
\end{figure}

\begin{definition}\label{def:winding_pair}\textbf{The local winding number at an ordered pair of oriented edges}
For a generic ordered pair $(e_k,e_{k+1})$ of oriented edges, define
\begin{equation}\label{eq:def_s}
s(e_k,e_{k+1}) = \left\{
\begin{array}{ll}
+1 & \mbox{ if the ordered pair is positively oriented }  \\
0  & \mbox{ if } e_k \mbox{ and } e_{k+1} \mbox{ are parallel }\\
-1 & \mbox{ if the ordered pair is negatively oriented }
\end{array}
\right.
\end{equation}
In the non generic case of ordered antiparallel edges, we slightly rotate the pair $(e_k, e_{k+1})$ to $(e_k^{\prime}, e_{k+1}^{\prime})$ as in Figure \ref{fig:antipar} and define
\begin{equation}\label{eq:s_antipar}
s(e_k,e_{k+1}) = \lim_{\epsilon\to 0^+} s(e_k^{\prime},e_{k+1}^{\prime}).
\end{equation}
Then the winding number of the ordered pair $(e_k,e_{k+1})$ with respect to the gauge ray direction $\mathfrak{l}$ is
\begin{equation}\label{eq:def_wind}
\mbox{wind}(e_k,e_{k+1}) = \left\{
\begin{array}{ll}
+1 & \mbox{ if } s(e_k,e_{k+1}) = s(e_k,\mathfrak{l}) = s(\mathfrak{l},e_{k+1}) = 1\\
-1 & \mbox{ if } s(e_k,e_{k+1}) = s(e_k,\mathfrak{l}) = s(\mathfrak{l},e_{k+1}) = -1\\
0  & \mbox{otherwise}.
\end{array}
\right.
\end{equation}
\end{definition}
We illustrate the rule in Figure \ref{fig:winding}.

\begin{figure}
  \centering{
	\includegraphics[width=0.7\textwidth]{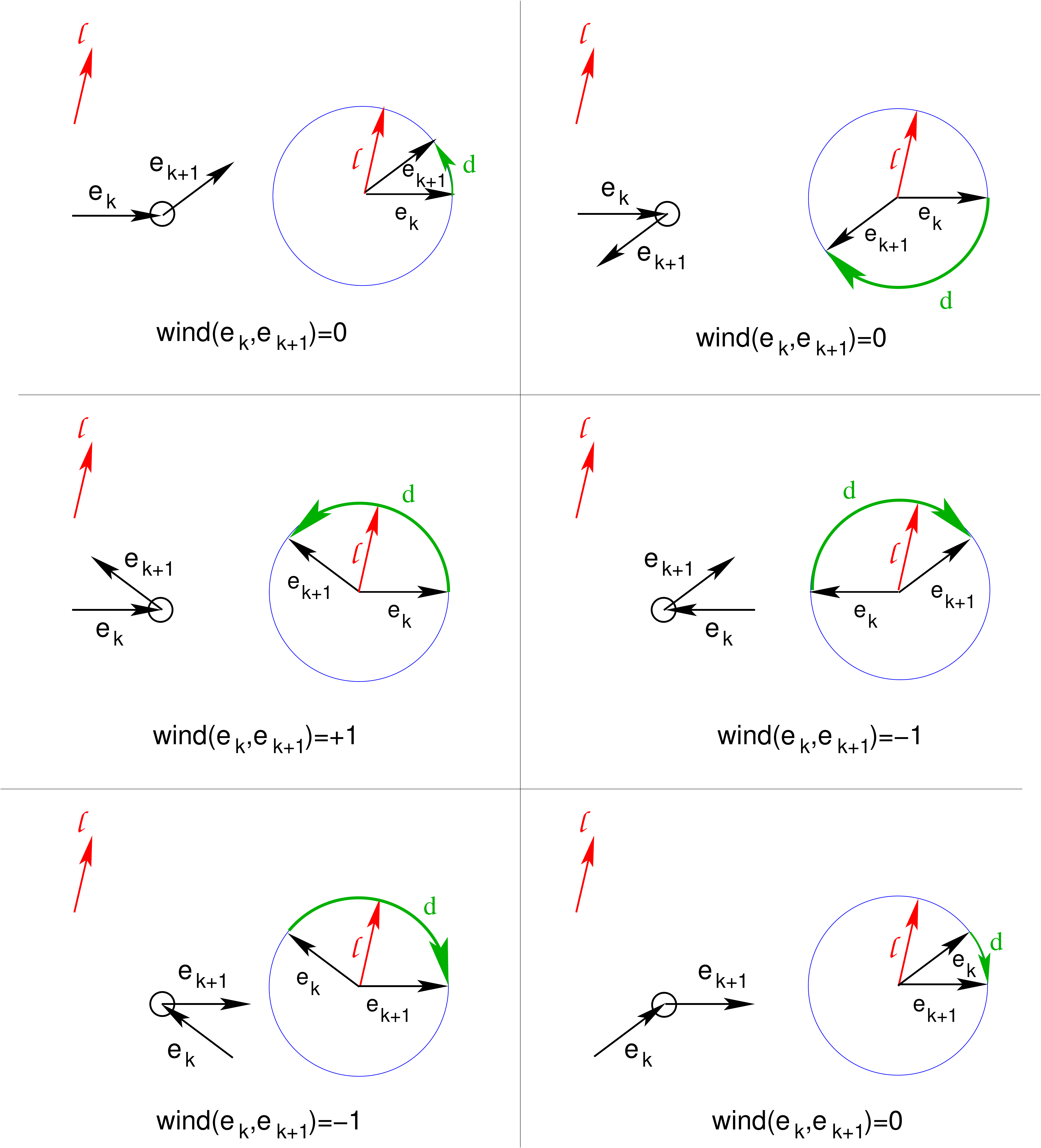}
  \caption{\label{fig:winding}\small{\sl The local rule to compute the winding number.}}
}
\end{figure}

Let us now consider a directed path ${\mathcal P}=\{e=e_1,e_2,\cdots, e_m\}$ starting at a vertex $V_1$ (either a boundary source or internal vertex) and ending at a 
boundary sink $b_j$, where $e_1=(V_1,V_2)$, $e_2=(V_2,V_3)$, \ldots, $e_m=(V_m,b_j)$. At each edge the orientation of the path coincides with the orientation of this edge in the graph.

We assign three numbers to ${\mathcal P}$:
\begin{enumerate}
\item The \textbf{weight $w({\mathcal P})$} is simply the product of the weights $w_l$ of all edges $e_l$ in ${\mathcal P}$, $w({\mathcal P})=\prod_{l=1}^m w_l$. If we pass the 
same edge $e$ of weight $w_e$ $r$ times, the weight is counted as $w_e^r$;
\item The \textbf{generalized winding number} $\mbox{wind}({\mathcal P})$ is the sum of the local winding numbers at each ordered pair of its edges
\[
\mbox{wind}(\mathcal P) = \sum_{k=1}^{m-1} \mbox{wind}(e_k,e_{k+1}),
\]
with $\mbox{wind}(e_k,e_{k+1})$ as in Definition \ref{def:winding_pair}; 
\item $\mbox{int}(\mathcal P)$ is the \textbf{number  of intersections} between the path and the rays ${\mathfrak l}_{i_r}$, $r\in[k]$: $\mbox{int}(\mathcal P) = \sum\limits_{s=1}^m \mbox{int}(e_s)$, where $\mbox{int}(e_s)$ is the number of intersections of gauge rays ${\mathfrak l}_{i_r}$ with $e_s$.
\end{enumerate}
The generalized winding of the path $\mathcal P$ depends on the gauge ray direction $\mathfrak{l}$ since it
counts how many times the tangent vector to the 
path is parallel and has the same orientation as ${\mathfrak l}$.

\begin{definition}\label{def:edge_vector}\textbf{The edge vector $E_e$.}
For any edge $e$, let us consider all possible directed paths ${\mathcal P}:e\rightarrow b_{j}$, 
in $({\mathcal N},{\mathcal O},{\mathfrak l})$ such that the first edge is $e$ and the end point is the boundary vertex $b_{j}$, $j\in[n]$.
Then the $j$-th component of $E_{e}$ is defined as:
\begin{equation}\label{eq:sum}
\left(E_{e}\right)_{j} = \sum\limits_{{\mathcal P}:e\rightarrow b_{j}} (-1)^{\mbox{wind}({\mathcal P})+ \mbox{int}({\mathcal P})} 
w({\mathcal P}).
\end{equation}
If there is no path from $e$ to $b_{j}$, the $j$--th component of $E_e$ is assigned to be zero. By definition, at the edge $e$ at the boundary sink $b_j$, the edge vector $E_{e}$ is
\begin{equation}\label{eq:vec_bou_sink}
\left(E_{e}\right)_{k} =  (-1)^{\mbox{int}(e)} w(e) \delta_{jk}.
\end{equation}
\end{definition}

In particular, for any $e$, all components of $E_e$ corresponding to the boundary 
sources in the given orientation are equal to zero. 
If $e$ is an edge belonging to the connected component of an isolated boundary sink $b_j$, 
then $E_e$ is proportional to the $j$--th vector of the canonical basis, while if $e$ is an edge belonging to the connected component of an isolated 
boundary source, then $E_e$ is the null vector.

If the number of paths starting at $e$ and ending at $b_j$ is finite for a given edge $e$ and destination $b_j$, the component $\left(E_{e}\right)_{j}$  
in (\ref{eq:sum}) is a polynomial in the edge weights.

If the number of paths starting at $e$ and ending at $b_j$ is infinite and the weights are sufficiently small, it is easy to check that the right hand side in (\ref{eq:sum}) converges. In Section \ref{sec:rational} we adapt the summation procedures of \cite{Pos} and \cite{Tal2} to prove that the edge vector components are rational expressions with subtraction-free denominators and provide explicit expressions (Theorem \ref{theo:null}). 

\subsection{Edge-loop erased walks, conservative and edge flows}\label{sec:flows}

Our next aim is to study the structure of the expressions representing the components of the edge vectors.

First, following \cite{Fom}, see also \cite{Law}, we adapt the notion of loop-erased walk to our situation, since our walks start at an edge, not at a vertex. 
\begin{definition}
\label{def:loop-erased-walk}
\textbf{Edge loop-erased walks.}  Let ${\mathcal P}$ be a walk (directed path) given by
$$
V_e \stackrel{e}{\rightarrow} V_1 \stackrel{e_1}{\rightarrow} V_2 \rightarrow \ldots \rightarrow b_j,
$$
where $V_e$ is the initial vertex of the edge $e$. The edge loop-erased part of  ${\mathcal P}$, denoted $LE({\mathcal P})$, is defined recursively as 
follows. If ${\mathcal P}$ does not pass any edge twice (i.e., all edges $e_i$ are distinct), then $LE({\mathcal P})={\mathcal P}$.
Otherwise, set $LE({\mathcal P})=LE({\mathcal P}_0)$, where ${\mathcal P}_0$ is obtained from ${\mathcal P}$, by removing the first
edge loop it makes; more precisely, find all pairs $l,s$ with $s>l$ and $e_l = e_s$, choose the one with the smallest value of $l$ and $s$ and remove the cycle
$$
V_l \stackrel{e_l}{\rightarrow} V_{l+1} \stackrel{e_{l+1}}\rightarrow V_{l+2} \rightarrow \ldots \stackrel{e_{s-1}}\rightarrow V_{s} ,
$$
from ${\mathcal P}$.
\end{definition}

\begin{remark}
An edge loop-erased walk can pass twice through the first vertex $V_e$, but it cannot pass twice any other vertex due to  trivalency.
For example, the directed path $1,2,3,4,5,6,7,8,12$ at Figure~\ref{fig:loop_erased} is edge loop-erased but it passes twice through the starting 
vertex $V_1$.
In general, the edge loop-erased walk does not coincide with the loop-erased walk defined in \cite{Fom, Law}. For instance, the directed path
$1,2,3,4,5,6,7,8,9,4,11$ has edge loop-erased walk $1,2,3,4,11$ and the loop-erased walk $7,8,9,4,11$. 

However, if $e$ starts at a boundary source, these two definitions coincide.  

In our text we never use loop-erased walks in the sense of \cite{Fom} and we use the notation $LE({\mathcal P})$ in the sense of 
Definition~\ref{def:loop-erased-walk}. 
\end{remark}

\begin{figure}
  \centering{
	\includegraphics[width=0.3\textwidth]{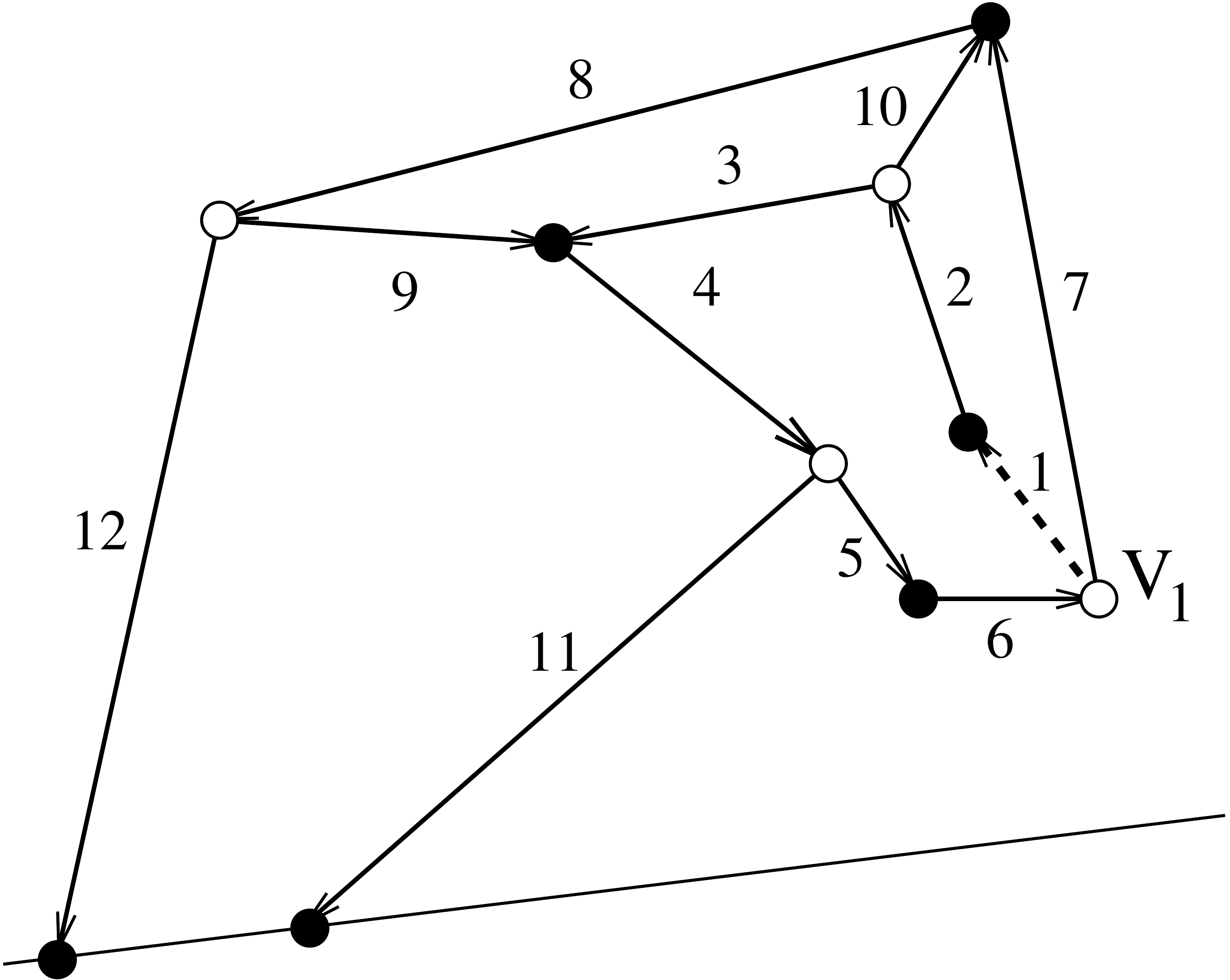}
	\vspace{-.7 truecm}
  \caption{\label{fig:loop_erased}\small{\sl  }}
}
\end{figure}

With this procedure, to each path starting at $e$ and ending at $b_j$ we associate an unique edge loop-erased walk $LE({\mathcal P})$, where the latter path is either acyclic or possesses one simple cycle passing through the initial vertex.
Then we formally reshuffle the summation over 
infinitely many paths starting at $e$ and ending at $b_j$ to a summation over a finite number of equivalent classes  $[LE({\mathcal P}_s)]$, each one consisting of all paths sharing the 
same edge loop-erased walk, $LE({\mathcal P}_s)$, $s=1,\ldots, S$. Let us remark that 
$ \mbox{int}({\mathcal P})-\mbox{int}(LE({\mathcal P}_s))=0 \,\,(\!\!\!\!\mod 2) $ for any ${\mathcal P}\in [LE({\mathcal P}_s)]$, and, moreover, 
$ \mbox{wind}({\mathcal P})-\mbox{wind}(LE({\mathcal P}_s))$ has the same parity as the number of simple cycles of ${\mathcal P}$ minus the number of 
simple cycles of $LE({\mathcal P}_s)$.
With this in mind, we rexpress (\ref{eq:sum}) as follows
\begin{equation}\label{eq:sum2}
\left(E_{e}\right)_{j} = \sum\limits_{s=1}^S (-1)^{\mbox{wind}(LE({\mathcal P}_s))+ \mbox{int}(LE({\mathcal P}_s))}\left[ 
\mathop{\sum\limits_{{\mathcal P}:e\rightarrow b_{j}}}_{{\mathcal P}\in [LE({\mathcal P}_s)] } (-1)^{\mbox{wind}({\mathcal P})-\mbox{wind}(LE({\mathcal P}_s))  } 
w({\mathcal P}) \right].
\end{equation}

We remark that the winding number along each simple closed loop introduces a $-$ sign in agreement with \cite{Pos}. Therefore the summation over paths may be interpreted as a discretization of path integration in some spinor theory. In typical spinor theories the change of phase during the rotation of the spinor corresponds to standard measure on the group $U(1)$ and requires the use of complex numbers. The introduction of the gauge direction forces the use of $\delta$--type measures instead of the standard measure on $U(1)$, and it permits to work with real numbers only.

Next we adapt the definitions of flows and conservative flows in \cite{Tal2} to our case.

\begin{definition}\label{def:cons_flow}\textbf{Conservative flow \cite{Tal2}}. A collection $C$ of distinct edges in a PBDTP graph $\mathcal G$ is called a conservative flow if 
\begin{enumerate}
\item For each interior vertex $V_d$ in $\mathcal G$ the number of edges of $C$ that arrive at $V_d$ is equal to the number of edges of $C$ that leave from $V_d$;
\item $C$ does not contain edges incident to the boundary.
\end{enumerate}
We denote the set of all conservative flows $C$ in $\mathcal G$  by ${\mathcal C}(\mathcal G)$. In particular, ${\mathcal C}(\mathcal G)$ contains the trivial flow with no edges to which we assign weight 1.
\end{definition}

\begin{definition}\label{def:edge_flow}\textbf{Edge flow at $e$}. A collection $F_e$ of distinct edges in a PBDTP graph $\mathcal G$ is called edge flow starting at the edge $e=e_1=(V_1,V_2)$ if 
\begin{enumerate}
\item $e\in F_e$;
\item For each interior vertex $V_d\ne V_1$ in $\mathcal G$ the number of edges of $F_e$ that arrive at $V_d$ is equal to the number of edges of $F_e$ that leave from $V_d$;
\item At $V_1$ the number of edges of $F_e$ that arrive at $V_1$ is equal to the number of edges of $F_e$ that leave from $V_1$ minus 1.
\end{enumerate}
We denote the set of all edge flows $F$ starting at an edge $e$ and ending at a boundary sink $b_j$ in $\mathcal G$ by ${\mathcal F}_{e,b_j}(\mathcal G)$.
\end{definition}

The conservative flows are collections of non-intersecting simple loops in the directed graph $\mathcal G$. We assign \textbf{no winding numbers} to conservative flows, exactly as in \cite{Tal2}. Due to the fact that the number of intersection of gauge rays with each connected component of a conservative flow is even, the assignment of intersection numbers to conservative flows is irrelevant. 

In our setting an edge flow  $F_{e,b_j}$ in ${\mathcal F}_{e,b_j}(\mathcal G)$ is either an edge loop-erased walk $P_{e,b_j}$ starting at the edge $e$ and ending at the boundary sink $b_j$ or the union of $P_{e,b_j}$ with a conservative flow with no common edges with $P_{e,b_j}$. In particular, our definition of edge flow coincides with the definition of flow in \cite{Tal2} if $e$ starts at a boundary source except for the winding and intersection numbers.
\begin{definition}\label{def:numb_flows}
\begin{enumerate}
\item We assign one number to each $C\in {\mathcal C}(\mathcal G)$: the \textbf{weight $w(C)$} is the product of the weights of all edges in $C$.
\item Let  $F_{e,b_j}\in{\mathcal F}_{e,b_j}(\mathcal G)$ be  the union of the edge loop-erased walk $P_{e,b_j}$ with a conservative flow with no common edges with $P_{e,b_j}$ (this conservative flow may be the trivial one). We assign three numbers to $F_{e,b_j}$:
\begin{enumerate} 
\item The \textbf{weight $w(F_e)$} is the product of the weights of all edges in $F_e$.
\item The \textbf{winding number} $\mbox{wind}(F_{e,b_j})$:
\begin{equation}
\label{eq:wind_flow}
\mbox{wind}(F_{e,b_j}) = \mbox{wind}(P_{e,b_j});
\end{equation}
\item The \textbf{intersection number} $\mbox{int}(F_{e,b_j})$:
\begin{equation}
\label{eq:int_flow}
\mbox{int}(F_{e,b_j}) = \mbox{int}(P_{e,b_j}).
\end{equation}
\end{enumerate}
\end{enumerate}
\end{definition}

\subsection{Rational representation of the components of $E_e$}\label{sec:rational}

A deep result of \cite{Pos}, see also \cite{Tal2}, is that each infinite summation in the square bracket of (\ref{eq:sum2}) is a subtraction-free rational expression when $e$ is the edge at a boundary source.
In this Section, adapting Theorem~3.2 in \cite{Tal2} to our purposes, we show that the components of $E_e$ defined in (\ref{eq:sum}) are rational expression in the weights with subtraction-free denominator and we provide an explicit expression for them. We remark that, contrary to the case in which the initial edge starts at a boundary source, if $e$ is an internal edge, the $j$--th component of $E_e$ may be null even if there exist paths starting at $e$ and ending at $b_j$ (see Figure \ref{fig:zero-vector}). 

\begin{theorem}\label{theo:null}\textbf{Rational representation for the components of vectors $E_e$}
Let $({\mathcal N},\mathcal O, \mathfrak l)$ be a PBDTP network representing a point $[A]\in \S \subset \GTNN$ with orientation $\mathcal O$ associated to the base $I =\{ 1\le i_1< i_2 < \cdots < i_k\le n\}$ in the matroid $\mathcal M$ and gauge ray direction $\mathfrak{l}$. Let us assign the $j$-th vectors of the canonical basis to the boundary sinks $b_j$, $j\in \bar I$.
Let the edge $e\in {\mathcal N}$ be such that there is a path starting at $e$ and ending at the boundary sink $b_j$. Then the $j$--th component of the edge vector at $e$, $\left(E_{e}\right)_{j}$, defined in (\ref{eq:sum2}) is a rational expression in the weights on the network with subtraction-free denominator: 
\begin{equation}
\label{eq:tal_formula}
\left(E_{e}\right)_{j}= \frac{\displaystyle\sum\limits_{F\in {\mathcal F}_{e,b_j}(\mathcal G)} \big(-1\big)^{\mbox{wind}(F)+\mbox{int}(F)}\ w(F)}{\sum\limits_{C\in {\mathcal C}(\mathcal G)} \ w(C)},
\end{equation}
where notations are as in Definitions~\ref{def:cons_flow},~\ref{def:edge_flow} and~\ref{def:numb_flows}.  
\end{theorem}
\begin{proof}
The proof is a straightforward adaptation of the proof in \cite{Tal2} for the computation of the Pl\"ucker coordinates. 
If the graph is acyclic, (\ref{eq:tal_formula}) coincides with (\ref{eq:sum}) because the denominator is one and edge flows $F_{e,b_j}$ are in one-to-one correspondence with directed paths connecting $e$ to $b_j$. Otherwise, in view of  (\ref{eq:sum}), (\ref{eq:tal_formula}) can be written as:
\begin{equation}
\label{eq:tal_formula_2}
 \sum\limits_{{\mathcal P}:e\rightarrow b_{j}}  \sum\limits_{C\in {\mathcal C}(\mathcal G)}  (-1)^{\mbox{wind}({\mathcal P})+ \mbox{int}({\mathcal P})}  w({\mathcal P})  w(C)  = \sum\limits_{F\in {\mathcal F}_{e,b_j}(\mathcal G)} \big(-1\big)^{\mbox{wind}(F)+\mbox{int}(F)}\ w(F),
\end{equation}
where in the left-hand side the first sum is over all directed paths from $e$ to $b_j$. In the left-hand side we have two types of terms (see also (\ref{eq:sum2})):
\begin{enumerate}
\item $\mathcal P$ is an edge loop-erased walk and $C$ is a conservative flow with no common edges with $\mathcal P$. By (\ref{eq:wind_flow}), the summation over this group coincides with the right-hand side of (\ref{eq:tal_formula_2});
\item $\mathcal P$ is not edge loop-erased or it is loop-erased, but has a common edge with  $C$.
\end{enumerate}
Following \cite{Tal2}, we prove that the summation over the second group gives zero by introducing a sign-reversing involution $\varphi$ on the set of pairs $(C,P)$. We first assign two numbers to each pair $(C,P)$ as follows:
\begin{enumerate}
\item Let $P=(e_1,\ldots,e_m)$. If $P$ is edge loop-erased, set $\bar s=+\infty$; otherwise, let $L_1=(e_l,e_{l+1},\ldots,e_{s_1})$ be the first loop erased according to Definition~\ref{def:loop-erased-walk} and set $\bar s=s$;
\item If $C$ does not intersect $P$, set $\bar t=+\infty$. Otherwise, set $\bar t$ the smallest $t$ such that $e_t\in P$ and $e_t\in C$. Denote the component of $C$ containing $e_{\bar t}$ by $L_2=(l_1,\ldots,l_p)$ with $l_1=e_{\bar t}$. 
\end{enumerate}
A pair $(C,P)$ belongs to the second group if and only if at least one of the numbers $\bar s$, $\bar t$ is finite. Moreover, in this case, $\bar s\ne \bar t$  due to the perfect orientation of the network. We then define $(C^*,P^*)=\varphi(C,P)$ as follows:
\begin{enumerate}
\item  If $\bar s < \bar t$, then $P$ completes its first cycle $L_1$ before intersecting any cycle in $C$. In this case $L_1\cap C=\emptyset$, and we remove $L_1$ from $P$ and add it to $C$. Then $P^*=(e_1,\ldots,e_{l-1},e_s,\ldots,e_m)$ and $C^*=C\cup L_1$;
\item If $\bar t < \bar s$, then $P$ intersects  $L_2$ before completing its first cycle. Then we remove $L_2$ from $C$ and add it to $P$: $C^*=C\backslash L_2$, $P^*=(e_1,\ldots,e_{\bar t-1},l_1=e_{\bar t},l_2,\ldots,l_p,e_{\bar t+1},\ldots,e_m)$. 
\end{enumerate}
From the construction of $\varphi$ it follows immediately that $(C^*,P^*)$ belongs to the second group, $\varphi^2=\mbox{id}$, and $\varphi$ is sign-reversing since 
$w(C^*) w(P^*) = w(C) w(P)$, $\mbox{wind}(P) +\mbox{wind}(P^*) =1 \,\,
(\!\!\!\!\mod 2)$ and $\mbox{int}(P) +\mbox{int}(P^*) =0 \,\,
(\!\!\!\!\mod 2)$. 
\end{proof}

\begin{corollary}
\label{cor:bound_source}\textbf{Edge vectors at boundary sources.} Under the hypotheses of Theorem \ref{theo:null}, let $e$ be the edge starting at the boundary source $b_{i_r}$. Then the number $\mbox{wind}(F)+\mbox{int}(F)$ has the same parity for all edge flows $F$ from $b_{i_r}$ to $b_j$ and it is equal to the number $N_{rj}$ of boundary sources between $i_r$ and $j$ in the orientation $\mathcal O$,
\begin{equation}\label{eq:index_source}
N_{rj} = \# \left\{ i_s \in I\ , \ i_s \in \big] \min \{i_r, j\}, \max \{ i_r , j \} \big[ \ \right\}.
\end{equation}
Therefore, for such edges (\ref{eq:tal_formula}) simplifies to
\begin{equation}
\label{eq:tal_formula_source}
\left(E_{e}\right)_{j}= \big(-1\big)^{N_{rj}}\ \frac{\sum_{F\in {\mathcal F}_{e,b_j}(\mathcal G)} \ w(F)}{\sum_{C\in {\mathcal C}(\mathcal G)} \ w(C)} =A^r_{j},
\end{equation}
where $A^r_j$ is the the entry of the reduced row echelon matrix $A$ with respect to the base $I=\{1\le i_1 < i_2 < \cdots < i_k \le n\}$.
\end{corollary}

\begin{proof}
First of all, in this case, each edge flow $F$ from $i_r$ to $j$ is either an acyclic edge loop--erased walk $\mathcal P$ or the union of $\mathcal P$ with a conservative flow $C$ with no common edges with $\mathcal P$. Therefore to prove that the number $\mbox{wind}(F)+\mbox{int}(F)$ has the same parity for all $F$ is equivalent to prove that $\mbox{wind}(\mathcal P)+\mbox{int}(\mathcal P)$ has the same parity for all edge loop--erased walks from $b_{i_r}$ to $b_j$. 
Any two such loop erased walks, ${\mathcal P}$ and ${\tilde {\mathcal P}}$, share at least the initial and the final edges and there exists $s\in \mathbb{N}$ and indices
\[
1\le l_1< r_1\le l_2 < r_2\le l_3 < \dots < l_s < r_s \le f,
\]
such that
\[
\begin{array}{c}
{\mathcal P} = (e, e_1,e_2, \dots, e_{l_1}, e_{l_1+1}, \dots, e_{r_1-1}, e_{r_1},e_{r_1+1}\dots, e_{l_s}, e_{l_s+1}, \dots, e_{r_s-1}, e_{r_s}, \dots , e_f),\\
{\tilde {\mathcal P}} = (e, e_1,\dots, e_{l_1}, {\tilde e}^{(1)}_1,\dots, {\tilde e}^{(1)}_{f_1}, e_{r_1}, \dots, e_{l_s}, {\tilde e}^{(s)}_1,\dots, {\tilde e}^{(s)}_{f_s}, e_{r_s}, \dots ,e_f).
\end{array}
\]
Indeed, due to acyclicity of both ${\mathcal P}$ and ${\tilde {\mathcal P}}$, if we add an edge $e_{j,i_r}$ from $b_j$ to $b_{i_r}$, we obtain a pair of simple cycles with the same orientation. Moreover, for each simple closed path the winding is equal to $1$ modulo $2$, therefore 
$$
\mbox{wind} (P) = \mbox{wind} (\tilde P)  = 1 - \mbox{wind}(f,e_{j,i_r}) - \mbox{wind}(e_{j,i_r},e) \quad (\!\!\!\!\!\!\mod 2),
$$
and, for any $p\in [s]$,
\[
\sum_{h=l_p+1}^{r_{p}-1} \mbox{int} (e_h) -\sum_{h=1}^{f_p} \mbox{int} ({\tilde e}^{(p)}_h) = 0 \quad (\!\!\!\!\!\!\mod 2),
\]
we easily conclude that
\[
\mbox{int} ({\mathcal P}) =\mbox{int} ({\tilde {\mathcal P}}), \quad (\!\!\!\!\!\!\mod 2).
\]
Since the right hand side in (\ref{eq:tal_formula_source}) coincides, up to the sign $(-1)^{N_{rj}}$, with the formula in \cite{Tal2}, the absolute value of the edge vector entry satisfies $|\left(E_{e}\right)_{j}|= |A^r_j|$. 

To complete the proof we need to show that $N_{rj}$ is the number of boundary sources in the interval $\big] \min \{ i_r, j\}, \max \{ i_r, j\}
\big[$. Without loss of generality, we may assume that $i_r<j$.
Since $\mathcal P$ is acyclic, all pivot rays ${\mathfrak l}_{i_l}$, $i_l\in [i_r -1] \cup [j, n]$ intersect $\mathcal P$ an even number of times, whereas all pivot rays ${\mathfrak l}_{i_l}$, $i_l\in [i_r +1, j]$ intersect $\mathcal P$ an odd number of times, while 
${\mathfrak l}_{i_r}$ intersects $\mathcal P$ either an even or an odd number of times. In the first case the 
winding number of the path is even, while in the second case it is odd (see Figure \ref{fig:pivot} [left]) and we get 
(\ref{eq:index_source}). 
\end{proof}

\subsection{The linear system on $({\mathcal N},\mathcal O,\mathfrak l)$}\label{sec:linear}

The edge vectors satisfy linear relations at the vertices of ${\mathcal N}$. In Theorem \ref{theo:consist} we prove that this set of linear relations provides an unique system of edge vectors on $({\mathcal N},\mathcal O,\mathfrak l)$ for any chosen set of independent vectors at the boundary sinks. Therefore the edge vectors computed in Theorem \ref{theo:null} may be also interpreted as the unique solution to this linear system of equations when we assign the $j$--th vector of the canonical basis to the $j$--th boundary vertex, $j\in [n]$.

\begin{lemma}
\label{lem:relations}
The edge vectors $E_e$ on $({\mathcal N},\mathcal O,\mathfrak l)$ satisfy the following linear equation at each vertex:  
\begin{enumerate}
\item  At each bivalent vertex with incoming edge $e$ and outgoing edge $f$:
\begin{equation}\label{eq:lineq_biv}
E_e =  (-1)^{\mbox{int}(e)+\mbox{wind}(e,f)} w_e E_f;
\end{equation}
\item At each trivalent black vertex with incoming edges $e_2$, $e_3$ and outgoing edge $e_1$ we have two relations:
\begin{equation}\label{eq:lineq_black}
E_2 =  (-1)^{\mbox{int}(e_2)+\mbox{wind}(e_2, e_1)}\ w_2 E_1,\quad\quad
E_3 =  (-1)^{\mbox{int}(e_3)+\mbox{wind}(e_3, e_1)}\  w_3 E_1;
\end{equation}
\item At each trivalent white vertex with incoming edge $e_3$ and outgoing edges $e_1$, $e_2$:
\begin{equation}\label{eq:lineq_white}
E_3 =  (-1)^{\mbox{int}(e_3)+\mbox{wind}(e_3, e_1)}\ w_3 E_1 + (-1)^{\mbox{int}(e_3)+\mbox{wind}(e_3, e_2)}\ w_3 E_2,
\end{equation}
\end{enumerate}
where $E_k$ denotes the vector associated to the edge $e_k$.
\end{lemma}
This statement follows directly from the definition of edge vector components as summations over all paths starting from this edge.

In Figure \ref{fig:rules23} we illustrate these relations at trivalent vertices assuming that the incoming edges do not 
intersect the gauge boundary rays. For instance, if $\mathfrak l$ belongs to the sector $S_1$, at the white vertex $E_3 = w_3 ( E_2 -E_1)$, where $E_j$ denotes the vector associated to the edge $e_j$, $j\in [3]$, whereas at the black vertex $E_3= w_3 E_1$ and $E_2 =-w_2 E_1$.

\begin{figure}
  \centering
	{\includegraphics[width=0.37\textwidth]{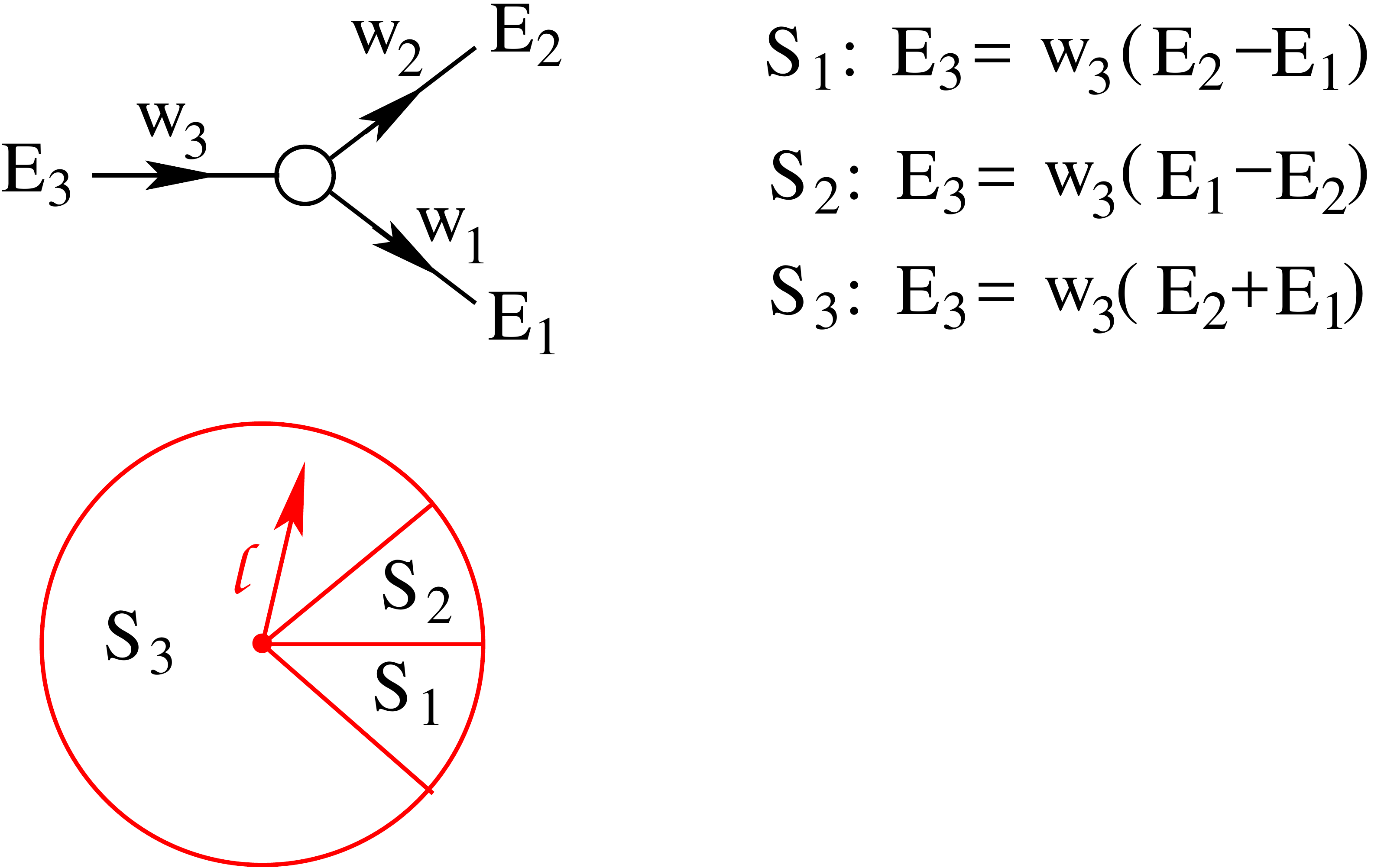}}
	\hfill
	{\includegraphics[width=0.49\textwidth]{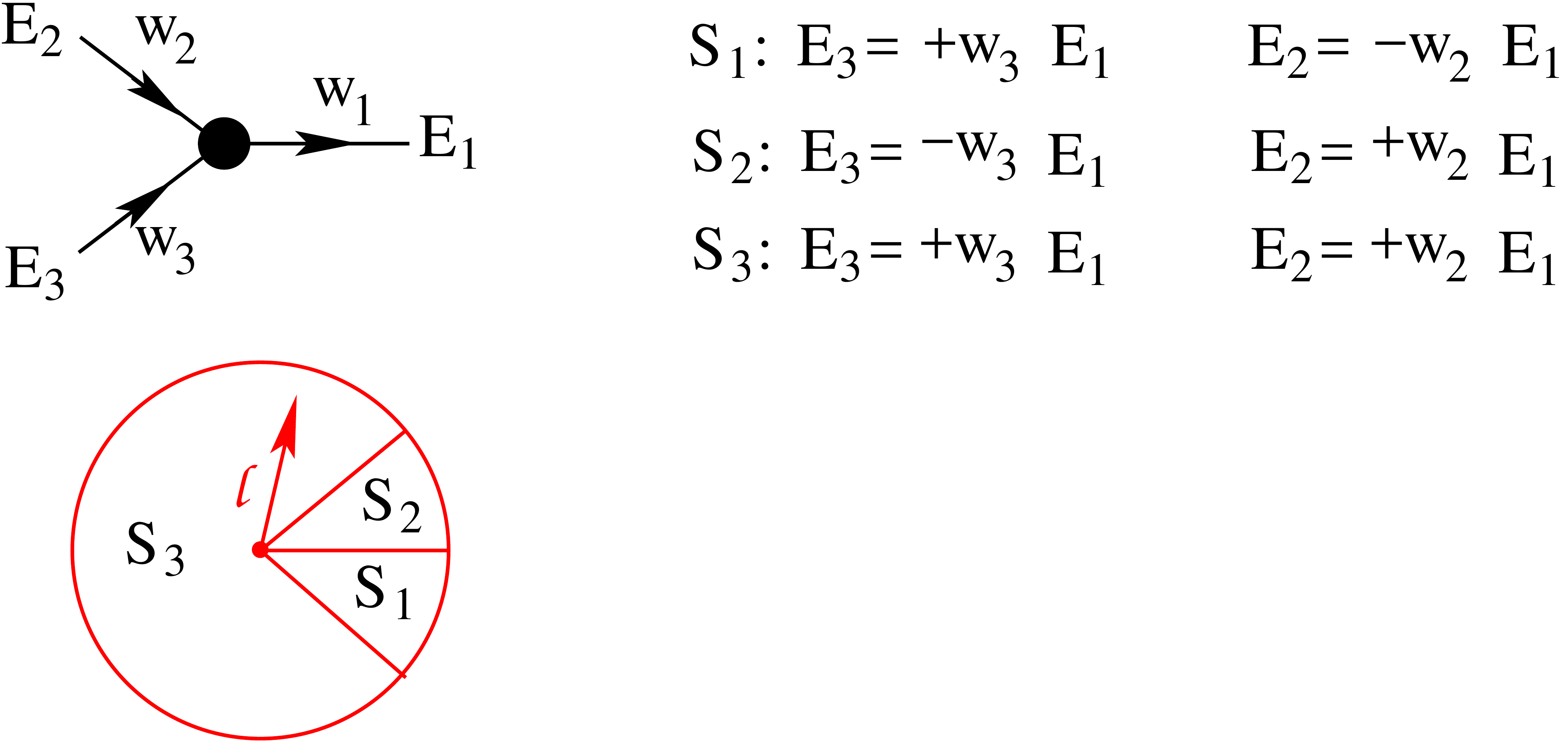}}
  \caption{\small{\sl The linear system at black and white vertices.}}\label{fig:rules23}
\end{figure}

Next we show that, for any given boundary condition at the boundary sink vertices, the linear system in Lemma \ref{lem:relations} defined by equations (\ref{eq:lineq_black}), (\ref{eq:lineq_white}) and (\ref{eq:lineq_biv}) at the internal vertices of $(\mathcal N, \mathcal O, \mathfrak{l})$ possesses an unique solution.

\begin{theorem}\textbf{Uniqueness of the system of edge vectors on $(\mathcal N, \mathcal O, \mathfrak{l})$.}\label{theo:consist}
Let $(\mathcal N, \mathcal O, \mathfrak{l})$ be a given PBDTP network with orientation $\mathcal O$ and gauge ray direction $\mathfrak{l}$. 

Then for any set $F_1,\dots, F_{n-k}$ of linearly independent vectors assigned to the boundary sinks $b_j$, the linear system of equations (\ref{eq:lineq_biv}), (\ref{eq:lineq_white}), (\ref{eq:lineq_black}) at all the internal vertices of $(\mathcal N, \mathcal O, \mathfrak{l})$ is consistent and provides an unique system of edge vectors on 
$(\mathcal N, \mathcal O, \mathfrak{l})$. 

Moreover, if we properly order variables and equations, the determinant of the matrix $M$ for this linear system is the sum of the weights of all conservative flows in $(\mathcal N,\mathcal O)$:
\begin{equation}
\label{eq:tal_den}
\det M = \sum\limits_{C\in {\mathcal C}(\mathcal G)}  w(C).
\end{equation}
\end{theorem}

\begin{proof} Let $g+1$ be the number of faces of ${\mathcal N}$ and let $t_W, t_B, d_W$ and $d_B$ respectively be the number of trivalent white, of trivalent black, of bivalent white, and of bivalent black internal vertices of ${\mathcal N}$ as in (\ref{eq:vertex_type}), where $n_I$ is the number of internal edges ({\sl i.e.} edges not connected to a boundary vertex) of ${\mathcal N}$. The total number of equations is $n_L=2t_B+ t_W+d_W+d_B =n_I+k$ whereas the total number of variables is equal to the total number of edges $n_I+n$. Therefore the number of free boundary conditions is $n-k$ and equals the number of boundary sinks. By definition, at the edge $e$ at a boundary sink $b_j$, the edge vector is $E_e = (-1)^{\mbox{int}(e)} w_e F_j$. 

Let us consider the inhomogeneous linear system obtained from equations (\ref{eq:lineq_biv}), (\ref{eq:lineq_white}), (\ref{eq:lineq_black}) in the $(n_I+k)$ unknowns given by the edge vectors not ending at the boundary sinks. Let us denote $M$ the $(n_I +k)\times (n_I+k)$ representative matrix of such linear system in which we enumerate edges so that each $r$-th row corresponds to the equation in (\ref{eq:lineq_biv}), (\ref{eq:lineq_white}), (\ref{eq:lineq_black}) in which the edge $e_r$ ending at the given vertex is in the l.h.s.. Then $M$ has unit diagonal by construction. 

If the orientation $\mathcal O$ is acyclic, then it is possible to enumerate the edges of $\mathcal N$ so that their indices in the right hand sight of each equation are bigger than that of the index on the left hand side. Therefore $M$ is upper triangular with unit diagonal, $\det M=1$ and the system of linear relations at the vertices has full rank.

Suppose now that the orientation is not acyclic. The standard formula expresses the determinant of $M$ as:
\begin{equation}
\label{eq:detM}
\det M = \sum\limits_{\sigma\in S_{n_L}} \mbox{sign}(\sigma)\prod\limits_{i=1}^{n_L}m_{i,\sigma(i)},
\end{equation}
where $S_{n_L}$ is the permutation group and $\mbox{sign}$ denotes the parity of the permutation $\sigma$. 

Any permutation can be uniquely decomposed as the product of disjoint cycles: 
$$
\sigma=(i_1,i_2,\ldots,i_{u_1+1})(j_1,j_2,\ldots,j_{u_2+1})\ldots(l_1,l_2,\ldots,l_{u_s+1}),
$$ 
and
$$
\mbox{sign}(\sigma)=(-1)^{u_1+u_2+\ldots+u_s}.
$$
On the other side, for $i\ne j$ $m_{i,j}\ne 0$ if and only if the ending vertex of the egde $i$ is the starting vertex of the edge $j$. Therefore $\prod\limits_{i=1}^{n_L} m_{i,\sigma(i)} \ne 0$ if and only if each cycle with $u_k>0$ in $\sigma$ coincides with a simple cycle in the graph, i.e. $\sigma$ encodes a conservative flow in the network. Therefore (\ref{eq:detM}) can be equivalently expressed as:
\begin{equation}
\label{eq:detM1}
\det M =\sum\limits_{C\in {\mathcal C}(\mathcal G)} \mbox{sign}(\sigma(C))\prod\limits_{i=1}^{n_L}m_{i,\sigma(i)}
\end{equation}
where
$\sigma(C)$ denotes the permutation corresponding to the conservative flow $C=C_1\cup C_2\cup\ldots\cup C_s$. Therefore 
\[
\mbox{sign}(\sigma(C))\prod\limits_{i=1}^{n_L}m_{i,\sigma(i)}=\prod\limits_{r=1}^s \left[(-1)^{u_r}\prod\limits_{t=1}^{u_r+1}(-1)^{1+\mbox{wind}(e_{i_t},e_{i_{t+1}})+\mbox{int}(e_{i_t})} w_{i_t}\right]= \prod\limits_{r=1}^s w(C_i)=  w(C),
\]
since the total winding of each simple cycle is $1\,\,(\!\!\!\mod 2)$, the total intersection number for each simple cycle is $0\,\,(\!\!\!\mod 2)$, and $w(C)=w(C_1)\cdots w(C_s)$. 
\end{proof}

\begin{example}
For the orientation and gauge ray direction as in Figure \ref{fig:Rules0}, the vectors $E_e$ on the Le--network coincide with those introduced
in the direct algebraic construction in \cite{AG3}.
\end{example}

\subsection{The dependence of $E_e$ on the gauge ray direction $\mathfrak l$}\label{sec:gauge_ray}
We now discuss the effect of a change of direction in the gauge ray $\mathfrak l$ on the vectors $E_e$. 

\begin{figure}
  \centering
	{\includegraphics[width=0.65\textwidth]{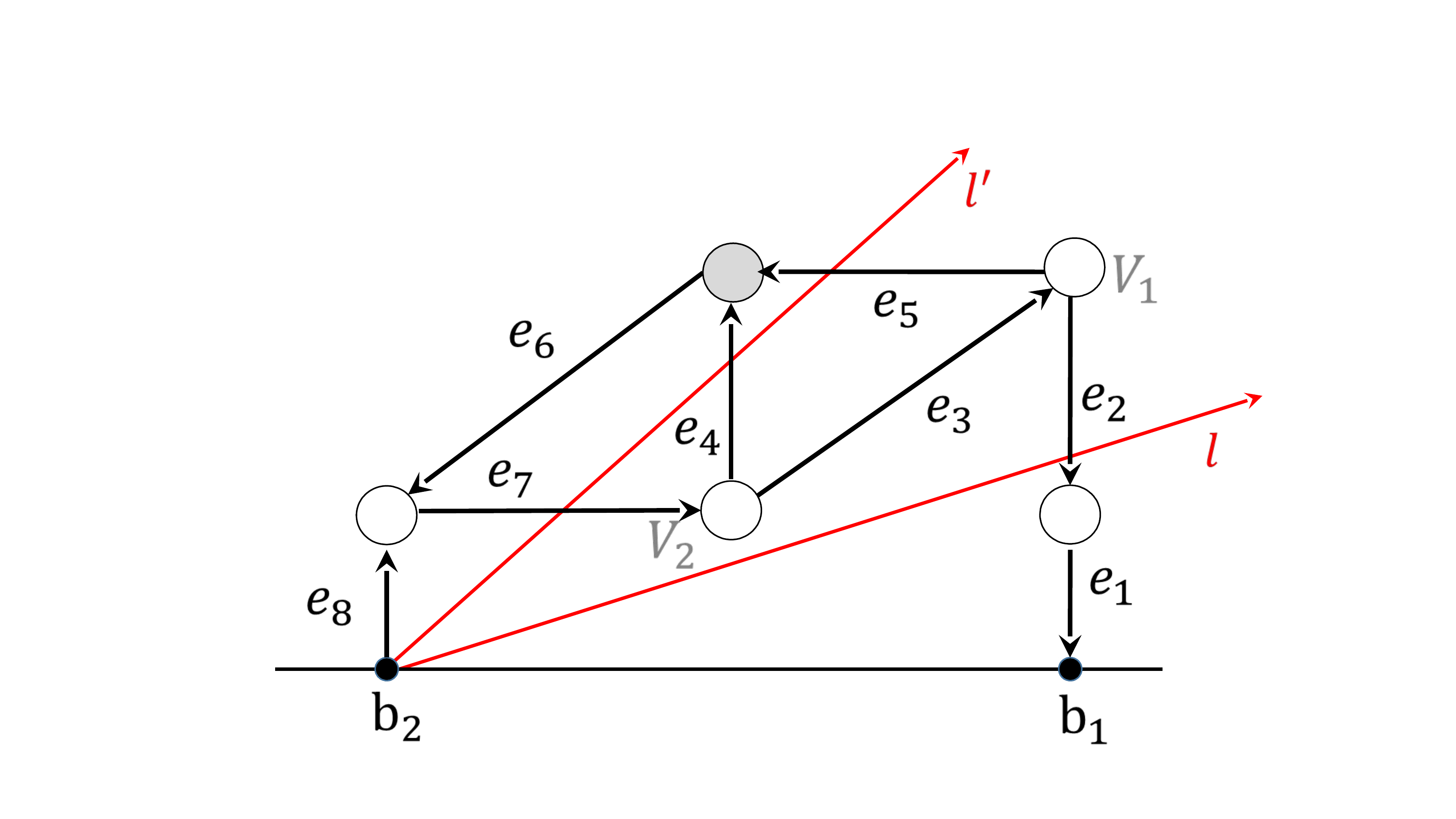}}
	\vspace{-.7 truecm}
  \caption{\small{\sl We illustrate Proposition \ref{prop:rays}.}}\label{fig:pivot}
\end{figure}

\begin{proposition}\label{prop:rays}\textbf{The dependence of the system of vectors on the ray direction $\mathfrak l$}
Let $({\mathcal N},\mathcal O)$ be an oriented network and consider two gauge directions $\mathfrak l$ and 
$\mathfrak{l}^{\prime}$.
\begin{enumerate}
\item For any boundary source edge $e_{i_r}$ the vector $E_{e_{i_r}}$ does not depend on the gauge direction $\mathfrak l$ 
and it coincides with the $r$-th row of the generalized RREF of $[A]$, associated to the pivot set $I$,  minus  the $i_r$--th vector of the canonical basis, which we denote $E_{i_r}$,
\begin{equation}
\label{eq:Ei}
E_{e_{i_r}} = A[r] -E_{i_r}.
\end{equation}
\item For any other edge $e$ we have
\begin{equation}\label{eq:vector_ray}
E^{\prime}_e = (-1)^{\mbox{int}(V_e)+\mbox{par(e)}} E_e,
\end{equation} 
where $E_e$ and  $E^{\prime}_e$ respectively are 
the edge vectors for $e$ for the gauge direction ${\mathfrak l}$ and ${\mathfrak l}^{\prime}$, $\mbox{par(e)}$ is 1 if 
we pass the ray $e$ during rotation of the gauge ray direction from ${\mathfrak l}$ to ${\mathfrak l}^{\prime}$ inside the disk and 0 otherwise,
whereas $\mbox{int}(V_e)$ denotes the number of gauge ray intersections with $e$ which pass the initial vertex $V_e$ of $e$ when we rotate from ${\mathfrak l}$ to ${\mathfrak l}^{\prime}$ inside the disk.
\end{enumerate}
\end{proposition}

\begin{proof}
Formula (\ref{eq:Ei}) follows from Corollary~\ref{cor:bound_source}: indeed (\ref{eq:index_source}) implies that the components of $E_{e_{\bar i}}$ are invariant with 
respect to changes of the gauge direction. Finally, since there is no path to the boundary source $b_{\bar i}$, therefore the corresponding component of the edge vector is zero. 

To prove the second statement, we show that, for a given initial edge $e$, the sign contribution of each edge loop--erased walk starting at $e$ is either the same before and after the gauge ray rotation or changes in the same way for every walk independently of the destination $b_j$.

Indeed let us consider a monotone continuous change of the gauge direction from initial $\mathfrak l(0)$ to final $\mathfrak l(1)$. For any given edge loop--erased walk $\mathcal P$, for every $t\in(0,1)$  such that  $\mathfrak l(t)$ forms zero angle with any edge of $\mathcal P$ distinct from the initial one, the parity of the  winding number remains unchanged. It changes  $1\  (\!\!\!\!\mod 2)$ only if  $\mathfrak l(t)$ forms zero angle with the initial edge $e$ of $\mathcal P$: in such case we settle $\mbox{par(e)}=1$. We remark that $\mathfrak l(t)$ can never form a zero angle with the edge at the boundary sink in $\mathcal P$. 

Similarly, if one of the gauge lines passes through a vertex in $\mathcal P$ distinct from the initial vertex, then the parity of the intersection number of $\mathcal P$ remains unchanged. It changes  $1\  (\!\!\!\!\mod 2)$ only if one of the gauge rays passes through the initial vertex of $\mathcal P$ (again it can never pass through the final vertex).    

Since the first edge $e$ and its initial vertex are common to all paths starting at $e$, all components of the 
vector $E_e$ either remain invariant, or are simultaneously multiplied by $-1$.
\end{proof}

\begin{example}
We illustrate Proposition \ref{prop:rays} in Figure \ref{fig:pivot}. In the rotation from $\mathfrak{l}$ to $\mathfrak{l}^{\prime}$ inside the disk, the gauge ray starting at $b_2$ passes the vertices $V_1$ and $V_2$ and the direction $e_3$. Therefore
$E^{\prime}_{e_i} = - E_{e_i}$, for $i=2,4,5$, whereas $E^{\prime}_{e_i} = E_{e_i}$ for all other edges.
\end{example}

\subsection{The dependence of $E_e$ on the orientation of the graph}\label{sec:orient}
We now explain how the system of vectors change when we change the orientation of the graph. Following \cite{Pos}, a change of orientation can be represented as a finite composition of elementary changes of orientation, each one consisting in
a change of orientation either along a simple cycle ${\mathcal Q}_0$ or
along a non-self-intersecting oriented path ${\mathcal P}$ from a boundary source $i_0$ to a boundary sink $j_0$.
Here we use the standard rule that we do not change the edge weight if the edge does not change orientation, otherwise we replace the original weight by its reciprocal.

\begin{theorem}\label{theo:orient}\textbf{The dependence of the system of vectors on the orientation of the network.}
Let ${\mathcal N}$ be a PBDTP network representing a given point $[A]\in \S \GTNN$ and $\mathfrak l$ be a gauge ray direction. Let $\mathcal O$, ${\hat {\mathcal O}}$ be
two perfect orientations of ${\mathcal N}$ for the bases $I,I^{\prime}\in {\mathcal M}$. Let $A[r]$, $r\in [k]$, denote 
the $r$-th row of a chosen representative matrix of $[A]$.
Let $E_e$ be the system of 
vectors associated to $({\mathcal N},\mathcal O, \mathfrak l)$ and satisfying the boundary conditions $E[j]$ at $b_j$, $j \in \bar I$, whereas $\hat E_e$ are those associated to $({\mathcal N},{\hat {\mathcal O}}, \mathfrak l)$ and satisfying the boundary conditions $E[l]$ at $b_l$, $l \in \bar I^{\prime}$.  

Then for any $e\in {\mathcal N}$, there exist real constants $\alpha_e\ne 0$, $c^r_e$, $r\in [k]$ such that
\begin{equation}\label{eq:orient}
\hat E_e = \alpha_e E_e + \sum_{r=1}^k c^r_e A[r].
\end{equation}
\end{theorem}

\begin{proof}
In Lemmas \ref{lemma:path} and \ref{lemma:cycle} we prove Theorem \ref{theo:orient} in the case of elementary changes of orientation. In the general case, the change of orientation is represented by the composition of a finite set of such elementary transformations and, by construction the edge vectors $\hat E_e$ solve the linear system of relations and satisfy the desired boundary conditions on $({\mathcal N},{\hat {\mathcal O}}, \mathfrak l)$. Since the solution of the linear system at the internal vertices is unique (Theorem \ref{theo:consist}), then the edge vectors $\hat E_e$ are the edge vectors in $({\mathcal N},{\hat {\mathcal O}}, \mathfrak l)$.
\end{proof}

Both in the case of an elementary change of orientation along a non-self-intersecting directed path $\mathcal P_0$ from a boundary source to a boundary sink or along a simple cycle ${\mathcal Q}_0$, we provide the explicit relation between the edge vectors in the two orientations.
If the change of orientation is ruled by $\mathcal P_0$, we use a two-steps proof:
\begin{enumerate}
\item We conveniently change the boundary conditions at the boundary sinks in the initial orientation of the
network $({\mathcal N},\mathcal O,\mathfrak l)$, we compute the system of vectors ${\tilde E}_e$ satisfying these new boundary conditions and we give explicit relations between the two systems of vectors $E_e$ and  ${\tilde E}_e$ on $({\mathcal N},\mathcal O,\mathfrak l)$;
\item Then, we show that the system of vectors ${\hat E}_e$ in (\ref{eq:hat_E_P}), defined on $({\mathcal N},{\hat {\mathcal O}},\mathfrak l)$ in terms of ${\tilde E}_e$, is the required system of vectors associated to the given change of orientation of the network. 
\end{enumerate}

For any elementary change of orientation, we assign an index ${\epsilon}(e)$ to each edge of the network in its \textbf{initial orientation}.
Let $\mathcal P_0$ be a non-self-intersecting oriented path from a boundary source $i_0$ to a boundary sink $j_0$ in the initial orientation of ${\mathcal N}$ and divide the interior of the disk into a finite number of regions bounded by the gauge ray ${\mathfrak l}_{i_0}$ oriented 
upwards, the gauge ray ${\mathfrak l}_{j_0}$ oriented downwards, the path $\mathcal P_0$ oriented as in $({\mathcal N},\mathcal O,\mathfrak l)$ and the boundary 
of the disk divided into two arcs, each oriented from $j_0$ to $i_0$. Then mark a region with a $+$ if its boundary is
oriented, otherwise mark it with $-$ (see Figure~\ref{fig:inv_symb}).

\begin{figure}
  \centering{\includegraphics[width=0.4\textwidth]{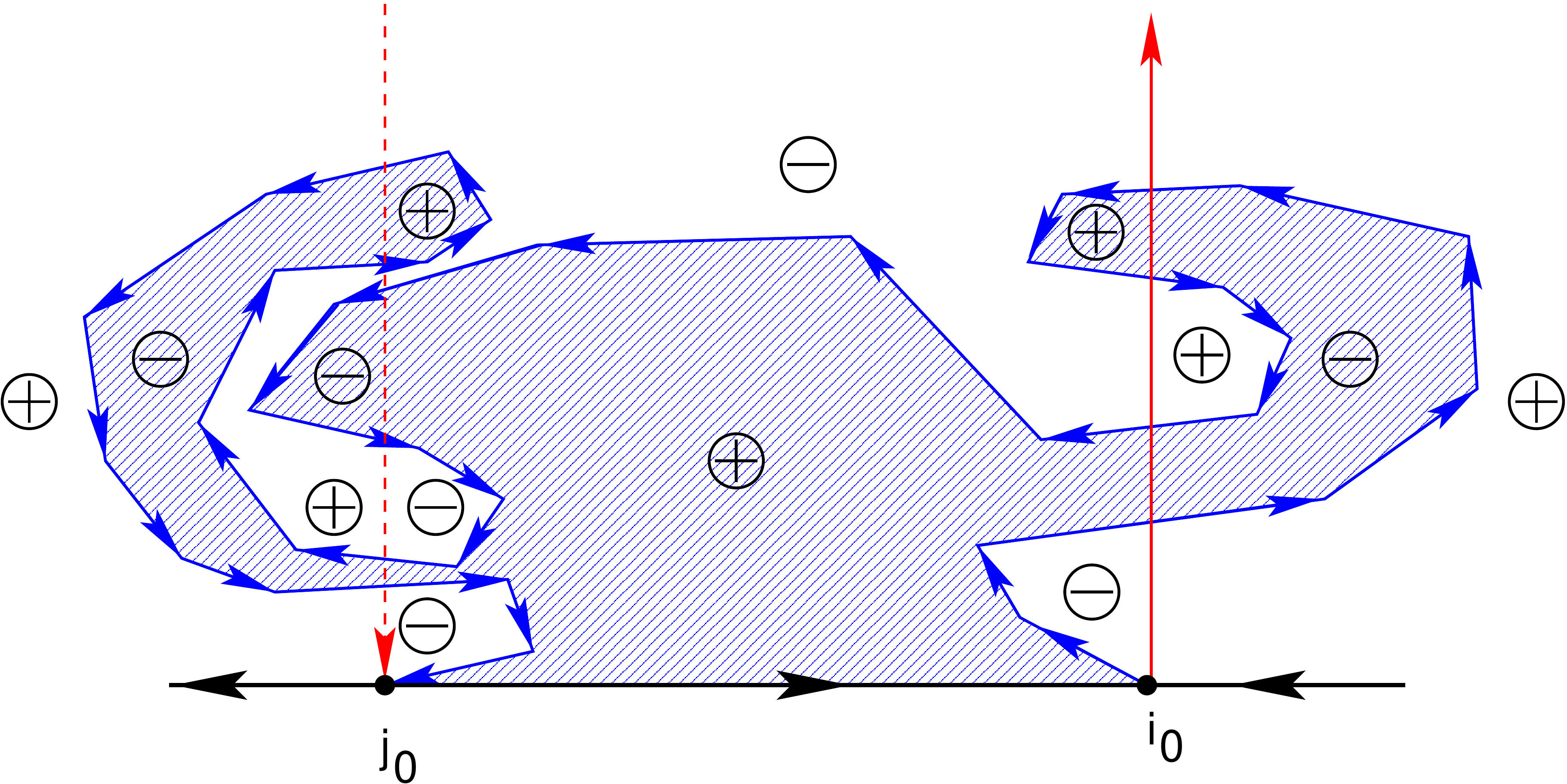}
	\hspace{.45 truecm}
	\includegraphics[width=0.4\textwidth]{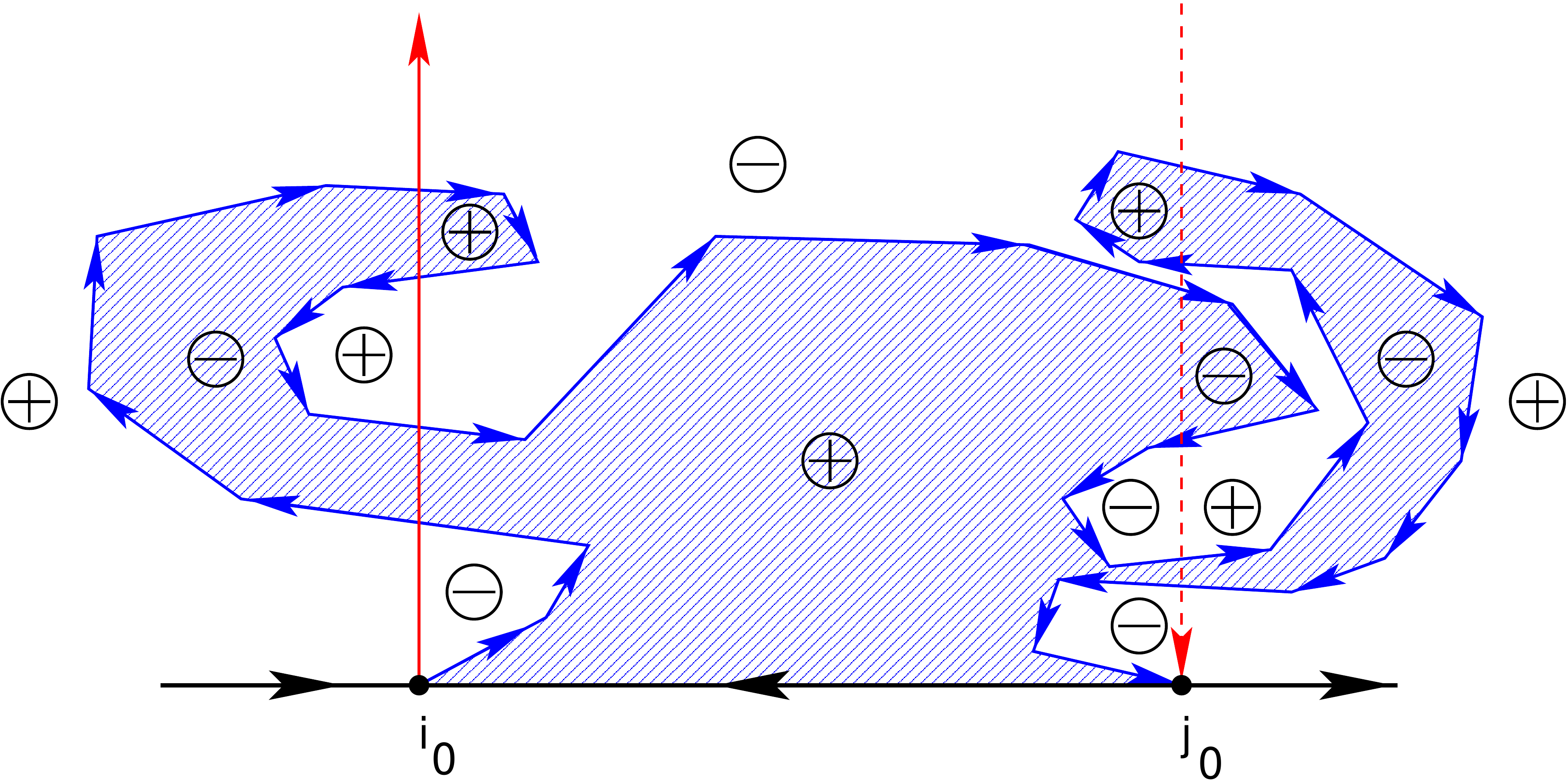}
  \caption{\small{\sl We illustrate the marking of the regions.}}\label{fig:inv_symb}}
\end{figure}

Similarly a closed oriented simple path $\mathcal Q_0$ divides the interior of the disk into two regions: we mark the region external to $\mathcal Q_0$ with a $+$ and the internal region with $-$. 

If the edge $e\not \in \mathcal P_0$ (respectively $e\not \in \mathcal Q_0$), we assign it an index $\epsilon(e)$ in the original orientation $\mathcal O$ of ${\mathcal N}$ as follows
\begin{equation}\label{eq:eps_not_path}
\epsilon (e) =
\left\{ 
\begin{array}{ll} 
0 & \mbox{ if the starting vertex of } e \mbox{ belongs to a } + \mbox{ region, }\\
1 & \mbox{ if the starting vertex of } e \mbox{ belongs to a } - \mbox{ region, }\\
\end{array}
\right.
\end{equation}
where, in case the initial vertex of $e$ belongs to $\mathcal P_0$ or $\mathcal Q_0$, we make an infinitesimal shift of the starting vertex in the direction of $e$ before assigning the edge to a region. 

If the edge $e\in \mathcal P_0$ (respectively $e\in \mathcal Q_0$), we assign it the index 
\begin{equation}\label{eq:eps_on_path}
\epsilon (e) = \epsilon_1 (e) + \epsilon_2 (e) +\epsilon_3(e),
\end{equation}
using the initial orientation $\mathcal O$ as follows:
\begin{enumerate}
\item We look at the region to the left and near the ending point of $e$, and assign index 
\[
\epsilon_1 (e) =
\left\{ 
\begin{array}{ll} 
0 & \mbox{ if the region is marked with } + ,\\
1 & \mbox{ if the region is marked with } - ;\\
\end{array}
\right.
\]
\item We consider the ordered pair $(e, \mathfrak{l})$ and assign index 
\[
\epsilon_2 (e) = \frac{1-s(e,\mathfrak{l})}{2}
\]
with $s(\cdot,\cdot)$ as in (\ref{eq:def_s})
\item We count the number of intersections with the gauge lines ${\mathfrak l}_{i_r}$, $r \in [k]$, and we define
\[
\epsilon_3 (e) = \mbox{ind}_{\mathcal O} (e) = \# \{ \mbox { intersections with gauge lines } \mathfrak{l}_{i_r}, \; r \in [k] \}.
\]
\end{enumerate}

\begin{lemma}\label{lemma:path}\textbf{The effect of a change of orientation along a non self--intersecting path from a boundary source to a boundary sink.}
Let $I = \{ 1\le i_1 < i_2 < \cdots < i_k\le n\}$ and $\bar I = \{ 1\le j_1 < j_2 < \cdots < j_{n-k}\le n\}$ respectively be the pivot and non--pivot indices in the representative RREF matrix $A$ associated to $({\mathcal N},\mathcal O,\mathfrak l)$. Assume that all the edges at the boundary vertices have unit weight and that no gauge ray intersects such edges in the initial orientation. 
Assume that we change the orientation  along a non-self-intersecting oriented path $\mathcal P_0$ from a boundary source $i_0$ to a 
boundary sink $j_0$. Let $E_e$ and $\tilde E_e$ be the systems of vectors on $({\mathcal N},\mathcal O,\mathfrak l)$ corresponding to the following choices of boundary conditions at edges $e_j$ ending at the boundary sinks $b_j$, $j\in \bar I$:
\begin{equation}
\label{eq:orient1}
E_{e_{j}}=E_{j}, \quad\quad\quad\quad
\tilde E_{e_{j}}= \left\{ \begin{array}{ll} E_{j} & \mbox{ if } j\not = j_0;\\
E_{j_0}-\frac{1}{A^{r_0}_{j_0}} A[r_0], &\mbox{ if } j = j_0,
\end{array}
\right.
\end{equation}
where $E_{j}$ is the $j$--th vector of the canonical basis, whereas $A[r_0]$ is the row of the matrix $A$ associated to the source $i_0$.

Define the system of vectors $\hat E_e$,  $e\in {\mathcal N}$, as follows: 
\begin{equation}\label{eq:hat_E_P}
{\hat E}_e = \left\{ \begin{array}{ll}
 (-1)^{\epsilon(e)} {\tilde E}_e, & \mbox{ if } e\not \in \mathcal P_0, \mbox{ with } \epsilon(e) \mbox{ as in (\ref{eq:eps_not_path})},\\
\displaystyle \frac{(-1)^{\epsilon(e)}}{w_e} {\tilde E}_e, & \mbox{ if } e\in \mathcal P_0, \mbox{ with } \epsilon(e) \mbox{ as in (\ref{eq:eps_on_path})}.
\end{array}\right.
\end{equation}

Then the system $\hat E_e$ defined in (\ref{eq:hat_E_P}) is the system of vectors on the network $({\mathcal N},{\hat {\mathcal O}},\mathfrak l)$ satisfying the boundary conditions
\begin{equation}
\label{eq:orient2}
\hat E_{e_{j}}= \left\{ \begin{array}{ll} (-1)^{\mbox{int} (e_j)^{\prime}} E_{j} & \mbox{ if } j\in \bar I \backslash \{ j_0\}\\
E_{i_0}, &\mbox{ if } j = i_0,
\end{array}
\right.
\end{equation}
where $\mbox{int}(e)^{\prime}$ is the number of intersections of the gauge ray $\mathfrak{l}_{j_0}$  with $e$.
\end{lemma}

\begin{remark}
To simplify the proof of Lemma \ref{lemma:path}, we assume without loss of generality that the edges at boundary vertices have unit weight and that no gauge ray intersects them in the initial orientation. This hypothesis may be always fulfilled modifying the initial network using the weight gauge freedom (Remark~\ref{rem:gauge_weight}) and adding, if necessary, bivalent vertices next to the boundary vertices using move (M3). 
In Sections~\ref{sec:different_gauge} and \ref{sec:middle}, we show that the effect of these transformations amounts to a well--defined non zero multiplicative constant for the edge vectors. Therefore, the statement in Lemma \ref{lemma:path} holds in the general case with obvious minor modifications in the boundary conditions for the three systems of vectors $E_e$, $\tilde E_e$ and $\hat E_e$.
\end{remark}

\begin{proof}
The system of vectors  $\tilde E_e - E_e$ is the solution to the system of linear relations on 
$({\mathcal N},\mathcal O,\mathfrak l)$ for the following boundary conditions: 
\[
\tilde E_{e_{j}}-E_{e_{j}} = \left\{ \begin{array}{ll} 0 & \mbox{ if } j\not = j_0;\\
-\frac{1}{A^{r_0}_{j_0}} A[r_0], &\mbox{ if } j = j_0.
\end{array}
\right.
\]
Then at all edges $e\in {\mathcal N}$ the difference $\tilde E_{e}-E_{e}$ is proportional to $A[r_0]$. In particular $\tilde E_{e_{i_0}} = -E_{i_0}$, since, by construction, $E_{e_{i_0}} = A[r_0] - E_{i_0}$. Therefore, each vector $\hat E_e$ in (\ref{eq:hat_E_P}) is a linear combination of the vector $E_e$ and $A[r_0]$.

In Appendix \ref{app:orient} we prove that the system $\hat E_e$ solves the linear system on $({\mathcal N},{\hat {\mathcal O}},\mathfrak l)$ at each internal vertex of the network. 

Finally the system of edge vectors ${\hat E}_e$ satisfies the boundary conditions in (\ref{eq:orient2}). First of all, any given boundary sink edge $e_j$, $j\not = j_0, i_0$, ends in a $+$ region, whereas it starts in a $-$ region only if it intersects $\mathfrak{l}_{j_0}$. The latter is exactly the unique case in which $\prec E_{e_j}, \hat E_{e_j} \succ =\prec \tilde E_{e_j}, \hat E_{e_j} \succ =-1$. 

The edge $e_{i_0}$ belongs to the path $\mathcal P_0$ and it does not intersect any gauge ray in both orientations of the network. $e_{i_0}$ has a $+$ region to the left and the pair $(e,\mathfrak{l})$ is negatively oriented or it has a $-$ region to the left and the pair $(e,\mathfrak{l})$ is positively oriented (see Figure \ref{fig:edgei_0}). Therefore
\[
\epsilon (e_{i_0}) = \epsilon_1 (e_{i_0}) + \epsilon_2 (e_{i_0}) =1.
\]
Finally $\hat E_{e_{i_0}} = (-1)^{\epsilon (e_{i_0})} \tilde E_{e_{i_0}} = E_{i_0}$ since $\tilde E_{e_{i_0}} = - E_{i_0}$. 
\end{proof}

\begin{figure}
  \centering{\includegraphics[width=0.4\textwidth]{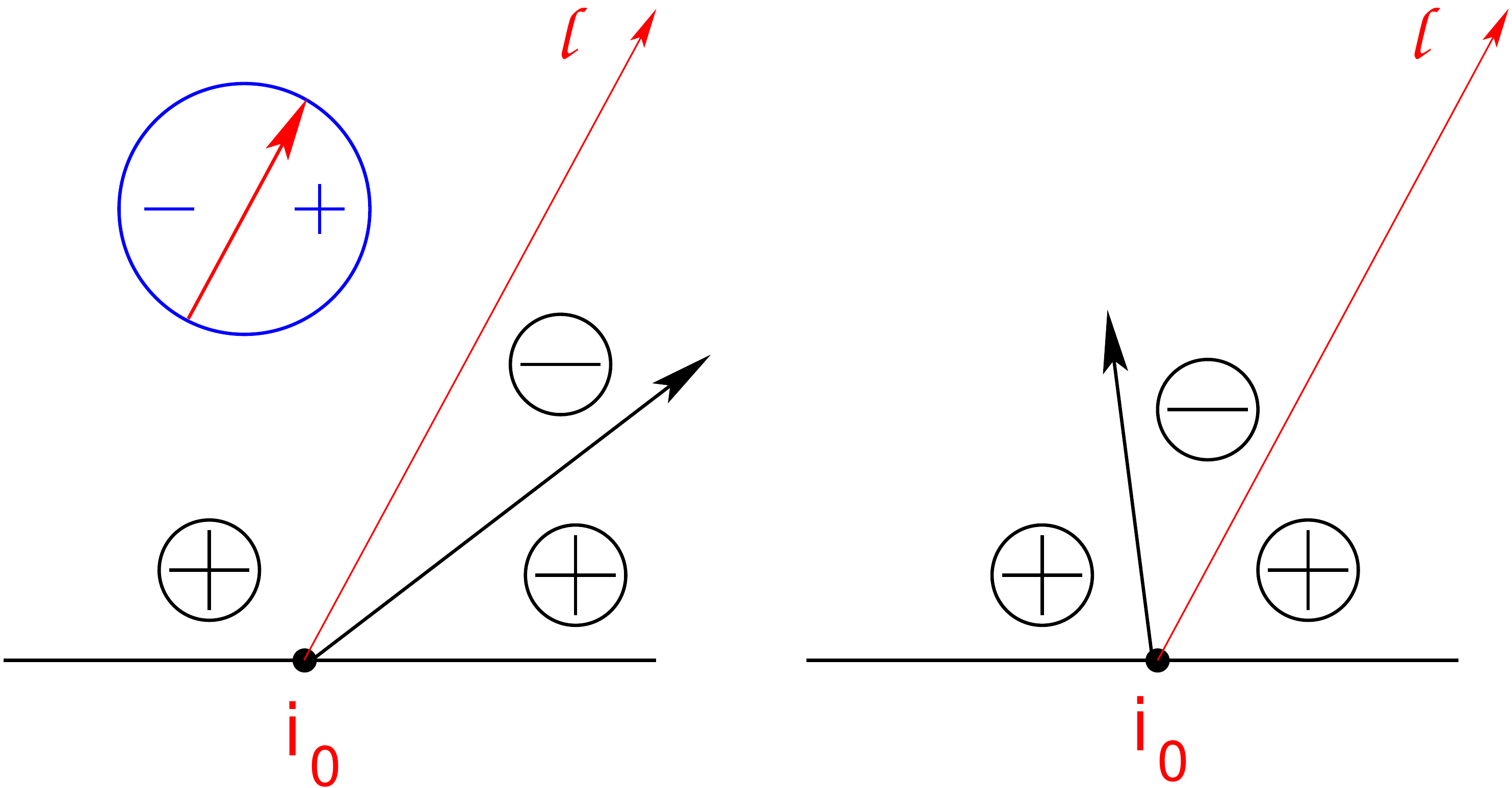}}
  \caption{\small{\sl The rule of the sign at $e_{i_0}$.}}\label{fig:edgei_0}
\end{figure}

The effect of a change of orientation along a closed simple path $\mathcal Q_0$ on the system of edge vectors follows along similar lines as above. 

 \begin{lemma}\label{lemma:cycle}\textbf{The effect of a change of orientation along a simple closed cycle.}
Let $I = \{ 1\le i_1 < i_2 < \cdots < i_k\le n\}$ be the pivot indices for $({\mathcal N},\mathcal O,\mathfrak l)$.
Let $E_e$ be the system of vectors on $({\mathcal N},\mathcal O,\mathfrak l)$ satisfying the boundary conditions $E_j$ at the boundary sinks $j\in \bar I$. Assume that we change the orientation along a simple closed cycle $\mathcal Q_0$ and let $({\mathcal N},{\hat {\mathcal O}},\mathfrak l)$ be the newly oriented network. 

Then, the system of edge vectors 
\begin{equation}\label{eq:hat_E_Q}
{\hat E}_e = \left\{ \begin{array}{ll}
 (-1)^{\epsilon(e)} E_e, & \mbox{ if } e\not \in \mathcal Q_0, \mbox{ with } \epsilon(e) \mbox{ as in (\ref{eq:eps_not_path})},\\
\displaystyle \frac{(-1)^{\epsilon(e)}}{w_e} E_e, & \mbox{ if } e\in \mathcal Q_0, \mbox{ with } \epsilon(e) \mbox{ as in (\ref{eq:eps_on_path})}.
\end{array}\right.
\end{equation}
is the system of vectors on the network  $({\mathcal N},{\hat {\mathcal O}},\mathfrak l)$ satisfying the same boundary conditions $E_j$ at the boundary sinks $j\in \bar I$.
\end{lemma}
The proof follows again from the Lemmas in Appendix~\ref{app:orient}.

\subsection{The dependence of the edge vectors on the weight and the vertex gauges freedom}\label{sec:different_gauge}

\begin{figure}
  \centering{\includegraphics[width=0.48\textwidth]{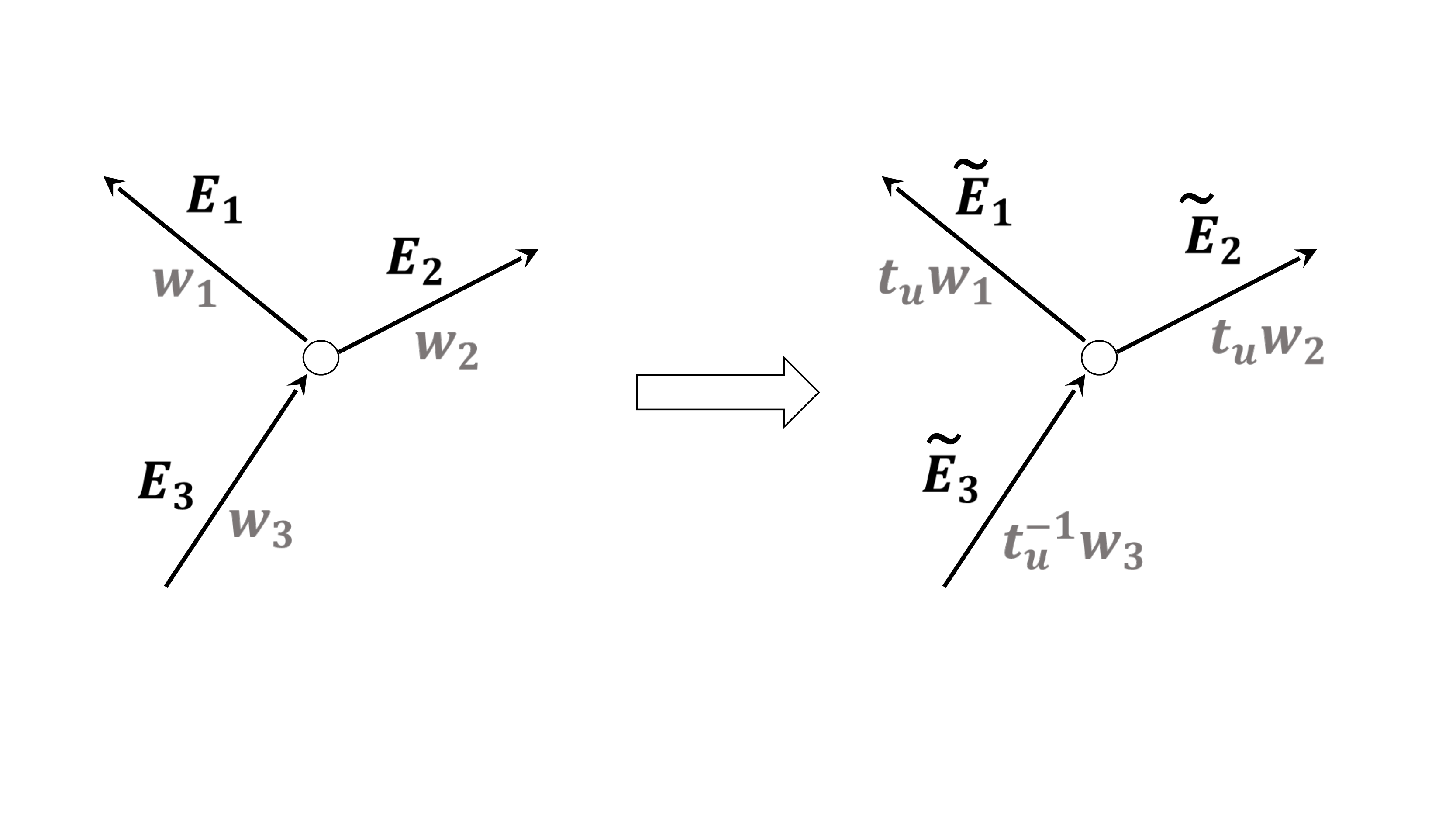}
	\hfill
	\includegraphics[width=0.48\textwidth]{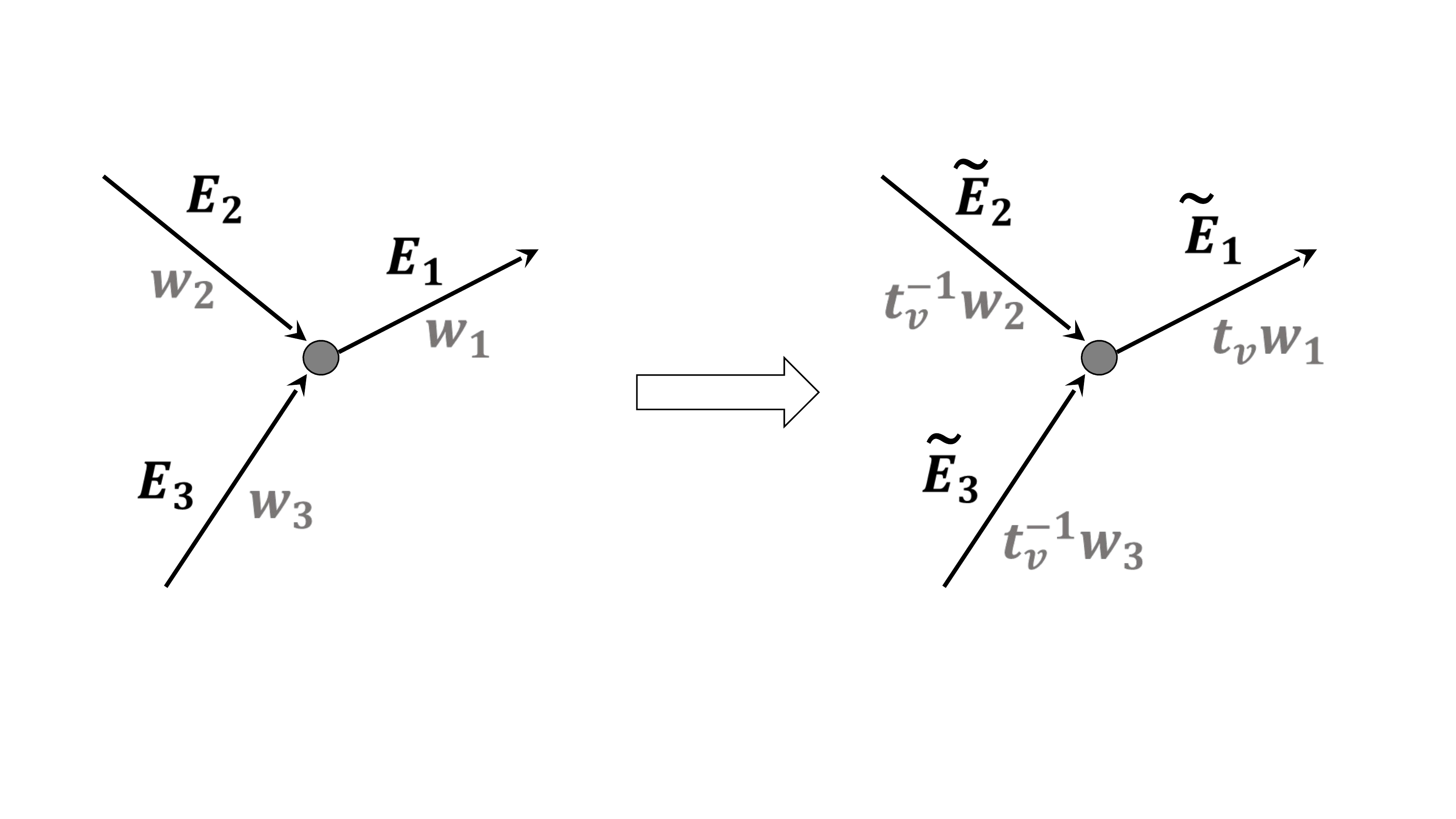}}
	\vspace{-1.5 truecm}
  \caption{\small{\sl The effect of the weight gauge transformation at a white [left] and at a black [right] vertex on the edge vectors.}\label{fig:weight_gauge_vectors}}
\end{figure}

In this Section we discuss the effect of the weight gauge and vertex gauge feeedom on the system of edge vectors. Both effect the edge vectors only locally.

\begin{lemma}\label{lem:weight_gauge}\textbf{Dependence of edge vectors on the weight gauge}
Let $E_e$ be the edge vectors on the network $({\mathcal N}, {\mathcal O}, \mathfrak{l})$.
\begin{enumerate}
\item Let ${\tilde E}_e$ be the system of edge vectors on $(\tilde {\mathcal N}, {\mathcal O}, \mathfrak{l})$, where $\tilde {\mathcal N}$ is obtained from ${\mathcal N}$ applying the weight gauge transformation at a white trivalent vertex as in Figure \ref{fig:weight_gauge_vectors} [left]. Then 
\begin{equation}\label{eq:wg_vector_white}
{\tilde E}_e = \left\{ \begin{array}{ll} E_e, &\quad  \forall e\in \mathcal N, \;\; e\not = e_1,e_2,\\
t_u E_e &\quad \mbox{ if } e = e_1,e_2.
\end{array}
\right.
\end{equation}
\item Let ${\tilde E}_e$ be the system of edge vectors on $(\tilde {\mathcal N}, {\mathcal O}, \mathfrak{l})$, where $\tilde {\mathcal N}$ is obtained from ${\mathcal N}$ applying the weight gauge transformation at a black trivalent vertex as in Figure \ref{fig:weight_gauge_vectors} [right]. Then
\begin{equation}\label{eq:wg_vector_black}
{\tilde E}_e = \left\{ \begin{array}{ll} E_e, &\quad \forall e\in \mathcal N, \;\;  e\not = e_1,\\
t_v E_e &\quad \mbox{ if } e = e_1.
\end{array}
\right.
\end{equation}
\end{enumerate}
\end{lemma}

The proof is straightforward and is omitted.

The vertex gauge freedom in the graph does do not effect the curve $\Gamma$.
Any such transformation may be decomposed in a sequence of elementary transformations in which a single vertex is moved whereas all other vertices remain fixed (see also Figure \ref{fig:vertex_gauge_vectors}).
This transformation effects only the three edge vectors incident at the moving vertex and the latter may only change of sign. 

\begin{lemma}\label{lem:vertex_gauge}\textbf{Dependence of edge vectors on the vertex gauge}
\begin{enumerate}
\item Let $E_e$ and ${\tilde E}_e$ respectively be the system of edge vectors on $({\mathcal N}, {\mathcal O}, \mathfrak{l})$ and on $({\tilde {\mathcal N}}, {\mathcal O}, \mathfrak{l})$, where ${\tilde {\mathcal N}}$ is obtained from ${\mathcal N}$ moving one internal white vertex as in Figure \ref{fig:vertex_gauge_vectors}[left]. Then ${\tilde E}_e =E_e$, for all $e\not = e_1,e_2,e_3$ and
\begin{equation}\label{eq:white_vertex_gauge}
{\tilde E}_{e_i} = (-1)^{\mbox{wind}({\tilde e}_i,f_i)- \mbox{wind}(e_i,f_i)+\mbox{int}({\tilde e}_i)-\mbox{int}(e_i)} E_{e_i}, \quad i=1,2,
\quad\quad
{\tilde E}_{e_3} = (-1)^{\mbox{wind}(f_3,{\tilde e}_3)-\mbox{wind}(f_3,e_3)  } E_{e_3};
\end{equation}
\item Let $E_e$ and ${\tilde E}_e$ respectively be the system of edge vectors on $({\mathcal N}, {\mathcal O}, \mathfrak{l})$ and on $({\tilde {\mathcal N}}, {\mathcal O}, \mathfrak{l})$, where ${\tilde {\mathcal N}}$ is obtained from ${\mathcal N}$ moving one internal black vertex as in Figure \ref{fig:vertex_gauge_vectors}[right]. Then ${\tilde E}_e =E_e$, for all $e\not = e_1,e_2,e_3$ and
\[
{\tilde E}_{e_1} = (-1)^{\mbox{wind}({\tilde e}_1,f_1)- \mbox{wind}(e_1,f_1)+\mbox{int}({\tilde e}_1)-\mbox{int}(e_1)} E_{e_1},
\quad\quad
{\tilde E}_{e_i} = (-1)^{\mbox{wind}(f_i,{\tilde e}_i)-\mbox{wind}(f_i,e_i)  } E_{e_i}, \quad i=2,3;
\]
\end{enumerate}
\end{lemma}

\begin{figure}
  \centering{\includegraphics[width=0.45\textwidth]{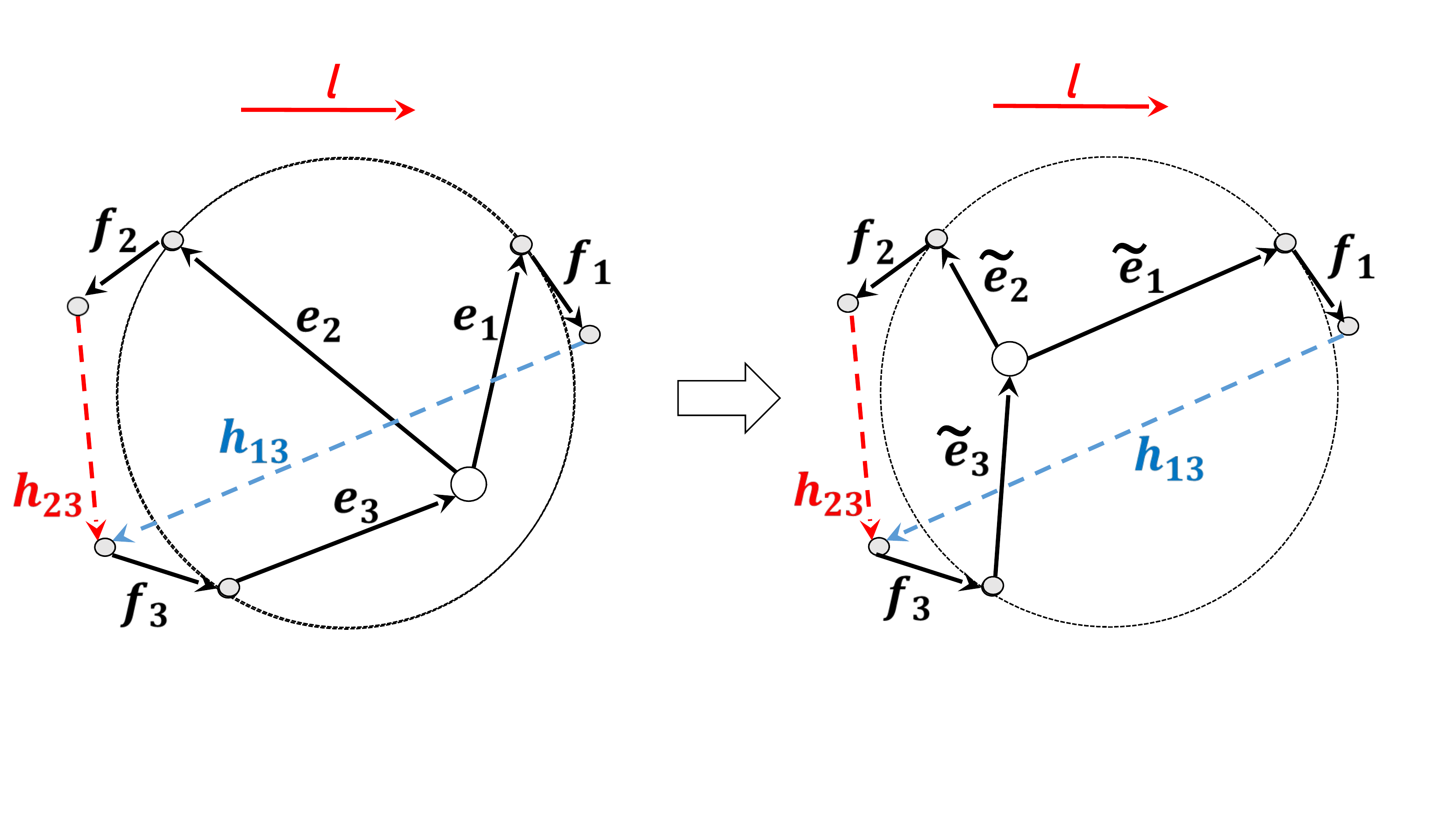}
	\hfill
	\includegraphics[width=0.45\textwidth]{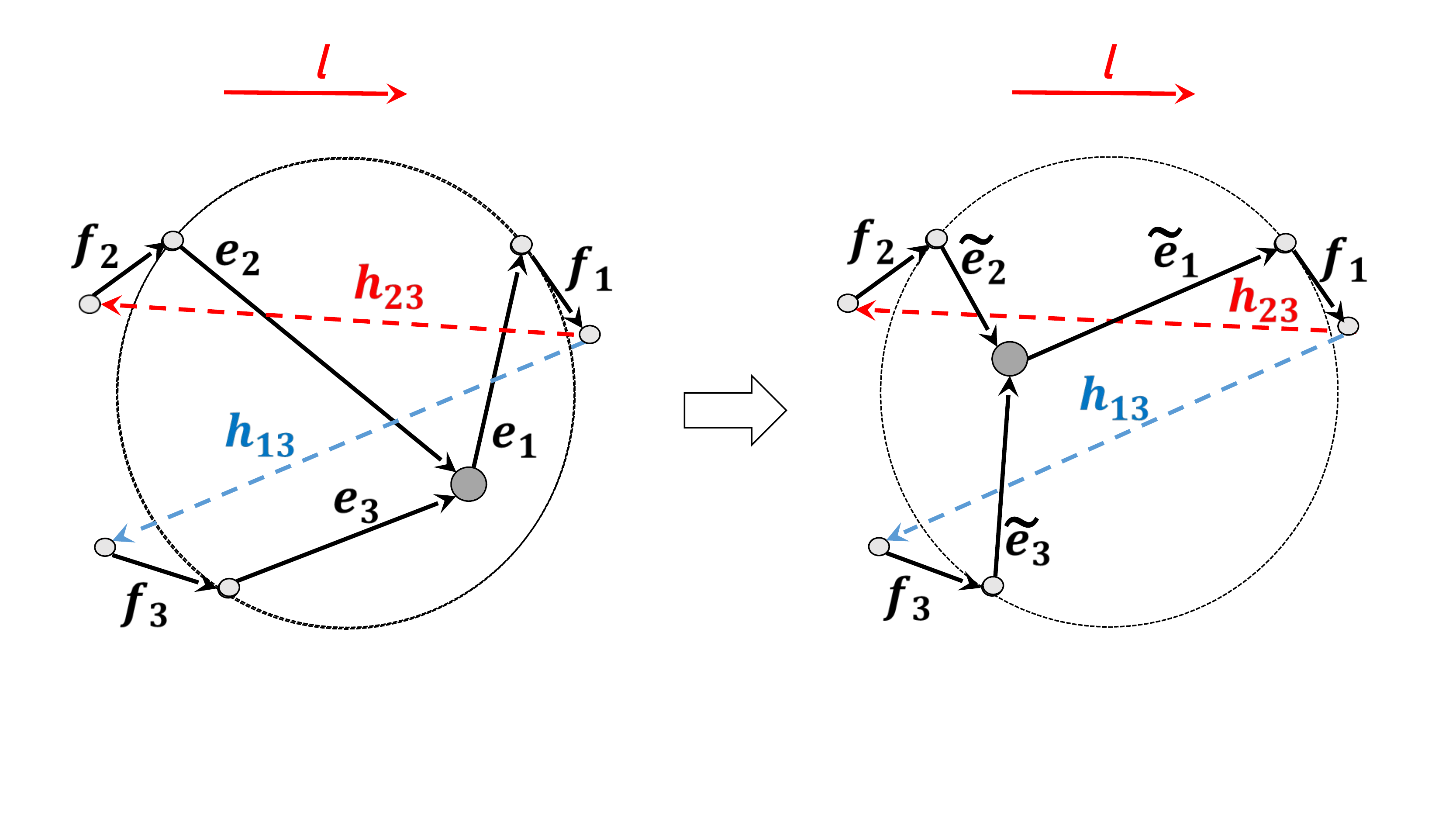}}
	\vspace{-.9 truecm}
  \caption{\small{\sl The effect of the vertex gauge transformation at a white [left] and at a black [right] vertex on the winding numbers of the edge vectors.}\label{fig:vertex_gauge_vectors}}
\end{figure}

\begin{proof}
The statement follows from the linear relations at vertices and the following identities at trivalent white vertices
\begin{equation}\label{eq:int_vertex_gauge}
\begin{array}{c}
\mbox{int} (e_i) +\mbox{int} (e_3) =\mbox{int} ({\tilde e}_i) +\mbox{int} ({\tilde e}_3), \quad (\!\!\!\!\!\!\mod 2),\\
\mbox{wind}(f_3,e_3) +  \mbox{wind}(e_3,e_i) +  \mbox{wind}(e_i,f_i) = \mbox{wind}(f_3,{\tilde e}_3) +  \mbox{wind}({\tilde e}_3,{\tilde e}_i) +  \mbox{wind}({\tilde e}_i,f_i) \quad (\!\!\!\!\!\!\mod 2),
\end{array}
\end{equation}
and the corresponding identities at black vertices.
\end{proof}

\subsection{Null edge vectors}\label{sec:null_vectors}

Edge vectors associated to the boundary source edges are always not null by construction if the boundary source is non isolated. On the contrary, a component of a vector associated to an internal edge 
can be equal to zero even if the corresponding boundary sink can be reached from that edge. 
In Figure~\ref{fig:zero-vector}[left] we present such an example: the network represents the point $[ 2p/(1+p+q),1] \in Gr^{\mbox{\tiny TP}}(1,2)$: all weights are equal to 1 except for the two edges carrying the positive weights 
$p$ and $q$. 
\begin{figure}
  \centering{
	\includegraphics[width=0.49\textwidth]{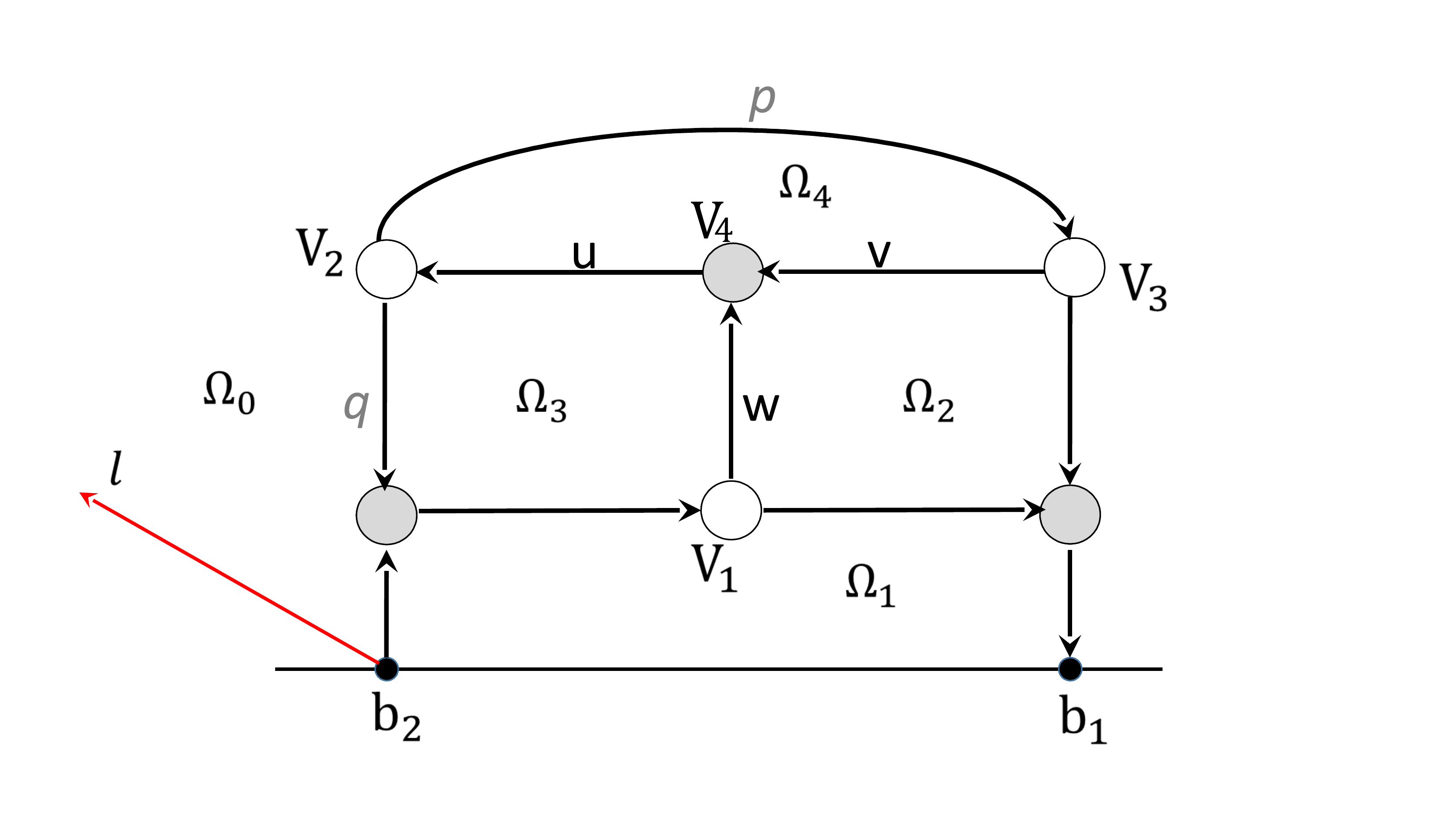}
	\hfill
	\includegraphics[width=0.49\textwidth]{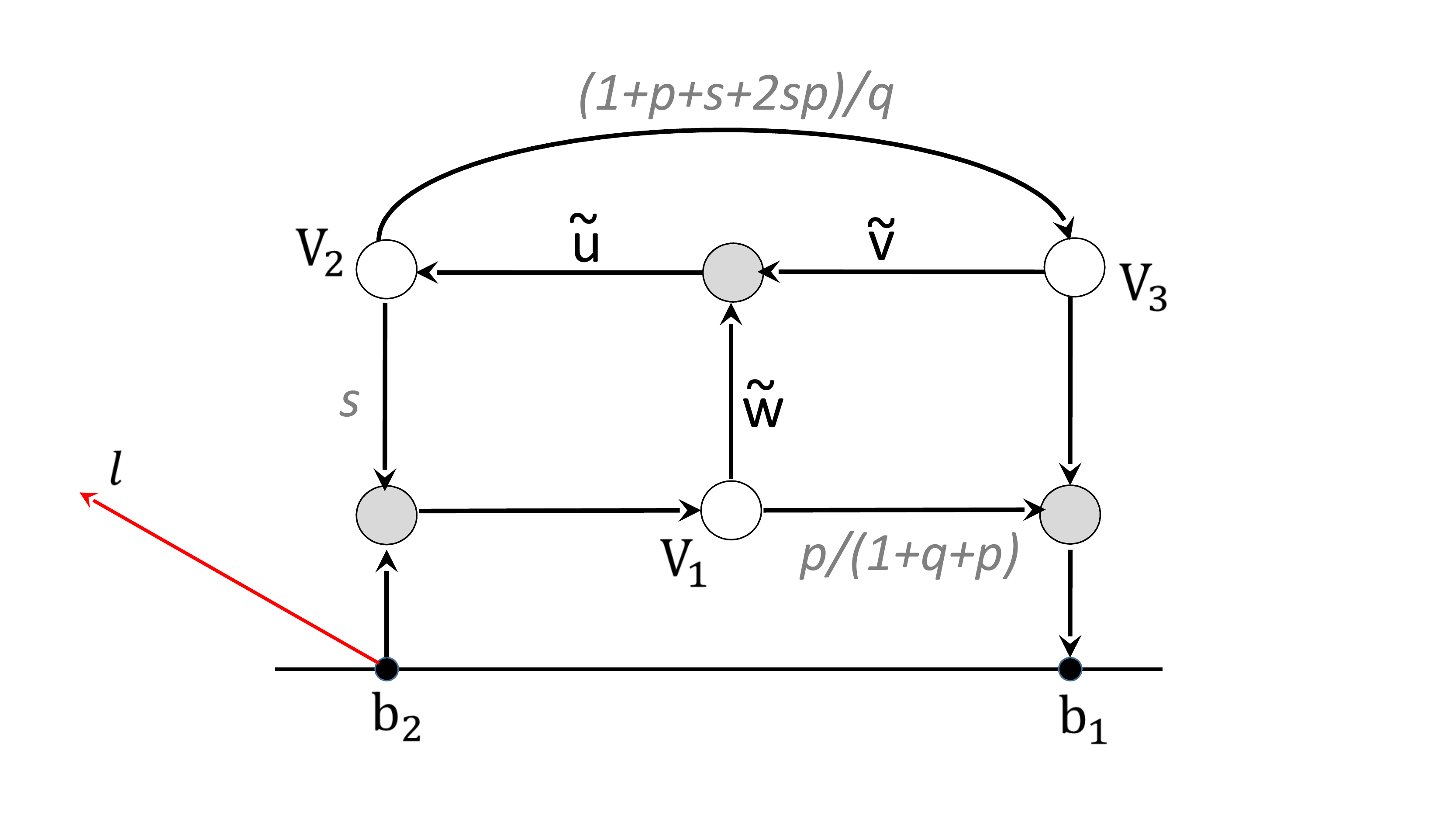}
  \caption{\small{\sl The appearance of null vectors on reducible networks [left] and their elimination using the 
	gauge freedom for unreduced graphs of Remark \ref{rem:gauge_freedom} [right].}}\label{fig:zero-vector}}
\end{figure}
On the edges $u,v,w$, using (\ref{eq:tal_formula}) we get
$$
E_u=E_v=-E_w=\left( \frac{q-p}{1+p+q}, 0\right).
$$
If $p=q$, then $E_u=E_v=E_w=(0,0)$. We remark that all other edge vectors associated to such network have non--zero first component for any choice of $p,q>0$. In particular the edge vector at the boundary source is equal to
$(\frac{1+2p}{1+p+q},0)$.

More in general, suppose that $E_e=0$ where $e$ is an edge ending at the vertex $V$. Then, if $V$ is black, all other edges at $V$ carry null vectors; the same occurs if $V$ is bivalent white. If $V$ is trivalent white, then the other edges at $V$ carry proportional vectors.

Of course, if the network possesses either isolated boundary sources or components isolated from the boundary, null-vectors are unavoidable. In \cite{AG3}, we provide a recursive construction of edge vectors for canonically oriented Le--networks and obtain as a by-product that null edge vectors are forbidden if the Le--network represents a point in an irreducible positroid cell. The latter property indeed is shared by systems of edge vectors on acyclically oriented networks as a consequence of the following Theorem:

\begin{theorem}
\label{thm:null_acyclic}
Let   $({\mathcal N}, {\mathcal O}, \mathfrak{l})$ be an acyclically oriented PBDTP network associated to a point in an irreducible positroid cell. Then all edge vectors $E_e$ are not-null.  

In particular Le-networks representing points in irreducible positroid cells do not possess null edge vectors.
\end{theorem}

\begin{proof}
By definition a PBDTP graph does not possess internal sources or sinks.
Let $e\in {\mathcal N}$ be an internal edge and let $b_j$ be a boundary sink such that there exists a directed path $P$ starting at $e$ and ending at $b_j$.  In this case, (\ref{eq:tal_formula}) in Theorem \ref{theo:null} simplifies to
\[
\left(E_{e}\right)_{j}= \displaystyle\sum\limits_{F\in {\mathcal F}_{e,b_j}(\mathcal G)} \big(-1\big)^{\mbox{wind}(F)+\mbox{int}(F)}\ w(F),
\]
where the sum runs on all directed paths $F$ starting at $e$ and ending at $b_j$. To complete the proof of the theorem we need to show that $\mbox{wind}(F)+\mbox{int}(F)$ has the same parity for all directed paths from $e$ to $b_j$. Thanks to acyclicity, the latter statement follows adapting its proof for boundary edges in Corollary \ref{cor:bound_source} and the observation that there exists a boundary source $b_i$ and a directed path $P_0$ from $b_i$ to $e$ such that any other directed path ${\tilde P}$ from $e$ to $b_j$ has a finite number of edges in common with $P$
and no edge in common with $P_0$.
Therefore 
\[
\begin{array}{c}
\mbox{int} (P_0)+ \mbox{int}( P) =\mbox{int} (P_0) + \mbox{int} ({\tilde P}) \quad (\!\!\!\!\!\!\mod 2),\\
\mbox{wind} (P_0) +\mbox{wind} (P) =\mbox{wind} (P_0\cup P) = \mbox{wind} (P_0\cup {\tilde P})
= \mbox{wind} (P_0) + \mbox{wind} ({\tilde P}) \quad (\!\!\!\!\!\!\mod 2),
\end{array}
\]
and the statement follows.
\end{proof}

The absence of null edge vectors for a given network is independent of its orientation and of the choices of the ray direction, of vertex gauge and of weight gauge, because of the transformation rules of edge vectors established in Sections \ref{sec:gauge_ray}, \ref{sec:orient} and \ref{sec:different_gauge}. 

In particular, if the edge vector $E_e =0$ for a given orientation $\mathcal O$ of the PBDTP network ${\mathcal N}$ representing $[A]\in \GTNN$, then it will remain null in any other orientation ${\mathcal O}^{\prime}$ of the network. Indeed in the new orientation such edge vector should be proportional to a linear combination of the rows of a representative matrix of $[A]$. However all such coefficients have to be null since the boundary conditions never contain the pivot vectors.

We summarize all the above properties of edge vectors in the following Proposition:

\begin{proposition}\label{prop:null_vectors}\textbf{Null edge vectors and changes of orientation, ray direction, weight and vertex gauges in $\mathcal N$.}
Let $E_e$ be the edge vector system on $({\mathcal N}, {\mathcal O}, \mathfrak{l})$ for given boundary conditions at the boundary sink vertices and suppose that $E_e\not =0$ for all $e\in {\mathcal N}$. Then the transformed edge vector system ${\tilde E}_e\not =0$ for all $e\in {\mathcal N}^{\prime}$ where  ${\mathcal N}^{\prime}$ is obtained from ${\mathcal N}$ changing either its orientation or its gauge ray direction or the weight gauge or the vertex gauge. 

Finally if the edge vector $E_e=0$ on the PBDTP network $({\mathcal N}, \mathcal O, \mathfrak l)$ for a given orientation $\mathcal O$, then it is null in any other orientation of ${\mathcal N}$.
\end{proposition}

In Section~\ref{sec:moves_reduc}, we discuss the effect of Postnikov moves and reductions on the transformation rules of edge vectors on equivalent networks. In particular, moves (M1), (M2)-flip and (M3) preserve both the PBDTP class and the acyclicity of the network. Then, in view of Theorem~\ref{thm:null_acyclic} we have the following result:

\begin{proposition}\label{prop:Le_net}\textbf{Absence of null vectors for PBDTP networks equivalent to the Le-network.}
Let the positroid cell $\S$ be irreducible and let the PBDTP network  $({\mathcal N},{\mathcal O},\mathfrak l)$ represent a point in  $\S$ and be equivalent to the Le-network via a finite sequence of moves (M1), (M3) and flip moves (M2). Then ${\mathcal N}$ does not possess null edge vectors.
\end{proposition}

Null-vectors may just appear in reducible networks representing irreducible positroid cells as in the example of Figure \ref{fig:zero-vector}. We plan to discuss thoroughly the mechanism of creation of null edge vectors in a future publication.
Here we just classify null edge vectors on PBDTP networks into two types.
Indeed edges carrying null vectors are contained in connected maximal subgraphs such that every edge belonging to one such subgraph carries a null vector and all edges belonging to its complement and having a vertex in common with it carry non zero vectors. For instance in the case of Figure \ref{fig:zero-vector}[left] there is one such subgraph and it consists of the edges $u,v,w$ and the vertices $V_1$, $V_2$, $V_3$ and $V_4$.

\begin{definition}\label{def:null_type_1_2}\textbf{Null edge vectors of type 1 and type 2}
Let the PBDTP network $({\mathcal N},{\mathcal O},\mathfrak l)$ of graph $\mathcal G$ represent a point in the irreducible positroid cell $\S$. Let $\mathcal G_0\subset \mathcal G$ be a connected maximal subgraph carrying null edge vectors.

Then we say that the edges in $\mathcal G_0$ are of type 1 if there exists a vector $E$ such that $E_f=c_fE$ for any edge $f\in \mathcal G\backslash \mathcal G_0$ having a vertex in common with $\mathcal G_0$, for some $c_f\ne 0$.

Otherwise we say that the edges in $\mathcal G_0$ are of type 2.
\end{definition}

In the case of Figure \ref{fig:zero-vector}[left] the null edge vectors are of type 1. An example of edge vector of type 2 
occurs in Figure \ref{fig:div_null_2} when $E_3$ and $E_1$ are not proportional to each other. 
\begin{figure}
  \centering{\includegraphics[width=0.7\textwidth]{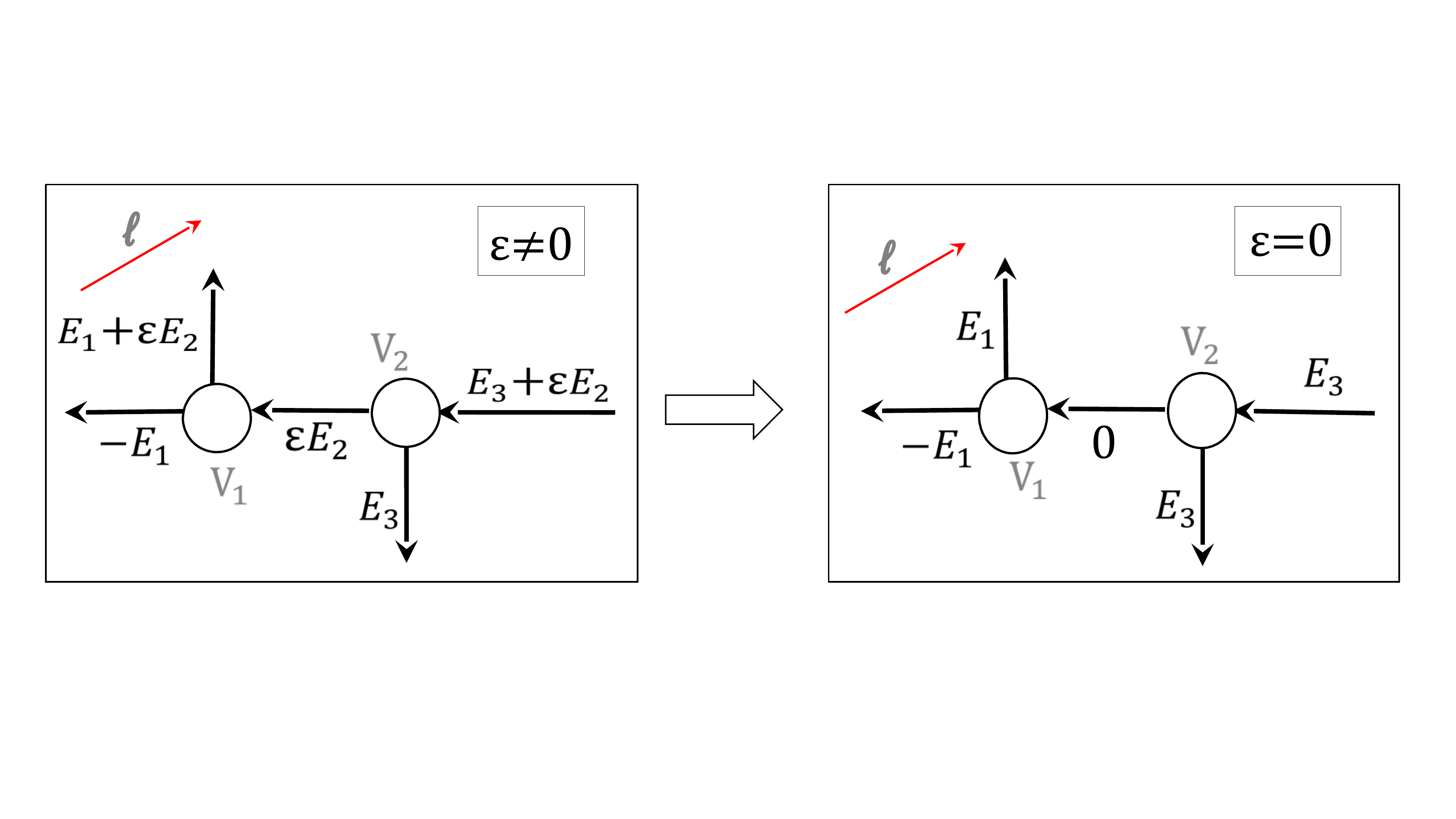}}
	\vspace{-2. truecm}
  \caption{\small{\sl If $E_1$ and $E_3$ are proportional the edge carrying the null vector when $\epsilon=0$ is of type 1, otherwise of type 2.}\label{fig:div_null_2}}
\end{figure}
If the network carries only null vectors of type 1 the construction of the divisor and the KP wave function can be carried out as in the following Section. On the contrary, in case of null vectors of type 2 we need to modify the PBDTP network, the curve to carry out the construction of a divisor and meromorphic wave function for the given soliton data (see Section \ref{sec:constr_null} for an example).

Finally, null vectors may just appear in reducible networks where we have extra freedom in fixing the edge weights (see Remark \ref{rem:gauge_freedom}). Therefore we conjecture that we may always choose the weights on reducible networks representing a given point so that all edge vectors are not null.

\begin{conjecture}\textbf{Elimination of null vectors on reducible PBDTP networks}\label{conj:null}
Let $({\mathcal N}, \mathcal O, \mathfrak l)$ be a reducible PBDTP network representing a given point $[A]\in\S \subset \GTNN$ for some irreducible positroid cell and such that it possesses a finite number of edges $e_1,\dots, e_s$ carrying null vectors,
$E_{e_l} = 0$, $l\in [s]$. Then using the gauge freedom for unreduced graphs of Remark \ref{rem:gauge_freedom}, we may always change the weights on $\mathcal N$ so that the resulting network $({\tilde {\mathcal N}}, \mathcal O, \mathfrak l)$ still represents $[A]$ and the edge vectors ${\tilde E}_e \not =0$, for all $e\in {\tilde {\mathcal N}}$.
\end{conjecture}

The example in Figure \ref{fig:zero-vector} satisfies the conjecture. Both networks represent the same point 
$[ 2p/(1+p+q),1] \in Gr^{\mbox{\tiny TP}}(1,2)$, but for the second network all edge vectors are different from zero. Indeed for any choice of $s>0$, using (\ref{eq:tal_formula}), we get:
$$
E_{\tilde w}=-E_{\tilde u}=-E_{\tilde v}=\left(\frac{1+p}{1+p+q},0\right).
$$

\section{Construction of the KP divisor on $\Gamma(\mathcal G)$}\label{sec:anycurve}

Throughout this Section, we fix the phases ${\mathcal K}= \{ \kappa_1 <\cdots< \kappa_n\}$ and the PBDTP graph in the disk $\mathcal G$ representing the irreducible positroid cell $\S \subset \GTNN$ as in Definition \ref{def:graph}. $\Gamma =\Gamma(\mathcal G)$ is the curve as in Construction
\ref{def:gamma}. We denote by $g+1$ the number of faces (ovals) of $\mathcal G$ ($\Gamma$).
$\mathcal O$ is the orientation of $\mathcal G$ and $\mathfrak l$ a fixed gauge direction. Finally $I=\{ 1\le i_1<\cdots < i_k\le n\}$ is the base in $\mathcal M$ associated to $\mathcal O$, whereas
$\bar I= [n] \backslash I$.

For any $[A]\in \S$, we denote $({\mathcal N}, \mathcal O)$ a directed network of graph $\mathcal G$ representing such point. $E_e$ is the system of edge vectors constructed in the previous Section on $({\mathcal N}, \mathcal O, \mathfrak l)$ with boundary conditions $E_j$ at the boundary vertices. 
\textbf{We assume that all edge vectors on $({\mathcal N}, \mathcal O, \mathfrak l)$ are not null. The construction may be extended to the case in which some of the edge vectors are null possibly after a convenient modification of the curve}. We leave a thorough investigation of the case of networks carrying null vectors to a future publication and we just present an example of the procedure for null vectors of type 1 and of type 2 in Section \ref{sec:constr_null}. We recall that, if an edge vector is null in one orientation then it is null in any other orientation of $\mathcal N$ (Proposition \ref{prop:null_vectors}) and that null vectors can't appear in PBDTP networks possessing an acyclic orientation (Theorem \ref{thm:null_acyclic}).

Then, we construct a degree $g$ effective KP divisor $\DKP$ contained in the union of all the $g+1$ ovals of $\Gamma$ and associated to the soliton data $(\mathcal K, [A])$. In this Section we also prove that $\DKP$ is invariant with respect to changes of: 1) gauge ray direction, 2) orientation, 3) position of the Darboux points, 4) weight gauge, 5) vertex gauge.  $\DKP$ depends on the soliton data, on the initial time and on the graph used to construct the reducible rational curve. Moreover, if the graph is reducible, $\DKP$ also depends on the gauge freedom for unreduced graphs.

In the next Section, we complete the proof of Theorem \ref{theo:exist} on $\Gamma$ using the combinatorics of the network to count the number of divisor points in each oval. 

In principle the construction of $\DKP$ may be carried out also on \textbf{reducible positroid cells}. 
In such case, the divisor points on the components of $\Gamma(\mathcal G)$ corresponding to the subgraph of $\mathcal G$ containing either an isolated boundary source or an isolated boundary sink, are all trivial (they do not depend on times).
Similarly, the assumption that  $\mathcal G$  is a \textbf{connected graph} is not restrictive, since edge vectors are null on all edges of disconnected components in the graph, and both the vacuum and the KP wave functions are trivial on the corresponding isolated components of the reducible curve.

\begin{figure}
  \centering{\includegraphics[width=0.45\textwidth]{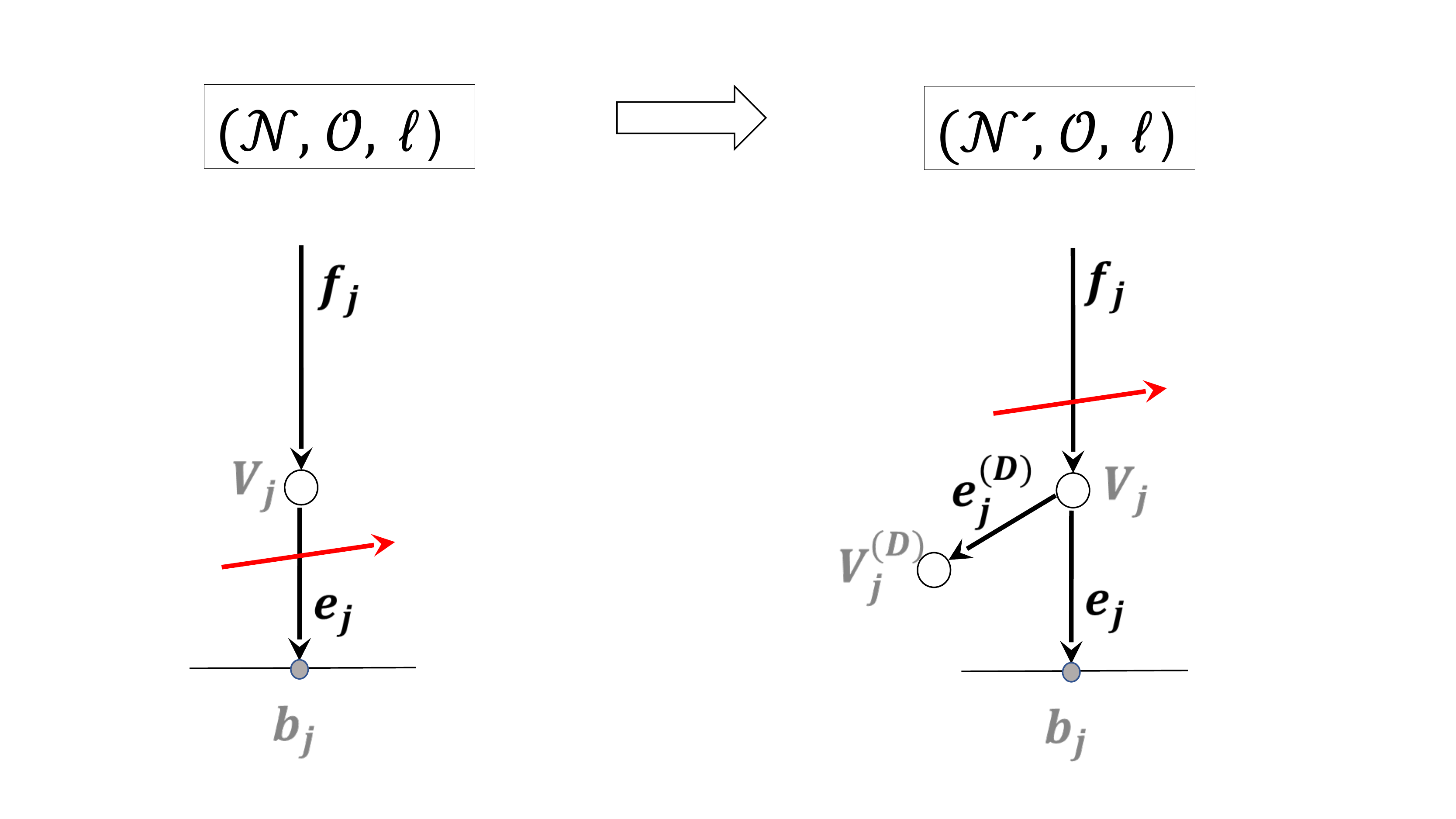}
	\hfill
	\includegraphics[width=0.45\textwidth]{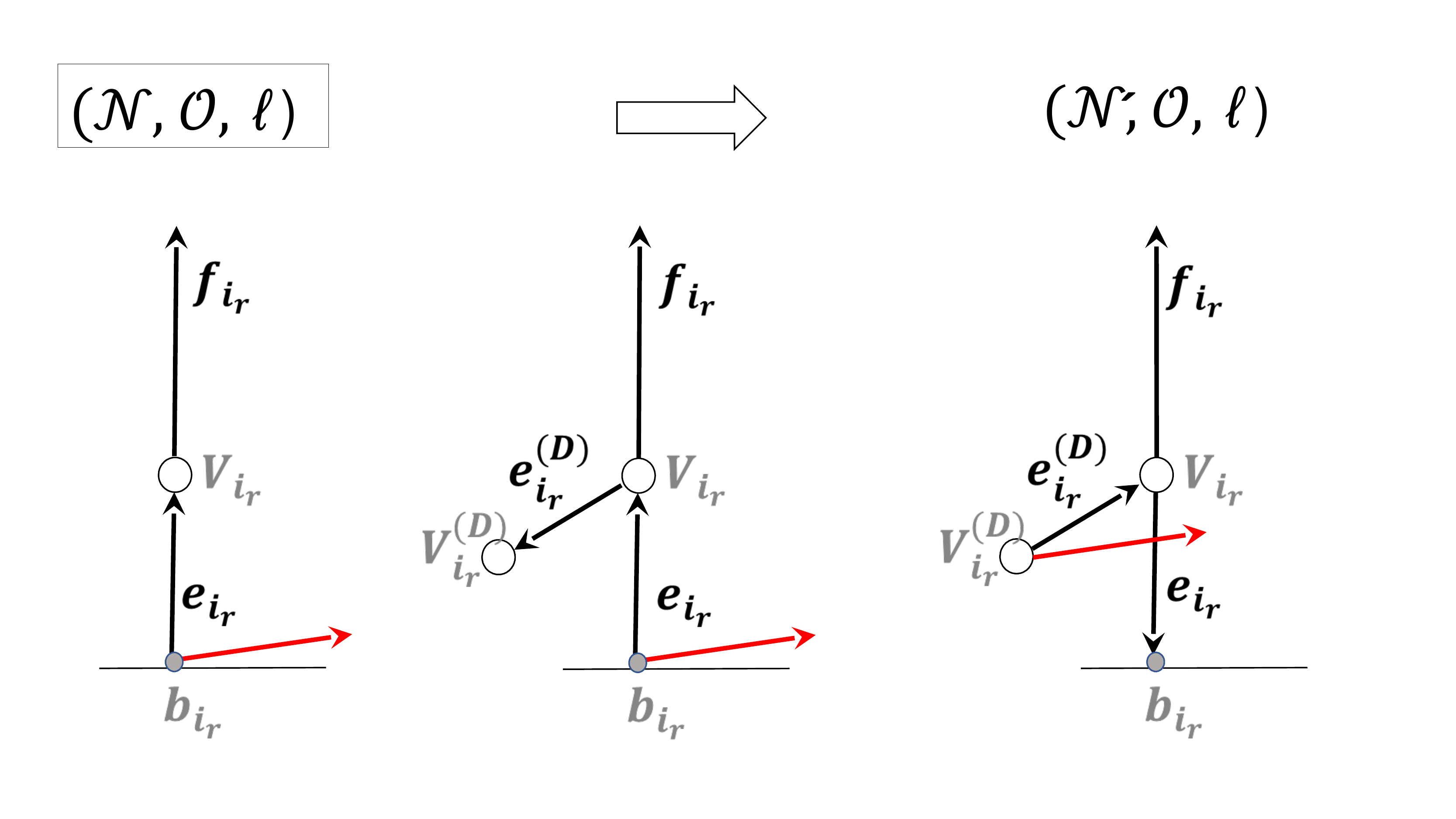}}
	\vspace{-.5 truecm}
  \caption{\small{\sl The insertion/removal of a Darboux vertex next to a boundary sink [left] and to a boundary source [right].}\label{fig:Darboux}}
\end{figure}

\textbf{Plan of the Section}: The construction of the KP divisor $\DKP$ and of the Baker--Akhiezer function $\hat \psi(P, \vec t)$ on $\Gamma(\mathcal G)$ is carried in several steps:
\begin{enumerate}
\item\textbf{The modified network ${\mathcal N}^{\prime}$.} First we modify $\mathcal N$ adding internal sources and sinks next to the boundary vertices and creating Darboux edges. The latter edges correspond to 
the Darboux sink/source marked points in $\Gamma$ as in \cite{AG1,AG2}. This transformation on the network is carried out using move (M2) \cite{Pos}, it transforms all boundary vertices into sinks and does not affect the curve $\Gamma$. It is a technical trick which simplifies the extension of the vacuum Sato wave function from $\Gamma_0$ to $\Gamma$;
\item\textbf{The edge wave functions} At the edges $e\in {\mathcal N}^{\prime}$ we define a non normalized vacuum edge wave function $\Phi_e (\vec t)$ and its dressing $\Psi_e (\vec t)$ using the system of vectors introduced in Section \ref{sec:vectors}. We then assign a degree $g+n-k$ vacuum network divisor $\DVN$ and a degree $g+n-k$ dressed network divisor $\DDN$ to ${\mathcal N}^{\prime}$ using the linear relations at the trivalent white vertices. 
We also explain the dependence of both the edge wave functions and of the network divisor on the gauge ray direction, the position of the Darboux edges, the weight gauge, the vertex gauge and the orientation of the network; 
\item\textbf{The vacuum divisor} Each vacuum network divisor number represents the coordinate of a vacuum divisor point on the corresponding component of $\Gamma$. Since the $n-k$ Darboux points associated to $\bar I$ may be interpreted as non effective divisor points, the degree $g$ vacuum divisor ${\mathcal D}_{\textup{\scriptsize vac},\Gamma}$ on $\Gamma$ is obtained from $\DVN$ eliminating the $n-k$ divisor points at the Darboux sink copies. By construction ${\mathcal D}_{\textup{\scriptsize vac},\Gamma}$ depends on the base in the matroid $\mathcal M$ associated to $\mathcal O$;
\item\textbf{The KP divisor on $\Gamma$} Similarly, each dressed network divisor number on a trivalent white vertex represents the coordinate of a dressed divisor point on the corresponding component of $\Gamma$. Since the $n$ Darboux points may be interpreted as non effective divisor points, the degree $g$ KP divisor $\DKP$ for the soliton data $(\mathcal K, [A])$ is then obtained from $\DDN$ eliminating the $n$ divisor points on the Darboux components and adding the degree $k$ Sato divisor, $\DS$.
By construction the KP divisor $\DKP = \DKP(\mathcal K, [A])$ is contained in  the union of the ovals of $\Gamma$ and it
does not depend on the gauge ray direction, on the position of the Darboux points, on the weight gauge and on the vertex gauge.
\item\textbf{The invariance of the KP divisor with respect to changes of base in ${\mathcal M}$ (Theorem \ref{theo:inv}).} On each component $\Gamma_l=\mathbb{CP}^1$ of $\Gamma(\mathcal G)$ corresponding to a trivalent white vertex $V_l$ of $\mathcal G$, we have four real ordered marked points corresponding to the edges at $V_l$ and to the dressed point divisor. A change of orientation of $\mathcal G$ induces a change in the local coordinate used on $\Gamma_l$, and leaves invariant the order and the position of the four points on $\Gamma_l$;
\item\textbf{The KP wave function on $\Gamma$} The normalized KP wave function $\hat \psi (P, \vec t)$ on $\Gamma$ is constructed imposing that it coincides for all times at each marked point $P$ with the value of the normalized dressed edge wave function $\hat \Psi_e (\vec t)$ at the corresponding edge $e=e(P)\in {\mathcal N}^{\prime}$. By construction it coincides with the Sato wave function on $\Gamma_0$ and, for any $\vec t$ its effective divisor of poles is $\DKP$: $\left(\hat \psi( \cdot, \vec t) \right) + \DKP\ge 0$.
\end{enumerate}

\begin{figure}
  \centering{
	\includegraphics[width=0.75\textwidth]{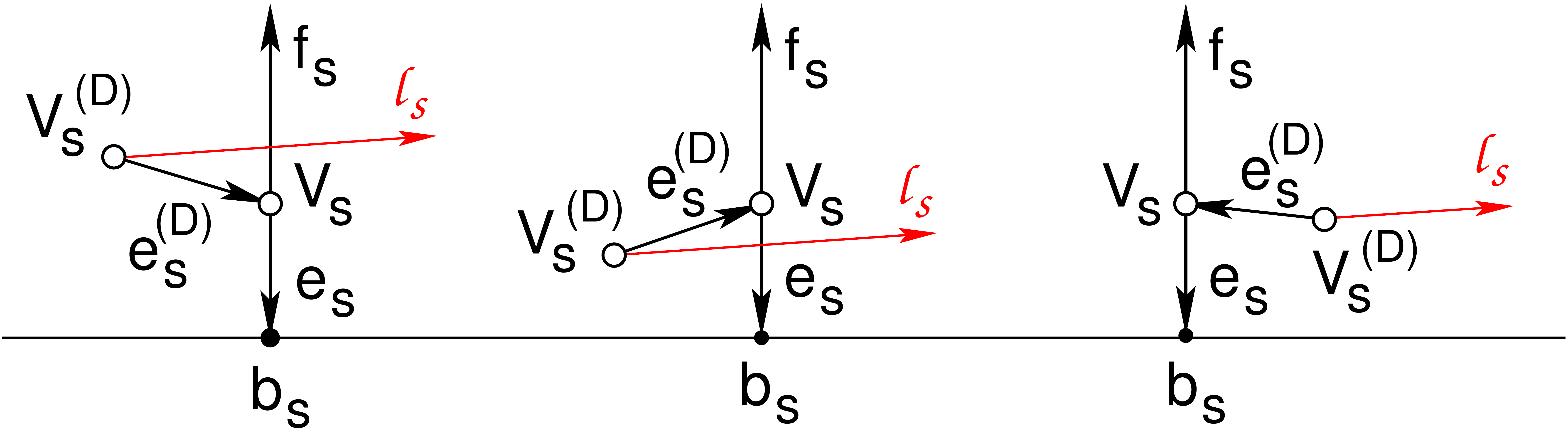}}
	\vspace{-.5 truecm}
  \caption{\small{\sl The possible configurations of the gauge ray ${\mathfrak l}_s$ intersections with $e_s$ and $f_s$.}\label{fig:Darboux1}}
\end{figure}

\subsection{The modified network ${\mathcal N}^{\prime}$}\label{sec:modN}

\begin{construction}\label{con:Nprime}\textbf{The modified network ${\mathcal N}^{\prime}$.}
Let $I$ be the base in $\mathcal M$ associated to the orientation of the PBDTP network $\mathcal N$ representing $[A]$. In the following we assume that the bivalent vertices $V_j$ are sufficiently near to the boundary vertices $b_j$, $j\in [n]$, and that the edges at  $V_j$ are all parallel.
We transform all boundary vertices of $\mathcal N$ into sinks by constructing the network  $({\mathcal N}^{\prime}, {\mathcal O}, \mathfrak l)$ where: 
\begin{enumerate}
\item ${\mathcal N}^{\prime}$ is the network obtained from ${\mathcal N}$ adding an edge with unit weight and a white vertex $V^{(D)}_{s}$ at each bivalent white vertex $V_{s}$ associated to the boundary vertex $b_s$, $s\in [n]$ (see Figure \ref{fig:Darboux}). Moreover, we assume that the  $V^{(D)}_{s}$ is sufficiently near to both $V_{s}$ and $b_s$;
\item $\mathcal O$ is the orientation of ${\mathcal N}^{\prime}$ such that
\begin{enumerate}
\item For any $s\in [n]$, the edge $e^{(D)}_s$ joining $V^{(D)}_{s}$ and $V_s$ ends at $V^{(D)}_{s}$ if $b_s$ is a boundary sink in $\mathcal N$, otherwise it starts at $V_s$;
\item For any $i\in I$, the edge $e_i\in \mathcal N$ starting at the boundary source $b_{i}$ is transformed into the edge $-e_i\in {\mathcal N}^{\prime}$ which ends at $b_{i}$;
\item All other edges in $\mathcal N^{\prime}$ are oriented as in $({\mathcal N},\mathcal O)$;
\end{enumerate}
\item The gauge direction is the same, but the line ${\mathfrak l}_{i}$ starts at $V^{(D)}_{i}$, for any $i\in I$. Since we assume that new vertices are close enough to boundary vertices, without loss of generality we assume that:
\begin{enumerate}
\item $\mbox{int}(e)$ is invariant for all edges $e$ except those at $V_i$, $i\in [n]$; 
\item If the gauge ray $\mathfrak{l}_s$ crosses $e_j$ in $\mathcal N$ for $j\ne s$, then it will cross $f_j$ in $\mathcal N$;
\item The gauge ray $\mathfrak{l}_s$ either intersects exactly one of the edges $e_s$, $f_s$ or none of them (see Figure~\ref{fig:Darboux1}).   
\end{enumerate}
\item We assign $E[s]$, the $s$--th vector of the canonical basis to the boundary vertex $b_s$, $s\in [n]$, both in $\mathcal N$ and $\mathcal N^{\prime}$. We assign the null vector to each internal Darboux source or sink in $\mathcal N^{\prime}$.
\end{enumerate}
\end{construction}

In Figure \ref{fig:Nprime}, we illustrate such modified network for the example of Figure \ref{fig:Rules0}.
We remark that the above construction is a generalization of the construction carried in \cite{AG3} in the special case of the Le--network oriented canonically. The system of vectors $E_e^{\prime}$ on ${\mathcal N}^{\prime}$ is related to the system $E_e$ on $\mathcal N$ as explained below.
\begin{figure}
  \centering{\includegraphics[width=0.6\textwidth]{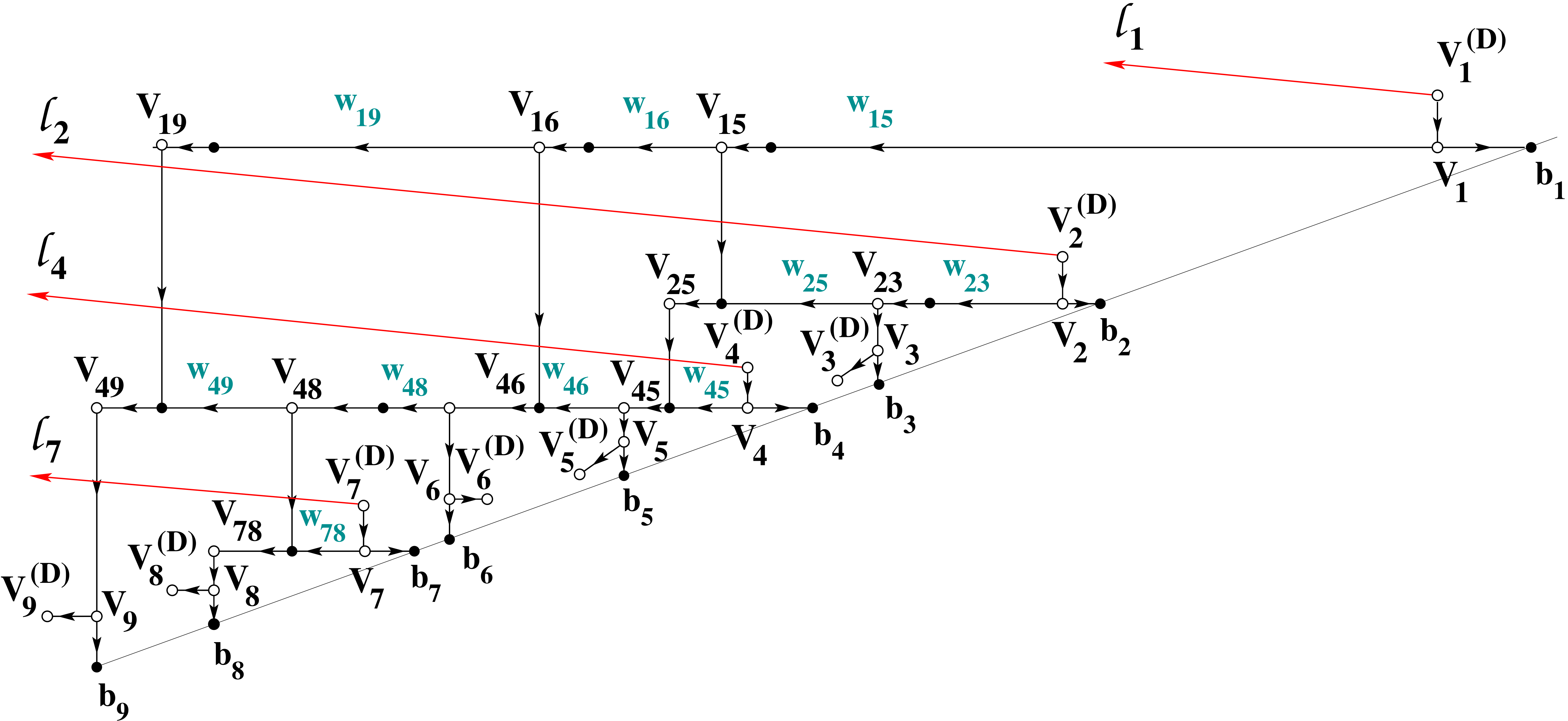}}
  \caption{\small{\sl The modified network $({\mathcal N}^{\prime}, {\mathcal O}, \mathfrak l)$ for the example of Figure \ref{fig:Rules0}. The choice of the position of the Darboux vertices is consistent with that of the Darboux points in Figure \ref{fig:net_curve}.}\label{fig:Nprime}}
\end{figure}
\begin{proposition}\label{prop:ext_syst}\textbf{Edge vectors on $({\mathcal N}^{\prime}, \mathcal O, \mathfrak{l})$.}
Let $(\mathcal N^{\prime}, \mathcal O, \mathfrak{l})$ be the transformed network obtained from $(\mathcal N, \mathcal O, \mathfrak{l})$ as described above.  Let us denote $E_e$ and $E_e^{\prime}$ the edge vector at the edge $e$ respectively in $\mathcal N$ and $\mathcal N^{\prime}$ as in Definition \ref{def:edge_vector}. Then
$E^{\prime}_e =E_e$, for any edge $e$ not ending or starting at $V_i$, $i\in [n]$,
\begin{equation}\label{eq:darboux_edge}
E^{\prime}_{e^{(D)}_{i_r}} = \left\{ \begin{array}{ll}
A[r], & \, \mbox{ if } i_r \in I,\\
0, & \, \mbox{ otherwise},
\end{array}
\right.
\end{equation}
\begin{equation}\label{eq:darboux_source}
E^{\prime}_{e_{i_r}} = \left\{ \begin{array}{ll}
(-1)^{\mbox{int}(V^{(D)}_{i_r},e_{i_r})} E [i_r], & \, \mbox{ if } i_r \in I,\\
(-1)^{\mbox{int}(e_{i_r})} E_{e_{i_r}}=E [i_r], & \, \mbox{ otherwise},
\end{array}
\right.
\end{equation}
where $\mbox{int}(V^{(D)}_{i_r},e_{i_r})$ is 1 if the gauge ray $\mathfrak{l}_{i_r}$ intersects the edge $e_{i_r}$ in $\mathcal N^{\prime}$, otherwise is 0,
and
\begin{equation}\label{eq:darboux_edge_2}
E^{\prime}_{f_{i_r}} = \left\{ \begin{array}{ll}
(-1)^{\mbox{int}(V^{(D)}_{i_r},f_{i_r})+\mbox{int}(e_{i_r})} E_{f_{i_r}}, & \, \mbox{ if } i_r \in I,\\
(-1)^{\mbox{int}(e_{i_r})}E_{f_{i_r}}, & \, \mbox{ otherwise},
\end{array}
\right.
\end{equation}
where $\mbox{int}(e_{i_r})$ is the number of intersections of gauge rays with $e_{i_r}\in \mathcal N$,  and $\mbox{int}(V^{(D)}_{i_r},f_{i_r})$ is 1 if the gauge ray $\mathfrak{l}_{i_r}$ intersects the edge $f_{i_r}$ in $\mathcal N^{\prime}$, otherwise is 0.
\end{proposition}

\begin{proof}
All edges vectors in $\mathcal N^{\prime}$ except those at the edges at $V_l$, $l\in [n]$, are not affected by the insertion/removal of a Darboux source/sink vertices, because of the assumptions on the boundary conditions and the position of the vertices 
$V_l$, $V^{(D)}_l$ sufficiently near to the boundary. 

If $i_r \in \bar I$, (\ref{eq:darboux_source}) and (\ref{eq:darboux_edge_2}) follow from the definition of edge vector and the fact that all the intersections of gauge rays with $e_{i_r}$ in $\mathcal N$ are transformed into intersections with $f_{i_r}$ in $\mathcal N^{\prime}$ since gauge rays start from Darboux source vertices in the latter. Moreover at Darboux sources, $\mathfrak{l}_{i_r}$ in $\mathcal N^{\prime}$ $\mathfrak{l}_{i_r}$ may intersect either $e_{i_r}$ or $f_{i_r}$ or none of them depending on the relative positions of vertices and the gauge ray direction. Therefore (\ref{eq:darboux_source}) and (\ref{eq:darboux_edge_2})
hold true also for $i_r\in I$. 

Finally let us prove (\ref{eq:darboux_edge}) for $i_r\in I$. By definition
\[
E_{e_{i_r}^{(D)}}^{\prime} = (-1)^{\mbox{wind}(e_{i_r}^{(D)}, f_{i_r})} E_{f_{i_r}}^{\prime} + (-1)^{\mbox{wind}(e_{i_r}^{(D)},e_{i_r})} E_{ e_{i_r}}^{\prime}.
\]
By Corollary \ref{cor:bound_source}, in $\mathcal N$ we have
\[
A[r]-E[i_r] = (-1)^{\mbox{int}(e_{i_r})} E_{f_{i_r}}, \quad \forall i_r \in I.
\]
Therefore (\ref{eq:darboux_edge}) is equivalent to prove that
\[
\mbox{wind}(e_{i_r}^{(D)}, a) + \mbox{int}(V^{(D)}_{i_r},a) =0 \quad (\!\!\!\!\!\!\mod 2),  \quad\quad a=f_{i_r},e_{i_r}\\
\]
and the latter equations easily follow by the definition of winding number. 
\end{proof}

\begin{remark}\label{rem:genusGN}\textbf{The relation between trivalent white vertices in ${\mathcal N}$ and ${\mathcal N}^{\prime}$}	
Let ${\mathcal N}$ be a PBDTP network with $g+1$ faces representing $[A]$ with graph $\mathcal G$ satisfying Definition \ref{def:graph}. Let $t_W, t_B, d_W$ and $d_B$ respectively be the number of trivalent white, of trivalent black, of bivalent white, and of bivalent black internal vertices of ${\mathcal N}$. Then on $\mathcal N$,
$t_W = g-k$, whereas
${\mathcal N}^{\prime}$ has by construction the same number of faces as $\mathcal N$ and $t_W+n=g+n-k$ trivalent white vertices.
It is straightforward to prove that the linear system on $\mathcal N^{\prime}$ is compatible and of maximal rank.  
\end{remark}

\begin{corollary}\textbf{The effect of a change in the orientation or in the gauge ray direction on the system of vectors on ${\mathcal N}^{\prime}$}
The system of vectors on ${\mathcal N}^{\prime}$ changes in agreement with Proposition \ref{prop:rays} and Theorem \ref{theo:orient} when we change either the gauge ray direction or the orientation. 
\end{corollary}

\begin{remark}\label{rem:comp3}\textbf{Relation with the construction in \cite{AG1}}. 
For data $[A]\in \GTP$, the algebraic construction in \cite{AG1} of a system of vectors and coefficients is associated to a representative matrix of $[A]$ in banded form. Such algebraic construction may be interpreted in terms of addition of Darboux points at convenient internal vertices. Let ${\mathcal N}^{\prime}_T$ be the modified Le-network for $[A]$ with the standard orientation and use move (M2) to add a white vertex with unit edge and null vector at each bivalent white vertex $V_{i_rj_{n-k+r}}$, $r\in [k]$. For instance for the example in Figure \ref{fig:Gr35}[left] representing a point in $Gr^{\mbox{\tiny TP}}(3,5)$, the added edges are coloured grey. The augmented system of vectors preserves its compatibility. Then invert the direction of each horizontal path from 
$V_{i_rj_{n-k+r}}$ to $V_{i_r}$, taking the reciprocal of the weights on each edge which changes of direction, we move the gauge lines to the new Darboux point. The transformed system of vectors is compatible since it satisfies the linear system at each vertex; it is straightforward to show that it coincides with the one used in \cite{AG1} for the choice of gauge ray direction of Figure \ref{fig:Gr35}[right]. In particular, at the new Darboux points the edge vector coincides with the $r$--th row of the representative matrix in banded form.
\end{remark}

\begin{figure}
  \centering{\includegraphics[width=0.47\textwidth]{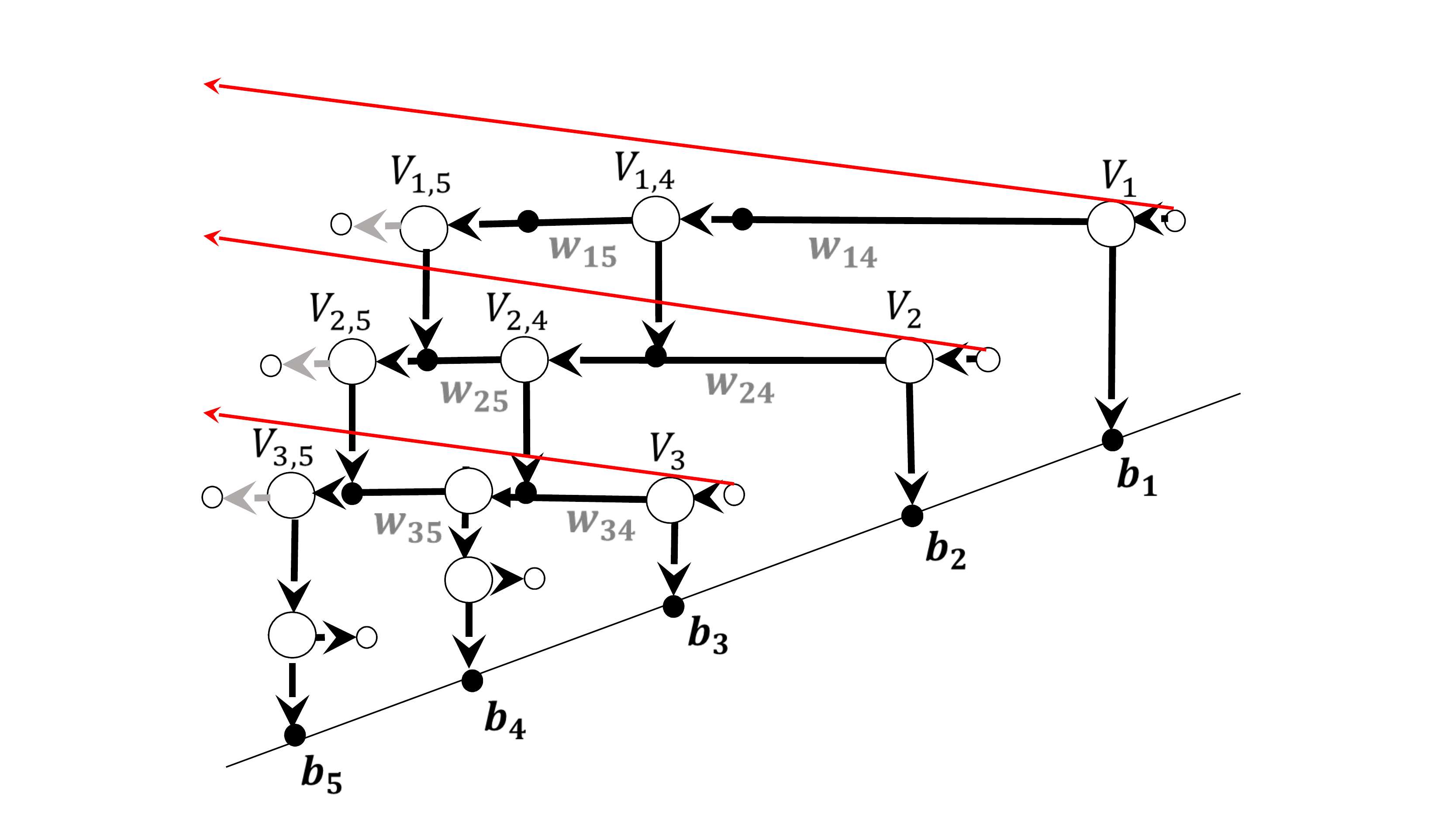}
	\hfill
	\includegraphics[width=0.47\textwidth]{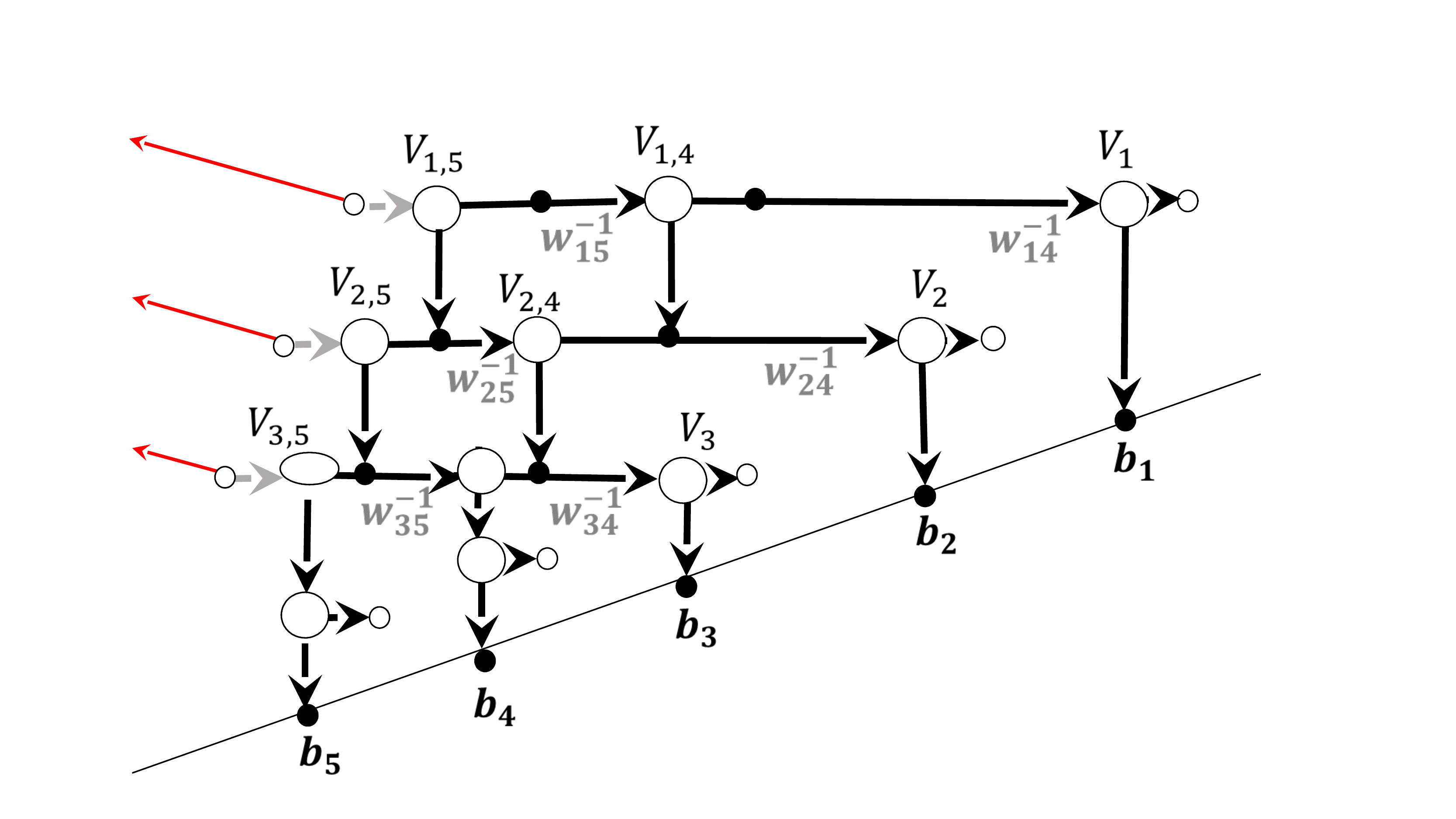}}
	\vspace{-.7 truecm}
  \caption{\small{\sl The graphical interpretation of the algebraic construction in \cite{AG1} in the case of $[A]\in Gr^{\mbox{\tiny TP}}(3,5)$. The network on the left is the one used in \cite{AG3} and provides RREF matrix at the Darboux source edges. The network on the right corresponds to the algebraic construction in \cite{AG1} and provides the matrix in banded form at the Darboux source edge. }\label{fig:Gr35}}
\end{figure}

\subsection{The vacuum and the dressed edge wave functions}\label{sec:vertex_wavefn_general_case}

In this Section we define the vacuum edge wave function and the dressed edge wave function on $({\mathcal N}^{\prime}, \mathcal O, \mathfrak l)$ and introduce an effective vacuum divisor and an effective dressed divisor. For the rest of the paper we denote ${\mathfrak E}_\theta (\vec t) = (e^{\theta_1(\vec t)}, \dots, e^{\theta_n(\vec t)})$, $\theta_j(\vec t) = \sum_{l\ge 1} \kappa_j^{l} t_l$, where $\vec t =(t_1=x,t_2=y,t_3=t, t_4,\dots)$ are the KP times, and 
 $\prec \cdot, \cdot\succ $ denotes the usual scalar product. Moreover we assume that only a finite number of entries of $\vec t$ are non zero.

\begin{definition}\label{def:vvw_gen}\textbf{The vacuum edge wave function (vacuum e.w.) and the dressed edge wave function (dressed e.w.) on $({\mathcal N}^{\prime}, \mathcal O, \mathfrak l)$.}
Let $({\mathcal N}^{\prime}, \mathcal O, \mathfrak l)$ be the modified oriented network of Construction \ref{con:Nprime} associated to the soliton data $(\mathcal K, [A])$ and let us denote $E_e^{\prime}$ its system of vectors as in Proposition \ref{prop:ext_syst}. Finally let $A$ be the 
RREF matrix of $[A]$ w.r.t. the base $I$ associated to the orientation $\mathcal O$, so that
the edge vectors at the Darboux sources are $E^{\prime}_{e_{i_r}^{(D)}} = A[r]$ and
\[
f^{(r)} (\vec t) = \sum_{j=1}^n A^r_j e^{\theta_j(\vec t)}, \quad\quad r\in [k],
\]
are the heat hierarchy solutions generating the Darboux transformation $\mathfrak D^{(k)}$ for the soliton data $(\mathcal K, [A])$.

We define the \textbf{vacuum edge wave function (vacuum e.w.)} on $({\mathcal N}^{\prime}, \mathcal O, \mathfrak l)$ as follows
\begin{equation}\label{eq:vvfN}
\Phi_{e, \mathcal O,\mathfrak l} (\vec t) \equiv \prec E_e^{\prime}, {\mathfrak E}_\theta (\vec t)\succ, \quad\quad e\in {\mathcal N}^{\prime}.
\end{equation}
In particular, using (\ref{eq:darboux_edge})--(\ref{eq:darboux_edge_2}), at Darboux source vertices we have
\begin{equation}\label{eq:vac_wf_source}
\begin{array}{l}
\Phi_{e^{(D)}_{i_r}, \mathcal O,\mathfrak l} (\vec t) =f^{(r)} (\vec t), \quad\quad
\Phi_{e_{i_r}, \mathcal O,\mathfrak l} (\vec t) = (-1)^{\mbox{int}(V^{(D)}_{i_r},e_{i_r})} e^{\theta_{i_r}(\vec t)},\\
\Phi_{f_{i_r}, \mathcal O,\mathfrak l} (\vec t) =(-1)^{\mbox{int}(f_j)} \left( f^{(r)} (\vec t) -e^{\theta_{i_r}(\vec t)} \right),
\end{array}
\quad\quad r\in k,
\end{equation}
whereas at Darboux sink vertices we have
\begin{equation}\label{eq:vac_wf_sink}
\Phi_{e^{(D)}_j, \mathcal O,\mathfrak l} (\vec t) \equiv 0, \quad\quad
\Phi_{e_j, \mathcal O,\mathfrak l} (\vec t) =e^{\theta_{j}(\vec t)},\quad\quad
\Phi_{f_j, \mathcal O,\mathfrak l} (\vec t) =(-1)^{\mbox{int}(f_j)}e^{\theta_{j}(\vec t)},
\quad\quad j\in \bar I,
\end{equation}
where $\mbox{int}(f_j)$ is the number of intersections of gauge rays with the edge $f_j$ in $\mathcal N^{\prime}$.

We also define the \textbf{dressed edge wave function (dressed e.w.)} on $({\mathcal N}^{\prime}, \mathcal O, \mathfrak l)$ as follows
\begin{equation}\label{eq:KPvfN}
\Psi_{e, \mathcal O,\mathfrak l} (\vec t) \equiv {\mathfrak D}^{(k)} \Phi_{e, \mathcal O,\mathfrak l} (\vec t).
\end{equation}
In particular, at all Darboux vertices we have $\Psi_{e^{(D)}_j, \mathcal O,\mathfrak l} (\vec t) \equiv 0$, for all $j\in [n]$, 
and
\begin{equation}\label{eq:KP_wf_source}
\begin{array}{lll}
\Psi_{e_j, \mathcal O,\mathfrak l} (\vec t) =(-1)^{\mbox{int}(V^{(D)}_{j},e_{j})} {\mathfrak D}^{(k)} e^{\theta_{j}(\vec t)},&\quad
\Psi_{f_j, \mathcal O,\mathfrak l} (\vec t) = (-1)^{1+\mbox{int}(f_j)}\Psi_{e_j, \mathcal O,\mathfrak l} (\vec t), 
&\quad\mbox {if } j\in I, 
\\
\Psi_{e_j, \mathcal O,\mathfrak l} (\vec t) ={\mathfrak D}^{(k)} e^{\theta_{j}(\vec t)}, &\quad
\Psi_{f_j, \mathcal O,\mathfrak l} (\vec t) = (-1)^{\mbox{int}(f_j)}\Psi_{e_j, \mathcal O,\mathfrak l} (\vec t), 
&\quad\mbox {if } j\in \bar I.
\end{array}
\end{equation}
\end{definition}

\begin{definition}\label{def:zeros_vvw}\textbf{The sets $\mathcal E$ and $\mathcal E_I$.} 
Let $(\mathcal N^{\prime}, \mathcal O)$ be given and let $I$ be the base associated to the orientation $\mathcal O$.
In the following we call
\[
\mathcal E = \{ e \in ({\mathcal N}^{\prime}, {\mathcal O}) \ : \, e\ne e^{(D)}_i \ , \ i\in [n] \ \}, \quad\quad \mathcal E_I = \{ e \in ({\mathcal N}^{\prime}, {\mathcal O}) \ : \, e\ne e^{(D)}_i \ , \ i\in \bar I \ \}
\]
respectively the sets of edges in the modified network which are not Darboux edges and those which are not Darboux sink edges in the given orientation.
\end{definition}

We remark that  $E_e=(0,\dots,0)$ for all $ e \in  (\mathcal N^{\prime}, \mathcal O)\backslash \mathcal E_I$. Therefore on such edges both the vacuum e.w. and the dressed e.w. are identically zero for all $\vec t$. 
Similarly, the only edges $e\in \mathcal N^{\prime}$ on which $E_e$ is a linear combination of the rows of $A$
are those in $\mathcal E_I \backslash \mathcal E$, {\sl i.e.} the edges at the Darboux sources in the given orientation.
On such edges the dressed e.w. is identically zero for all $\vec t$, since the vacuum e.w. is one of the heat hierarchy solution generating the Darboux transformation for the soliton data $(\mathcal K,[A])$. Therefore the following statement holds true.

\begin{lemma}
Let the dressed e.w. $\Psi_{e, \mathcal O,\mathfrak l} (\vec t)$ and the set $\mathcal E$ be as in Definitions \ref{def:vvw_gen} and \ref{def:zeros_vvw}. Then $\Psi_{e, \mathcal O,\mathfrak l} (\vec t) \not \equiv 0$ for any $e\in \mathcal E$ and the latter set is the same for any orientation of the network.
\end{lemma}

\begin{remark}\label{rem:t_0}\textbf{The initial time $\vec t_0$.}
Let the vacuum and the dressed e.w. and $\mathcal E_I$ and $\mathcal E$ be as in Definitions \ref{def:vvw_gen}  and \ref{def:zeros_vvw}. In the following we fix an initial time $\vec t_0= (x_0,0,\dots)$ such that
\begin{enumerate}
\item For any $e\in \mathcal E_I$, the sign of the vacuum e.w. $\Phi_{e, \mathcal O,\mathfrak l} (\vec t_0)$ is the same of the coefficient of the highest phase appearing in it;
\item For any $e\in \mathcal E$, the dressed e.w. $\Psi_{e, \mathcal O,\mathfrak l} (\vec t_0) \not = 0$.
\end{enumerate}
\end{remark}

The full rank linear system satisfied by the edge vectors (see Section \ref{sec:linear}) induces a full rank linear system satisfied by the edge wave functions:
\begin{enumerate}
\item At any bivalent vertex incident with the incoming edge $e$ and the outgoing edge $f$, 
\begin{equation}\label{eq:lin_Phi1}
\Phi_{e, \mathcal O, \mathfrak l} (\vec t) =  (-1)^{\mbox{int}(e)+\mbox{wind}(e,f)} w_e \Phi_{f, \mathcal O, \mathfrak l} (\vec t), \\
\Psi_{e, \mathcal O, \mathfrak l} (\vec t) =  (-1)^{\mbox{int}(e)+\mbox{wind}(e,f)} w_e \Psi_{f, \mathcal O, \mathfrak l} (\vec t), 
\end{equation}
where $\mbox{int}(e)$ denotes the number of intersections of the gauge boundary rays with the edge $e$, while 
$\mbox{wind}(e,f)$ denotes the winding number of the path containing the two edges $e$ and $f$ with respect to the 
gauge direction  $\mathfrak l$;
\item 
At any trivalent white vertex incident with the incoming edge $e_3$ and the outgoing edges $e_1$, $e_2$,
\begin{equation}\label{eq:lin_Phi}
\begin{array}{l}
\Phi_{e_3, \mathcal O, \mathfrak l} (\vec t) =  (-1)^{\mbox{int}(e_3)} \, w_3 \left( (-1)^{\mbox{wind}(e_3, e_1)} \Phi_{e_1, \mathcal O, \mathfrak l} (\vec t) + (-1)^{\mbox{wind}(e_3, e_2)}\  \Phi_{e_2, \mathcal O, \mathfrak l} (\vec t)\right),\\
\Psi_{e_3, \mathcal O, \mathfrak l} (\vec t) = (-1)^{\mbox{int}(e_3)} \, w_3 \left( (-1)^{\mbox{wind}(e_3, e_1)} \Psi_{e_1, \mathcal O, \mathfrak l} (\vec t) + (-1)^{\mbox{wind}(e_3, e_2)}\  \Psi_{e_2, \mathcal O, \mathfrak l} (\vec t)\right);
\end{array}
\end{equation}
\item At any trivalent black vertex incident with the incoming edges  $e_2$, $e_3$ and the outgoing edge $e_1$,
\begin{equation}\label{eq:lin_Phi3}
\begin{array}{ll}
\Phi_{e_m, \mathcal O, \mathfrak l} (\vec t) &=  (-1)^{\mbox{int}(e_m)+\mbox{wind}(e_m, e_1)}\ w_m \Phi_{e_1, \mathcal O, \mathfrak l} (\vec t),\quad\quad m=2,3,\\
\Psi_{e_m, \mathcal O, \mathfrak l} (\vec t) &=  (-1)^{\mbox{int}(e_m)+\mbox{wind}(e_m, e_1)}\ w_m \Psi_{e_1, \mathcal O, \mathfrak l} (\vec t),\quad\quad m=2,3.
\end{array}
\end{equation}
\end{enumerate}

In view of the above we associate vacuum and dressed network divisors to $({\mathcal N}^{\prime},\mathcal O, \mathfrak l)$ as follows.

\begin{definition}\label{def:vac_div_gen}\textbf{The vacuum network divisor $\DVN$ and the dressed network divisor $\DDN$.}
At each trivalent white vertex $V$ of $({\mathcal N}^{\prime},\mathcal O,\mathfrak l)$, 
let $\Phi_{e_m, \mathcal O,\mathfrak l} (\vec t)$ be the vacuum e.w. on the edge $e_m$, $m\in [3]$, defined above. Then
we assign the following vacuum network divisor number $\gamma_{\textup{\scriptsize vac},V}$ to $V$: 
\begin{equation}\label{eq:vac_pole_def}
\gamma_{\textup{\scriptsize vac},V} = \frac{(-1)^{\mbox{wind}(e_3, e_1)} \Phi_{e_1,\mathcal O,\mathfrak l} (\vec t_0) }{(-1)^{\mbox{wind}(e_3, e_1)} \Phi_{e_1,\mathcal O,\mathfrak l} (\vec t_0)+(-1)^{\mbox{wind}(e_3, e_2)} \Phi_{e_2,\mathcal O,\mathfrak l} (\vec t_0)},
\end{equation}
where $\mbox{wind}(e_3,e_i)$ is the winding number of the directed path $[e_3,e_i]$, $i=1,2$.

We call $\DVN= \{ (\gvac, V_l) \; l\in [g+n-k]  \}$ the vacuum network divisor on ${\mathcal N}^{\prime}$, where $V_l$, $l\in [g+n-k]$, are the trivalent white vertices of the network.

At each trivalent white vertex $V$ of $({\mathcal N}^{\prime},\mathcal O,\mathfrak l)$, 
let $\Psi_{e_m, \mathcal O,\mathfrak l} (\vec t)$ be the dressed e.w. on the edge $e_m$, $m\in [3]$, defined above. Then, to $V$ we assign the dressed network divisor number 
\begin{equation}\label{eq:dress_pole_def}
\gamma_{\textup{\scriptsize dr},V} = \frac{(-1)^{\mbox{wind}(e_3, e_1)} \Psi_{e_1,\mathcal O,\mathfrak l} (\vec t_0) }{(-1)^{\mbox{wind}(e_3, e_1)} \Psi_{e_1,\mathcal O,\mathfrak l} (\vec t_0)+(-1)^{\mbox{wind}(e_3, e_2)} \Psi_{e_2,\mathcal O,\mathfrak l} (\vec t_0)},
\end{equation}
with $\mbox{wind}(e_3,e_i)$ as above. If $V$ is a trivalent vertex, associated to a Darboux source, the denominator in (\ref{eq:dress_pole_def}) is 
zero, and we set $\gamma_{\textup{\scriptsize dr},V}=\infty$.

We call $\DDN = \{ (\gdr, V_l), \; l\in [g+n-k]  \}$ the dressed network divisor on ${\mathcal N}^{\prime}$, where $V_l$, $l\in [g+n-k]$, are the trivalent white vertices of the network.
\end{definition}

\begin{remark}\label{rem:div_N}\textbf{Network divisors on $(\mathcal N, \mathcal O, \mathfrak{l})$.}
We remark that we may define a vacuum and an edge wave function also on the original network $(\mathcal N, \mathcal O, \mathfrak{l})$. Such edge wave functions coincide with those in Definition \ref{def:vvw_gen} for $(\mathcal N^{\prime}, \mathcal O, \mathfrak{l})$ at all edges $e$ not ending or starting at a vertex $V_l$, $l\in [n]$, associated to a Darboux (source or sink) point. At the edges at $V_l$ in $\mathcal N$ the two dressed wave functions may differ only by sign, whereas the two vacuum wave functions differ by a heat hierarchy solution if $V_l$ is associated to a source. 

Therefore we may also define a degree $g-k$ dressed network divisor on $(\mathcal N, \mathcal O, \mathfrak{l})$
\begin{equation}\label{eq:net_div_N}
{\mathcal D}_{\textup{\scriptsize dr},{\mathcal N}} \; = \; \DDN \; \backslash \; \{ (\gdr, V_l), \; l\in [n]  \}.
\end{equation}
The latter divisor is exactly the one used in Definition \ref{def:DKP} to complete the Sato divisor on $\Gamma$ for the soliton data $(\mathcal K, [A])$.
\end{remark}

If $\mathcal N$ is the acyclically oriented Le--network representing $[A]$, then the vacuum and dressed network divisors constructed in \cite{AG3} coincide with the above definition for the choice of the gauge ray $\mathfrak l$ as in Figure \ref{fig:Nprime}.

\begin{definition}\label{def:triv_div}\textbf{Trivial network divisor numbers.}
By definition, both the vacuum and the dressed network divisor numbers at each trivalent white vertex are real and represent the local coordinate of a point on the corresponding copy of $\mathbb{CP}^1$ in $\Gamma$. In general, on $({\mathcal N}^{\prime}, \mathcal O,\mathfrak l)$, the value of the vacuum or dressed network divisor number at a given trivalent white vertex $V$ depends on the choice of $\vec t_0$. If this is not the case we call the corresponding \textbf{network divisor number trivial}.
\end{definition}

If ${\mathcal N}^{\prime}$ is the Le--network, by construction the only trivial network divisor numbers occur at Darboux vertices \cite{AG3}. In the general case, trivial network divisor numbers occur also at all vertices where the linear system in (\ref{eq:lineq_white}) involves linearly dependent vectors.

\begin{lemma}\label{lemma:trivial_div}\textbf{Trivial network divisor numbers}
Let $E_{e_m}$, $m\in [3]$, be the edge vectors at a trivalent white vertex $V$ of $({\mathcal N}^{\prime},\mathcal O,\mathfrak l)$ and $\gamma_{\textup{\scriptsize vac},V}$, $\gamma_{\textup{\scriptsize dr},V}$ respectively be the vacuum and dressed network divisor numbers at $V$. If there exists a non zero constant $c_V$ such that either $E_{e_2} = c_V E_{e_1}$ or $E_{e_3} = c_V E_{e_1}$ or $E_{e_3} = c_V E_{e_2}$, then at $V$ the vacuum and dressed network divisor numbers concide
$\gamma_{\textup{\scriptsize dr},V} = \gamma_{\textup{\scriptsize vac},V}$
and the corresponding network divisor numbers are trivial.
\end{lemma}

\begin{proof}
Indeed in the first case for all $\vec t$,
$\Phi_{e_2,\mathcal O,\mathfrak l} (\vec t) =c_V\Phi_{e_1,\mathcal O,\mathfrak l} (\vec t)$, $\Psi_{e_2,\mathcal O,\mathfrak l} (\vec t) =c_V\Psi_{e_1,\mathcal O,\mathfrak l} (\vec t)$,
and 
\[
\Phi_{e_3,\mathcal O,\mathfrak l} (\vec t) =(-1)^{\mbox{int}(e_3)} w_3\left( (-1)^{\mbox{wind}(e_3, e_1)}+c_V(-1)^{\mbox{wind}(e_3, e_2)}\right) \Phi_{e_1,\mathcal O,\mathfrak l} (\vec t),\]
\[ 
\Psi_{e_3,\mathcal O,\mathfrak l} (\vec t) =(-1)^{\mbox{int}(e_3)} w_3\left( (-1)^{\mbox{wind}(e_3, e_1)}+c_V(-1)^{\mbox{wind}(e_3, e_2)}\right) \Psi_{e_1,\mathcal O,\mathfrak l} (\vec t),
\]
so that
\[
\gamma_{\textup{\scriptsize dr},V} = \gamma_{\textup{\scriptsize vac},V} = \frac{1}{1+c_V(-1)^{\mbox{wind}(e_3, e_2)-\mbox{wind}(e_3, e_1)}}.
\]
The other cases may be treated in a similar way.
\end{proof}

Next we use the transformation properties of the edge vectors w.r.t. changes of orientation, of ray direction, weight and vertex gauges settled in Sections \ref{sec:gauge_ray}--\ref{sec:different_gauge}, to characterize the transformation properties of both the edge wave functions and the networks divisors.

Any change of gauge ray direction only effect the signs of the vacuum and of the dressed e.w.s at the given vertex $V$ 
(Proposition \ref{prop:change_orient}), and leaves invariant both the vacuum and the dressed network divisor numbers (Corollary \ref{cor:indep_gauge}).

Every change of orientation in ${\mathcal N}^{\prime}$ (either induced by a change of base in the matroid or due to a change of orientation in closed cycles) corresponds to a well defined change of the local coordinate of each copy if $\mathbb{CP}^1$, and the dressed network divisor numbers 
change according to such coordinate transformation (Corollary \ref{cor:indep_orient}). Therefore, all the positions of the corresponding divisor 
points on $\Gamma$ remain invariant (Theorem \ref{theo:inv}).
On the contrary, any change of orientation in the network associated to a change of base in the matroid acts untrivially on the values of the vacuum e.w. and therefore also on the vacuum network divisor numbers and the positions of the corresponding vacuum divisor points in $\Gamma$ (Remark \ref{rem:vac_div_orient}).

\begin{proposition}\label{prop:change_orient}\textbf{The dependence of the vacuum e.w. and of the dressed e.w. on the gauge ray directions and on the orientation of the network}
Let ${\mathcal N}^{\prime}$ be a network representing $[A]\in \GTNN$ and let $\mathcal O_s$, $\mathfrak l_s $, $s=1,2$, respectively  be two perfect orientations and two gauge ray directions of ${\mathcal N}^{\prime}$. Let ${\mathcal K}$ be a system of phases. Let $f^{(r)} (\vec t)$ be a system of heat hierarchy solutions generating the Darboux transformation ${\mathfrak D}^{(k)}$ for the soliton data $(\mathcal K, [A])$. 
Then, for any fixed $\vec t$ and $e\in {\mathcal N}^{\prime}$, 
\begin{enumerate}
\item The vacuum e.w. $\Phi$ and the dressed e.w. $\Psi$ may only change of sign if we change the gauge ray direction keeping the orientation $\mathcal O$ fixed:
\begin{equation}\label{eq:phi_gauge} 
\Phi_{e, \mathcal O, \mathfrak l_2} (\vec t) = (-1)^{\mbox{int}(V_e)+\mbox{par(e)}}  \Phi_{e, \mathcal O, \mathfrak l_1} (\vec t), \quad\quad
\Psi_{e, \mathcal O, \mathfrak l_2} (\vec t) = (-1)^{\mbox{int}(V_e)+\mbox{par(e)}}  \Psi_{e, \mathcal O, \mathfrak l_1} (\vec t),
\end{equation}
where the indices $\mbox{int}(V_e)$ and $\mbox{par(e)}$ are as in Proposition \ref{prop:rays};
\item If we change the orientation in ${\mathcal N}^{\prime}$ from $\mathcal O_1$ to $\mathcal O_2$, then, for any edge $e\in {\mathcal N}^{\prime}$ there exist real constants $\alpha_e$, $c^r_e$, $r\in [k]$, such that the vacuum e.w. $\Phi$ and the dressed e.w. $\Psi$ satisfy
\begin{equation}\label{eq:phi_orient}
\Phi_{e, \mathcal O_2, \mathfrak l} (\vec t) = \alpha_e \Phi_{e, \mathcal O_1, \mathfrak l} (\vec t) + \sum_{r=1}^k c^r_e f^{(r)} (\vec t),\quad\quad
\Psi_{e, \mathcal O_2, \mathfrak l} (\vec t) = \alpha_e \Psi_{e, \mathcal O_1, \mathfrak l} (\vec t) .
\end{equation}
Moreover, if the two orientations are associated to the same base $I$, then for any edge $e\in {\mathcal N}^{\prime}$ there exists a real constant $\alpha_e$ such that the vacuum e.w. 
satisfies
\begin{equation}\label{eq:phi_orient_2}
\Phi_{e, \mathcal O_2(I), \mathfrak l} (\vec t) = \alpha_e \Phi_{e, \mathcal O_1(I), \mathfrak l} (\vec t) 
\end{equation}
\end{enumerate}
\end{proposition}

\begin{corollary}\label{cor:indep_gauge}\textbf{Independence of the network divisors on the gauge ray direction.}
Let ${\mathfrak l}_1$ and ${\mathfrak l}_2$ be two gauge directions for the oriented network $({\mathcal N}^{\prime}, \mathcal O)$.
Then for any trivalent white vertex $V\in ({\mathcal N}^{\prime}, \mathcal O)$, both the vacuum network divisor numbers and the dressed network divisor numbers 
computed using either ${\mathfrak l}_1$ or ${\mathfrak l}_2$ are the same:
\begin{equation}\label{eq:inep_gauge}
\gamma_{\textup{\scriptsize vac},V,\mathfrak l_2} = \gamma_{\textup{\scriptsize vac},V,\mathfrak l_1}, \quad\quad\gamma_{\textup{\scriptsize dr},V,\mathfrak l_2} = \gamma_{\textup{\scriptsize dr},V,\mathfrak l_1}.
\end{equation}
\end{corollary}

The proof immediately follows by direct computation of the network divisor numbers inserting (\ref{eq:phi_gauge}) into (\ref{eq:dress_pole_def}) and (\ref{eq:vac_pole_def}) since $\mbox{int}(V_{e_1})=\mbox{int}(V_{e_2})$ and $\mbox{wind}_{\mathfrak{l_2}} (e_3,e_j)=
\mbox{par}(e_3)+\mbox{wind}_{\mathfrak{l_1}} (e_3,e_j)+\mbox{par}(e_j)$, $j=1,2$, where $\mbox{wind}_{\mathfrak{l_s}}$ is the winding computed w.r.t. the gauge direction $\mathfrak{l}_s$.

\begin{corollary}\label{cor:indep_orient}\textbf{The dependence of the dressed network divisor on the orientation.}
Let ${\mathcal O}_1$ and ${\mathcal O}_2$ be two perfect orientations for network ${\mathcal N}$.
Then for any trivalent white vertex $V$, such that all edges at $V$ have the same versus in both orientations, $(e_1^{(2)},e_2^{(2)},e_3^{(2)})=(e_1^{(1)},e_2^{(1)},e_3^{(1)})$, the dressed network divisor number is the same in both
orientations 
\begin{equation}\label{eq:inep_gauge_1}
\gamma_{\textup{\scriptsize dr},V,\mathcal O_2} = \gamma_{\textup{\scriptsize dr},V,\mathcal O_1}.
\end{equation}
If at the vertex $V$ we change orientation of edges from $(e_1^{(1)},e_2^{(1)},e_3^{(1)})$ to $(e_2^{(2)},e_3^{(2)},e_1^{(2)})=(e_1^{(1)},-e_2^{(1)},-e_3^{(1)})$ (Figure \ref{fig:phi_orient}[left]), then the relation between the dressed network divisor numbers at $V$ in the two orientations is
\begin{equation}\label{eq:inep_gauge_2}
\gamma_{\textup{\scriptsize dr},V,\mathcal O_2} = \frac{1}{1-\gamma_{\textup{\scriptsize dr},V,\mathcal O_1}}.
\end{equation}
If at the vertex $V$ we change orientation of edges from $(e_1^{(1)},e_2^{(1)},e_3^{(1)})$ to $(e_3^{(2)},e_1^{(2)},e_2^{(2)})=(-e_1^{(1)},e_2^{(1)},-e_3^{(1)})$ (Figure \ref{fig:phi_orient}[right]), then the relation between the dressed network divisor numbers at $V$ in the two orientations is
\begin{equation}\label{eq:inep_gauge_3}
\gamma_{\textup{\scriptsize dr},V,\mathcal O_2} = \frac{\gamma_{\textup{\scriptsize dr},V,\mathcal O_1}}{\gamma_{\textup{\scriptsize dr},V,\mathcal O_1}-1}.
\end{equation}
\end{corollary}

\begin{figure}
  \centering{\includegraphics[width=0.55\textwidth]{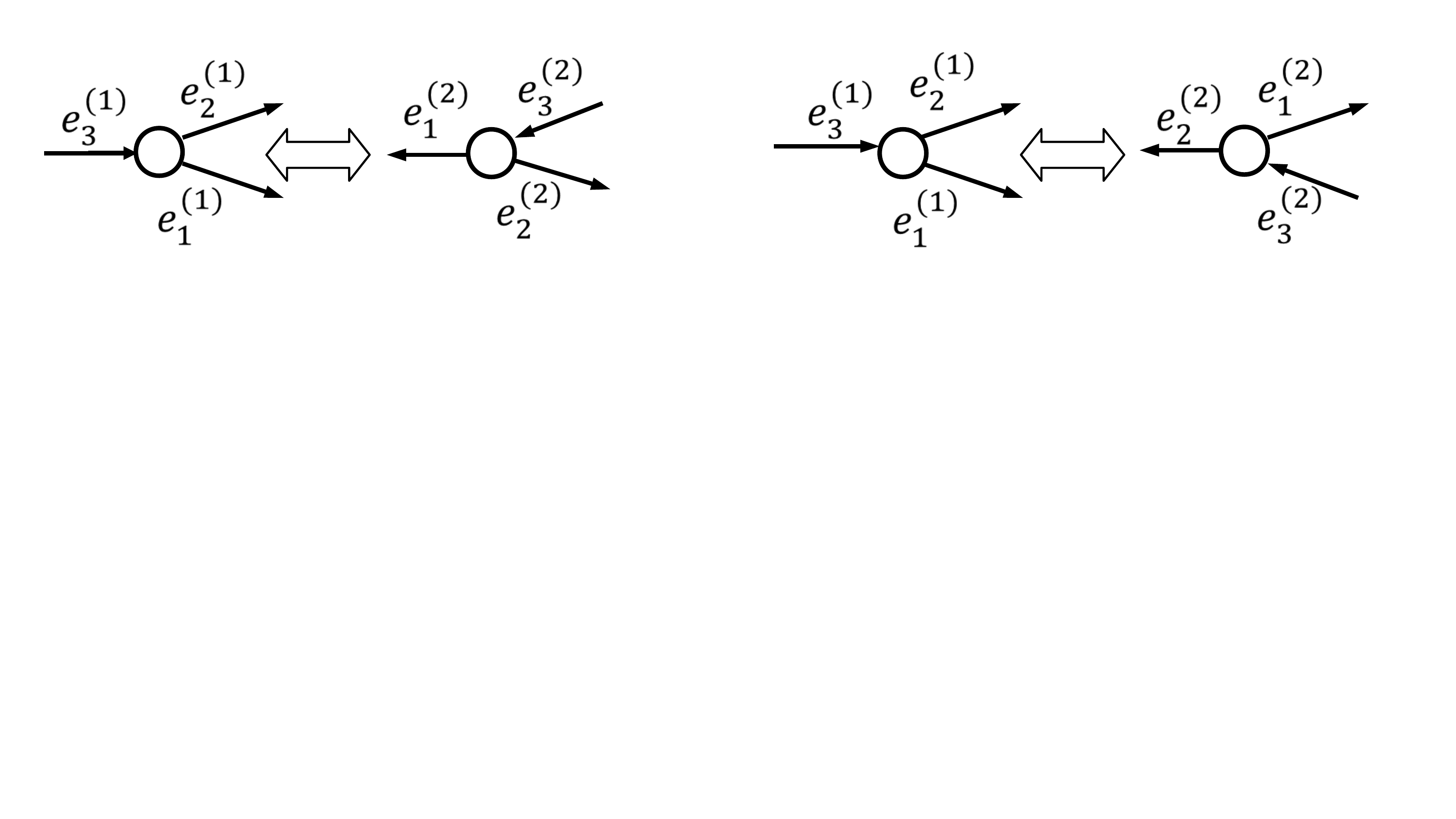}}
	\vspace{-3.7 truecm}
  \caption{\small{\sl The change of orientation at a vertex $V$.}\label{fig:phi_orient}}
\end{figure}

{\sl Proof}
The proof follows by direct computation of the dressed network divisor numbers using (\ref{eq:phi_orient}) and the linear system (\ref{eq:lin_Phi}) at $V$ for both orientations. For instance, (\ref{eq:inep_gauge_2}) follows observing that
\[
\Psi_{e_m^{(2)}, \mathcal O_2, \mathfrak l} (\vec t) = \alpha_{m-1} \Psi_{e_{m-1}^{(1)}, \mathcal O_1, \mathfrak l} (\vec t) \quad (\!\!\!\!\!\!\mod 3),
\]
and that the linear system w.r.t. both orientations imposes the constraint
\[
\displaystyle
\frac{\alpha_1}{\alpha_3} \, (-1)^{\mbox{wind}_{\mathcal O_2} (e^{(2)}_3,e^{(2)}_2)-\mbox{wind}_{\mathcal O_2} (e^{(2)}_3,e^{(2)}_1)} = -w_3(-1)^{\mbox{int}_{\mathcal O_1} (e^{(1)}_3)+\mbox{wind}_{\mathcal O_1} (e^{(1)}_3,e^{(1)}_1)},
\] 
where $w_3$ is the weight of $e^{(1)}_3$ in the initial orientation,
so that
\[
\begin{array}{ll}
\gamma_{\textup{\scriptsize dr},V,\mathcal O_2} &= \displaystyle\left(1+\frac{(-1)^{\mbox{wind}(e_3^{(2)}, e_2^{(2)})} \Psi_{e_2^{(2)},\mathcal O_2,\mathfrak l}}{(-1)^{\mbox{wind}(e_3^{(2)}, e_1^{(2)})} \Psi_{e_1^{(2)},\mathcal O_2,\mathfrak l} }\right)^{-1} =\left(1+\frac{(-1)^{\mbox{wind}(e_3^{(2)}, e_2^{(2)})} \alpha_1 \Psi_{e_1^{(1)},\mathcal O_1,\mathfrak l} }{(-1)^{\mbox{wind}(e_3^{(2)}, e_1^{(2)})} \alpha_3 \Psi_{e_3^{(1)},\mathcal O_1,\mathfrak l} }\right)^{-1}\\
&\\
& =(1-\gamma_{\textup{\scriptsize dr},V,\mathcal O_1})^{-1}.\quad\quad\quad \square
\end{array}
\]
\begin{remark}
In next Section we show that the transformation of the network divisor numbers of Corollary~\ref{cor:indep_orient} corresponds to a change of coordinates on the copy of 
$\mathbb{CP}^1$ associated with such white vertex. 
\end{remark}

\begin{remark}\label{rem:vac_div_orient}\textbf{The dependence of the vacuum network divisor numbers on the change of orientation}
If the two orientations ${\mathcal O}_1$ and ${\mathcal O}_2$ of ${\mathcal N}^{\prime}$ are associated to the same base $I$ ({\sl i.e.} correspond to changes of orientation along cycles), then (\ref{eq:inep_gauge_1})--(\ref{eq:inep_gauge_3}) hold true with $\gamma_{\textup{\scriptsize vac},V,\mathcal O_s}$, $s=1,2$ instead of $\gamma_{\textup{\scriptsize dr},V,\mathcal O_s}$ at any trivalent white vertex $V$ . 

Instead, if the change of orientation is induced by a change in the base of ${\mathcal M}$, then the rule for the transformation of the vacuum network divisor numbers is more complicated and we do not write it explicitly since we shall not use it in the following.  
\end{remark}

The position of each Darboux edge in $(\mathcal N^{\prime},\mathcal O, \mathfrak{l})$ rules the position of the corresponding Darboux point on $\Gamma$. There is a certain gauge freedom in our construction since any Darboux point (edge) may lay in one of two ovals in $\Gamma$ (faces in $\mathcal N$). By our assumptions on the position of Darboux vertices, the switch of the Darboux source edge $e^{(D)}_s$ from one face to the other corresponds to a change of sign in the edge vectors at $e_s$ and $f_s$ (Figure \ref{fig:Darboux1}), whereas the switch of the Darboux sink edge $e^{(D)}_s$ from one face to the other leaves invariant the edge vectors at $e_s$ and $f_s$. Therefore the following statement holds.

\begin{proposition}\label{prop:darboux}\textbf{Dependence of divisor numbers on the position of Darboux edges}
Let ($\mathcal N^{\prime}, \mathcal O, \mathfrak l)$ as in Construction \ref{con:Nprime}. Then, if we rotate the Darboux vertex $V^{(D)}_s$ for some $s\in [n]$ inside a given face both the vacuum and the dressed divisor numbers keep the same value. If we switch $V^{(D)}_s$ from one face to the other, then the relation between the old (superscript $(1)$) and the new (superscript $(2)$) divisor numbers is
\begin{equation}\label{eq:div_no_Darboux_pos}
\gamma_{\textup{\scriptsize vac},V_i,}^{(2)} = 1-\gamma_{\textup{\scriptsize vac},V_i,}^{(1)},\quad\quad \gamma_{\textup{\scriptsize dr},V_i,}^{(2)} = 1-\gamma_{\textup{\scriptsize dr},V_i,}^{(1)},
\end{equation}
where, in the above formulas, both dressed divisor numbers are $\infty$ if $V^{(D)}_{i,1}$ and $V^{(D)}_{i,2}$ are Darboux sources.
\end{proposition}

The proof is trivial. The only interesting case is the behavior of the vacuum divisor point at a Darboux source vertex (see also Figure \ref{fig:Darboux1}): if, for $s=i_r\in I$, the Darboux source vertex $V^{(D)}_s$ lays to the left of $V_s$ then $\gamma_{\textup{\scriptsize vac},V_s} = e^{\theta_{s}(\vec t_0)}/f^{(r)}(\vec t_0)$, otherwise $\gamma_{\textup{\scriptsize vac},V_s} = (f^{(r)}(\vec t_0)-e^{\theta_{s}(\vec t_0)})/f^{(r)}(\vec t_0)$.

Finally, using Lemmas \ref{lem:weight_gauge} and \ref{lem:vertex_gauge}, it is immediate to prove that the network divisor numbers are independent on both the weight gauge and the vertex gauge:

\begin{proposition}\label{prop:gauge}\textbf{Independence of divisor numbers on the weight gauge and on the vertex gauge}
Let ($\mathcal N^{\prime}, \mathcal O, \mathfrak l)$ as in Construction \ref{con:Nprime}. Then, both the vacuum and the dressed divisor numbers are the same on each white vertex for any choice of weight gauge and of vertex gauge.
\end{proposition}

\begin{remark}\textbf{The dependence of the divisor on the gauge freedom for unreduced graphs}\label{rem:div_unred}
The gauge freedom for unreduced graphs of Remark \ref{rem:gauge_freedom} acts untrivially on network divisor numbers.
We show that on the example already considered in Section \ref{sec:null_vectors}. The network in Figure~\ref{fig:zero-vector}[right] represents the same point $[ 2p/(1+p+q),1] \in Gr^{\mbox{\tiny TP}}(1,2)$ for any $s>0$. It is easy to check that the network divisor is
\[
\left( \frac{p}{1+2p}, V_1\right), \quad \left( -\frac{1+p}{q}, V_2\right), \quad \left( 1+s\frac{1+2p}{1+p}, V_3\right), \quad
\] 
for any choice of $\vec t_0$.
\end{remark}

Next we normalize both the vacuum and the dressed e.w.s. 

\begin{definition}\label{def:norm_vac_gen}\textbf{The normalized vacuum edge wave function $\hat \Phi$ and the KP edge wave function $\hat \Psi$.}
Let $\mathcal E$ and $\mathcal E_I$ be as in Definition \ref{def:zeros_vvw}, where $I=\{1\le i_1 <\cdots < i_k \le n\}$ is the base in ${\mathcal M}$ for the gauge--oriented network $({\mathcal N}^{\prime}, \mathcal O, \mathfrak l)$.
Let $\vec t_0$ be as in the Remark \ref{rem:t_0}. 
Then 
\begin{enumerate}
\item The \textbf{normalized vacuum edge wave function (normalized vacuum e.w)} on $({\mathcal N}^{\prime}, I)$ is 
\begin{equation}\label{eq:vvfnorN}
\hat \Phi_{e, I} (\vec t) = \left\{ \begin{array}{ll} \displaystyle \frac{\Phi_{e, \mathcal O,\mathfrak l} (\vec t) }{\Phi_{e, \mathcal O,\mathfrak l} (\vec t_0) }, &\quad\quad \mbox{for all } e\in {\mathcal E}(I), \quad \forall \vec t,\\
e^{\theta_{j} (\vec t -\vec t_0)}, &\quad\quad \mbox{if } e=e^{(D)}_j, \;\; j\in \bar I.
\end{array}\right.
\end{equation}
\item The \textbf{KP edge wave function (KP e.w)} on $({\mathcal N}^{\prime}, I)$ is 
\begin{equation}\label{eq:KPvfnorN}
\hat \Psi_{e} (\vec t) = \left\{ \begin{array}{ll} \displaystyle\frac{\Psi_{e, \mathcal O,\mathfrak l} (\vec t) }{\Psi_{e, \mathcal O,\mathfrak l} (\vec t_0) }, &\quad\quad \mbox{for all } e\in \mathcal E, \quad \forall \vec t,\\
\displaystyle\frac{{\mathfrak D}^{(k)} e^{i\theta_{j} (\vec t)}}{{\mathfrak D}^{(k)} e^{i\theta_{j} (\vec t_0)}},&\quad\quad \mbox{if } e=e^{(D)}_j, \;\; j\in [n].
\end{array}\right.
\end{equation}
\end{enumerate}
\end{definition}

\begin{remark}\label{rem:indep}\textbf{The KP edge wave function on $\mathcal N^{\prime}$.} 
The name KP wave function for $\hat \Psi_e (\vec t)$ on $\mathcal N^{\prime}$ is fully justified. Indeed $\hat \Psi_e (\vec t)$ just depends on the soliton data $(\mathcal K, [A])$ and on the chosen network representing $[A]$ since $\hat \Psi_e (\vec t)$ takes the same value on a given $e$ for any choice of orientation, gauge ray direction, weight gauge and vertex gauge. Moreover, by construction it satisfies the Sato boundary conditions at the edges at the boundary, and takes real values for real $\vec t$. This function is a common eigenfunction to all KP hierarchy auxiliary linear operators $-\partial_{t_j} + B_j$, where $B_j =(L^j)_+$, and the Lax operator $L=\partial_x+\frac{u(\vec t)}{2}\partial_x^{-1}+ u_2(\vec t)\partial_x^{-2}+\ldots$, the coefficients of these operators are the same for all edges.

On the contrary the normalized vacuum edge wave function $\hat \Phi_{e,I} (\vec t)$ depends on the base $I$ associated to the orientation of the network.

Finally we remark that $\hat \Psi_e (\vec t)$ ($\hat \Phi_{e,I} (\vec t)$) takes equal values at all edges $e$ at any given bivalent or black trivalent vertex. Finally at any trivalent white vertex $\hat \Psi_e (\vec t)$ ($\hat \Phi(\vec t)$) takes either the same value at all three edges for all times or distinct values at some $\vec t\not= \vec t_0$. 
\end{remark}

\subsection{The vacuum wave function $\hat \phi_I$ on $\Gamma$}\label{sec:vac_KP_gen}

In this Section, we define the vacuum edge wave function $\hat \phi_I (P, \vec t)$ on $\Gamma=\Gamma (\mathcal G)$, where $\mathcal G$ is the graph of $\mathcal N^{\prime}$ and $I$ is the base associated to the chosen orientation of $\mathcal N^{\prime}$ using the normalized vacuum edge wave function of Definition \ref{def:norm_vac_gen}. 

\begin{construction}\label{con:vac_gen}\textbf{The vacuum wave function $\hat \phi_I$ on $\Gamma$}
Let the soliton data $({\mathcal K}, [A])$ be given, with $\mathcal K = \{ \kappa_1 < \cdots < \kappa_n\}$ and $[A]\in \S \subset \GTNN$. Let $I\in \mathcal M$ be fixed.
Let $\mathcal G$ be a PBDTP graph representing $\S$ as in Definition \ref{def:graph} with $(g+1)$ faces and let the curve $\Gamma=\Gamma(\mathcal G)$ be as in Construction \ref{def:gamma}. 
Let $({\mathcal N}, \mathcal O(I), \mathfrak l)$ be a network of graph $\mathcal G$ representing $[A]$ with no null edge vectors. Let ${\mathcal N}^{\prime}$ be the modified network defined in Section \ref{sec:modN} for a given choice of Darboux vertices. Let $\DVN$ be the vacuum network divisor defined in the previous Section and $\hat \Phi_{e,I} (\vec t)
$ be the normalized vacuum edge wave function of Definition \ref{def:norm_vac_gen}. Finally on each component of $\Gamma$ let the coordinate $\zeta$ be as in Definition \ref{def:loccoor} (see also Figure \ref{fig:lcoord}).

Then on $\Gamma$, we construct the vacuum wave function $\hat \phi_I(P, \vec t)$ as follows:
for such data. 
\begin{enumerate}
\item The restriction of ${\hat \phi}$ to $\Gamma_{0}$ is denoted ${\hat \phi}_0$ and is the normalized  Sato vacuum wave function 
\[
{\hat \phi}_0 (\zeta, \vec t) = e^{\theta (\vec t-\vec t_0)},\quad\quad 	\theta(\vec t) = \sum\limits_{l\ge 1} \zeta^l t_l.
\]
\item On any component of $\Gamma$ corresponding to a bivalent vertex or a trivalent black vertex $V$, we define $\hat \phi_I (P, \vec t) = \hat \Phi_{e,I} (\vec t)$, where $e$ is one of the edges at $V$;
\item For any fixed orientation $\mathcal O$ for the base $I$, on the component $\Gamma_l$, $l\in [g+n-k]$, corresponding to the trivalent white vertex $V_l$, we have four real ordered marked points: $P_m^{(l)}$, $m\in [3]$, correspond to the edges $e_m^{(l)}$ at $V_l$, and the forth 
point  $P_{\textup{\scriptsize vac}}^{(l)}$ is defined by $\zeta(P_{\textup{\scriptsize vac}}^{(l)})=\gamma_{\textup{\scriptsize vac}, V_l}$. A change of orientation which preserves the base $I$ acts as a well defined change of coordinates on the four points $\Gamma_l$ which leaves invariant the value of the normalized vacuum wave function at the corresponding edges. For these reasons, we define 
\[
\hat \phi_I (P_m^{(l)}, \vec t) = \hat \Phi_{e_m^{(l)},I} (\vec t),\quad\quad m\in [3],
\]
and extend it on $\Gamma_l$ to a meromorphic function of degree less than or equal to one. It is straightforward to verify that
either $\hat \phi_I (P, \vec t)$ has a simple pole at $P_{\textup{\scriptsize vac}}^{(l)}$ or is independent of the spectral parameter. The latter case occurs if and only if the network divisor point $\gamma_{\textup{\scriptsize vac}, V_l}$ is trivial, {\sl i.e.} the normalized vacuum edge wave function takes the same value at the three edges at $V_l$ for any given $\vec t$.
\end{enumerate}
\end{construction}

By construction, $\hat \phi$ is meromorphic of degree $\mathfrak d\le g$, its poles are all real and simple and belong to the copies $\Gamma_l=\mathbb{CP}^1$ corresponding to trivalent white vertices $V_l$ carrying non trivial vacuum network divisor numbers on ${\mathcal N}^{\prime}$. 
Since the vacuum wave function  $\hat \phi_I (P,\vec t)$ is constant w.r.t. the spectral parameter $P$ on each component associated to a boundary sink in $\mathcal N$, the following definition of vacuum divisor on $\Gamma$ is natural.

\begin{definition}\label{def:vac_div_G}{\textbf{The vacuum divisor $\DVG$}}. We denote by $\DVG$ the following degree $g$ effective divisor: it is the sum of the $g$ simple poles located at the points 
$P_{\textup{\scriptsize vac}}^{(l)}$ different from the Darboux sinks. 
\end{definition}

$\DVG$ is contained in the union of all ovals since all network divisor numbers are real (we thus have proven item (1) in Theorem \ref{theo:vac_div}).
$\DVG$ depends on the base $I$ (Proposition \ref{prop:change_orient}) and, for any fixed base $I$, the position of the pole belonging to a component associated to a Darboux source depends on the position of the latter (Proposition \ref{prop:darboux}).

In the following Theorem we summarize the properties 
of the vacuum wave function which follow directly from its definition. 

\begin{theorem}\label{lemma:noneffvac}\textbf{Characterization of the vacuum divisor of $\hat \phi$ on $\Gamma$.}
Let $\hat \phi$, $\DVN$, $\DVG$ on $\Gamma$ be as in Construction \ref{con:vac_gen} and in Definitions \ref{def:vac_div_gen} and \ref{def:vac_div_G}. Then $\hat \phi_I(P,\vec t)$ is the unique meromorphic extension of the Sato vacuum wave function to
$\Gamma\backslash\{P_0\}$ satisfying:
\begin{enumerate}
\item At $\vec t=\vec t_0$ $\hat \phi (P, \vec t_0)=1$ at all points $P\in \Gamma\backslash \{P_0\} $;
\item ${\hat \phi} (\zeta(P), \vec t)$ is real for real values of the local coordinate $\zeta$ and for all real $\vec t$ on each component of $\Gamma$;
\item $\hat \phi$ takes equal values at pairs of glued points $P,Q\in \Gamma$, for all $\vec t$:  $\hat \phi(P, \vec t) = \hat \phi(Q, \vec t)$;
\item $\hat \phi(P^{(3)}_{i_r}, \vec t)\equiv f^{(r)} (\vec t)$ at each Darboux source point $P^{(3)}_{i_r}$, $r\in [k]$, for all $\vec t$, where $f^{(r)} (\vec t) = \sum_{j=1}^N A^i_j e^{\theta_j (\vec t- \vec t_0)}$ is $r$--th heat hierarchy solution associated to the RREF representative matrix of $[A]$ with respect to the base $I$;
\item $\hat \phi(\zeta, \vec t)$ is either constant or meromorphic of degree one w.r.t. to the spectral parameter on each component $\Gamma_l$ corresponding to a trivalent white vertex; 
\item On each $\Gamma_l$, $l\in [g+n-k]$, corresponding to a trivalent white vertex the vacuum wave function satisfies
\[
\hat \phi_I (\zeta(P),\vec t) = \frac{\hat \Phi_{e_3, I} (\vec t) \zeta - \gvac\hat \Phi_{e_1, I} (\vec t) }{ \zeta- \gvac}, \quad\quad \forall P\in \Gamma_l, \;\; \forall \vec t.
\]
\item $\hat \phi(\zeta, \vec t)$ is constant w.r.t. to the spectral parameter on each other copy of $\mathbb{CP}^1$;
\item For any $\vec t$, $(\hat \phi(P,\vec t))+\DVG\ge 0$, where $(\hat \phi(P,\vec t))$ denotes the divisor for   $\hat \phi(P,\vec t)$.
\item The vacuum divisor $\DVG$ is contained in the union of all ovals.  
\end{enumerate}
\end{theorem}

In \cite{AG3}, we have proven the above Theorem in the case in which  ${\mathcal G}$ is the Le--graph oriented canonically.

\subsection{The construction of the KP wave function $\hat \psi$ and of the KP divisor $\DKP$ on $\Gamma$.}\label{sec:inv}

In this Section we define the KP wave function $\hat \psi$ on the curve $\Gamma=\Gamma (\mathcal G)$ and the KP divisor $\DKP$. By construction $\DKP$ is contained in the union of all the ovals of $\Gamma$. In Theorem \ref{theo:inv} we show that $\DKP$ does not depend on the base $I$ associated to the orientation $\mathcal O$ of $\mathcal N^{\prime}$ used to construct the divisor and the wave function. 

\begin{construction}\label{con:dress_gen}\textbf{The KP wave function $\hat \psi$ on $\Gamma$.}
Let the soliton data $({\mathcal K}, [A])$ be given, with $\mathcal K = \{ \kappa_1 < \cdots < \kappa_n\}$ and $[A]\in \S \subset \GTNN$. Let $I\in \mathcal M$ be fixed.
Let $\mathcal G$ be a PBDTP graph representing $\S$ as in Definition \ref{def:graph} with $(g+1)$ faces and let the curve $\Gamma=\Gamma(\mathcal G)$ be as in Construction \ref{def:gamma}.  
Let $({\mathcal N}, \mathcal O(I), \mathfrak l)$ be a network of graph $\mathcal G$ representing $[A]$ with no null edge vectors. Let ${\mathcal N}^{\prime}$ be the modified network defined in Section \ref{sec:modN} for a given choice of Darboux vertices. Let $\DDN$ and $\hat \Psi_{e} (\vec t)$ respectively be the dressed network divisor of Definition \ref{def:vac_div_gen} and the KP edge wave function of Definition \ref{def:norm_vac_gen}. Finally on each component of $\Gamma$ let the coordinate $\zeta$ be as in Definition \ref{def:loccoor} (see also Figure \ref{fig:lcoord}).

We define the KP wave function $\hat \psi (P, \vec t)$ on $\Gamma\backslash \{P_0\}$ as follows:
\begin{enumerate}
\item The restriction of ${\hat \psi}$ to $\Gamma_{0}$ is the normalized dressed Sato wave function 
\[
{\hat \psi} (\zeta, \vec t) = \frac{{\mathfrak D}^{(k)} {\hat \phi}_0 (\zeta, \vec t)}{{\mathfrak D}^{(k)} {\hat \phi}_0 (\zeta, \vec t_0)};
\]
\item If $\Sigma_l$ is the component of $\Gamma$ corresponding to the black vertex $V^{\prime}_l$, then, for any $P\in \Sigma_l$ and for all $\vec t$, ${\hat \psi} (P, \vec t)$ is assigned the value of the KP e.w. $\hat \Psi_{e} (\vec t)$, where $e$ is one of the edges at $V^{\prime}_l$:
${\hat \psi} (\zeta(P), \vec t) = \hat \Psi_{e} (\vec t)$;
\item\label{it:non_int_dress} Similarly if $\Gamma_{l}$ is the component of $\Gamma$ corresponding to $V_{l}$, which is either a bivalent white vertex or a trivalent white vertex with KP e.w. coinciding on all edges at $V_l$, then, for any $P\in \Gamma_{l}$ and for all $\vec t$, ${\hat \psi} (P, \vec t)$ is assigned the value $\hat \Psi_{e} (\vec t)$, where $e$ is one of the  edges at $V_{l}$: ${\hat \psi} (\zeta(P), \vec t) = \hat \Psi_{e} (\vec t)$;
\item\label{it:int_dress} If $\Gamma_{l}$ is the component of $\Gamma$ corresponding to a trivalent white vertex with KP e.w. taking distinct values on the edges at $V_l$ for some $\vec t\not=\vec t_0$, then we define $\hat \psi$ at the marked points as
\[
{\hat \psi} (\zeta(P^{(m)}_{l}), \vec t) = \hat \Psi_{e_m} (\vec t), \quad \quad m\in [3],\;\; \forall \vec t,
\]
where $e_m$, $m\in [3]$ are the edges at $V_{l}$, and 
uniquely extend $\hat \psi$ to a degree one meromorphic function on $\Gamma_{l}$ imposing that it has a simple pole at $P_{\textup{\scriptsize dr}}^{(l)}$ with real coordinate $\zeta (P_{\textup{\scriptsize dr}}^{(l)})=\gamma_{\textup{\scriptsize dr}, V_l}$, with $\gdr$ as in (\ref{eq:dress_pole_def}):
\begin{equation}\label{eq:99}
{\hat \psi} (\zeta(P), \vec t) = \frac{\hat \Psi_{e_3} (\vec t) \zeta - \gdr\hat \Psi_{e_1} (\vec t) }{ \zeta- \gdr}= \frac{(-1)^{\mbox{wind}(e_3,e_1)} \Psi_{e_1} (\vec t) (\zeta -1) + (-1)^{\mbox{wind}(e_3,e_2)} \Psi_{e_2} (\vec t)\zeta}{((-1)^{\mbox{wind}(e_3,e_1)} \Psi_{e_1} (\vec t_0)+ (-1)^{\mbox{wind}(e_3,e_2)} \Psi_{e_2} (\vec t_0))(\zeta - \gdr)}.
\end{equation}
\end{enumerate}
\end{construction}
By construction, the KP wave function is constant w.r.t. the spectral parameter on any component containing a Darboux (source or sink) point and has $k$ real simple poles on $\Gamma_0$. Therefore the following definition of KP divisor is fully justified.

\begin{definition}\label{def:DKP}{\textbf{The KP divisor on  $\Gamma$.}} The effective KP divisor $\DKP$ is a sum of $g$ simple poles, namely 
\begin{enumerate}
\item $k$ poles on $\Gamma_0$ coinciding with the Sato divisor at $\vec t= \vec t_0$;
\item $g-k$ poles $P_{\textup{\scriptsize dr}}^{(l)}\in \Gamma_l$ uniquely identified by the condition that, in the local coordinate induced by the orientation $\mathcal O$, $\zeta(P_{\textup{\scriptsize dr}}^{(l)}) =\gamma_{\textup{\scriptsize dr}, V_l}$,
and $V_l$, $l\in [n+1,g+n-k]$, are the trivalent white vertices not containing Darboux edges. 
\end{enumerate}
\end{definition}

\begin{remark} 
We remark that the position of the $g-k$ poles on $\Gamma\backslash \Gamma_0$ is the same if we use the network divisor ${\mathcal D}_{\textup{\scriptsize dr},{\mathcal N}}$  defined in (\ref{eq:net_div_N}) instead of $\DDN$. Therefore in applications one can work directly with the original network $\mathcal N$ to construct the divisor on $\Gamma$ for the given soliton data.
\end{remark}

By construction $\DKP$ is independent on the gauge ray direction, on the weight gauge, on the vertex gauge and on the position of the Darboux points. We now prove that it is also independent on the orientation of the network.
Indeed for any fixed orientation $\mathcal O$, on the copy $\Gamma_l$ corresponding to the trivalent white vertex $V_l$, we have four real ordered marked points: $P_m^{(l)}$, $m\in [3]$, represented by the ordered edges at $V_l$, and the dressed divisor point $P_{\textup{\scriptsize dr}}^{(l)}$ with local coordinate $\zeta(P_{\textup{\scriptsize dr}}^{(l)})=\gamma_{\textup{\scriptsize dr}, V_l}$.

Using Corollary \ref{cor:indep_orient}, in Theorem \ref{theo:inv} we show that any change of orientation in $\mathcal N^{\prime}$acts on the four marked points in $\Gamma_l$ as a change of coordinate, so in particular it leaves invariant the position of the pole divisor in the oval.

\begin{theorem}\textbf{Invariance of $\DKP$ on network orientation}\label{theo:inv}
Let ${\mathcal D}_{\textup{\scriptsize dr}, {\mathcal N}^{\prime}, \mathcal O_1}$ and ${\mathcal D}_{\textup{\scriptsize dr}, {\mathcal N}^{\prime}, \mathcal O_2}$, be the dressed network divisors on $\mathcal N^{\prime}$ w.r.t the orientations $\mathcal O_s$, $s=1,2$, as in Definition \ref{def:vac_div_gen}, and let $\DKP^{(s)}$, $s=1,2$, be the corresponding KP divisors on $\Gamma$ as in Definition \ref{def:DKP}. Then 
\[
\DKP^{(2)} = \DKP^{(1)},
\] 
that is the divisor point on each component $\Gamma_l$,  $P_{\textup{\scriptsize dr}}^{(l)}$, is the same in both constructions
and the divisor numbers $\gdr^{(s)}$ are the local coordinates of $P_{\textup{\scriptsize dr}}^{(l)}$. Indeed, let $e^{(s)}_m$, $m\in [3]$, $s=1,2$, be the edges at the trivalent white vertex $V_l$, where $e^{(s)}_3$ is the unique edge pointing inwards at $V_l$ in orientation $\mathcal O_s$ (see also Figure \ref{fig:phi_orient}).
Let $\zeta_s$, $s=1,2$, be the local coordinates on $\Gamma_l$ respectively for orientations $\mathcal O_1$ and $\mathcal O_2$ of ${\mathcal N}^{\prime}$ and let $P^{(s)}_m$ be the marked point on $\Gamma_l$ corresponding to $e^{(s)}_m$ as in Definition
\ref{def:loccoor} (see also Figure \ref{fig:lcoord}).
Then 
\begin{enumerate}
\item If all edges at $V_l$ have the same versus in both orientations, then $P^{(2)}_m=P^{(1)}_m$, $m\in [3]$, the local coordinate on $\Gamma_l$ is the same $\zeta_1=\zeta_2$ in both orientations and (\ref{eq:inep_gauge_1}) implies that the divisor point on $\Gamma_l$ is the same and has the same coordinate
\[
\zeta_2(P_{\textup{\scriptsize dr}}^{(l)})=\zeta_1(P_{\textup{\scriptsize dr}}^{(l)});
\]
\item If at $V_l$ we change orientation of edges from $(e_1^{(1)},e_2^{(1)},e_3^{(1)})$ to $(e_2^{(2)},e_3^{(2)},e_1^{(2)})$ (see Figure \ref{fig:phi_orient}[left]), then $P^{(2)}_m=P^{(1)}_{m-1}, \,\,
(\!\!\!\!\mod 3)$ and $\zeta_2 = (1-\zeta_1)^{-1}$ on $\Gamma_l$. Therefore (\ref{eq:inep_gauge_2}) implies that the divisor point on $\Gamma_l$ is the same and its local coordinate changes as  
\[
\zeta_2(P_{\textup{\scriptsize dr}}^{(l)})=(1-\zeta_1(P_{\textup{\scriptsize dr}}^{(l)}))^{-1};
\] 
\item If at $V_l$ we change orientation of edges from $(e_1^{(1)},e_2^{(1)},e_3^{(1)})$ to $(e_3^{(2)},e_1^{(2)},e_2^{(2)})$ (see Figure \ref{fig:phi_orient}[right]), then $P^{(2)}_m=P^{(1)}_{m+1}, \,\, (\!\!\!\!\mod 3)$ and $\zeta_2 = \zeta_1/(\zeta_1-1)$ on $\Gamma_l$. Therefore (\ref{eq:inep_gauge_3}) implies that the divisor point on $\Gamma_l$ is the same and its local coordinate changes as
\[  
\zeta_2(P_{\textup{\scriptsize dr}}^{(l)}) = \frac{\zeta_1(P_{\textup{\scriptsize dr}}^{(l)})}{\zeta_1(P_{\textup{\scriptsize dr}}^{(l)})-1}.
\]
\end{enumerate}
\end{theorem} 

The proof follows directly from the definition of the KP wave function on $\Gamma$, Corollary \ref{cor:indep_orient} and the definition of $\DKP$.

By construction, each double point in $\Gamma$ corresponds to an edge; therefore the KP wave function takes the same value at all double points for all times and, at the double points $\kappa_j$ corresponding to edges at boundary vertices $b_j$, $j\in [n]$, it coincides for all times with the normalized Sato wave function.   
$\hat \psi$ is meromorphic of degree $\mathfrak d_{\mbox{\tiny KP}}\le g$, and its poles are all simple and belong to $\DKP$, which is contained in the union of all the real ovals of $\Gamma$. Therefore $\hat \psi$ satisfies the properties in Definition \ref{def:KPwave}, that is it is the KP wave function for the soliton data $(\mathcal K, [A])$ and the divisor $\DKP$ on $\Gamma$.

\begin{theorem}\label{lemma:KPeffvac}\textbf{$\hat \psi$ is the unique KP wave function on $\Gamma$ for $(\mathcal K, [A])$ and the divisor $\DKP$.}
Let $\hat \psi$, ${\mathcal D}_{\textup{\scriptsize dr}, {\mathcal N}^{\prime}}$, $\DKP$ on $\Gamma$ be as in Construction \ref{con:dress_gen} and Definitions \ref{def:vac_div_gen} and \ref{def:DKP}. Then $\hat \psi$ satisfies the following properties of Definition \ref{def:rrss} on $\Gamma\backslash\{P_0\}$. 
\begin{enumerate}
\item At $\vec t=\vec t_0$ $\hat \psi (P, \vec t_0)=1$ at all points $P\in \Gamma\backslash \{P_0\} $;
\item ${\hat \psi} (\zeta(P), \vec t)$ is real for real values of the local coordinate $\zeta$ and for all real $\vec t$ on each component of $\Gamma$;
\item $\hat \psi$ takes the same value at pairs of glued points $P,Q\in \Gamma$, for all $\vec t$:  $\hat \psi(P, \vec t) = \hat \psi(Q, \vec t)$;
\item $\hat \psi(\zeta, \vec t)$ is constant w.r.t. to the spectral parameter on each copy of $\mathbb{CP}^1$ corresponding to a boundary source or sink vertex $V_{l}$, $l\in [n]$;
\item $\hat \psi(\zeta, \vec t)$ is either constant or meromorphic of degree one w.r.t. to the spectral parameter on each copy of $\mathbb{CP}^1$ corresponding to any other trivalent white vertex. $\hat \psi(\zeta, \vec t)$ is constant w.r.t. to the spectral parameter on each other copy of $\mathbb{CP}^1$; 
\item $\DKP+(\hat\psi(P,\vec t))\ge 0$ for all $\vec t$;
\item $\DKP$ does not depend on the choice of position of the Darboux points in $\Gamma_l$, $l\in [n]$.
\item The $\DKP$ divisor is contained in the union of all ovals.   
\end{enumerate}
\end{theorem}

The proof of the assertions is straightforward and is omitted. We remark that, in the special case of the Le--network the divisor $\DKP$ coincides with the one constructed in \cite{AG3}.
We complete the proof of Theorem \ref{theo:exist} in Section \ref{sec:comb}, where we show that
the KP divisor $\DKP$ satisfies the regularity and reality conditions (Items (\ref{item:defodd_KP}) and (\ref{item:defeven_KP}) of Definition \ref{def:real_KP_div}).

\subsection{Global parametrization of positroid cells via KP divisors: the case $Gr^{TP}(1,3)$}\label{sec:global}
\begin{figure}
  \centering{\includegraphics[width=0.7\textwidth]{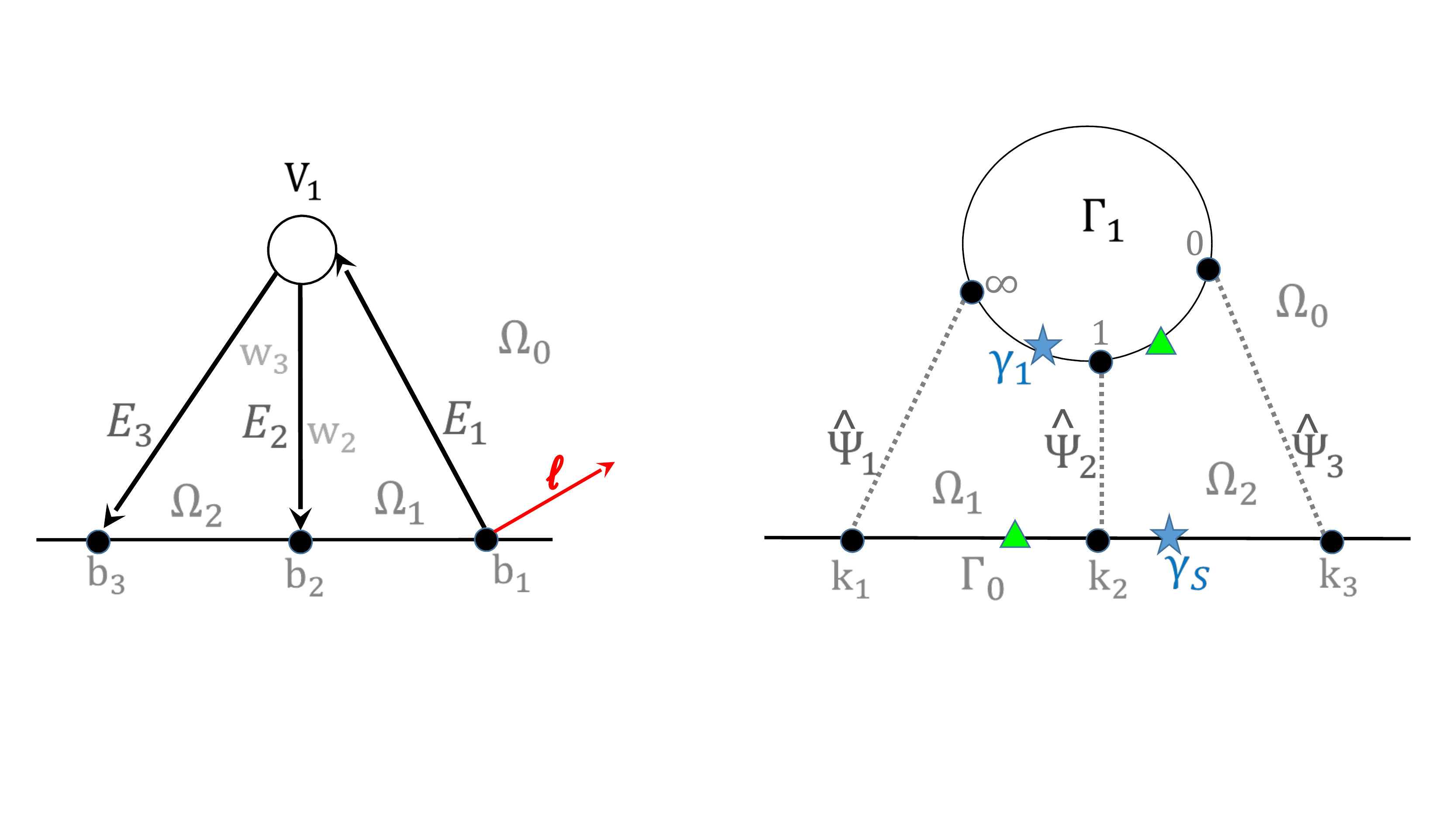}}
	\vspace{-1.7 truecm}
  \caption{\small{\sl We consider the issue of the global parametrization of positroid cells via divisors in the simplest example $Gr^{TP}(1,3)$.}\label{fig:global}}
\end{figure}
Next we address the issue of the global parametrization of positroid cells via divisors in the simplest example of blow-up. As a case study let us take the totally positive Grassmannian $Gr^{TP} (1,3)$.
With usual affine coordinates, we have the cell parametrization $[1,w_2,w_3]$, the heat hierarchy solution $f(\vec t)= e^{\theta_1 (\vec t)}+w_2 e^{\theta_2 (\vec t)}+w_3 e^{\theta_3 (\vec t)}$ and the Darboux transformation ${\mathfrak D}=\partial_x -\frac{\partial_x f(\vec t)}{f(\vec t)}$. 
Then on the oriented network (Figure \ref{fig:global}[left]) the vectors are 
\[
E_1 =E_2+E_3, \quad\quad E_2 =(0,w_2,0), \quad\quad E_3 =(0,0,w_3),
\]
the edge wave function takes the value
\[
\Psi_1 (\vec t) =\Psi_2(\vec t)+\Psi_3(\vec t), \quad\quad \Psi_2(\vec t) =w_2 (\kappa_2-\gamma_S) e^{\theta_2(\vec t)}, \quad\quad \Psi_3(\vec t) =w_3 (\kappa_3-\gamma_S) e^{\theta_3(\vec t)},
\]
where $\gamma_S$ is the coordinate of the Sato divisor point at the initial time $\vec t_0=\vec 0=(0,0,0,\dots)$,
\[
\gamma_S= \gamma_S (w_2,w_3,\vec 0)= \frac{\kappa_1 +w_2 \kappa_2 + w_3 \kappa_3}{1 +w_2  + w_3 }.
\]
On the curve $\Gamma=\Gamma_0\sqcup\Gamma_1$, at the double points the normalized KP wave function is $\hat \Psi_j(\vec t) =\frac{\Psi_j(\vec t)}{\Psi_j(\vec 0)}$ and the divisor point $\gamma_1 \in \Gamma_1$ in the local coordinates induced by the orientation of the network takes the value
\[
\gamma_1 = \gamma_1 (w_2,w_3,\vec 0)= \frac{w_3(\kappa_3-\gamma_S)}{w_3(\kappa_3-\gamma_S)+w_2(\kappa_2-\gamma_S)}.
\]
It is easy to check that the positivity of the weights is equivalent to $\gamma_S \in]\kappa_1,\kappa_3[$, $\gamma_1>0$ and the fact that there is exactly one divisor point in each one of the finite ovals, $\Omega_1$ and $\Omega_2$, that is
\begin{enumerate}
\item Either $\kappa_1 <\gamma_S<\kappa_2$ and $\gamma_1<1$, {\sl i.e.} the Sato divisor point belongs to $\Omega_1$ and the other divisor point to $\Omega_2$. In Figure \ref{fig:global} [right] we illustrate this case representing divisor points by triangles;
\item Or $\kappa_2 <\gamma_S<\kappa_3$ and $\gamma_1>1$, {\sl i.e.} the Sato divisor point belongs to $\Omega_2$ and the other divisor point to $\Omega_1$. In Figure \ref{fig:global} [right] we illustrate this case representing divisor points by stars.
\end{enumerate}
We also remark that the transformation from $(w_2,w_3)$ to $(\gamma_S,\gamma_1)$ looses injectivity and has null Jacobian along the line $w_3=\frac{\kappa_2-\kappa_1}{\kappa_3-\kappa_2}$ so that, for any $w_2>0$,
\[
\gamma_S ( w_2, \frac{\kappa_2-\kappa_1}{\kappa_3-\kappa_2} )=\kappa_2, \quad\quad \gamma_1 ( w_2, \frac{\kappa_2-\kappa_1}{\kappa_3-\kappa_2} )=1.
\]

If we invert the relation between divisor numbers and weights, we get
\[
w_2 (\gamma_S,\gamma_1)= \frac{(\gamma_1 -1)(\gamma_S-\kappa_1)}{\gamma_S-\kappa_2},\quad\quad w_3 (\gamma_S,\gamma_1)= \frac{\gamma_1(\gamma_S-\kappa_1)}{\kappa_3-\gamma_S}.
\]
Therefore in the non--generic case when $\gamma_S\to\kappa_2$ and $\gamma_1\to 1$, we need to apply the blow-up procedure at the point $(\gamma_S,\gamma_1)=(\kappa_2,1)$, by setting $\gamma_S=\kappa_2+\epsilon$, $\gamma_1=1+z\epsilon$ and take the limit $\epsilon\to 0$, so that
\[
w_2 (\kappa_2,1) = z(\kappa_2-\kappa_1),\quad\quad w_3 (\kappa_2,1) = \frac{\kappa_2-\kappa_1}{\kappa_3-\kappa_2}.
\]

\begin{figure}
  \centering{\includegraphics[width=0.47\textwidth]{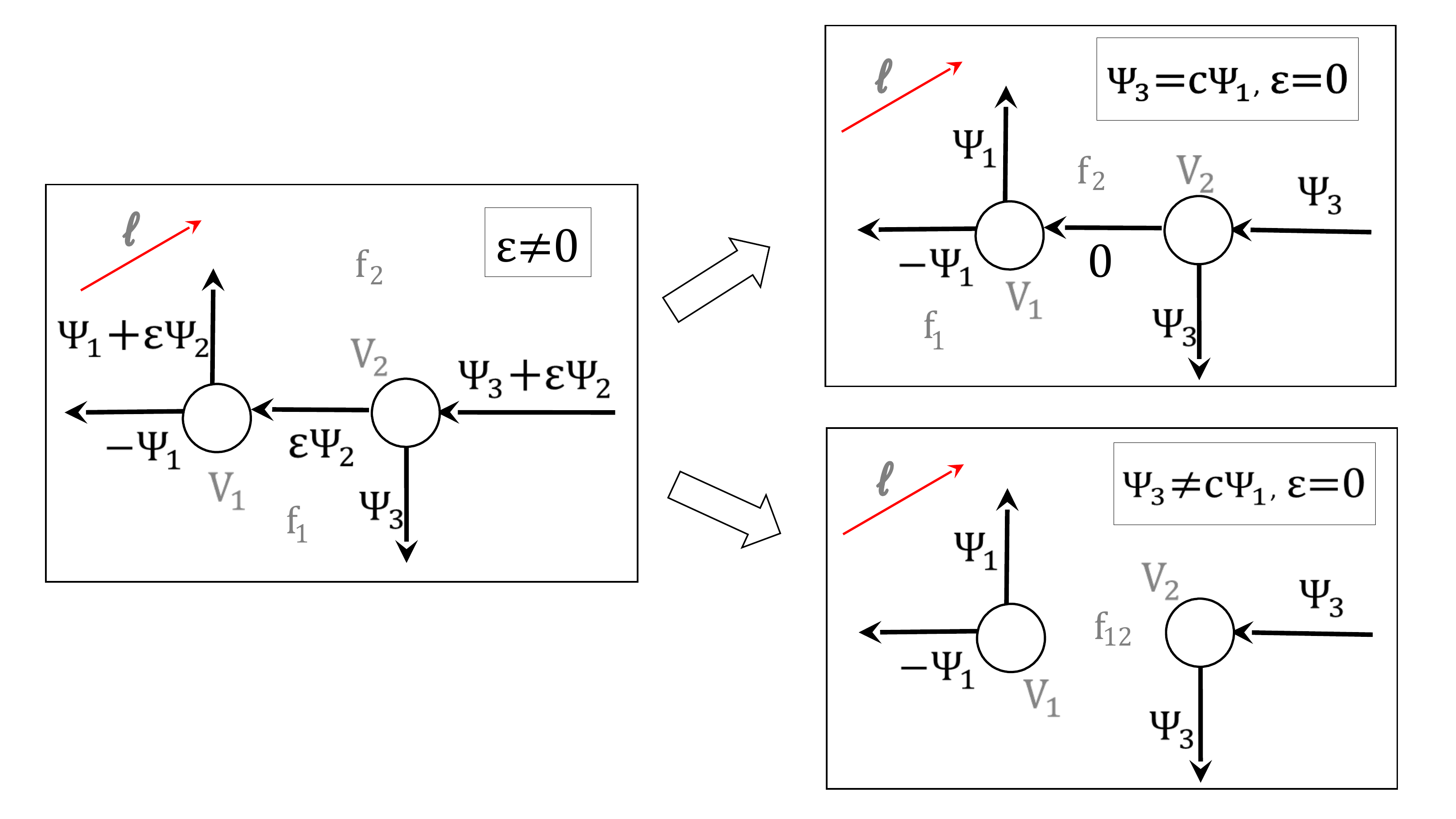}
	\hfill
	\includegraphics[width=0.47\textwidth]{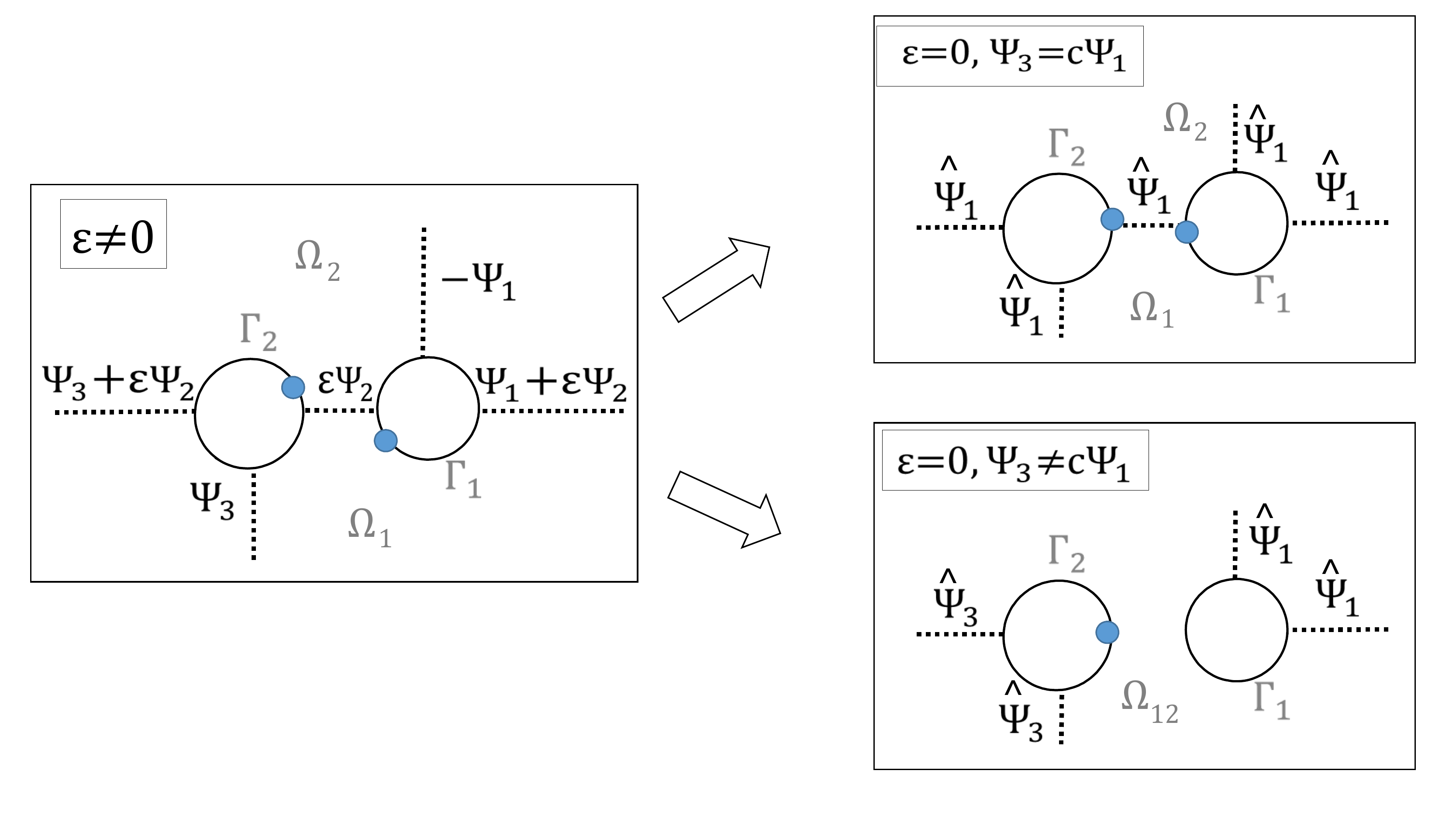}}
	\vspace{-.4 truecm}
  \caption{\small{\sl We present the construction of the divisor in case of edges carrying null vectors in the simplest case: we show both the edge wave function $\Psi$ on the oriented network [left] and the normalized wave function and divisor on the curve [right] and in the case of edges carrying null vectors of type 1 [top] and type 2 [bottom].}\label{fig:div_null}}
\end{figure}
\subsection{Construction of the divisor in the case of null edge vectors}\label{sec:constr_null}
We now  discuss the modification of the technical construction in case of null edge vectors on a simple example.
In Definition \ref{def:null_type_1_2} edges carrying null vectors have been classified of type 1 and 2 depending on the properties of the not--null edge vectors on the edges adjacent to them. Without loss of generality, in the following we assume that all vertices belonging to a given maximal connected subgraph carrying null vectors are trivalent.

\textbf{Construction of the wave function and divisor for soliton data associated to a network carrying null vectors of type 1:}
If the null edge vectors are just of type 1, then the normalized wave function takes the same value at all edges adjacent to a given connected maximal subgraph carrying null edge vectors. Then we simply extend analytically the normalized KP wave function to the components of $\Gamma$ associated to the given subgraph, by attributing the same value of the normalized KP wave function also at the edges carrying null vectors. Since the resulting wave function is regular and constant with respect to the spectral parameters on all components associated to such connected maximal subgraph, 
the position of the divisor point on components corresponding to trivalent white vertices carrying null vectors is completely irrelevant: we just attribute a divisor point to each component corresponding to a trivalent white vertex carrying null vectors using the counting rule established in \cite{AG1}. We remark that there is a certain freedom in the attribution of divisor points to ovals. In any case, the final divisor is real and contained in the union of all the ovals. Finally, we claim that it is always possible to attribute exactly one divisor point in each finite oval and no divisor point in the infinite oval in all cases of connected maximal subgraphs carrying null edge vectors of type one. We plan to discuss thoroughly this issue in a future publication and below we present two examples.

In the simplest case the connected maximal subgraph carrying null edge vectors contains only one edge and its boundary consists of two trivalent white vertices. We present such an example in Figure \ref{fig:div_null}, case $\epsilon =0$ and $\Psi_3 = c \Psi_1$, $c\not =0$. In such case, the null edge necessarily bounds two distinct faces, $f_1$ and $f_2$, the network divisor numbers at the vertices $V_1$ and $V_2$ coincide with the coordinates of the corresponding points on the components $\Gamma_1$ and $\Gamma_2$ and the divisor points on $\Gamma$ are trivial by definition. Then we may use the counting rule to attribute exactly one divisor point to each the corresponding ovals $\Omega_1$ and $\Omega_2$. 

In more complicated cases, like that one constructed in Figure \ref{fig:div_null_new}, the construction of the divisor is carried out in a similar way. We attribute one divisor point to each component corresponding to a trivalent white vertex using the counting rule if there is an edge carrying a null vector and we do not attribute any divisor point to black vertices. 
For instance, in the case of Figure \ref{fig:div_null_new} we have that the Sato divisor point $\gamma_s$ belongs to the intersection of the oval $\Omega_1$ and the Sato component $\Gamma_0$ which corresponds to the boundary of the disk.
The remaining divisor numbers are
\[
\gamma_1 = \frac{q+p+1}{2p+1}, \quad\quad \gamma_2 = \frac{p(1+2q)}{p-q}, \quad\quad \gamma_3 = \frac{q-p}{1+2q}.
\]
Therefore, upon denoting with the same symbol both the divisor point and its local coordinate, we get: 
\begin{enumerate}
\item $\gamma_1 \in \Gamma_1 \cap \Omega_3$, $\gamma_2 \in \Gamma_2 \cap \Omega_2$ and $\gamma_3 \in \Gamma_3 \cap \Omega_4$ if
$p>q$. We represent one configuration with stars in Figure \ref{fig:div_null_new} [right];
\item $\gamma_1 \in \Gamma_1 \cap \Omega_2$, $\gamma_2 \in \Gamma_2 \cap \Omega_4$ and $\gamma_3 \in \Gamma_3 \cap \Omega_3$ if
$p<q$. We represent one configuration with triangles in Figure \ref{fig:div_null_new} [right].
\end{enumerate}

\begin{figure}
  \centering{\includegraphics[width=0.49\textwidth]{NullVector_2.pdf}
	\hfill
	\includegraphics[width=0.49\textwidth]{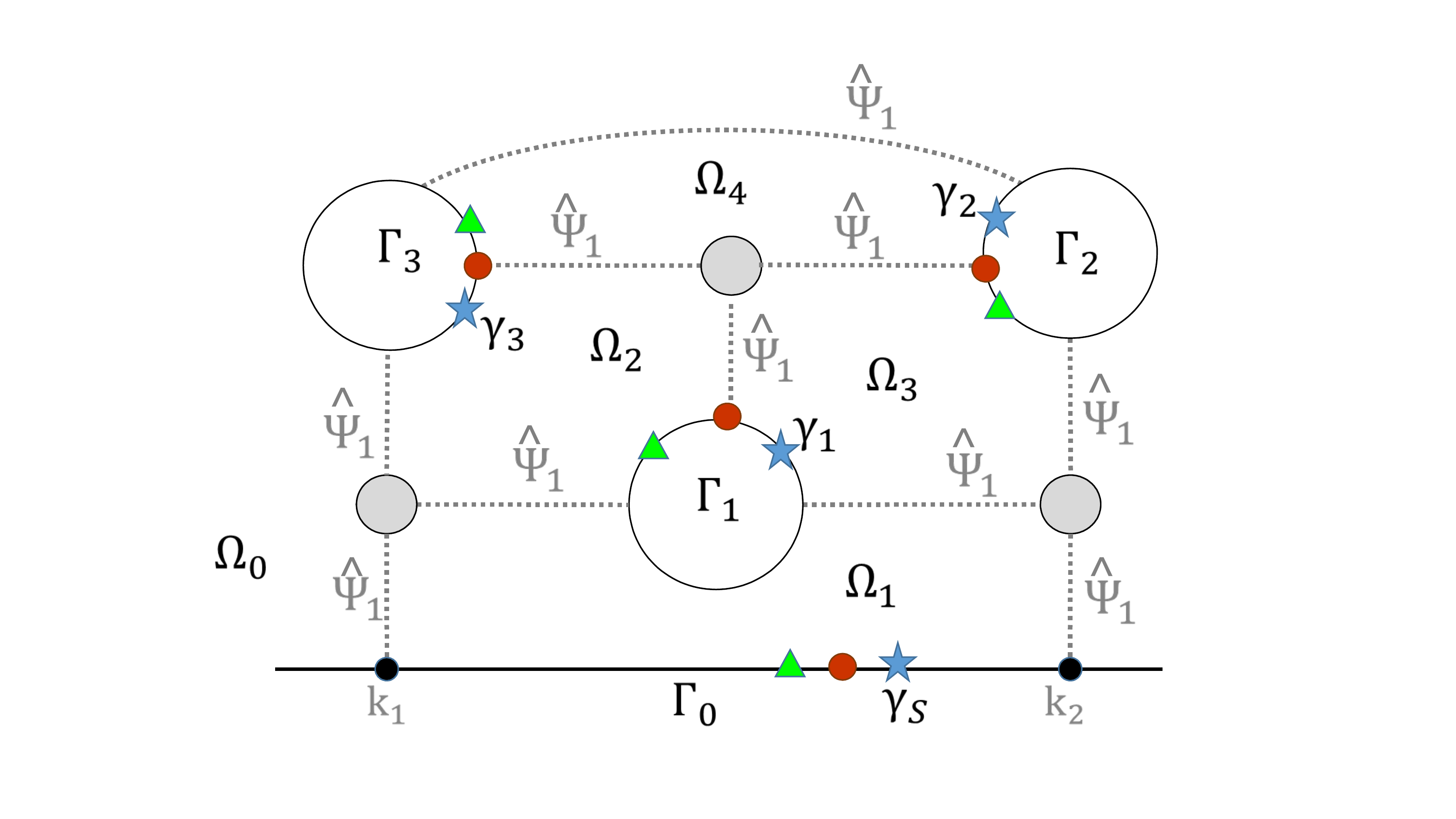}}
	\vspace{-.7 truecm}
  \caption{\small{\sl We reproduce the network in Figure \ref{fig:zero-vector} possessing null edge vectors $E_u,E_v$ and $E_w$ when $p=q$ [left] and show the KP divisor on the curve in the case $q<p$ (stars), $q=p$ (balls) and $q>p$ (triangles), where $\gamma_S$ is the Sato divisor point.}\label{fig:div_null_new}}
\end{figure}

\smallskip

\textbf{Construction of the wave function and divisor for soliton data associated to a network carrying null vectors of type 2:}
If there do appear also null edge vectors of type 2, then the problem of continuation of the wave function is untrivial since
the normalized wave function takes different value at edges adjacent to a given connected maximal subgraph carrying null edge vectors. Then we need to modify the graph eliminating such edges of type 2 and the curve $\Gamma$ accordingly. Since elimination of edges causes a decrease in the number of faces, we need also to eliminate some trivial divisor points in order to keep the degree of the divisor consistent with the genus of the curve. Again there is a certain freedom in the attribution of divisor points to ovals. In any case, the final divisor is real and contained in the union of all the ovals. Finally, we claim that it is always possible to attribute exactly one divisor point in each finite oval and no divisor point in the infinite oval also in this case. 

We plan to discuss thoroughly this issue in a future publication and below we present just the simplest example in Figure \ref{fig:div_null}, case $\epsilon =0$ and $\Psi_3 \ne c \Psi_1$, $c\not =0$.
The connected maximal subgraph carrying null edge vectors contains only one edge and its boundary consists of two trivalent white vertices. In such case, the null edge necessarily bounds two distinct faces, $f_1$ and $f_2$, the network divisor numbers at the vertices $V_1$ and $V_2$ coincide with the coordinates of the corresponding points on the components $\Gamma_1$ and $\Gamma_2$ and the divisor points on $\Gamma$ are trivial by definition. Contrary to the case of null edge vectors of type 1, we cannot extend analytically the wave function to include the edge carrying the null vector. Therefore we eliminate such edge from the network and we merge the two faces $f_1,f_2$ into a single one $f_{12}$ eliminating a double point and a face in $\Gamma$. We also eliminate one divisor point on the reduced curve so to have only one divisor point in the face $\Omega_{12}$: for instance we eliminate the divisor point on $\Gamma_1$ and keep that on $\Gamma_2$ in the Figure, but we could also do the opposite.

\section{Combinatorial characterization of the regularity of $\DKP$}\label{sec:comb}

In this Section we show that the position of each (vacuum or KP) divisor point in the intersection of $\Gamma_l$ with the ovals is determined by an index of each pair of edges at the corresponding trivalent white vertex $V_l$. We use such index to provide a combinatorial proof that there is exactly one KP divisor point in each finite oval thus completing the proof of Theorem \ref{theo:exist}. Then, in Section \ref{sec:lam1} we use the formalism of relations on half--edges introduced in \cite{Lam2} to restate the main Theorems of our paper in such framework.

In the following $a,b,c$ denote either the edge at $V_l$ or the corresponding marked point on $\Gamma_l$, $w_a$ is the weight of the edge $a$, $f_{ab}$ is the face of ${\mathcal N}$ bounded by the edges $a,b$ and $\Omega_{ab}$ is the corresponding oval in $\Gamma(\mathcal G)$ (see also Figure \ref{fig:corr_V_G}).
\begin{figure}
  \centering
  \includegraphics[width=0.56\textwidth]{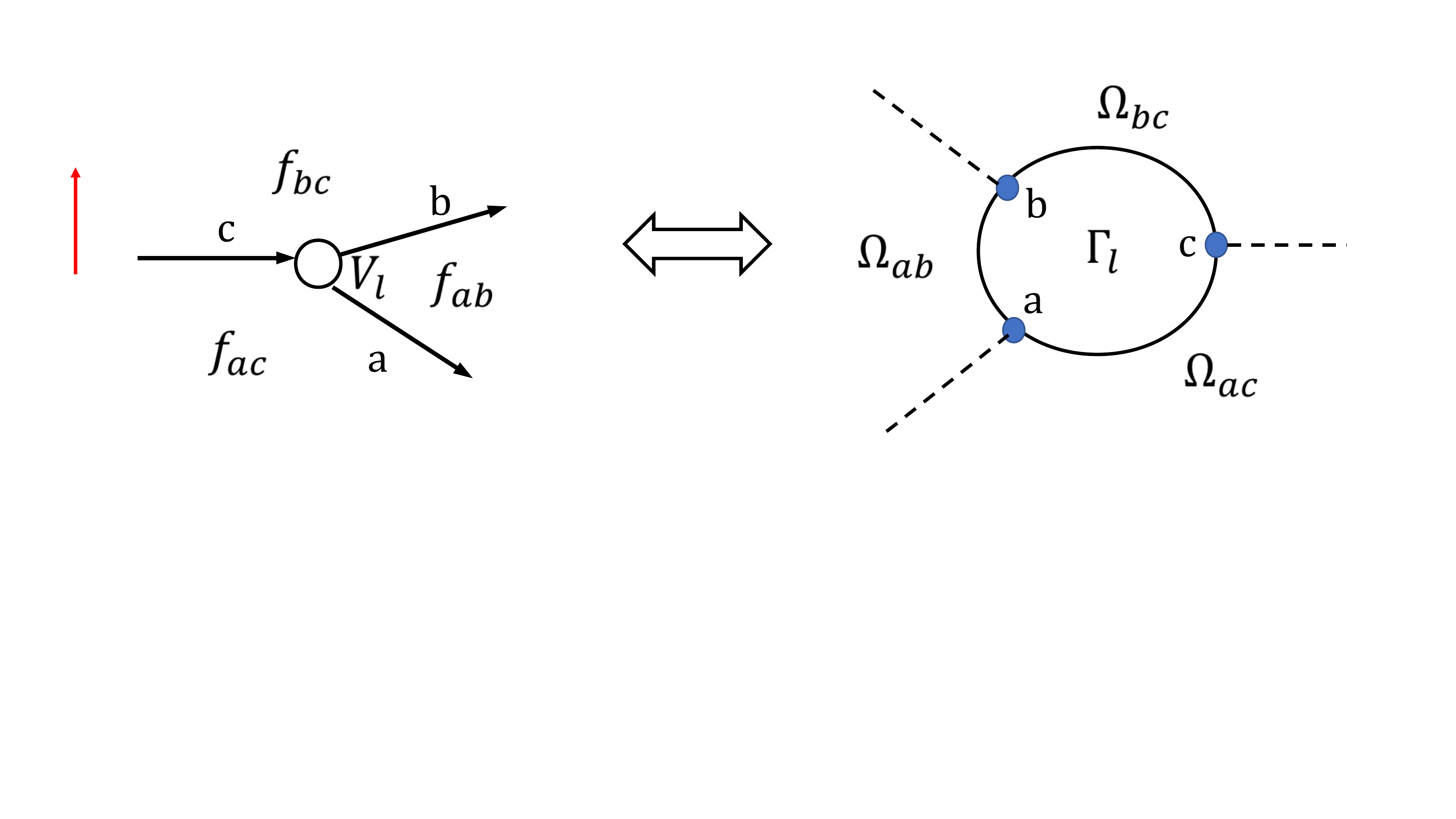}
	\vspace{-2.5 truecm}
  \caption{\small{\sl The correspondence between faces at $V_l$ (left) and ovals bounded by $\Gamma_l$ (right) under the assumption that the curve is constructed reflecting the graph w.r.t. a vertical ray ({\sl i.e.} we assume that all boundary source and sink vertices in the original network lay on a horizontal line).}}
	\label{fig:corr_V_G}
\end{figure}

Let us start recalling the definition of the (vacuum or dressed) network divisor number 
associated to $V_l$. Let $\Psi_s = \Psi_{s,{\mathcal N}, \mathcal O, \mathfrak l} (\vec t_0)$, be either the vacuum or the dressed e.w. at the 
edges $s=a,b,c$ of a given trivalent white vertex $V_l$ as in Definition \ref{def:vvw_gen}.
Then the linear relation at $V_l$ is
\[
\Psi_c = (-1)^{\mbox{int}(c)} w_c \left( (-1)^{\mbox{wind}(c,a)}\Psi_a + (-1)^{\mbox{wind}(c,b)}\Psi_b  \right),
\]
and, in the local coordinates induced by the orientation of ${\mathcal N}^{\prime}$, the local coordinate of the (vacuum or KP) divisor point is equal to the vacuum divisor number 
\begin{equation}\label{eq:formula_div}
\begin{array}{ll}
\gamma &\displaystyle= \frac{ (-1)^{\mbox{wind}(c,a)}\Psi_a}{(-1)^{\mbox{wind}(c,a)}\Psi_a +  (-1)^{\mbox{wind}(c,b)}\Psi_b} = w_c(-1)^{\mbox{int}(c)+\mbox{wind}(c,a)}\frac{\Psi_a}{\Psi_c}\\
&\displaystyle =\left( 1+ (-1)^{\mbox{wind}(c,b)-\mbox{wind}(c,a)}\frac{\Psi_b}{\Psi_a}\right)^{-1}= 1-w_c(-1)^{\mbox{int}(c)+\mbox{wind}(c,b)} \frac{\Psi_b}{\Psi_c}.
\end{array}
\end{equation}

\begin{definition}\textbf{The set of indices at each pair of edges at a given vertex}\label{def:index_pair}
To each pair of edges $f,g$ at a white or black vertex $V_l$ we assign the following indices: 
\begin{enumerate}
\item \textbf{Index of change of direction} $\epsilon_{cd}(f,g)$: it marks the fact that both edges at the vertex $V_l$ are pointing outwards or inwards: 
\[
\epsilon_{cd}(f,g) = \left\{ \begin{array}{ll}  1 & \mbox{ if both } f, g \mbox{ point outwards at the white vertex } V_l,\\
1 & \mbox{ if both } f, g \mbox{ point inwards at the black vertex } V_l,\\
0 & \mbox{ otherwise}.
\end{array}\right.
\]
\item  \textbf{Index of change of sign} $\epsilon_{cs}(f,g)$: it marks the sign of the product of the (vacuum or dressed) wave function $\Psi_f\Psi_g$ at the corresponding marked points $f,g\in\Gamma_l$ : 
\[
\epsilon_{cs}(f,g) = \left\{\begin{array}{ll}  1 & \mbox{ if  } \Psi_f\Psi_g<0, \quad \mbox{ or } \quad \Psi_f\Psi_g=0  \mbox{ and } \Psi_f+ \Psi_g<0,\\
0 & \mbox{ otherwise}.
\end{array}\right.
\]
\item \textbf{Index of intersection} $\epsilon_{int}(f,g)$: it marks the number of intersections with gauge rays with edges pointing inward at $V_l$ : 
\[
\epsilon_{int}(f,g) = \left\{\begin{array}{ll}  \mbox{int} (f) & \mbox{ if  } f \mbox{ points inwards at }  V_l \mbox{ and } g \mbox{ outwards},\\
\mbox{int} (g) & \mbox{ if  } g \mbox{ points inwards at } V_l \mbox{ and } f \mbox{ outwards },\\
\mbox{int} (f)+\mbox{int} (g)  & \mbox{ if  both } f, g \mbox{ point inwards at } V_l,\\
0 & \mbox{ otherwise}.
\end{array}\right.
\]
\item \textbf{Index of winding} $\epsilon_{wind}(f,g)$: it coincides with the winding of the pair $(f,g)$ if one of the edges points inwards and the other outwards at $V_l$, 
otherwise it is the difference of the winding indices at $V_l$: 
\[
\epsilon_{wind}(f,g) = \left\{\begin{array}{ll}  \mbox{wind} (f,g) & \mbox{if $f$ points inwards and $g$ outwards at } V_l,\\
-\mbox{wind} (g,f) & \mbox{if $g$ points inwards and $f$ outwards at } V_l,\\
\mbox{wind} (h,g) -\mbox{wind} (h,f)& \mbox{if } \epsilon_{cd}(f,g)=1 \mbox{ and }  h \mbox{ is the third edge at } V_l.\\
\end{array}\right.
\]
\item \textbf{Total index} $\epsilon_{tot} (f,g)$: it is the sum of the indices of change of direction, change of sign, of intersection and of winding at the pair $(e,f)$:
\begin{equation}\label{eq:tot}
\epsilon_{tot} (f,g) = \epsilon_{cd}(f,g)+\epsilon_{cs}(f,g)+\epsilon_{int}(f,g)+\epsilon_{wind}(f,g)    \quad
(\!\!\!\!\!\!\mod 2).
\end{equation}
\end{enumerate}
\end{definition}

The following Lemma explains the relations among such indices at each vertex.

\begin{lemma}\label{lemma:count_eps}\textbf{Relations among the indices of a pair of edges at a vertex}.
Let $(f,g)$ be a pair of edges at the vertex $V_l$ and let $\epsilon_{cd} (f,g)$, $\epsilon_{wind} (f,g)$, $\epsilon_{int} (f,g)$, $\epsilon_{cs} (f,g)$ and  $\epsilon_{tot} (f,g) $ be as above.
Then the following holds true:
\begin{enumerate}
\item If $V_l$ is a bivalent vertex then 
\begin{equation}\label{eq:tot_biv}
\epsilon_{cs} (f,g) = \epsilon_{wind} (f,g)+\epsilon_{int} (f,g) \quad
(\!\!\!\!\!\!\mod 2),\quad\quad \epsilon_{tot} (f,g) = 0 \quad
(\!\!\!\!\!\!\mod 2);
\end{equation}
\item If $V_l$ is a trivalent black vertex then 
\begin{equation}\label{eq:tot_triv_bl}
\epsilon_{cs} (f,g) = \epsilon_{wind} (f,g)+\epsilon_{int} (f,g) \quad
(\!\!\!\!\!\!\mod 2),\quad\quad \epsilon_{tot} (f,g) = \epsilon_{cd} (f,g) \quad
(\!\!\!\!\!\!\mod 2);
\end{equation}
\item If $V_l$ is a trivalent white vertex and $e_m$, $m\in [3]$, are the edges at $V_l$, then
\begin{equation}\label{eq:tot_triv_wh}
\epsilon_{tot} (e_1,e_2) +\epsilon_{tot} (e_2,e_3)+\epsilon_{tot} (e_3,e_1) = 1 \quad
(\!\!\!\!\!\!\mod 2),
\end{equation}
and, moreover, exactly one among $\epsilon_{tot}(e_i,e_j)$ is $1 \,\,
(\!\!\!\!\mod 2)$.
\end{enumerate}
\end{lemma}

The proof of the lemma follows directly from the definitions of the indices and the linear relations at the vertices. We illustrate it in Figures \ref{fig:table_divisor1} and \ref{fig:table_divisor2}.

\subsection{The position of the divisor points in the ovals}\label{sec:position}
\begin{figure}
	 \includegraphics[width=0.46\textwidth]{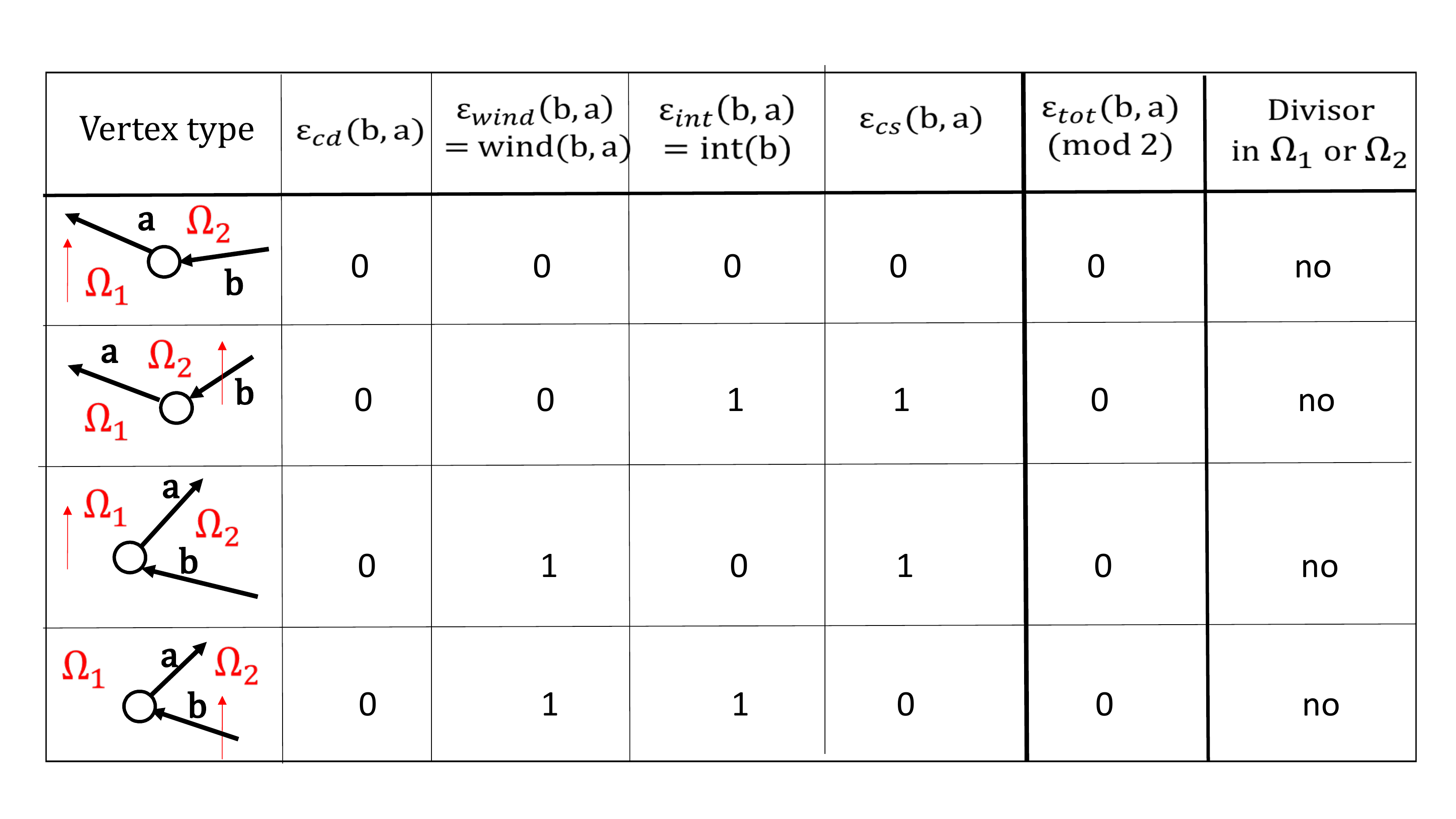}
	\hspace{.5 truecm}
	 \includegraphics[width=0.46\textwidth]{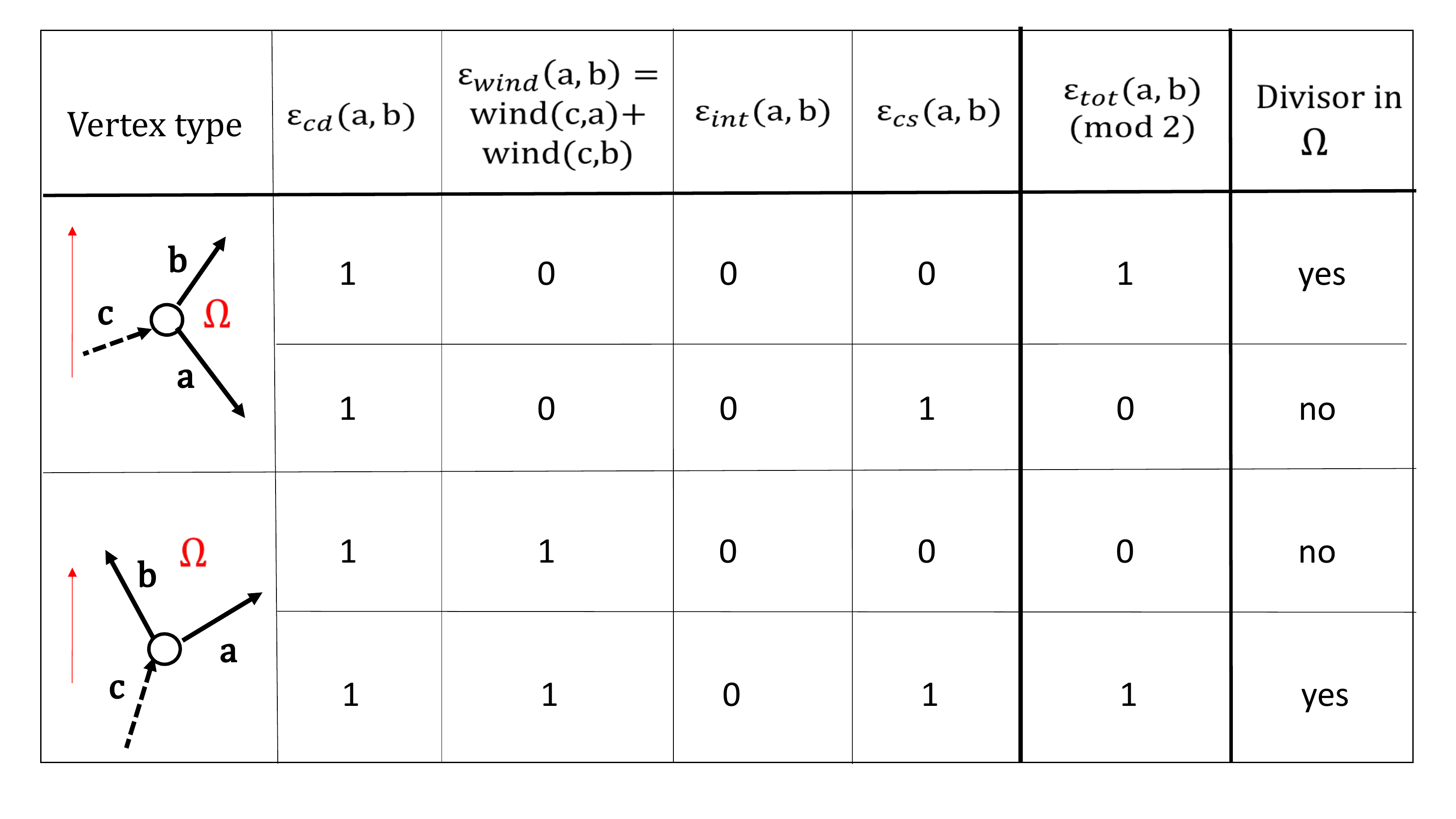}
	\includegraphics[width=0.46\textwidth]{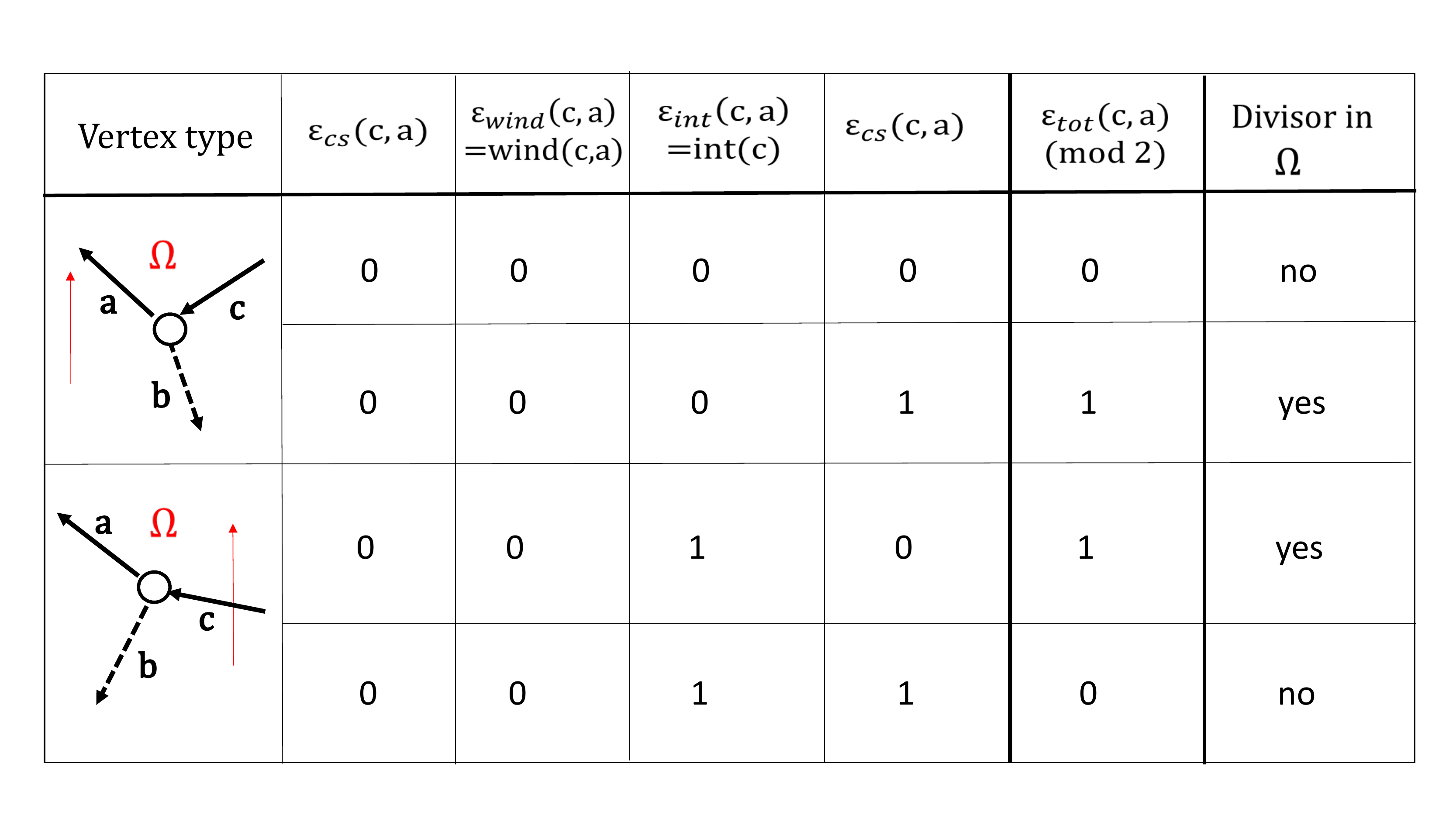}
	\hspace{.5 truecm}
	 \includegraphics[width=0.46\textwidth]{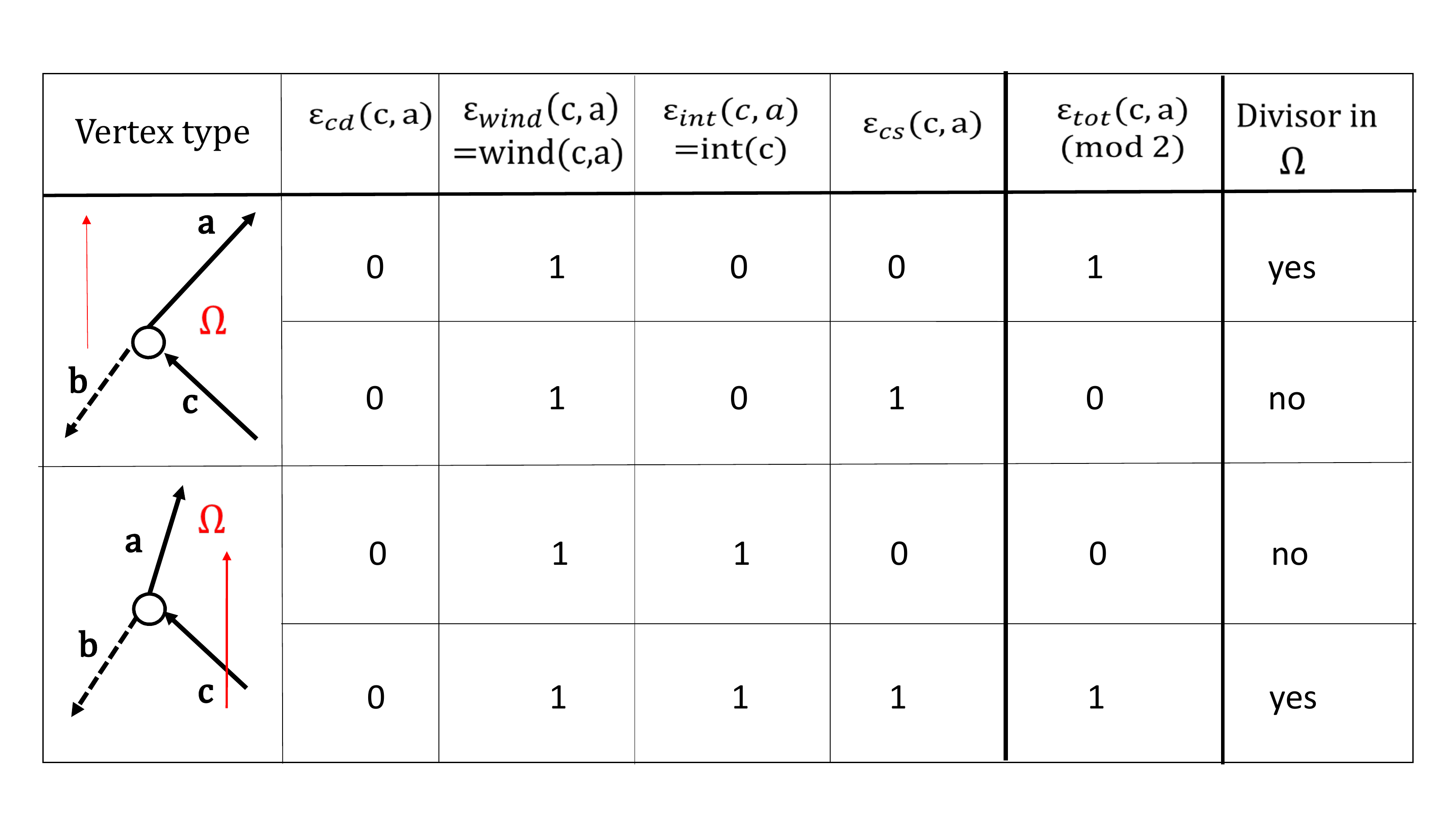}
	\includegraphics[width=0.46\textwidth]{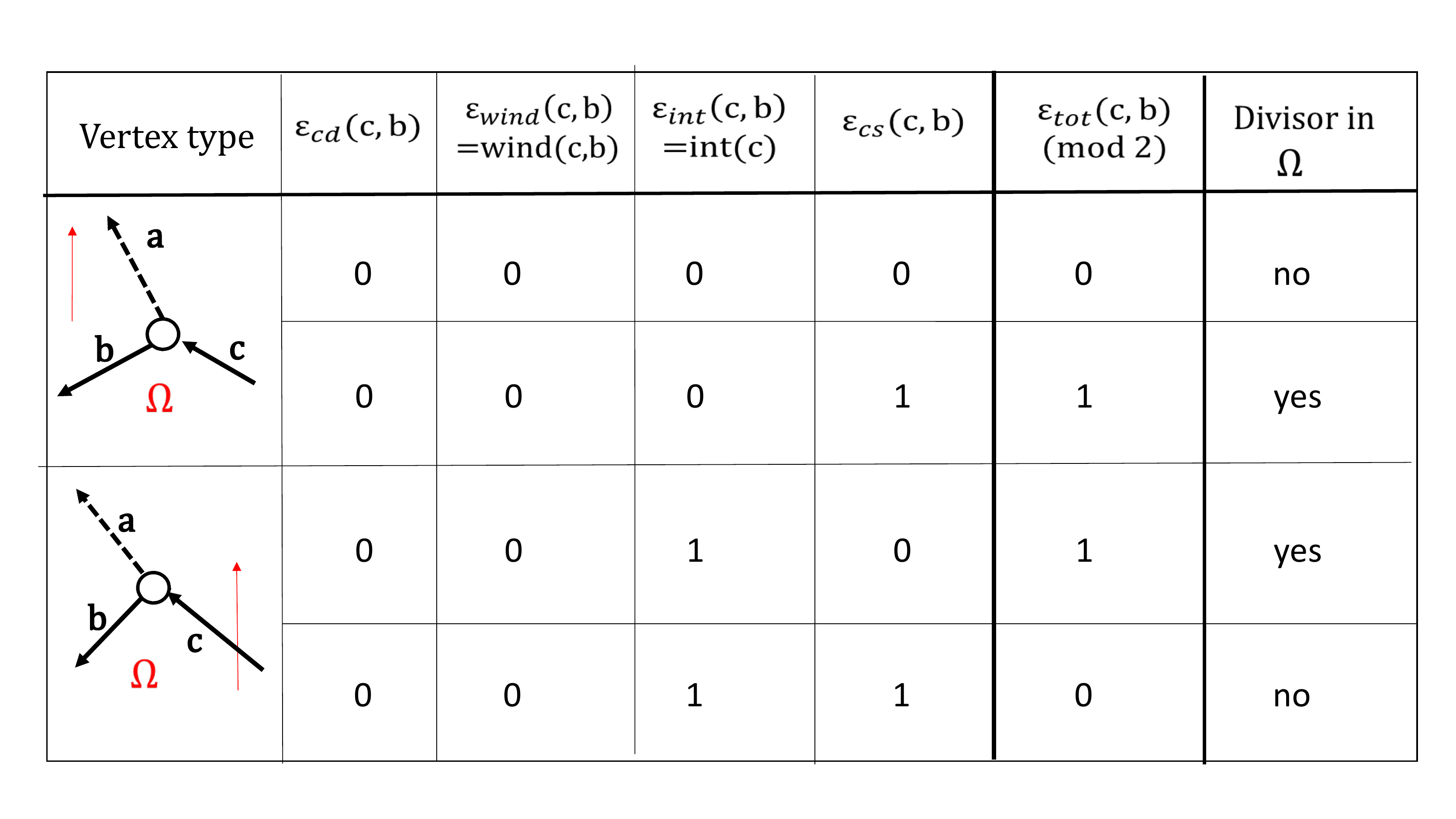}
	\hspace{.5 truecm}
	 \includegraphics[width=0.46\textwidth]{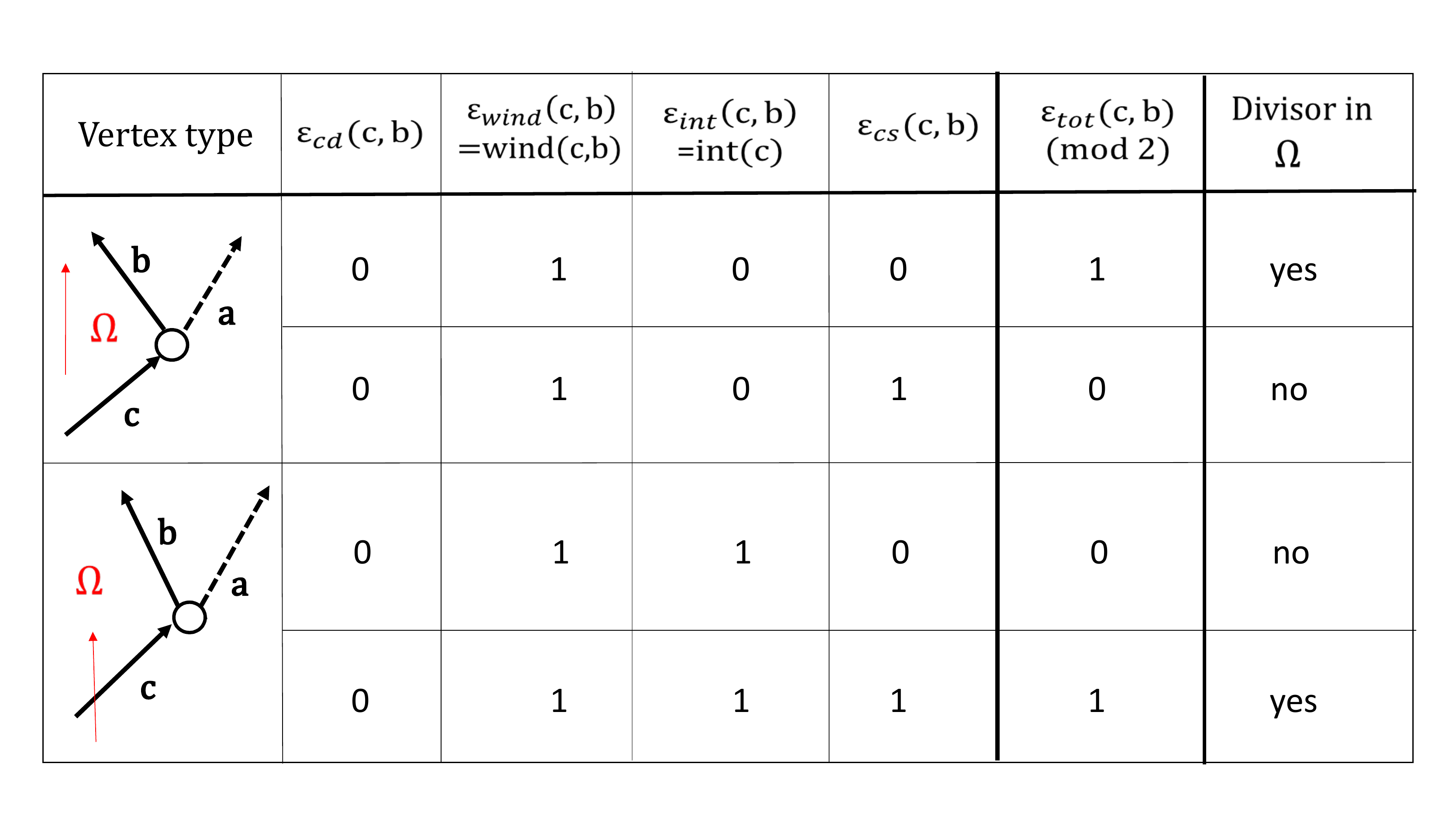}
\includegraphics[width=0.46\textwidth]{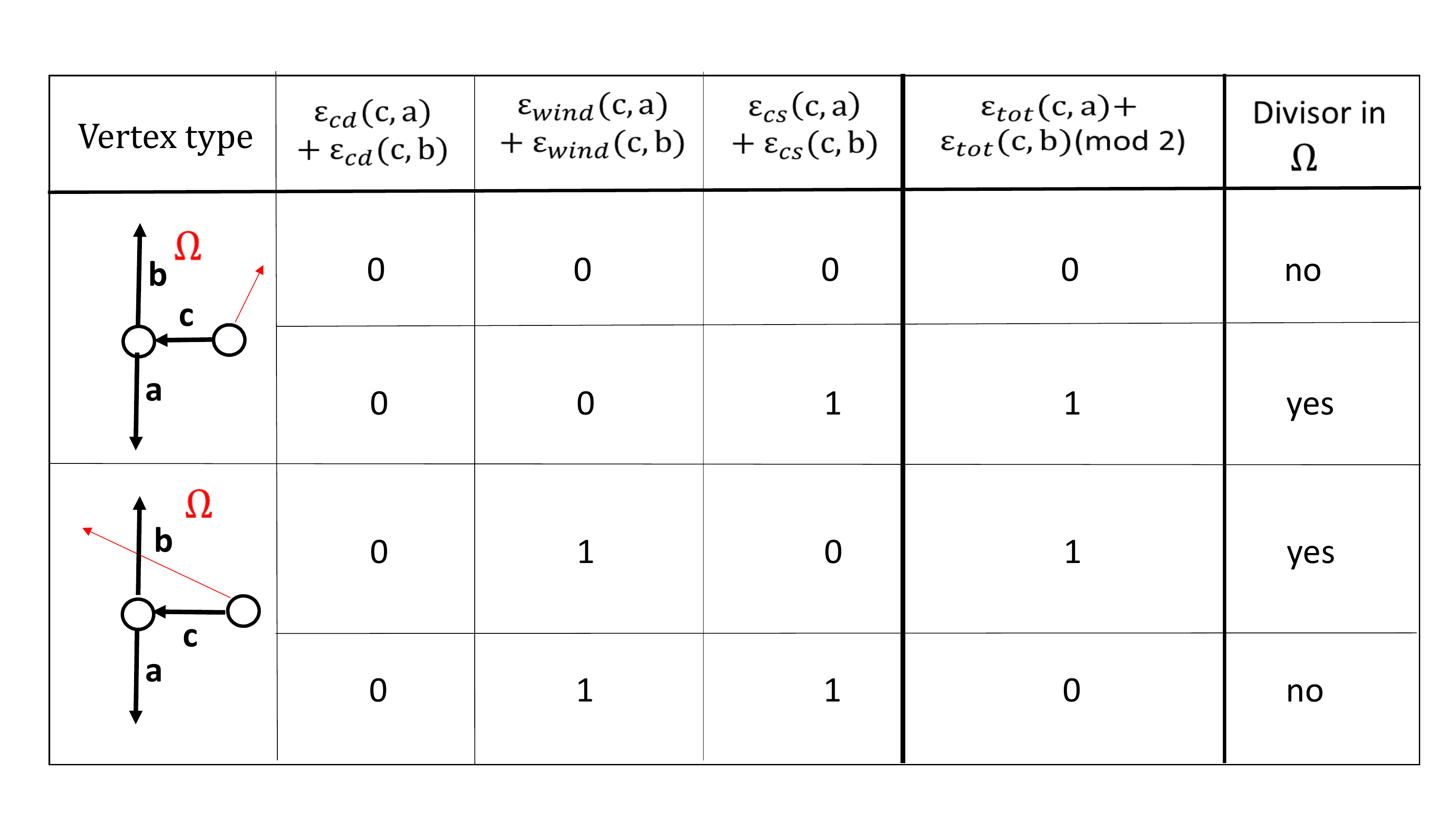}
	\hspace{.5 truecm}
	 \includegraphics[width=0.46\textwidth]{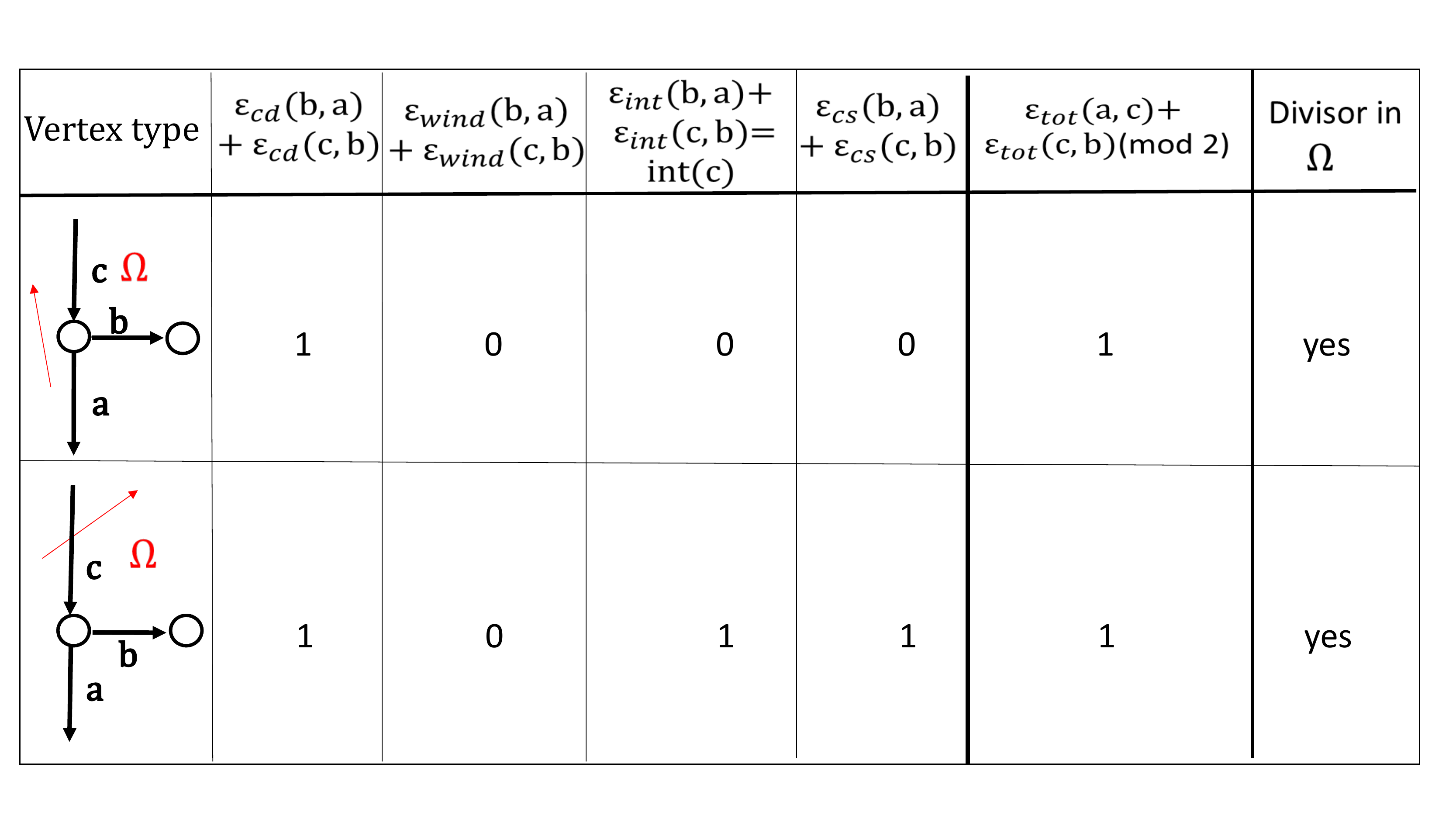}
	\includegraphics[width=0.46\textwidth]{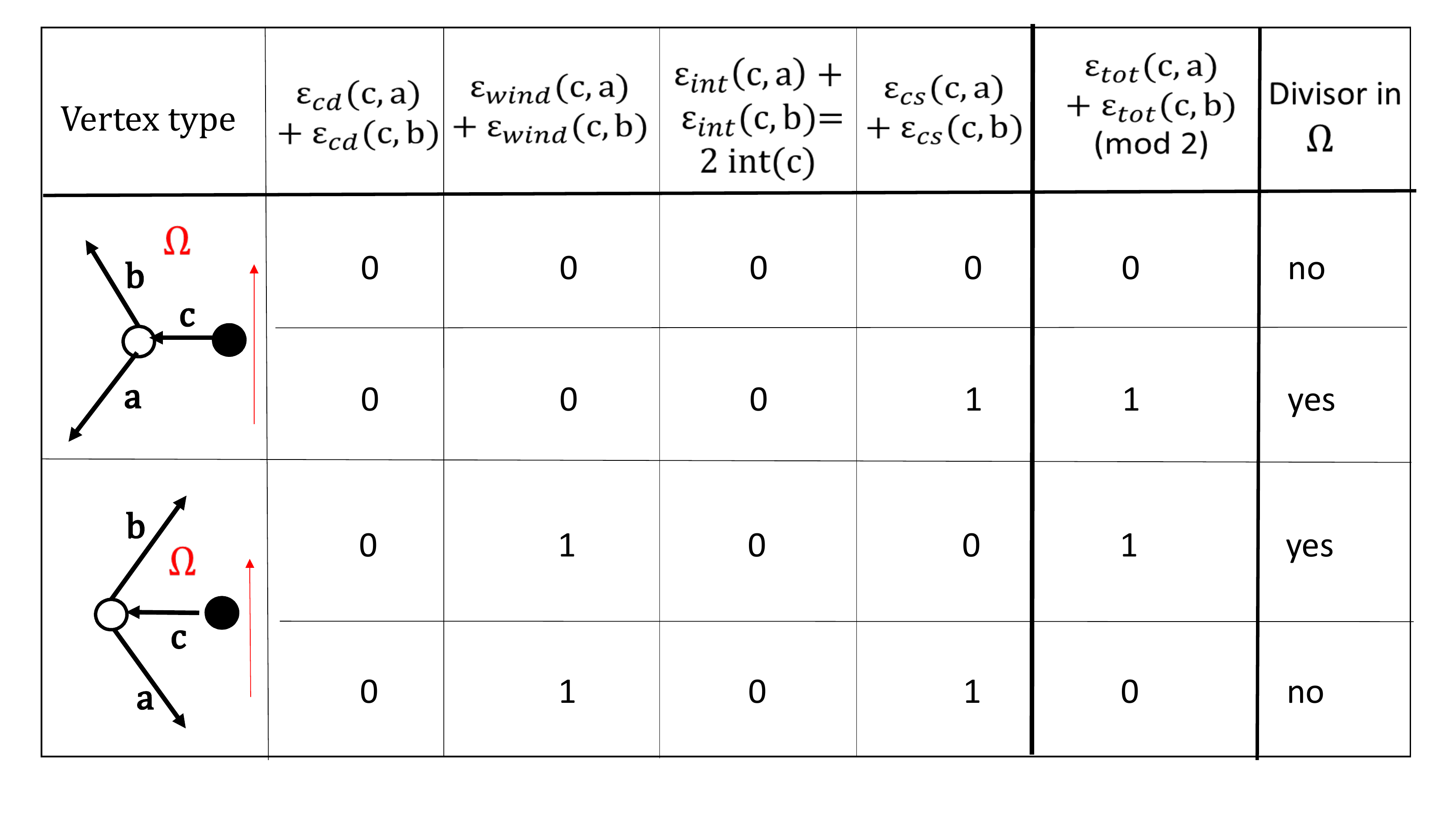}
	\hspace{.5 truecm}
	 \includegraphics[width=0.46\textwidth]{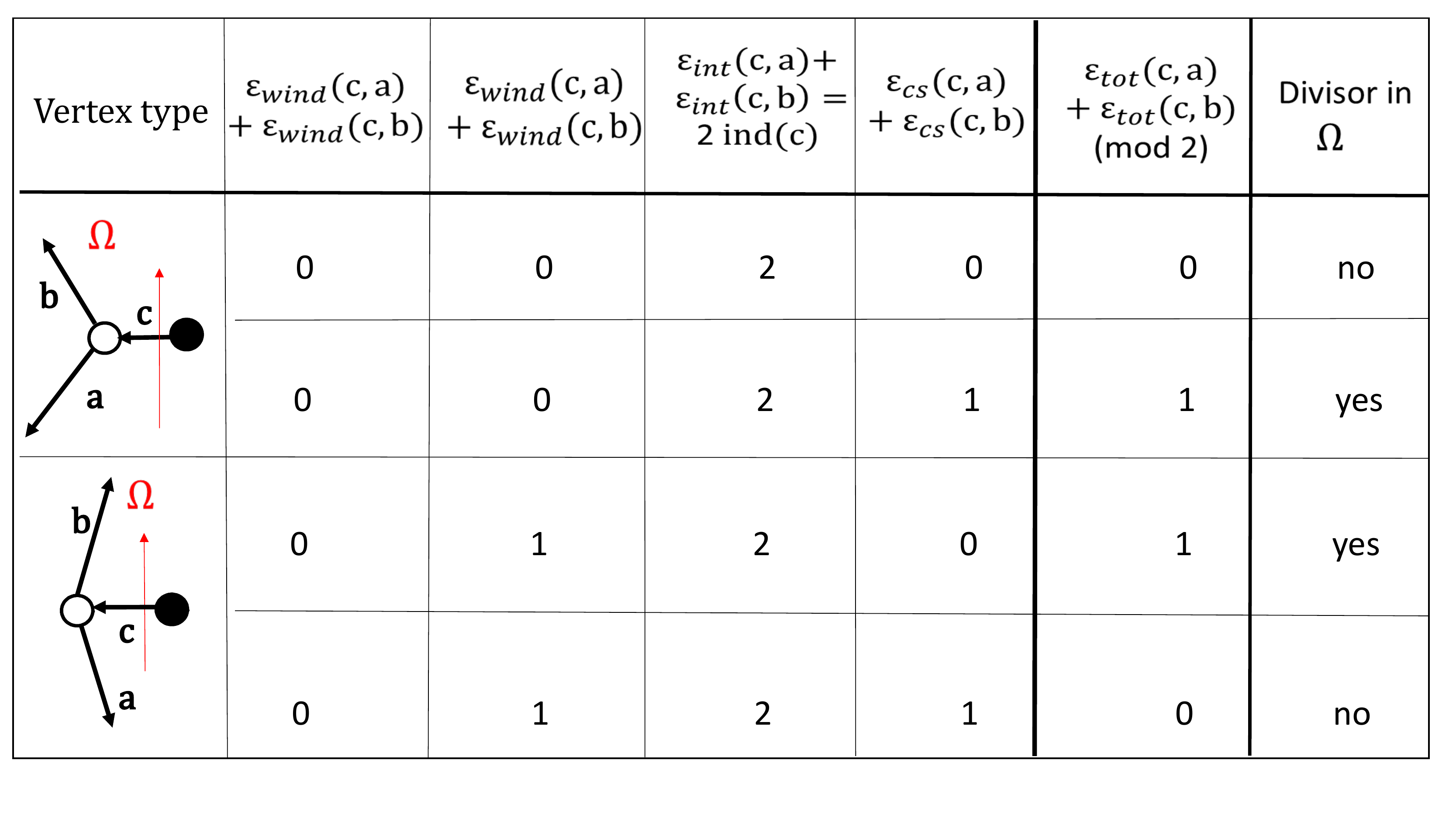}
  \caption{\small{\sl We associate an index $\epsilon_{tot}$ to each pair of edges at a white vertex. The divisor point associated to the vertex $V_l$ belongs to $\Omega\cap \Gamma_l$ if the face $\Omega$ is bounded by a pair of edges at $V_l$ with odd total index.}}
	\label{fig:table_divisor1}
\end{figure}

\begin{figure}
  \centering
  \includegraphics[width=0.46\textwidth]{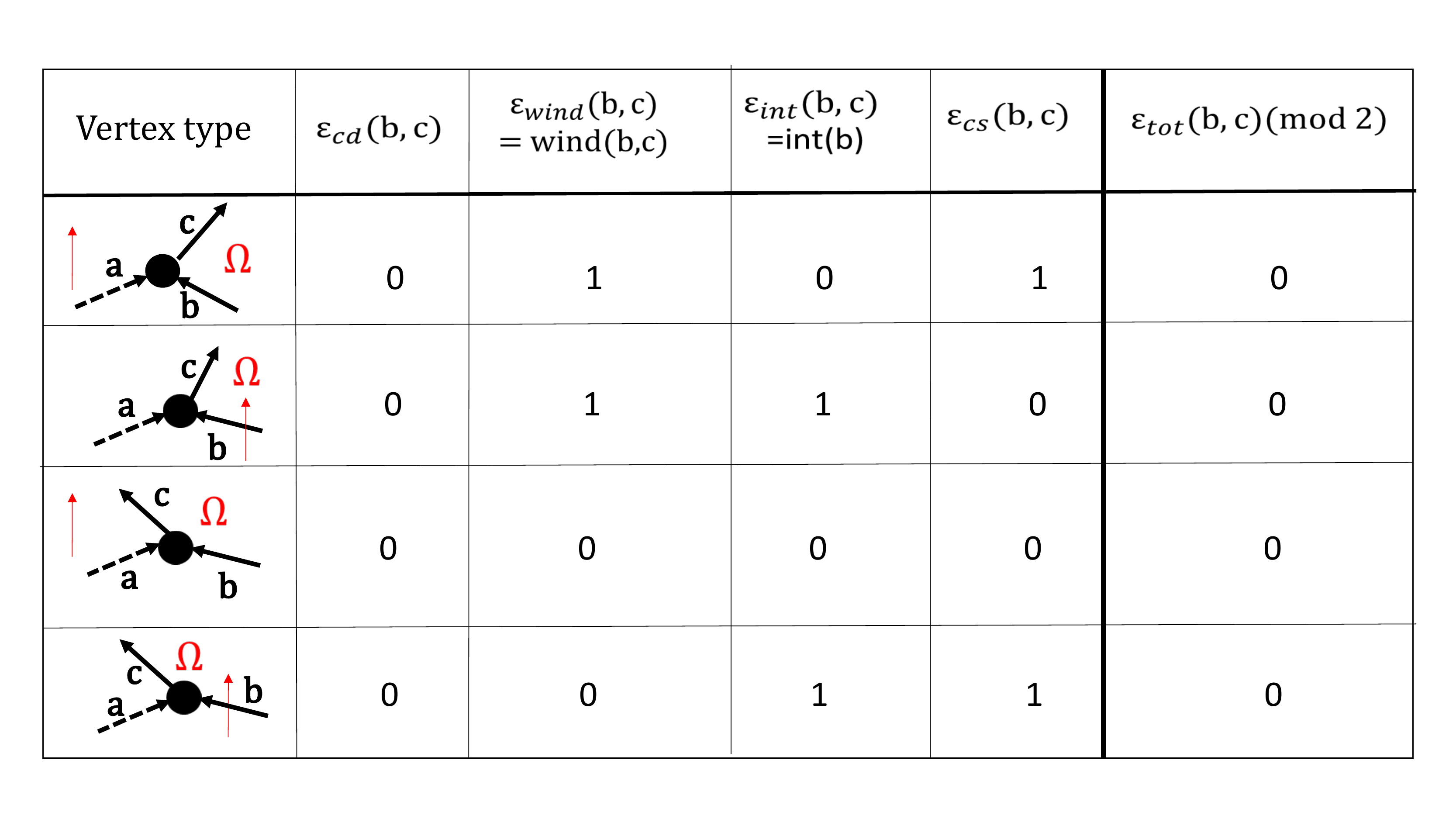}
	\hspace{.5 truecm}
	 \includegraphics[width=0.46\textwidth]{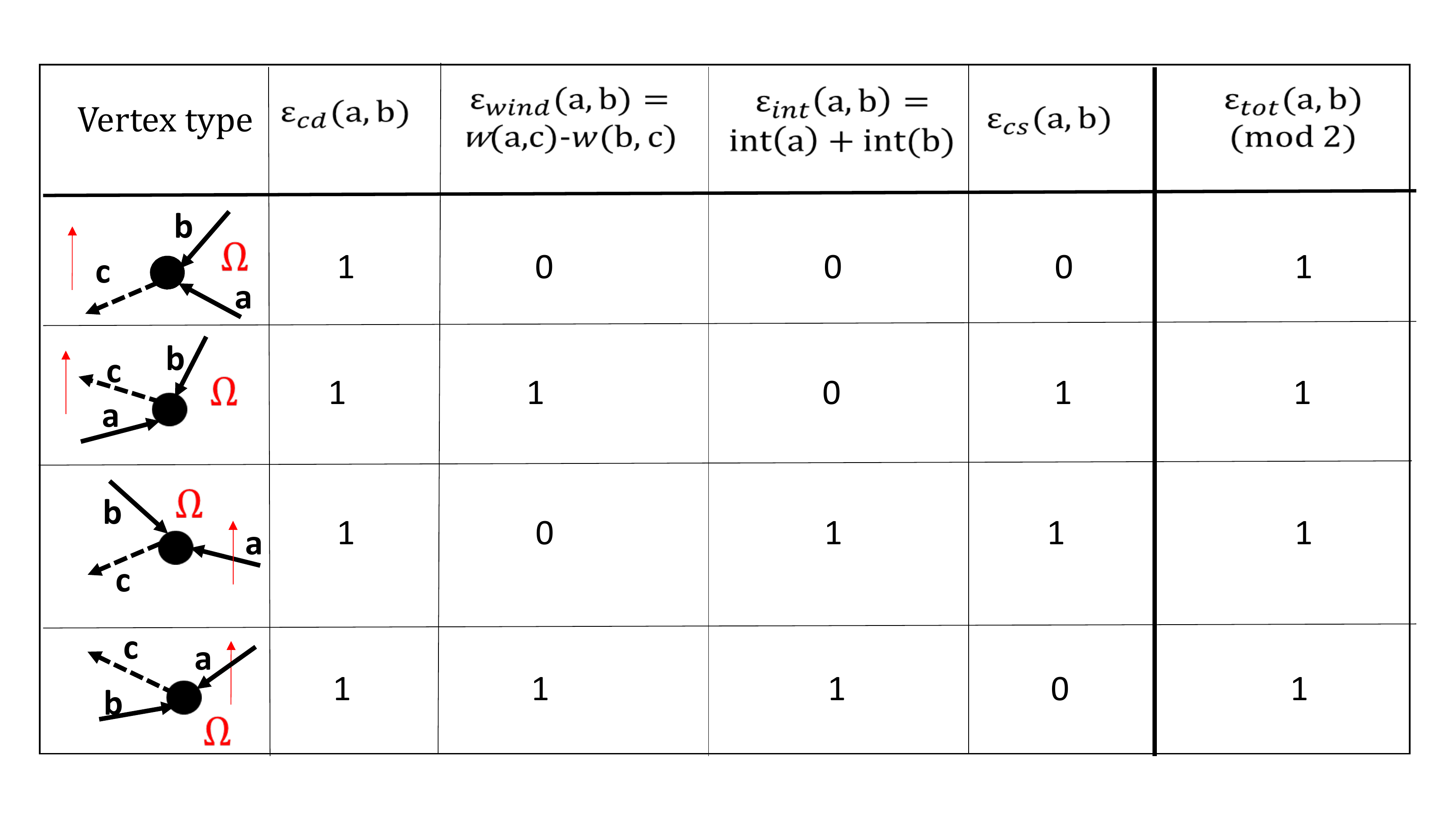}
  \vspace{-.7 truecm}
  \caption{\small{\sl We associate an index $\epsilon_{tot}$ to each pair of edges at a black vertex.}}
	\label{fig:table_divisor2}
\end{figure}
We now discuss the position of vacuum or dressed divisor points in the ovals of $\Gamma$. We recall that $\vec t_0$ has been fixed so that the vacuum (resp. dressed) e.w. is equal to zero only at the edges corresponding to Darboux sink points (resp. all Darboux points) on $\Gamma$. To simplify notations we use the symbol $\Omega_s$ to denote both the face in the network and the corresponding oval in the curve, and the symbol $\Gamma_0$ to denote both the boundary of the disk in the network and the Sato component in the curve.

Any trivalent white vertex bounds either two or three faces by construction (see Figure \ref{fig:table_divisor1}).
Using the definition of divisor point in (\ref{eq:formula_div}) and that of the indices, it is straightforward to verify that the intersection of an oval $\Omega_s$ with a component $\Gamma_l$ contains a vacuum or dressed divisor point if and only if the total index of a pair of edges at $V_l$ bounding the face $\Omega_s$ in ${\mathcal N}^{\prime}$ is odd (see also Figure \ref{fig:table_divisor1}). 

\begin{theorem}\label{theo:pos_div}\textbf{Characterization of the position of the vacuum or dressed divisor point in $\Gamma_l$}
Let $(e_1,e_2)$ be a pair of edges at the trivalent white vertex $V_l$ bounding the face $\Omega_s$ in ${\mathcal N}^{\prime}$ and let $P_1, P_2$ be the corresponding marked 
points in the component $\Gamma_l$ bounding the oval $\Omega_s$ in $\Gamma$. Let $\zeta$ be the local coordinate on $\Gamma_l$ induced by the 
orientation of ${\mathcal N}^{\prime}$ and, to fix notations, suppose that $\zeta(P_1)<\zeta(P_2)$. Let $\gamma_l$ be the (vacuum or dressed) network divisor number for $V_l$ and let $\bar P_l$ be the corresponding (vacuum or dressed) divisor point in $\Gamma_l$. Moreover, if $l\in [n]$, let $P^{(D)}_l$ be the Darboux point corresponding to the Darboux edge $e^{(D)}_l$ at $V_l$.
Then 
\begin{enumerate}
\item If $l\in \bar I$, that is $V_l$ is a boundary vertex connected to a Darboux sink vertex by the edge $e^{(D)}_l$, then the vacuum or dressed divisor number satisfies $\gamma_l =\zeta (P^{(D)}_l )$;
\item If $l\in I$, that is $V_l$ is a boundary vertex connected to a Darboux source vertex by the edge $e^{(D)}_l$ and $\gamma_l$ is the dressed divisor number, then $\gamma_l =\zeta (P^{(D)}_l)$; 
\item If $l\in I$, that is $V_l$ is a boundary vertex connected to a Darboux source vertex by the edge $e^{(D)}_l$ and $\gamma_l$ is the vacuum divisor number, then the corresponding vacuum divisor point $\bar P^{(l)}$ belongs to $\Omega_s\cap \Gamma_l$, that is $\zeta(\bar P^{(l)}) =\gamma_l \in ]\zeta(P_1), \zeta (P_2)[$, if and only 
if $\epsilon_{tot} (e_1,e_2)$ is odd:
\[
P^{(l)} \in \Omega_s	\cap \Gamma_l\quad \iff \quad \epsilon_{tot} (e_1,e_2) = 1 \quad
(\!\!\!\!\!\!\mod 2) .
\]
\item If $l\in [n+1, g+n-k]$, that is $V_l$ is not a vertex connected by an edge to a Darboux vertex, then
the corresponding divisor point $\bar P^{(l)}$ belongs to $\Omega_s\cap \Gamma_l$, that is $\zeta(\bar P^{(l)}) =\gamma_l \in ]\zeta(P_1), \zeta (P_2)[$, if and only 
if $\epsilon_{tot} (e_1,e_2)$ is odd:
\[
P^{(l)} \in \Omega_s	\cap \Gamma_l \quad \iff \quad \epsilon_{tot} (e_1,e_2) = 1 \quad
(\!\!\!\!\!\!\mod 2).
\]
\end{enumerate}
\end{theorem}

\begin{remark}
In the theorem above, we use the cyclic order on real part of $\mathbb{CP}^1$, therefore $\zeta(P_1)=\infty$, and  $\zeta(P_2)=0$, then $\zeta(P_1)<\zeta(P_2)$ in our notations.
\end{remark}

\begin{figure}
  \centering
  \includegraphics[width=0.46\textwidth]{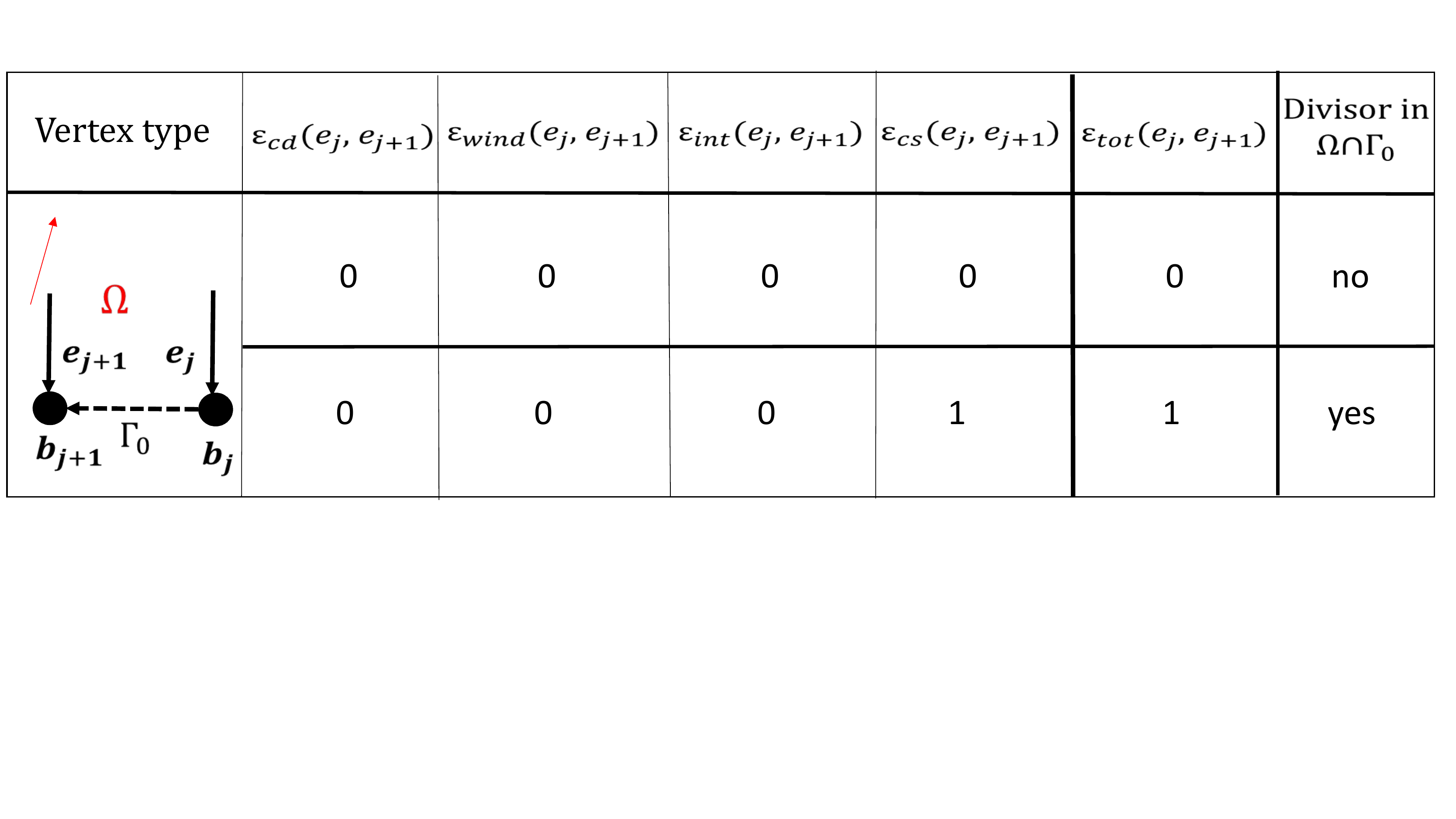}
  \vspace{-1.9 truecm}
  \caption{\small{\sl We associate an index $\epsilon_{tot}$ to each pair of consecutive edges at the boundary of the disk.}}
	\label{fig:table_divisor3}
\end{figure}

\begin{corollary}
\label{cor:eps_tot}
\textbf{Invariance of $\epsilon_{tot}(e_i,e_j)$ at trivalent white vertices.} Under the hypotheses of Theorem~\ref{theo:pos_div} for any choice of orientation, gauge ray direction, weight gauge and vertex gauge the value  $\epsilon_{tot}(e_i,e_j)$ is the same for any give pair of edges at the same trivalent white vertex. 
\end{corollary}
The proof follows from Proposition~\ref{prop:gauge}, Corollary~\ref{cor:indep_gauge}, Theorem~\ref{theo:inv} and Theorem~\ref{theo:pos_div}.

To complete the characterization of the positions of KP divisor points, we must include the Sato divisor.
If the boundary of the face $\Omega_s$ contains portions of the boundary of the disk we need to add the contribution coming from the boundary 
(see Figure \ref{fig:table_divisor3}). In our counting rule, the boundary of the disk is oriented clockwise. In our construction, 
all edges at boundary vertices are parallel and, using the properties of divisor numbers with respect to a change of gauge line direction without loss of generality we assume that the intersection index at each such edge is zero. Each portion of the boundary of the disk bounding the given face is marked by two consecutive boundary vertices $b_j,b_{j+1}$
 and we assign an index to each such pair. In the case of the vacuum divisor, 
$\epsilon_{cs} (b_j,b_{j+1})=0$ by definition and there is no divisor point on $\Gamma_0\cap \Omega_s$. In the case of the dressed e.w., 
let $e(b_j),e(b_{j+1})$ respectively be the edges at $b_j,b_{j+1}$. Then $\Psi_{e(b_j)}(\vec t_0)\Psi_{e(b_{j+1})}(\vec t_0)<0$ implies that there is an odd number of Sato 
divisor points in $\Gamma_0\cap \Omega_s$, while $\Psi_{e(b_j)}(\vec t_0)\Psi_{e(b_{j+1})}(\vec t_0)>0$ implies that there is an even number of Sato divisor points 
belonging to the interval $]\kappa_j, \kappa_{j+1}[$. We show below that the number of Sato divisor points at each finite edge $b_j,b_{j+1}$ is, indeed, zero or one.

In the following Lemma we give relations among the indices inside a given oval which will be used to complete the proof of Theorem \ref{theo:exist} in the general case.

\begin{lemma}\label{lemma:count_face}
Let $\Omega_s$, $s\in [0,g]$ denote the faces in ${\mathcal N}^{\prime}$ corresponding to ovals in $\Gamma$. Let $(V_l; f_l,g_l)$, $l\in [n_s]$, be the triples of (black or white) vertices $V_l$ and pair of edges $(f_l, g_l)$ at $V_l$ bounding $\Omega_s$, where $n_s$ is the total number of pair of edges at internal vertices bounding $\Omega_s$.
Let $\mu_{s, sou}$, $\mu_{s, si}$ respectively be the number of univalent white vertices corresponding to Darboux internal sources and sinks and belonging to $\Omega_s$, and let $\rho_{s}$ be the total number of changes of signs of the wave function occurring at the edges of $\Omega_s$ ending at the boundary of the disk.
Then
\begin{enumerate}
\item For any $s\in [g]$ (finite oval)
\begin{equation}\label{eq:eps_face_s}
\begin{array}{ll}
\mu_{s,si} + \displaystyle \sum_{l=1}^{n_s} \epsilon_{wind} (f_l,g_l) + \mathop{\sum_{l=1}^{n_s}}_{V_l \, \mbox{white}} \epsilon_{cd} (f_l,g_l) =1 &\quad
(\!\!\!\!\!\!\mod 2),\\
\displaystyle\mu_{s,sou} + \sum_{l=1}^{n_s} \epsilon_{int} (f_l,g_l) = 0 &\quad
(\!\!\!\!\!\!\mod 2),\\
\displaystyle\rho_{s}+\sum_{l=1}^{n_s} \epsilon_{cs} (f_l,g_l)  =  0  &\quad
(\!\!\!\!\!\!\mod 2).
\end{array}
\end{equation}
\item If $s=0$ (infinite oval)
\begin{equation}\label{eq:eps_face_0}
\begin{array}{ll}
\mu_{0,si} + \displaystyle \sum_{l=1}^{n_0} \epsilon_{wind} (f_l,g_l) + \mathop{\sum_{l=1}^{n_0}}_{V_l \, \mbox{white}} \epsilon_{cd} (f_l,g_l) =0 &\quad
(\!\!\!\!\!\!\mod 2),\\
\mu_{0,sou} + k + \sum_{l=1}^{n_0} \epsilon_{int} (f_l,g_l) = 0 &\quad
(\!\!\!\!\!\!\mod 2),\\
\rho_0+\sum_{l=1}^{n_0} \epsilon_{cs} (f_l,g_l)  =  0  &\quad
(\!\!\!\!\!\!\mod 2).
\end{array}
\end{equation}
\end{enumerate}
\end{lemma}

\begin{proof}
The only untrivial statement is the first one in both (\ref{eq:eps_face_s}) and (\ref{eq:eps_face_0}). We prove the statement by induction in the number of trivalent white vertices changing direction. 

First of all let $\Omega_s$ be an internal face of the network ${\mathcal N}^{\prime}$, {\sl i.e.} a face whose boundary does not contain portions of the boundary of the disk. If there are no internal sources or sinks in $\Omega_s$ and no changes of directions at the vertices bounding $\Omega_s$ then the statement holds true since the total winding number of the face is an odd number
\[
\sum_{l=1}^{n_s} \epsilon_{wind} (f_l,g_l) = 1 \quad
(\!\!\!\!\!\!\mod 2).
\]
Now suppose to insert a change of direction in $\Omega_s$ along the path $\pi= (e_{i},e_{i+1},\dots,e_{j-1},e_j)$ changing all directions of the edges $e_r$, $r\in [i+1,j-1]$ (see Figure \ref{fig:wind_change}[left]). Then the winding of $\Omega_s$ will change by one:
\[
\sum_{r=i}^{j-1} \left(\epsilon_{wind} (e_r, e_{r+1}) -\epsilon_{wind} (e_r^{\prime}, e_{r+1}^{\prime})\right) =1 \quad
(\!\!\!\!\!\!\mod 2).
\]
If the face $\Omega_s$ intersects the boundary of the disk, then each boundary sink edge in $\Omega_s$ increases by one the counter of the trivalent white vertices changing of direction and keeps the winding of $\Omega_s$ invariant. If we move a boundary source edge from $\Omega_s$ to
the other face, then the overall contribution  from the edges at such vertex due to winding and changes of direction is the same at each face
before and after such move (see Proposition \ref{prop:ext_syst}, Definition \ref{def:vvw_gen} and Figure \ref{fig:wind_change}[middle]). Finally it is straightforward to check that internal black sources behave 
like boundary Darboux sources (see Figure \ref{fig:wind_change}[right]).

The proof in the case $s=0$ (infinite oval) goes along the same lines with the only difference that in the simplest case where there are neither internal sources nor sinks in $\Omega_0$ nor changes of directions at the internal vertices bounding $\Omega_0$ then 
\[
\sum_{l=1}^{n_0} \epsilon_{wind} (f_l,g_l) = 0\quad
(\!\!\!\!\!\!\mod 2).
\]
\end{proof}

\begin{figure}
  \centering
  \includegraphics[width=0.3\textwidth]{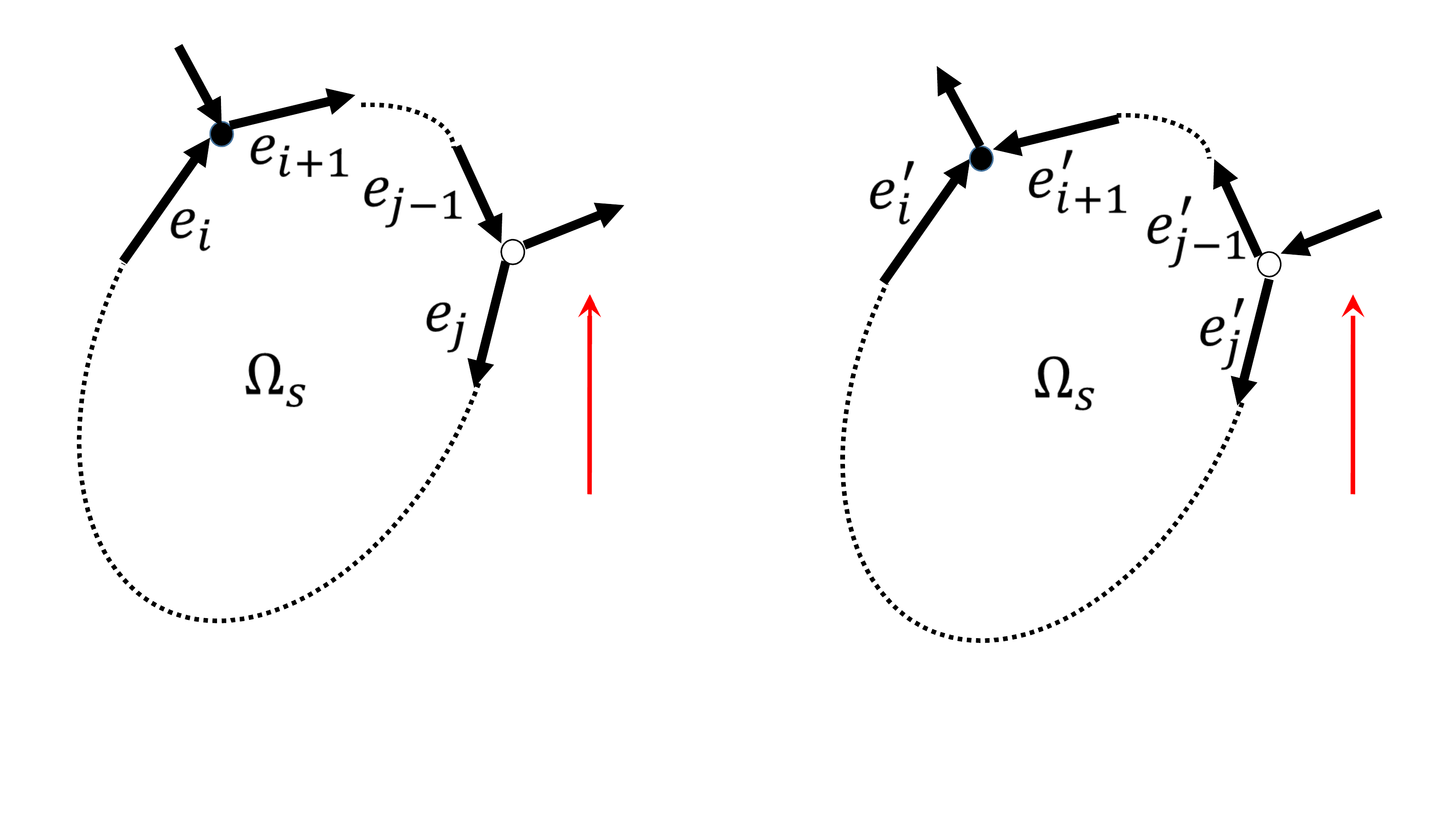}
\hfill
  \includegraphics[width=.3\textwidth]{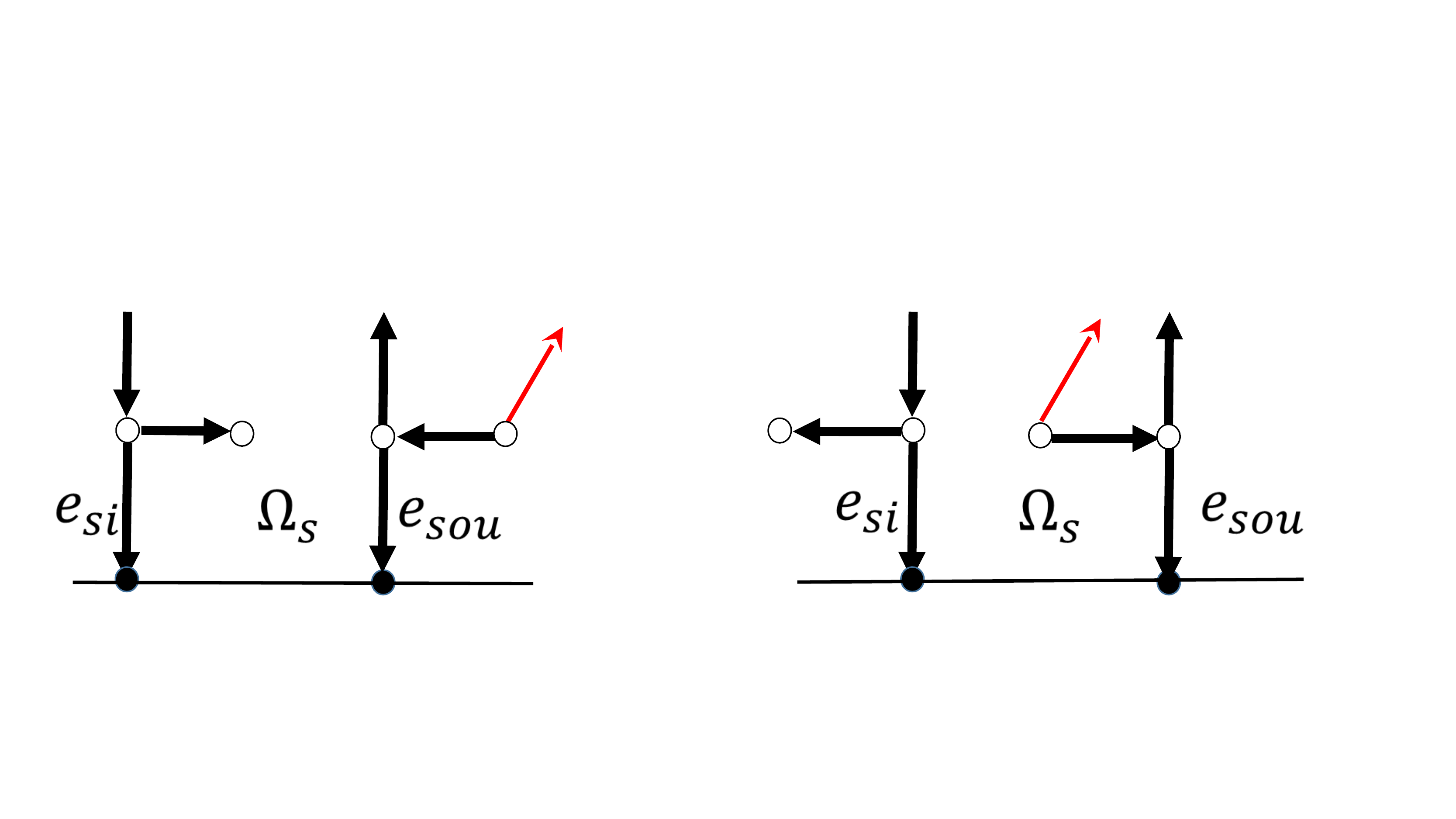}
\hfill
  \includegraphics[width=.3\textwidth]{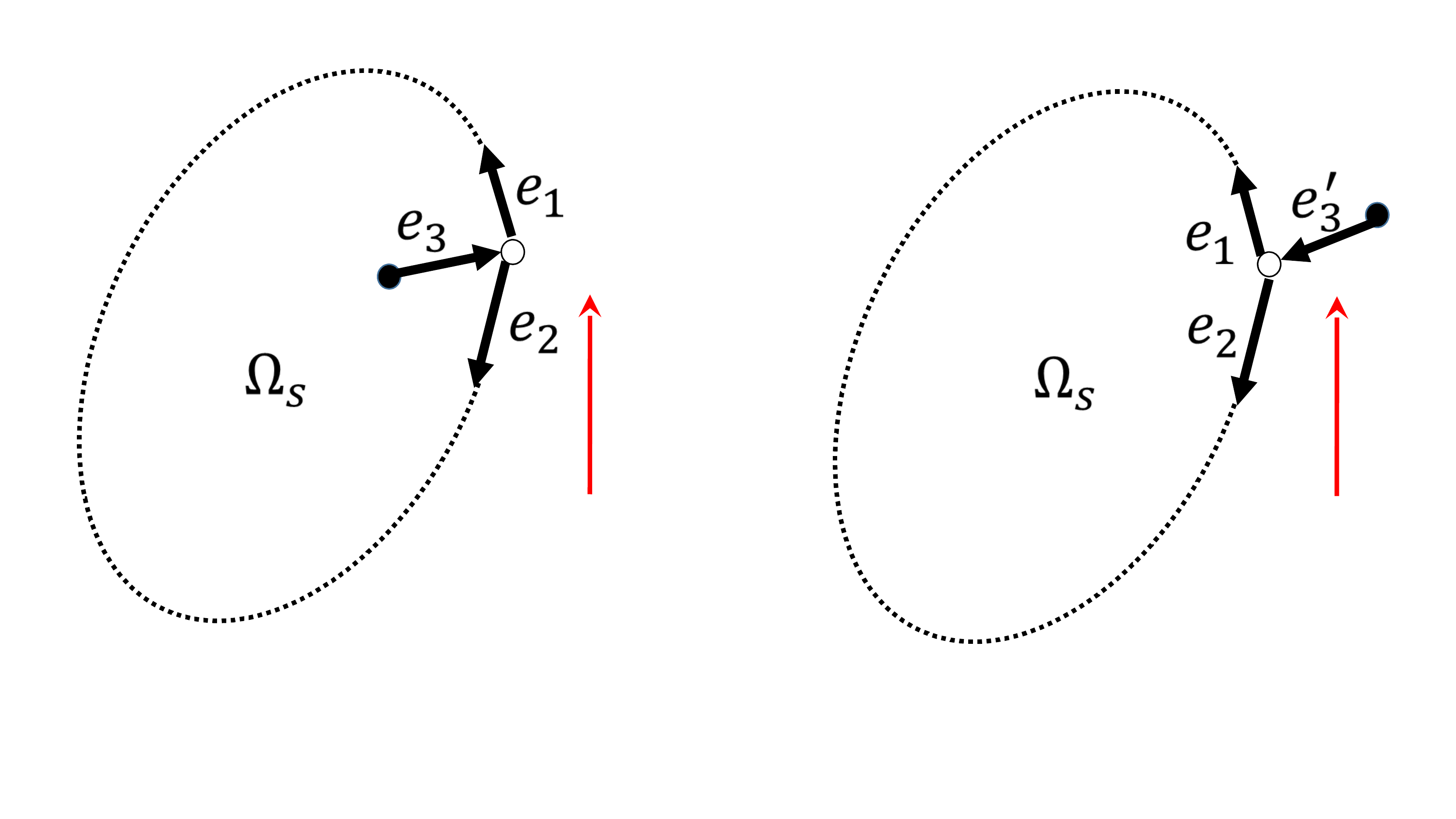}
  \vspace{-0.7 truecm}
  \caption{\small{\sl The relation between winding and changes of direction.}}
	\label{fig:wind_change}
\end{figure}

We are now ready to complete the proof of Theorem \ref{theo:exist} using the Lemmas above.

\begin{remark}\textbf{Notations used in the next Theorem}
In the following Theorem ${\mathcal D} = \{ \gamma_l, \,\; l\in [g+n-k]\}$ denotes either the vacuum network divisor ${\mathcal D}_{\textup{\scriptsize vac}, {\mathcal N}^{\prime}, I}$ or the dressed network divisor ${\mathcal D}_{\textup{\scriptsize dr}, {\mathcal N}^{\prime}}$, where ${\mathcal N}^{\prime}$ is a given oriented network with $g+1$ faces representing $[A]$. 

$\Gamma=\Gamma(\mathcal G)$ is the reducible rational curve for $\mathcal G$ and $\Omega_s$, $s\in[0,g]$, denotes both the oval in $\Gamma$ and its corresponding face in the network. As usual, $\Omega_0$ is the infinite oval which contains the essential singularity. 

For simplicity, we use the same symbol ${\mathcal D}\subset \Gamma $ to denote the set of divisor points $P_l \in \Gamma_l$, $l\in [g+n-k]$, such that $\zeta(P_l )= \gamma_l$, with $\zeta$ the coordinate on $\Gamma_l$ associated to the given orientation of the network. 

In every oval $\Omega_s$, $s\in [0,g]$, of $\Gamma$, $\nu_{s}$ denotes the number of divisor points of ${\mathcal D}$ belonging to the oval:
\[
\nu_s = \# ({\mathcal D}\cap \partial \Omega_s),
\]
$\mu_{s}$, $\mu_{s,sou}$ and $\mu_{s,si}$ respectively denote the total number of Darboux univalent vertices, of Darboux source univalent vertices and of Darboux sink univalent vertices belonging to the face $\Omega_s$. 

To count the number of divisor points on $\Gamma_0\cap \Omega_s$, where $\Gamma_0$ is the $\mathbb{CP}^1$ component corresponding to the boundary of the disk, we use the index $\rho_s$, $s\in [g]$.
$\rho_{s}$ is simply the total number of changes of signs of the e.w. on the edges at the boundary vertices belonging to the intersection of the face $\Omega_s$ with the boundary of the disk. We remark that, by construction, $\rho_s=0$ for all $s\in [0,g]$, if $\mathcal D$ is the vacuum network divisor. 
\end{remark}

\begin{theorem}\label{prop:comb_oval}\textbf{The number of divisor points in each oval (regularity property of the KP divisor in Theorem \ref{theo:exist}).}
Let $\mathcal G$, $\Gamma$, ${\mathcal D}$, $\nu_s$, $\mu_{s}$, $\mu_{s,sou}$ $\mu_{s,si}$ and $\rho_s$, $s\in [g]$ be as above. Then
\begin{equation}\label{eq:odd_N}
\nu_{s}+ \mu_{s} +\rho_s = \left\{ \begin{array}{ll} 1 \quad
(\!\!\!\!\!\!\mod 2) & \mbox{ if } s\in [g],\\
k \quad
(\!\!\!\!\!\!\mod 2) & \mbox{ if } s=0.
\end{array}
\right.
\end{equation}
\end{theorem}

\begin{proof}
The proof follows from Theorem \ref{theo:pos_div} and Lemmas \ref{lemma:count_eps} and \ref{lemma:count_face}. For any fixed $s\in [g]$
let $(V_r^{(s)}; f_r^{(s)},g_{r}^{(s)})$, $r\in [n_s]$, be the set of pair of edges $f_r^{(s)},g_{r}^{(s)}$ at internal vertices $V_r^{(s)}$ bounding $\Omega_s$, where $[n_s]$ is the total number of pair of edges bounding $\Omega_s$. 
We remark that a given vertex $V$ may appear twice if two pairs of edges at $V$ belong to the boundary of $\Omega_s$ (see Figure \ref{fig:table_divisor1}). 

By construction, a divisor point $P_l\in \mathcal D$ belongs to $\Gamma_l \cap\Omega_s$ 
if and only if $\Gamma_l$ corresponds to a white trivalent vertex $V_l$ carrying an odd index at the edges $(V_l^{(s)}, f_l^{(s)},g_l^{(s)})$ bounding $\Omega_s$ (see Theorem \ref{theo:pos_div}). Therefore the number $\nu_s$ of divisor points in $\mathcal D\cap \Omega_s$ satisfies
\begin{equation}\label{eq:formula_1}
\nu_{s} = \mathop{\sum_{r=1}^{n_s}}_{V_r \mbox{ tr.w.}} \epsilon_{tot} (f_r,g_r) = \mathop{\sum_{r=1}^{n_s}}_{V_r \mbox{ tr.w.}} \epsilon_{cd} (f_r,g_r) +  \mathop{\sum_{r=1}^{n_s}}_{V_r \mbox{ tr.w.}}\Big( \epsilon_{cs} (f_r,g_r) + \epsilon_{wind} (f_r,g_r) +\epsilon_{int} (f_r,g_r) \Big),
\end{equation}
where the sums above run on all pair of edges at trivalent white vertices belonging to the boundary of $\Omega_s$. Using (\ref{eq:tot_biv}) and (\ref{eq:tot_triv_bl}), the second sum in the r.h.s. of (\ref{eq:formula_1}) has the same parity if we substitute it with the total contribution to the winding, change of sign and intersection indices due to \textbf{all} pair of edges at internal vertices bounding $\Omega_s$, that is 
\begin{equation}\label{eq:formula_2}
\mathop{\sum_{r=1}^{n_s}}_{V_r \mbox{ tr.w.}} \left( \epsilon_{cs} (f_r,g_r) + \epsilon_{wind} (f_r,g_r) +\epsilon_{int} (f_r,g_r) \right) =
\sum_{r=1}^{n_s} \left( \epsilon_{cs} (f_r,g_r) + \epsilon_{wind} (f_r,g_r) +\epsilon_{int} (f_r,g_r) \right)  \quad
(\!\!\!\!\!\!\mod 2).
\end{equation}
Therefore, inserting (\ref{eq:eps_face_s}), (\ref{eq:formula_1}) and (\ref{eq:formula_2}) in the r.h.s. of (\ref{eq:odd_N}), for any $s\in [g]$, we easily get
\begin{equation}\label{eq:mainformula_s}
\begin{array}{l}
\nu_{s}+ \mu_{s} +\rho_s \displaystyle \equiv \mathop{\sum_{r=1}^{n_s}}_{V_r \mbox{ tr.w.}}\epsilon_{tot} (f_r,g_r) + \mu_{s}+\rho_s =
\mathop{\sum_{r=1}^{n_s}}_{V_r \mbox{ tr.w.}}\epsilon_{cd} (f_r,g_r) +\\
\displaystyle \quad\quad+  \sum_{r=1}^{n_s} \Big( \epsilon_{cs} (f_r,g_r) + \epsilon_{wind} (f_r,g_r) 
+ \epsilon_{int} (f_r,g_r) \Big)+ \mu_{s,so} +\mu_{s,si}+\rho_s  =1 \quad
(\!\!\!\!\!\!\mod 2).
\end{array}
\end{equation}

The proof of (\ref{eq:odd_N}) at the infinite oval follows using  (\ref{eq:eps_face_0}) instead of (\ref{eq:eps_face_s}):
\begin{equation}\label{eq:mainformula_0}
\begin{array}{l}
\nu_{0}+ \mu_{0} + \rho_0 + k \displaystyle = \sum_{l, V_l \mbox{ tr.w.}} \epsilon_{tot} (f_l,g_l) + \mu_{0} + \rho_0   + k = \\
\quad \displaystyle=  \sum_{l, V_l \mbox{ tr.w.}} \epsilon_{cd} (f_l,g_l) +  \sum_{l, V_l \mbox{ tr.w.}} \left( \epsilon_{cs} (f_l,g_l) + \epsilon_{wind} (f_l,g_l) +
\epsilon_{int} (f_l,g_l) \right)+ \mu_{0} + \rho_0   +k  = \\
\quad\displaystyle=\sum_{l, V_l \mbox{ tr.w.}} \epsilon_{cd} (f_l,g_l) +  \sum_{l=1}^{n_s} \left( \epsilon_{cs} (f_l,g_l) + \epsilon_{wind} (f_l,g_l) +\epsilon_{int} (f_l,g_l) \right)+ \mu_{0}+ \rho_0 +k =0 \quad
(\!\!\!\!\!\!\mod 2).
\end{array}
\end{equation}
\end{proof}

The vacuum divisor on $\Gamma$ is obtained from the vacuum network divisor ${\mathcal D}\equiv 
{\mathcal D}_{\textup{\scriptsize vac}, {\mathcal N}^{\prime}, I}$ by eliminating the Darboux sink divisor points $P_{j_s}$, $s\in[n-k]$: $\DVG={\mathcal D}\backslash \{P_{j_s},s\in[n-k]\}$. In this case, we denote $\nu'_s=\nu_s- \mu_{s,si}$ the number of divisor points in $\DVG\cap\Omega_s$ for 
$s\in[0,g]$, and  (\ref{eq:odd_N}) becomes
\begin{equation}\label{eq:odd_N_2}
\nu'_{s}+ \mu_{s,so} = \left\{ \begin{array}{ll} 1 \quad
(\!\!\!\!\!\!\mod 2) & \mbox{ if } s\in [g],\\
k \quad
(\!\!\!\!\!\!\mod 2)& \mbox{ if } s=0,
\end{array}
\right.
\end{equation}
since $\rho_s=0$ for any $s\in[0,g]$. 
So we completed the proof of Theorem~\ref{theo:vac_div} for curves associated to general networks.

If ${\mathcal D}\equiv {\mathcal D}_{\textup{\scriptsize dr}, {\mathcal N}^{\prime}}$ is the dressed network divisor, then  $\DKP$ is obtained from 
${\mathcal D}$ by eliminating the Darboux sink and source divisor points $P_{j_s}$, $s\in[n]$, and by adding the Sato divisor ${\mathcal D}_S$:
$\DKP=\left({\mathcal D}\cup {\mathcal D}_S\right) \backslash \{P_{j_s},s\in[n]\}$. In this case we denote $\nu''_s=\nu_s- \mu_{s}$ the number 
of divisor points in $\DKP\cap\left(\Omega_s\backslash \Gamma_0 \right)$. Therefore, (\ref{eq:odd_N}) becomes
\begin{equation}\label{eq:odd_N_3}
\nu''_{s}+ \rho_{s} = \left\{ \begin{array}{ll} 1 \quad
(\!\!\!\!\!\!\mod 2), & \mbox{ if } s\in [g],\\
k \quad
(\!\!\!\!\!\!\mod 2), & \mbox{ if } s=0.
\end{array}
\right.
\end{equation}
Thus we have completed the proof of Theorem \ref{theo:exist},
since $\DKP$ is an effective divisor of degree $g$, and each finite oval $\Omega_s$, $s\in[g]$ contains an odd number of divisor points, then $\# \left( \Omega_s \cap \DKP\right) =1$, $s\in [g]$, and $\# \left( \Omega_0 \cap \DKP\right) =0$. As a corollary 
we see, that each $\Omega_s$,   $s\in[g]$ contains either 0 or 1 Sato divisor point, and either 0 or 1 non-Sato  divisor point. We summarize 
this as:
\begin{corollary}
\label{cor:DKP}
Let $\Omega_s$ be a finite oval of $\Gamma(\mathcal G)$, then 
\begin{enumerate}
\item $\Omega_s$ contains exactly one divisor point of $\DKP$;
\item $\nu''_{s} = 0$, $\rho_s=1$ if and only if $\Omega_s$ contains one Sato divisor point;
\item $\nu''_{s} = 1$, $\rho_s=0$ if and only if $\Omega_s$ does not contain Sato divisor points.
\end{enumerate}
\end{corollary}

\section{Signatures of edges and total non--negativity}\label{sec:lam1}
In this Section, we provide an invariant formulation of the construction of both the KP wave function and the KP divisor using a space of relations at vertices analogous to the space of relations introduced in \cite{Lam2} (section 14) to provide a mathematical formulation of on shell diagrams \cite{AGP1, AGP2} (see also \cite{ATT}).
In \cite{Lam2} (section 14) it is defined a relation space on bicolored networks associating a formal variable $z_{(U,e)}$ to each half edge with the following properties:
\begin{enumerate}
\item If $w(u,v)$ is the weight of the oriented edge $e=(U,V)$, then $ z_{(U,e)} = w(U,V) z_{(V,e)}$;
\item At each black vertex $V$ is associated the equation $z_{(V,e)} = z_{(V,e^{\prime})}$ for every pair of edges $e,e^{\prime}$ incident to $V$;
\item If $V$ is a white vertex, then $\sum\limits_{e\,\mbox{\scriptsize{at}}\ V} z_{(V,e)} = 0$.
\end{enumerate}
Lam \cite{Lam2} observes that it is necessary to assign proper signatures to all edges in order that the system of relations has maximal rank and that the formal variables are associated to totally non-negative Grassmannians. In \cite{Lam2} the question of providing a simple rule to assign such signatures is left open.  

In \cite{AG2}, we have provided an edge signature in the case $\mathcal N$ is canonically oriented Le--network by explicitly assigning a $+$ to all horizontal edges and to each vertical edge a signature equal to one minus the number of boundary sources shadowed in the SE direction by the corresponding box in the Le-diagram. 

In this paper we have introduced well defined signatures at vertices associated to the geometry of $\mathcal G$ through the choice of perfect orientations and the use of gauge ray directions to measure winding and intersections indices at pairs of edges. Clearly, each such vertex signature induces a well defined edge signature at all internal edges $e$ for instance w.r.t. a perfect orientation as we did in \cite{AG2} on the Le--network; moreover the signature at boundary edges is uniquely defined in our geometric construction and respects the total non--negativity property. Below we propose an alternative description of such signatures and restate our main theorems in this framework.

In the following, let $\mathcal G$ be a reduced PBDTP graph representing an irreducible positroid cell $\S \in \GTNN$ of dimension $|D|$. Let $[A]\in \S$ and $\mathcal N$ be the PBDTP network of graph $\mathcal G$ representing $[A]$, uniquely defined up to the weight gauge. The reduced property of the graph ensures the absence of null edge vectors in our construction, moreover, the number of white trivalent vertices is $|D|-k$. All the results of this section may be extended also to unreduced graphs with the additional condition of absence of null vectors on one of the networks representing the given point $[A]$.

In Section \ref{sec:vectors} we have used winding at pair of edges and gauge ray directions to assign signatures to the linear relations at the internal vertices of directed PBDTP networks. We have also proven that the linear system at the vertices has maximal rank for any given perfect orientation $\mathcal O$ and gauge ray direction $\mathfrak l$ on $\mathcal N$ (Theorem \ref{theo:consist}) and have provided its explicit solution in Theorem \ref{theo:null} adapting the approach in \cite{Tal2} to our setting. In particular, the action of the linear system on the boundary is the following:
\begin{enumerate}
\item If we fix the base $I$ and as boundary conditions we assign the canonical vectors $E[j]$ to the boundary sinks $b_j$, $j\in \bar I$, then, for any perfect orientation $\mathcal O=\mathcal O(I)$ and gauge ray direction $\mathfrak l$, the edge vectors at the boundary sources $b_{i_r}$, $r\in [k]$, will be $A[r]-E[i_r]$, where $A[r]$ is the $r$--th row of the reduced row echelon representation of $[A]$ with respect to the base $I$;
\item If we pass from the base $I$ to the base $J$ and from the boundary conditions $E[j]$ at the boundary sinks $b_j$, $j\in \bar I$ to the boundary conditions $E[l]$ at $b_l$, $l\in \bar J$, we have an explicit transformation of the edge vectors which will transform the vectors $A[r]-E[i_r]$ at the boundary sources $i_r\in I$, to the vectors $\hat A[s]-E[i_s]$ at the boundary sources $i_s\in J$. Here $\hat A [s]$ is the $s$--th row of the reduced row echelon form matrix represetning $[A]$ with respect to the base $J$.
\end{enumerate} 
Before continuing reviewing our construction, we remark that the alternative choice $-E[j]$ at a boundary sink vertex corresponds to the multiplication by $(- 1)$ of the corresponding column in the reduced row echelon matrix with respect to the base $I$. This change of sign preserves the Pl\"ucker relations, but not the total non--negativity property, {\sl i.e.} the resulting point $[A]$ belongs to the stratum of $Gr(k,n)$ associated to the matroid $\mathcal M$ if we choose the ''wrong'' signs.

In Section \ref{sec:anycurve} we have applied the above construction to the soliton data $(\mathcal K = \{\kappa_1<\cdots <\kappa_n\}, [A])$: we have transformed the edge vector $E_e$ into a dressed edge wave function $\Psi_e(\vec t)$ and we have associated a real and regular KP divisor $\DKP$ on $\Gamma$ to the linear relations at white vertices on $\mathcal N$. In particular, we have shown that the divisor is invariant with respect to changes of both the orientation and the gauge ray direction. Finally in the previous Section we have combinatorially characterized the position of the divisor in the ovals of $\Gamma$ and proven that there is exactly one divisor point in each oval except the one containing the essential singularity of the wave function on $\Gamma$. The position of the divisor in the oval is combinatorially detected by the condition $\epsilon_{tot}(e_i,e_{i+1})=1$ at pair of edges $(e_i, e_{i+1})$ at white vertices.

All of the above implies that there must exist an invariant formulation of the results obtained in this paper.
In the following $z_{(V,e)}$ may be interpreted as an half-edge vector or an half-edge (vacuum or dressed) wave function; moreover we enumerate edges from 1 to $m$ counterclockwise at any vertex $V$ of valency $3$, without reference to their orientation and we use cyclical order in summations.
Then, we may re-express the linear relations both for the edge vectors (\ref{eq:lineq_biv})--(\ref{eq:lineq_white}) and the edge wave functions (\ref{eq:lin_Phi1})--(\ref{eq:lin_Phi3}) as linear relations for half-edge vectors and half-edge dressed wave functions $z_{(U,e)}$ as follows.
\begin{definition}\textbf{Linear relations at vertices on PBDTP networks}\label{def:sign_rela}
Let $\mathcal G$ be a reduced PBDTP graph representing a $|D|$--dimensional irreducible positroid cell $\S \subset \GTNN$ and let $\mathcal N$ be a network of graph $\mathcal G$ with positive weights $w(U,V)$ at edges $e=(U,V)$ with  $U$, $V$ respectively the initial and final vertices of $V$ (so that $w(V,U)= w(U,V)^{-1}$). Then we introduce the following system of relations:
\begin{enumerate}
\item  If $w(U,V)$ is the weight of the oriented edge $e=(U,V)$, then 
\begin{equation}\label{eq:lam_edge}
z_{(U,e)}=  w(U,V) z_{(V,e)};
\end{equation}
\item  At each bivalent ($m=2$) or black trivalent ($m=3$) vertex $V$ and for every pair of edges $(e_i,e_{i+1})$ incident to $V$ 
\begin{equation}\label{eq:lam_black}
z_{(V,e_i)} = (-1)^{\epsilon_{cs}(e_i,e_{i+1})} z_{(V,e_{i+1})}, \quad\quad i\in [m],
\end{equation}
where the indices satisfy 
\begin{equation}\label{eq:lam_cs}
\sum_{i=1}^m \epsilon_{cs}(e_i,e_{i+1}) =0 \quad (\!\!\!\!\!\!\mod 2), \quad\quad \epsilon_{cs}(e_i,e_{i+1}) \in \{0, 1\}, \quad i\in [m];
\end{equation}
\item  If $V$ is a trivalent white vertex, then 
\begin{equation}\label{eq:lam_white}
(-1)^{\epsilon_{tot} (e_1,e_2)-\epsilon_{cs} (e_1,e_2)  } z_{(V,e_3)} +  c.p. =0,
\end{equation}
where $e_m$, $m\in [3]$, are the edges at $V$, and the indices satisfy (\ref{eq:lam_cs}) and
\begin{equation}\label{eq:lam_tot}
\sum_{i=1}^3 \epsilon_{tot}(e_i,e_{i+1}) =1, \quad\quad \epsilon_{tot}(e_i,e_{i+1}) \in \{ 0, 1\}, \quad
\epsilon_{tot}(e_i,e_{i+1})\epsilon_{tot}(e_{i+1},e_{i+2}) =0, \quad i\in [m].
\end{equation}
\end{enumerate}
\end{definition}
The above definition implies that at each pair of edges at a given vertex we assign a $\pm$ sign as follows.
\begin{definition}\textbf{Admissible vertex signatures for the linear relations on PBDTP networks}\label{def:vertex_sign}
Let $\mathcal G$, $\mathcal N$ as in Definition \ref{def:sign_rela}. Then to any system of linear relations satisfying Definition \ref{def:sign_rela} we assign a vertex signature at pairs of half edges as follows
\begin{enumerate}
\item  At each bivalent ($m=2$) or black trivalent ($m=3$) vertex $V$ and for every pair of edges $(e_i,e_{i+1})$ incident to $V$, we assign
\[
\sigma_V ( e_i, e_{i+1}) =\left\{ \begin{array}{ll} +, &\quad \mbox{ if } \,\,\epsilon_{cs}(e_i,e_{i+1}) = 0 \quad (\!\!\!\!\!\!\mod 2),\\
-, &\quad \mbox{ if } \,\,\epsilon_{cs}(e_i,e_{i+1}) = 1 \quad (\!\!\!\!\!\!\mod 2),
\end{array}\right.
\] 
{\sl i.e.} the admissible vertex signatures of pairs of half edges cyclically ordered at a trivalent black vertex are either
$\sigma_V=(+,+,+)$ or $\sigma_V= (+,-,-)$ with respect to a convenient labeling of the edges at $V$;
\item  If $V$ is a trivalent white vertex, then for every pair of edges $(e_i,e_{i+1})$ incident to $V$, we assign
\[
\sigma_V  ( e_i, e_{i+1}) =\left\{ \begin{array}{ll} +, &\quad \mbox{ if  } \,\,\epsilon_{tot} (e_i,e_{i+1})-\epsilon_{cs} (e_{i},e_{i+1}) = 0 \quad (\!\!\!\!\!\!\mod 2),\\
-, &\quad \mbox{ if  } \,\,\epsilon_{tot} (e_i,e_{i+1})-\epsilon_{cs} (e_{i},e_{i+1}) = 1 \quad (\!\!\!\!\!\!\mod 2),
\end{array}\right.
\]
{\sl i.e.} the admissible vertex signatures of pairs of half edges cyclically ordered at a trivalent white vertex are either
$\sigma_V= (+,+,-)$ or $\sigma_V =(-,-,-)$ with respect to a convenient labeling of the edges at $V$.
\end{enumerate}
\end{definition}

Conditions (\ref{eq:lam_cs}) and (\ref{eq:lam_tot}) are modeled on Definition \ref{def:index_pair} and Lemma \ref{lemma:count_eps}; therefore the linear relations satified by edge vectors and edge wave functions at vertices may be restated in the invariant form of Definition \ref{def:sign_rela}. 
In our construction of edge vectors, the linear system on $(\mathcal N,\mathcal O, \mathfrak l)$ has maximal rank and induces a signature satisfying Definitions \ref{def:sign_rela} and \ref{def:vertex_sign}, where the signature at vertices depends on both the orientation and the gauge ray direction. The linear relations for the edge wave functions on $(\mathcal N,\mathcal O, \mathfrak l)$ are formally obtained substituting the symbol $E_e$ with the symbol $\Psi_e(\vec t)$. If we fix the phases $\mathcal K$ and the initial time $\vec t_0$ (here and in the following we assume that only a finite number of KP times are different from zero), we have a well defined way to compute both indices $\epsilon_{tot}$ and $\epsilon_{cs}$ on $(\mathcal N,\mathcal O, \mathfrak l)$ and they satisfy the following compatibility conditions (\ref{eq:admis_eps}) at internal ovals.

\begin{lemma}\label{lem:adm_1}\textbf{Compatibility conditions on internal ovals for the indices $\epsilon_{tot}$ and $\epsilon_{cs}$.}
Let $\mathcal N$  be a network as in the above Definition with perfect orientation $\mathcal O$ and
gauge ray direction $\mathfrak l$. Then the system of relations for the edge vectors (\ref{eq:lineq_biv})--(\ref{eq:lineq_white}) 
on $(\mathcal N,\mathcal O, \mathfrak l)$ may be expressed as a system on half edge vectors fulfilling Definition \ref{def:sign_rela} and it has maximal rank.
Moreover, for any given internal oval $\Omega_s$ the indices  $\epsilon_{tot}$ and $\epsilon_{cs}$ satisfy the following conditions: 
\begin{equation}\label{eq:admis_eps}
\sum\limits_{V\in\Omega_s} \epsilon_{cs}(e_V,f_V) = 0, \ \ \ \ \ \sum\limits_{V\in\Omega_s} \epsilon_{tot}(e_V,f_V) = 1, \ \ \ \ (\!\!\!\!\!\!\mod 2),
\end{equation}
where $e_V$, $f_V$ denote the pair of edges at vertex $V$ bounding $\Omega_s$.
\end{lemma}

The invariance of the KP divisor implies that $\epsilon_{tot}(e_i, e_{i+1})$ is the same on $(\mathcal N, \mathcal K, \vec t_0)$ independently of the perfect orientation $\mathcal O$ and the gauge ray direction $\mathfrak l$. 

\begin{lemma}\label{lem:adm_2}\textbf{Invariance of the $\epsilon_{tot}$ index at white vertices.}
Let $\mathcal N$ be a reduced network representing the point $[A]$ in the Grassmannian. For any fixed choice  $\mathcal K$ of ordered phases and initial time $\vec t_0$ the pair of edges $e_i$, $e_{i+1}$ at a given white trivalent vertex $V$ such that $\epsilon_{tot}(e_i,e_{i+1})=1$ is independent on both the orientation of the network and the gauge ray direction and the weight and vertex gauges.
\end{lemma}

\begin{figure}
  \centering
  \includegraphics[width=0.49\textwidth]{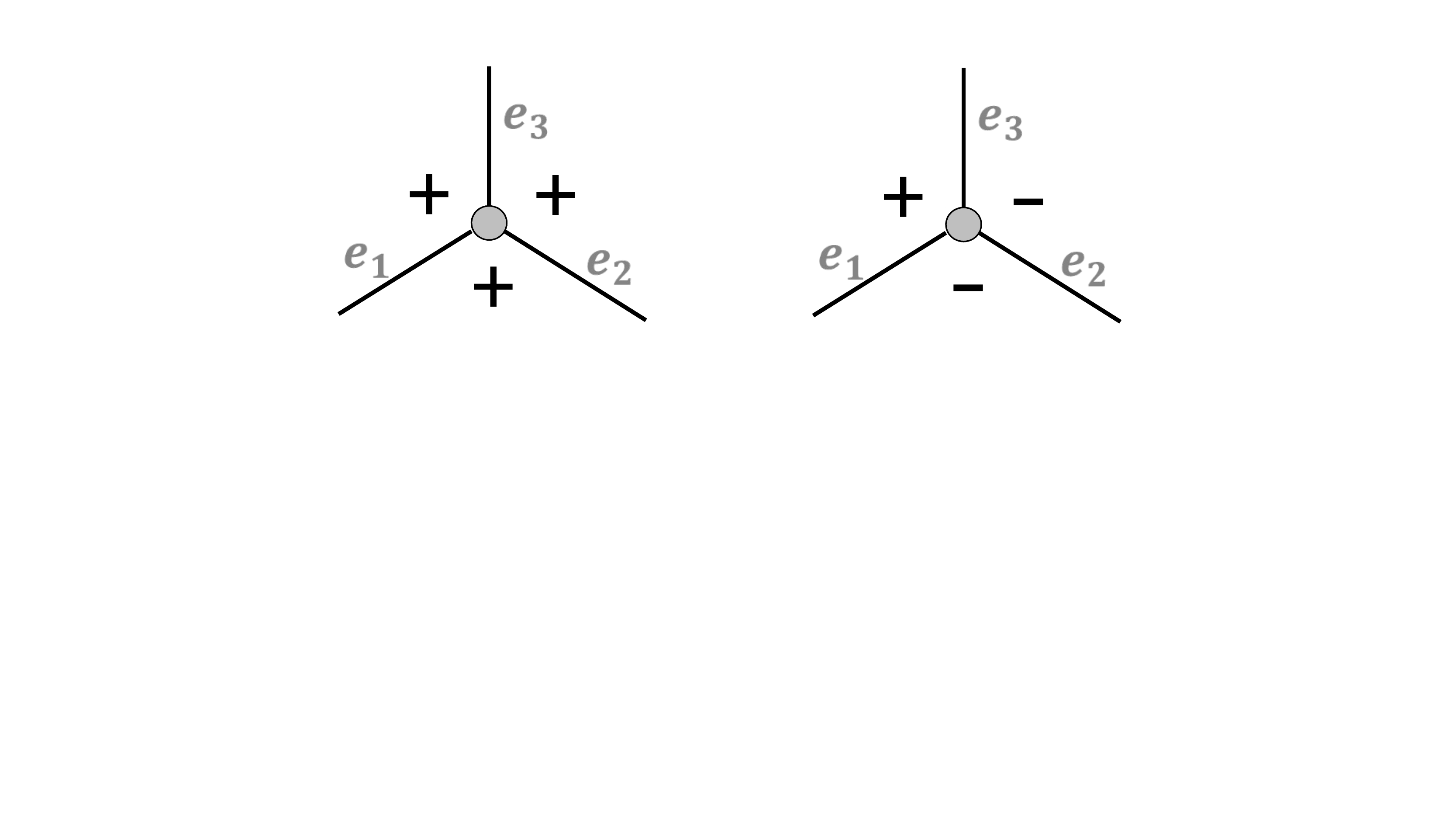}
	\hfill
	\includegraphics[width=0.49\textwidth]{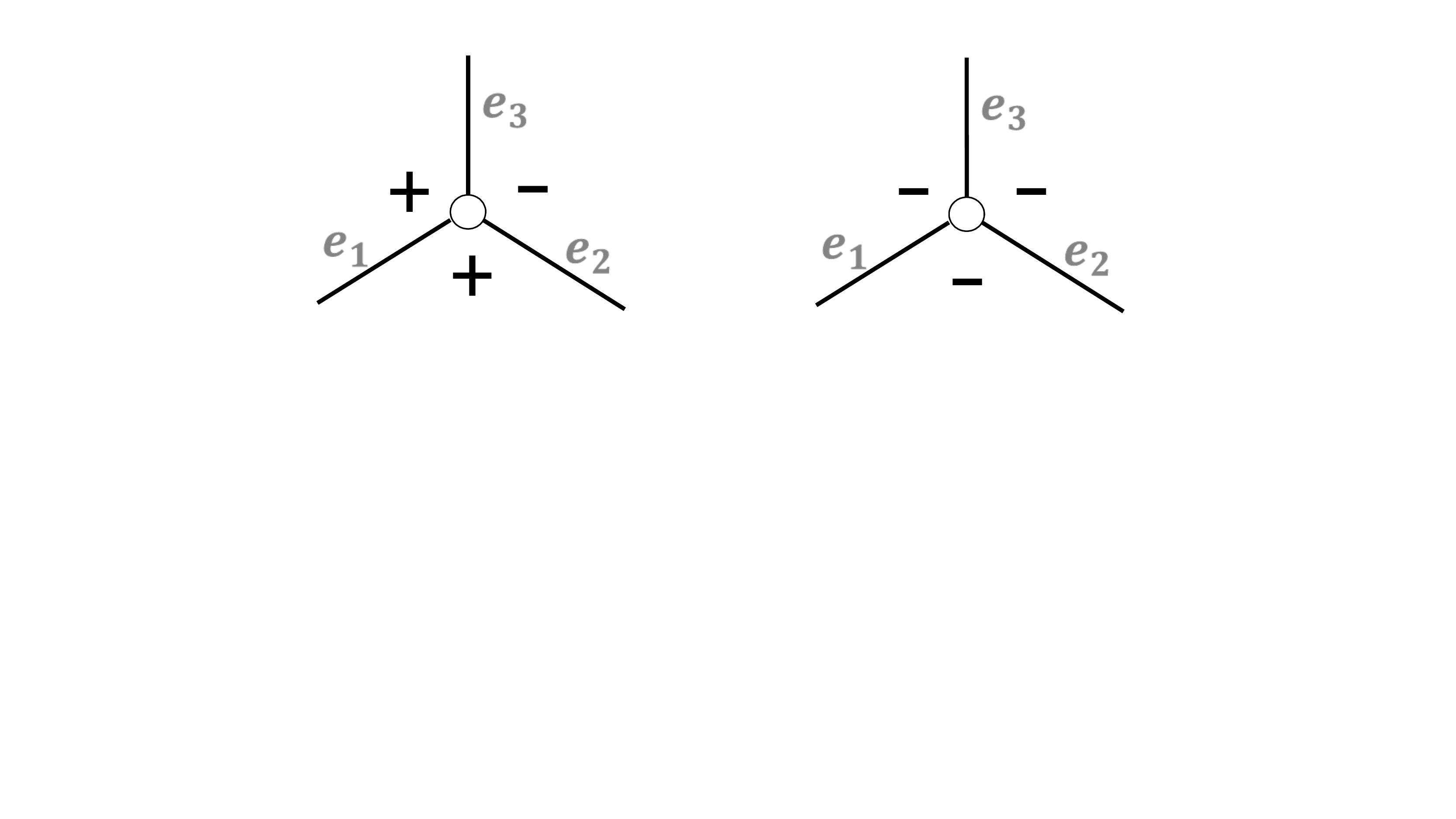}
  \vspace{-2.7 truecm}
  \caption{\small{\sl The geometric construction in Section \ref{sec:vectors} induces a well defined signature for pair of edges at internal vertices. The latter corresponds to relations (\ref{eq:form_3}) and (\ref{eq:form_4}) which link the geometrical picture on $\mathcal G$ to the differential form associated to the soliton data $(\mathcal K, [A])$. For almost any fixed time $\vec t_0$ these relations are solvable and provide a unique real and regular KP divisor on $\Gamma$.}}
	\label{fig:vert_sign}
\end{figure}

\begin{remark}\textbf{Open problems connected to vertex signatures}\label{conj:sign}
Let us fix the reduced PBDTP graph $\mathcal G$ representing the $|D|$--dimensional irreducible positroid cell $\S \subset \GTNN$.
In our construction, the signatures at vertices are completely defined by the geometry of $\mathcal G$ through the selection of perfect orientations, the choice of the gauge ray direction and a convenient choice of the vertex gauge.
Indeed, from Definition \ref{def:index_pair} and Lemma \ref{lemma:count_eps}, for each pair of edges $(e_i, e_{i+1})$ at either a bivalent or a black trivalent vertex $V$
\begin{equation}\label{eq:form_1}
\sigma_V (e_i, e_{i+1})= (-1)^{\epsilon_{wind} (e_i,e_{i+1})+\epsilon_{int} (e_i,e_{i+1}) },
\end{equation}
whereas
\begin{equation}\label{eq:form_2}
\sigma_V (e_i, e_{i+1})= (-1)^{ \epsilon_{wind} (e_i,e_{i+1})+\epsilon_{int} (e_i,e_{i+1}) +\epsilon_{cd} (e_i,e_{i+1})},
\end{equation}
for each $(e_i, e_{i+1})$ at a white trivalent vertex. 
We call a signature as in (\ref{eq:form_1}) and (\ref{eq:form_2}) on $(\mathcal N, \mathcal O, \mathfrak l)$ totally non--negative since the signature of edges at the boundary is uniquely defined by the geometry (see Section \ref{sec:vectors}).
We conjecture that this geometric formulation encodes all full rank sets of linear relations at vertices on $\mathcal G$ for $\S$.
If this conjecture is true, then our construction of edge vectors using gauge ray directions exhausts all signatures at vertices associated to the non-negativity property of maximal minors.

In Definition \ref{def:sign_rela} the indices $\epsilon_{cs}$ and $\epsilon_{tot}$ play the role of formal discrete differential forms. The signatures at pair of edges rule the transformation from the geometric setting encoded by $\epsilon_{wind}$, $\epsilon_{int}$ and $\epsilon_{cd}$ to the differential form setting encoded by $\epsilon_{cd}$ and $\epsilon_{tot}$.
Indeed, for each pair of edges $(e_i, e_{i+1})$ at either a bivalent or a black trivalent vertex $V$
\begin{equation}\label{eq:form_3}
\epsilon_{cs} (e_i,e_{i+1}) = \epsilon_{wind} (e_i,e_{i+1})+\epsilon_{int} (e_i,e_{i+1}) \quad (\!\!\!\!\!\!\mod 2),
\end{equation}
whereas for each $(e_i, e_{i+1})$ at a white trivalent vertex
\begin{equation}\label{eq:form_4}
\epsilon_{tot} (e_i,e_{i+1}) +\epsilon_{cs} (e_i,e_{i+1}) = \epsilon_{wind} (e_i,e_{i+1})+\epsilon_{int} (e_i,e_{i+1}) +\epsilon_{cd} (e_i,e_{i+1})\quad (\!\!\!\!\!\!\mod 2).
\end{equation}

The right hand side in (\ref{eq:form_3}) and (\ref{eq:form_4}) encodes the possible values of the signatures associated to a given graph depending on its orientation and the choice of gauge ray direction. The construction of the previous sections shows that for any soliton data $(\mathcal K, [A])$, $[A]\in \S$, and for almost every fixed normalizing time $\vec t_0$ there exists a unique solution to the system (\ref{eq:form_3}), (\ref{eq:form_4}) satisfying Definition \ref{def:sign_rela} and (\ref{eq:admis_eps}). Such solution explicitly gives the real and regular divisor configuration associated to the soliton data $(\mathcal K, [A])$ at $\vec t_0$. 

Therefore a natural question is which real and regular divisor configurations satisfying Definition \ref{def:sign_rela} and (\ref{eq:admis_eps}) are possible for a given network as we vary either $\mathcal K$ or $\vec t_0$. In particular, for given soliton data $(\mathcal K, [A])$, the dependence of $\epsilon_{tot}$ on $\vec t_0$ rules the dynamics of the zero divisor associated to the KP wave function $\hat \psi(P, \vec t)$ on $\Gamma\backslash \{ P_0\}$ and it is naturally connected to the problem of classifying realizable asymptotic soliton graphs in \cite{KW2} since the signatures rule also the admissible asymptotic positions of the KP zero divisor. Next Lemma gives some natural constraints.

One can also start from a formal assignment of indices $\epsilon_{tot}$ and $\epsilon_{cs}$ on $\mathcal G$ satisfying Definition \ref{def:sign_rela} and (\ref{eq:admis_eps}) and check whether they correspond to a KP divisor for some soliton data $(\mathcal K, [A])$ and $\vec t_0$. Next Lemma gives the constraints for the positions of the Sato divisor at time $\vec t_0$ and of the essential singularity of the wave function in the ovals. If one were able to solve (\ref{eq:form_3}) and (\ref{eq:form_4}), {\sl i.e.} construct a perfect orientation $\mathcal O=\mathcal O(I)$ and a gauge ray direction $\mathfrak{l}$, then Theorem \ref{theo:null} would allow to construct the wave function for any given soliton data $(\mathcal K, [A])$ and to check whether the position of the divisor in the ovals corresponds to $\epsilon_{tot}$ for some $\vec t_0$.
\end{remark}

\begin{lemma}\label{lem:signature}\textbf{Dependence of the vertex signature on $\vec t_0$.}
To the soliton data $(\mathcal K, [A])$ and the fixed network $\mathcal N$ of graph $\mathcal G$ representing $[A]$ as above, we associate a finite number of admissible vertex signatures satisfying both Lemmas \ref{lem:adm_1} and \ref{lem:adm_2} as $\vec t_0$ varies. All these vertex signatures share the following properties:
\begin{enumerate}
\item Exactly $k+1$ of the faces intersecting the boundary of the disk have the following property: all pairs of edges $(e_j, e_{j+1})$ at white vertices contained in the boundary of such faces satisfy $\epsilon_{tot} (e_j, e_{j+1})=0$. Moreover one of these faces has this property for all $\vec t_0$ (infinite oval); 
\item The remaining $|D|-k$ faces $f$ possess exactly one pair of edges at a white vertex $V$ in $\partial f$ of total index 1.
\end{enumerate}
Such signatures rule the possible positions of the KP zero divisor in the ovals of the corresponding curve $\Gamma$ as $\vec t_0$ varies.
\end{lemma}

Finally, let us we summarize the main constructions and restate the main Theorems in the framework of the half-edge dressed wave function.
In the setting of Definition \ref{def:sign_rela}, we have shown that for any $I \in{\mathcal M}$, any orientation $\mathcal O =\mathcal O(I)$ and any gauge ray direction $\mathfrak l$, the linear system at the vertices has an unique solution, and moreover, given the boundary conditions 
\begin{equation}\label{eq:lam_bound}
z_{(b_j,e)} \equiv z_{(b_j,e)} (\vec t) = {\mathfrak D}^{(k)}\phi^{(0)} (\kappa_j; \vec t), \quad\quad j\in \bar I,
\end{equation}
where ${\mathfrak D}^{(k)}\phi^{(0)} (\zeta; \vec t)$ is the Sato wave function defined in (\ref{eq:Satowf}), the solution also satisfies 
\[
z_{(b_i,e)} \equiv z_{(b_i,e)} (\vec t) = {\mathfrak D}^{(k)}\phi^{(0)} (\kappa_i; \vec t), \quad\quad i\in  I,
\]
where we adopt the convention that $z_{(b_i,e)}=-z_{(V_i,e)}$ at boundary sources, since $z_{(V_i,e)}= -{\mathfrak D}^{(k)}\phi^{(0)} (\kappa_i; \vec t)$ by Corollary \ref{cor:bound_source}.
Therefore $z_{(V,e)}$ is the value of the non--normalized dressed wave function at each half-edge in the network for the soliton data $(\mathcal K, [A])$.

By construction (see Section \ref{sec:vertex_wavefn_general_case}) there always exists $\vec t_0$ such that $z_{(V,e)}(\vec t_0) \not =0$ for all pair of half edges $(U,e)$ in $\mathcal N$. This is a consequence of the fact that we work with reduced networks which admit an acyclic orientation and therefore do no possess null edge vectors. Therefore, the normalized KP wave function  
\[
\hat z_{(V,e)}(\vec t) = \frac{z_{(V,e)}(\vec t)}{z_{(V,e)}(\vec t_0)},
\] 
takes the same value at each double point on the curve $\Gamma$ for all times, {\sl i.e. }
\[
\hat z_{(V,e)}(\vec t) = \hat z_{(U,e)}(\vec t), \quad\quad \forall e=(U,V).
\]
Finally we may reconstruct the real regular divisor for the soliton data $(\mathcal K, [A])$ on the reducible curve $\Gamma$ using the half--edge wave function.

Let $\DS = \DS (\vec t_0)$ be the Sato divisor for the soliton data $(\mathcal K, [A])$ as in Definition \ref{def:Sato_data}. For any trivalent white vertex $V$ let us define a real number $\gamma_V$ by
\[
\gamma_V = \frac{(-1)^{\epsilon_{tot} (e_2,e_3)-\epsilon_{cs} (e_2,e_3)} z_{(V,e_1)}}{(-1)^{\epsilon_{tot} (e_1,e_2)-\epsilon_{cs} (e_1,e_2)  } z_{(V,e_3)}},
\]
where edges are labeled anticlockwise. Let $\Gamma_V\subset \Gamma$ be the component corresponding to $V$, and $P_V^{(m)}$, $m\in[3]$ be three points in $\Gamma_V$ corresponding to the edges $e_m$ at $V$, and $\zeta$ be the coordinate on $\Gamma_V$ such that
\[
\zeta(P_V^{(1)}) =0, \quad\quad \zeta(P_V^{(2)}) =1, \quad\quad 
\zeta(P_V^{(3)}) =\infty,
\]
and let $P_V \in \Gamma_V$ be such that $\gamma_V$ is its $\zeta$--coordinate,
\[
\gamma_V = \zeta (P_V).
\]
Then 
\[
\DKP = \DS \cup \{ P_V \, : \, V \mbox{ trivalent white vertex in } \mathcal G \}
\]
is the real and regular KP divisor of degree $d=\dim (\S)$, associated to the soliton data $(\mathcal K, [A])$ on $\Gamma$.

Therefore, Theorems~\ref{theo:consist}, \ref{theo:pos_div} and \ref{prop:comb_oval} and Corollaries \ref{cor:eps_tot} and \ref{cor:DKP} can be restated as follows:

\begin{theorem}\label{theo:sig_rela}\textbf{Signatures generating totally non-negative positroid cells and real regular KP divisors.}
Let $\mathcal G$ be a reduced PBDTP graph representing an irreducible positroid cell $\S \in \GTNN$. Let $\mathcal N$ be a network of graph $\mathcal G$. Then:
\begin{enumerate}
\item Any choice of orientation and gauge ray direction $(\mathcal O, \mathfrak{l})$ on $\mathcal N$ induces a vertex system of relations of maximal rank associated to the same point $[A]\in\S$ which fulfils Definitions \ref{def:sign_rela} and \ref{def:vertex_sign}, and Lemmas \ref{lem:adm_1} and \ref{lem:adm_2};
\item Let $z_{(V,e)}$ be the half-edge dressed wave function for the soliton data $(\mathcal K, [A])$. Then the system of relations at vertices induces a divisor point $P_V$ on each component $\Gamma_V$ corresponding to a trivalent white vertex $V$ such that $P_V$ belongs to the oval $\Omega_s$, $s\in \dim (\S)$, of $\Gamma$ corresponding to the face $f_s$ in $\mathcal N$ bounded by the pair of edges of index $\epsilon_{tot}=1$. Finally, each (finite) oval contains exactly one divisor point,
$\# (\Omega_s \cap \DKP ) = 1$, $s\in \dim (\S)$, whereas the remaining (infinite) oval contains the essential singularity of the KP wave function. 
\end{enumerate}
\end{theorem}

\section{Effect of moves and reductions on curves, edge vectors and divisors}\label{sec:moves_reduc}

In \cite{Pos} it is introduced a set of local transformations - moves and reductions - on planar bicolored networks in the disk which leave invariant the boundary measurement map. Therefore two networks in the disk connected by a sequence of such moves and reductions represent the same point in $\GTNN$. 
There are three moves: 
\begin{enumerate}
\item[(M1)] The square move (see Figure \ref{fig:squaremove});
\item[(M2)] The unicolored edge contraction/uncontraction (see Figure \ref{fig:unicolor});
\item[(M3)] The middle vertex insertion/removal (see Figure \ref{fig:middle});
\end{enumerate}
and three reductions:
\begin{enumerate}
\item[(R1)] The parallel edge reduction (see Figure \ref{fig:parall_red_poles});
\item[(R2)] The dipole reduction (see Figure \ref{fig:dip_leaf}[left]);
\item[(R3)] The leaf reduction (see Figure \ref{fig:dip_leaf}[right]).
\end{enumerate}
In our construction each such transformation induces a well defined change in both the curve, the system of edge vectors and the (vacuum or dressed) divisor which we describe below.

In the following, we restrict ourselves to PBDTP networks and, without loss of generality, we fix both the orientation and the gauge ray direction. Indeed changes of orientation or of gauge direction produce effects on both the system of vectors, the edge wave functions and the vacuum and dressed divisors which are completely under control in view of the results of Section \ref{sec:vectors}.

We label with the same indices vertices, edges, faces in the network corresponding to components, double points, ovals in the curve. We denote $({\mathcal N}, \mathcal O, \mathfrak l)$ the initial oriented 
network and $({\tilde {\mathcal N}}, {\tilde{\mathcal O}}, \mathfrak l)$ the oriented network obtained from it by applying one move (M1)--(M3) or one reduction 
(R1)--(R3). We assume that the orientation ${\tilde {\mathcal O}}$ coincides with $\mathcal O$ at all edges except at those involved in the move or reduction where we use Postnikov rules to assign the orientation. We denote with the same symbol and a tilde any quantity referring to the transformed network or the transformed curve. For instance, $g$, $\Omega_s$, $s\in [0, g]$ and ${\tilde g}$, ${\tilde \Omega}_s$, $s\in [0, {\tilde g}]$ respectively denote the genus and the ovals in the initial and transformed curves. 
To simplify notations, we use the same symbol $\gamma_l$, respectively ${\tilde \gamma}_l$ for the divisor number and the divisor point before and after the transformation.

\subsection{(M1) The square move}\label{sec:square}

If a network has a square formed by four trivalent vertices
whose colors alternate as one goes around the square, then one can switch the colors of these
four vertices and transform the weights of adjacent faces as shown in Figure \ref{fig:squaremove}[left].

The relation between the face weights before and after the square move is \cite{Pos}
\[
{\tilde f}_5 = (f_5)^{-1}, \quad {\tilde f}_1 = \frac{f_1}{1+1/f_5}, \quad {\tilde f}_2 = f_2 (1+f_5), \quad {\tilde f}_3 = f_3 (1+f_5),\quad {\tilde f}_4 = \frac{f_4}{1+1/f_5},
\]
so that the relation between the edge weights with the orientation in Figure \ref{fig:squaremove} is
\begin{equation}\label{eq:alpha}
{\tilde \alpha}_1 = \frac{\alpha_3\alpha_4}{{\tilde\alpha}_2}, 
\quad\quad {\tilde \alpha}_2 = \alpha_2 + \alpha_1\alpha_3\alpha_4, \quad\quad  {\tilde \alpha}_3 = \frac{\alpha_2\alpha_3}{{\tilde \alpha}_2}, \quad\quad \displaystyle {\tilde \alpha}_4 = \frac{\alpha_1\alpha_3}{{\tilde \alpha}_2}.
\end{equation}
The system of equations on the edges outside the square is the same before and after the move and also the boundary
conditions remain unchanged. The uniqueness of the solution implies that all vectors outside the square 
including $E_1$, $E_2$, $E_3$, $E_4$ remain the same. In the following Lemma, $e_j, h_k$ respectively are the edges carrying the vectors $E_j,F_k$ in the initial configuration. For instance, $\i(h_1)$ is the number of intersections of gauge rays with the edge carrying the vector $F_1$ in the initial configuration, whereas $\w(-h_1,h_2)$ is the winding number of the pair of edges carrying the vectors $\tilde F_1, \tilde F_2$ after the move because the edge $h_1$ has changed of versus.

\begin{figure}
  \centering{\includegraphics[width=0.45\textwidth]{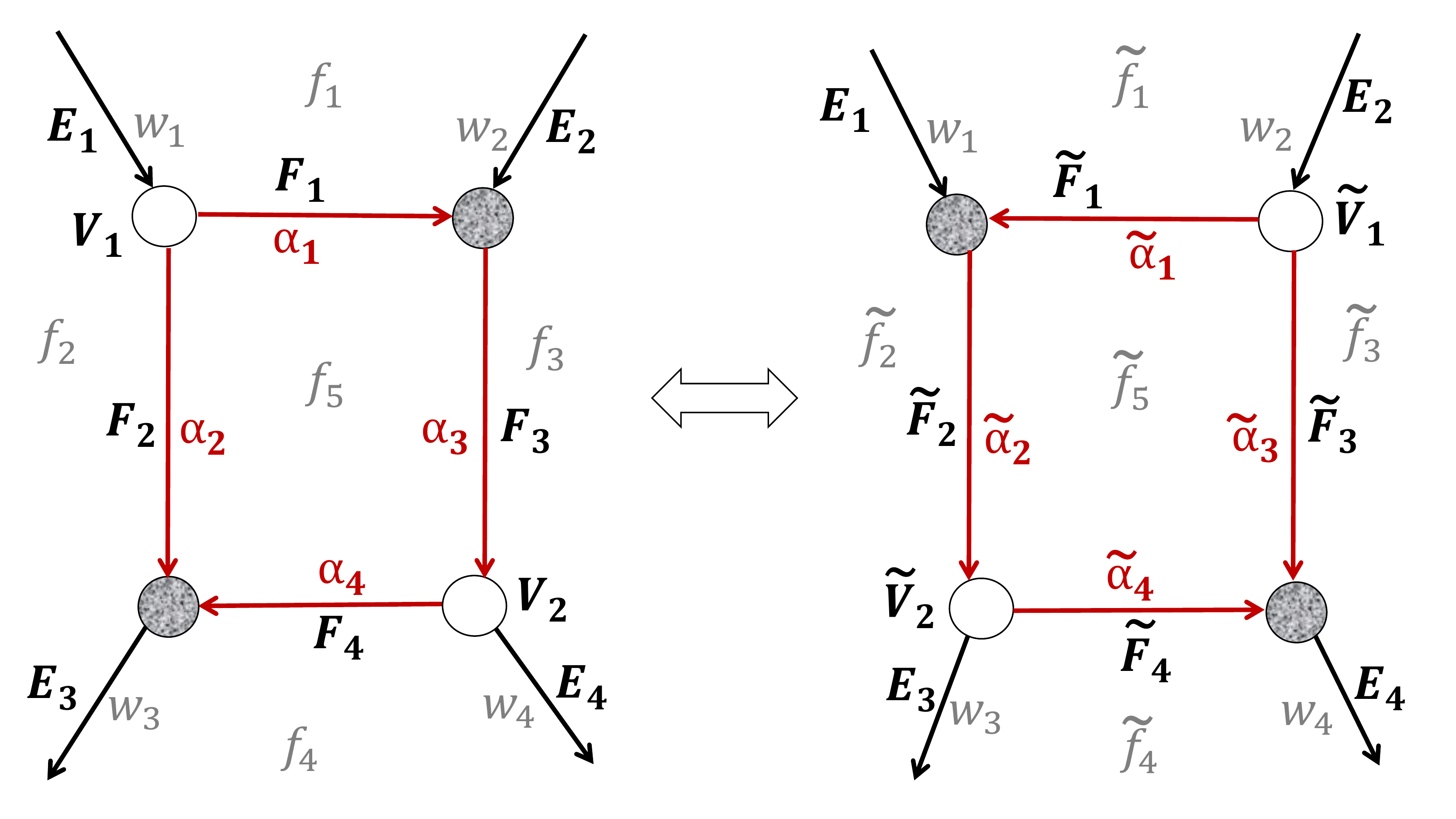}
	\hfill
	\includegraphics[width=0.45\textwidth]{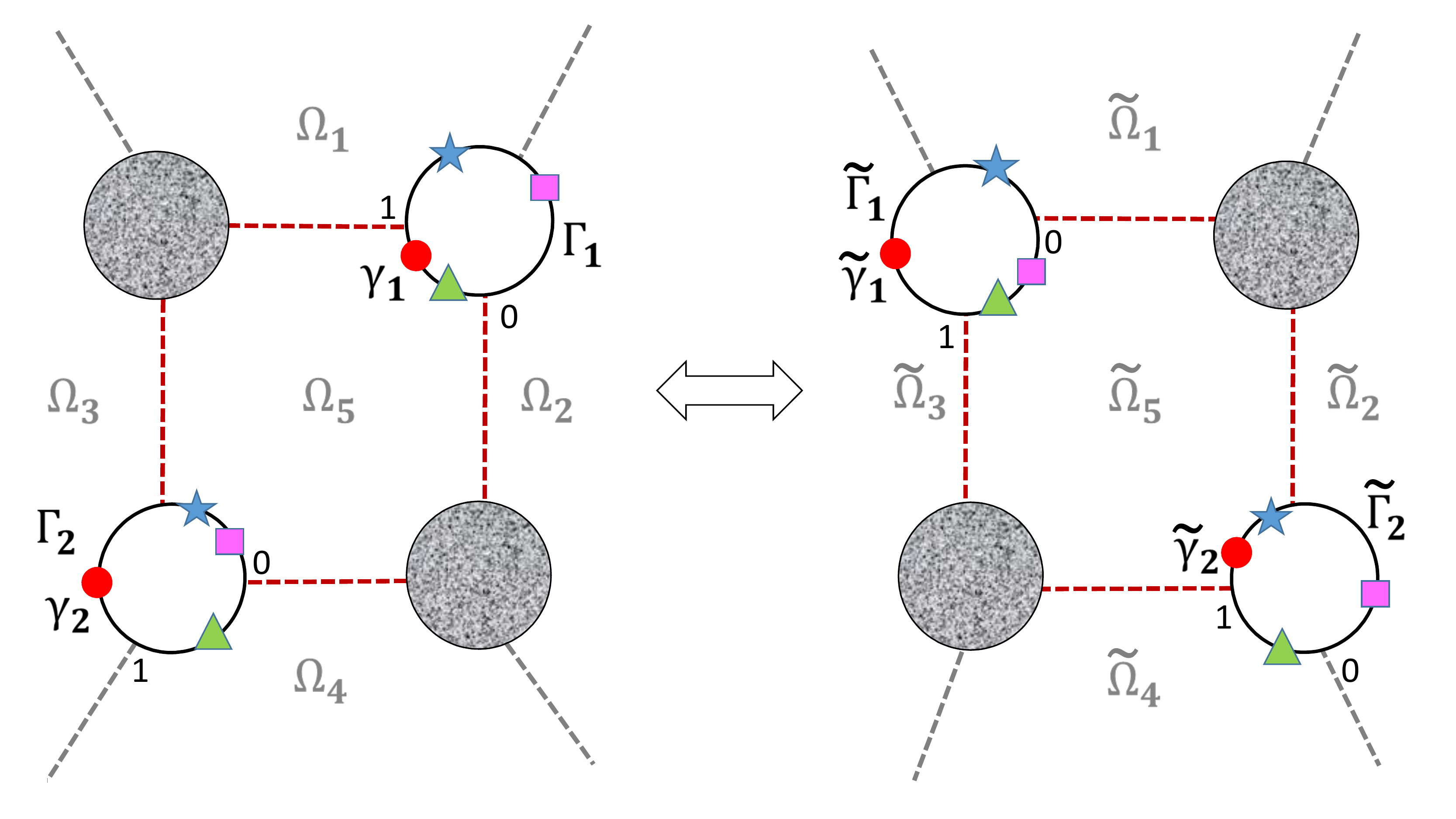}}
  \caption{\small{\sl The effect of the square move [left] on the possible configurations of vacuum or dressed divisor points [right].}\label{fig:squaremove}}
\end{figure}

\begin{lemma}\label{lem:square_edge_vectors}\textbf{The effect of the square move on the edge vectors}
The relations between the system of vectors $F_j$,${\tilde F}_j$ and $E_j$, $j\in [4]$, for the orientation in Figure \ref{fig:squaremove} are
\begin{equation}\label{eq:move1}
\begin{array}{l}
\displaystyle F_1 = (-1)^{\i(h_1)+\w(h_1,h_3)} \alpha_1 F_3, \quad F_2 = (-1)^{i(h_2)+\w(h_2,e_3)} \alpha_2E_3,\quad F_4 = (-1)^{\i(h_4)+\w(h_4,e_3)} \alpha_4 E_3,
\\
F_3 = (-1)^{\i(h_3)} \alpha_3 \left( (-1)^{\i(h_4)+\w(h_3,h_4)+\w(h_4,e_3)}\alpha_4 E_3 + (-1)^{\w(h_3,e_4)} E_4  \right),
\end{array}
\end{equation}
\begin{equation}\label{eq:move2}
\begin{array}{l}
{\tilde F}_1 =  (-1)^{\i(h_1)+\w(-h_1,h_2)} {\tilde F}_2, \quad {\tilde F}_3 =  (-1)^{i(h_3)+\w(h_3,e_4)} {\tilde \alpha}_3 E_4, \quad {\tilde F}_4 =  (-1)^{\i(h_4)+\w(-h_4,e_4)} {\tilde \alpha}_4 E_4, \\
{\tilde F}_2 =  (-1)^{\i(h_2)} {\tilde \alpha}_2 \left( \,  (-1)^{\w(h_2,e_3)} E_3 + (-1)^{\i(h_4)+\w(h_2,-h_4)+\w(-h_4,e_4)} {\tilde \alpha}_4  E_4  \right),
\end{array}
\end{equation}
where we have used the abridged notations $i(e)=\mbox{int}(e)$, $w(e,h)=\mbox{wind}(e,h)$.
\end{lemma}

The proof is straightforward and is omitted.
\begin{table}
\caption{The effect of the square move on the dressed (vacuum) divisor} 
\centering
\begin{tabular}{|c|c|c|c|}
\hline\hline
Position of poles in $\Gamma^{(1)}$ & Position of poles in $\Gamma^{(2)}$ & Symbol for divisor point & Range of parameter\\[0.5ex]
\hline
$\gamma_1 \in \Omega_5, \;  \gamma_2 \in \Omega_4$    
& ${\tilde \gamma}_1 \in {\tilde \Omega}_5, \;  {\tilde \gamma}_2 \in {\tilde \Omega}_4$ & $\triangle$ & $\psi_0 > 0$\\
$\gamma_1 \in \Omega_5, \;  \gamma_2 \in \Omega_2$    
& ${\tilde \gamma}_1 \in {\tilde \Omega}_2, \;  {\tilde \gamma}_2 \in {\tilde \Omega}_5$  
& $\bigcircle$   & $-\alpha_4 < \psi_0 < 0$\\
$\gamma_1 \in \Omega_1, \;  \gamma_2 \in \Omega_5$    
& ${\tilde \gamma}_1 \in {\tilde \Omega}_1, \;  {\tilde \gamma}_2 \in {\tilde \Omega}_5$ 
& $\largestar$  & $- ({\tilde \alpha}_4)^{-1} < \psi_0 < -\alpha_4$\\
$\gamma_1 \in \Omega_3, \;  \gamma_2 \in \Omega_5$    
& ${\tilde \gamma}_1 \in {\tilde \Omega}_5, \;  {\tilde \gamma}_2 \in {\tilde \Omega}_3$   
& $\square$  & $\psi_0 < - ({\tilde \alpha}_4)^{-1}$\\[1ex]
\hline
\end{tabular}
\label{table:SM}
\end{table}

Using Lemma \ref{lem:square_edge_vectors}, $\alpha_2 < {\tilde \alpha}_2$, $\alpha_4 < ({\tilde \alpha}_4)^{-1}$ and
\[
\begin{array}{c}
\mbox{wind}(e_1,h_1)+\mbox{wind}(h_1,h_3)+\mbox{wind}(h_3,e_4) =\mbox{wind}(e_1,h_2)+\mbox{wind}(h_2,-h_4)+\mbox{wind}(-h_4,e_4) \quad (\!\!\!\!\!\!\mod 2),\\ 
\mbox{wind}(e_1,h_1)+\mbox{wind}(h_1,h_3)+\mbox{wind}(h_3,h_4)+\mbox{wind}(h_4,e_3) =\mbox{wind}(e_1,h_2)+\mbox{wind}(h_2,e_3) \quad (\!\!\!\!\!\!\mod 2), \\
\sum_{i=1}^4\mbox{int}(h_i) =0 \quad (\!\!\!\!\!\!\mod 2),
\end{array}
\]
it is easy to verify the following relation between the (vacuum or dressed) divisor points before and after the square move in the local coordinates induced by the chosen orientation (see Figure \ref{fig:squaremove}).

\begin{corollary}\textbf{The effect of the square move on the position of the divisor}
Let the local coordinates on $\Gamma_i$, ${\tilde \Gamma_i}$, $i=1,2$, be as in Figure \ref{fig:squaremove} and let
\[
\psi_0 = (-1)^{\mbox{int}(h_4) + \mbox{wind}(h_3,e_4)-\mbox{wind}(h_3,h_4)-\mbox{wind}(h_4,e_3)} \frac{\Psi_{e_4} (\vec t_0)}{\Psi_{e_3} (\vec t_0)},
\]
where $\Psi_{e_j} (\vec t_0)$ is the value of the dressed (vacuum) edge wave function at the edges $e_j$, $j=3,4$. Then
\[
\gamma_1 = \frac{\alpha_2\tilde\alpha_2^{-1}}{1+{\tilde \alpha}_4 \psi_0}, \quad\quad \gamma_2 = \frac{\alpha_4}{\alpha_4 + \psi_0},\quad\quad {\tilde \gamma}_1 = \frac{\alpha_4 (1 + {\tilde \alpha}_4 \psi_0)}{\alpha_4 + \psi_0},\quad\quad {\tilde \gamma}_2 = \frac{1}{1+{\tilde \alpha}_4 \psi_0},
\]
and the position of the divisor points in the ovals depends on $\psi_0$ as shown in Table \ref{table:SM}. In particular,
there is exactly one (vacuum or dressed) divisor point in $(\Gamma_1 \cup \Gamma_2)\cap\Omega_5$, $({\tilde \Gamma}_1 \cup {\tilde \Gamma}_2)\cap{\tilde \Omega}_5$.
\end{corollary}

The square move leaves the number of ovals invariant, eliminates the divisor points $\gamma_1,\gamma_2$ and creates the divisor points ${\tilde \gamma}_1$, ${\tilde \gamma}_2$. We summarize such properties in the following Lemma.

\begin{lemma}\label{lemma:poles_move1}\textbf{The effect of the square move (M1) on the curve and the divisor} 
Let ${\tilde {\mathcal N}}$ be obtained from ${\mathcal N}$ via move (M1). Let ${\mathcal D} ={\mathcal D}({\mathcal N})$,  ${\tilde {\mathcal D}} ={\mathcal D}({\tilde {\mathcal N}})$ respectively  be the dressed (vacuum) network divisor before and after the square move.
Then
\begin{enumerate}
\item ${\tilde g}=g$, and the number of ovals is invariant;
\item The number of dressed (vacuum) divisor points is invariant in every oval: ${\tilde \nu}_{l} =\nu_{l}$, $l\in [0, g]$;
\item ${\tilde {\mathcal D}} = \left( {\mathcal D} \backslash \{ \gamma_1, \gamma_2 \} \right) \cup \{ {\tilde \gamma}_1, {\tilde \gamma}_2 \}$, where
$\gamma_l$ (respectively $\tilde\gamma_l$), $l=1,2$, is the divisor point on $\Gamma_l$ (respectively  $\tilde\Gamma_l$), the component of $\mathbb{CP}^1$ associated to the white vertex $V_l$ (respectively $\tilde V_l$) involved in the square move transforming ${\mathcal N}$ into $\tilde{\mathcal N}$;
\item Either $\gamma_l$, $\tilde\gamma_l$, $l=1,2$, are all untrivial divisor points or all trivial divisor points.
\end{enumerate}
\end{lemma}

The proof is straightforward and is omitted. 

\subsection{(M2) The unicolored edge contraction/uncontraction}\label{sec:flip}

\begin{figure}
\includegraphics[width=0.6\textwidth]{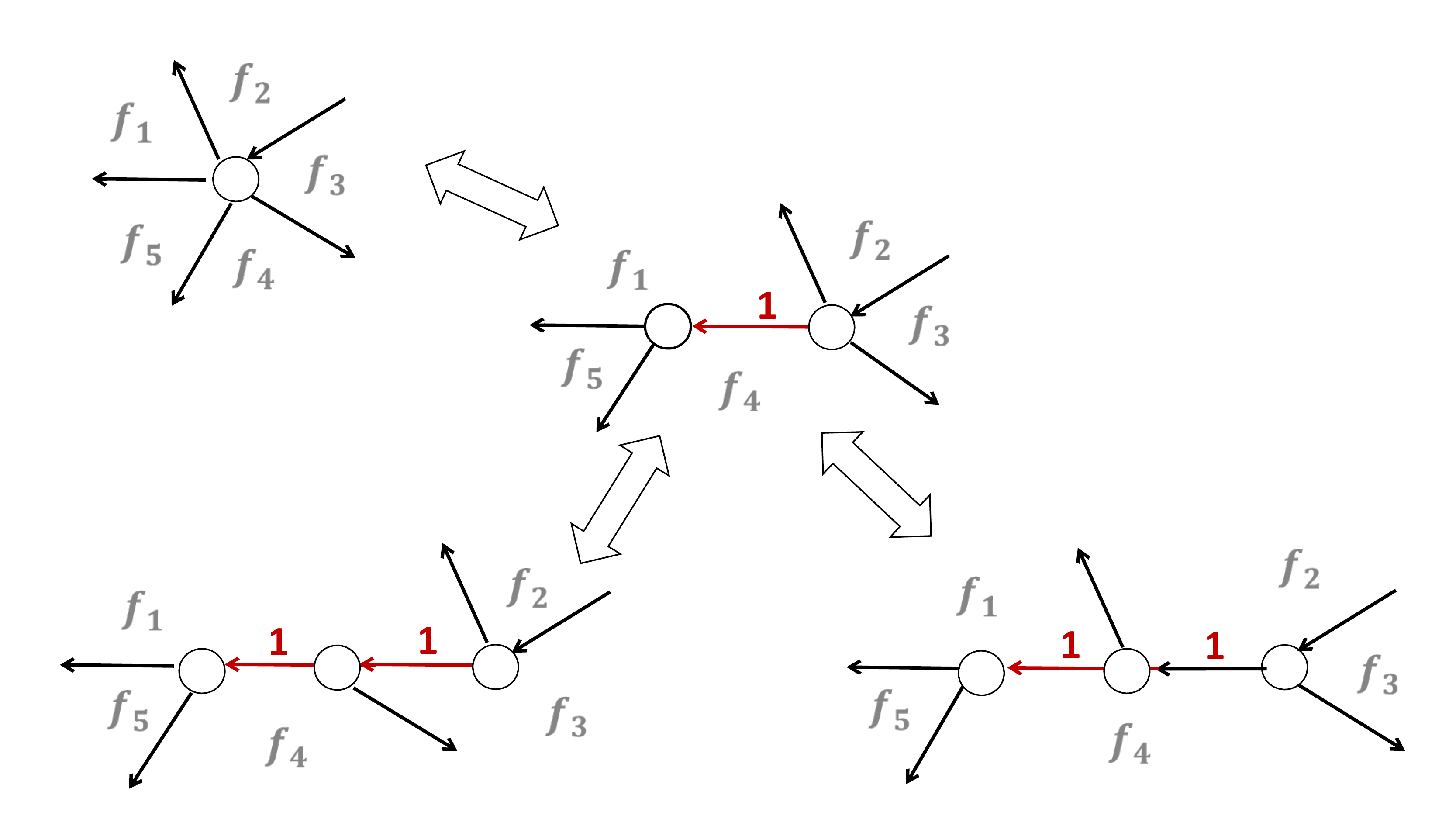}
\vspace{-0.5 truecm}
\caption{\small{\sl The unicolored edge contraction/uncontraction }}
\label{fig:unicolor}
\end{figure}

The unicolored edge contraction/un\-con\-traction consists in the elimination/addition of an internal vertex of equal color and of an unit edge, and it leaves invariant the face weights and the boundary measurement map. In Figure \ref{fig:unicolor} we show an example of possible configurations equivalent w.r.t to such move.
Since we work with trivalent graphs, there are two cases of interest to us: the insertion/removal of a Darboux vertex next to a boundary vertex (see Figure \ref{fig:Darboux}) and the flip move (see Figure \ref{fig:flipmove}).

The insertion/removal of a Darboux vertex leaves invariant the system of vectors except for the edges $e_{l}=[b_l, V_l]$, $e^{(D)}_l=[V_l, V^{(D)}_l]$. Its effect is just local and has been explained in Section \ref{sec:modN} where we have used it to modify the PBDTP network $\mathcal N$. 

\begin{figure}
  \centering{\includegraphics[width=0.6\textwidth]{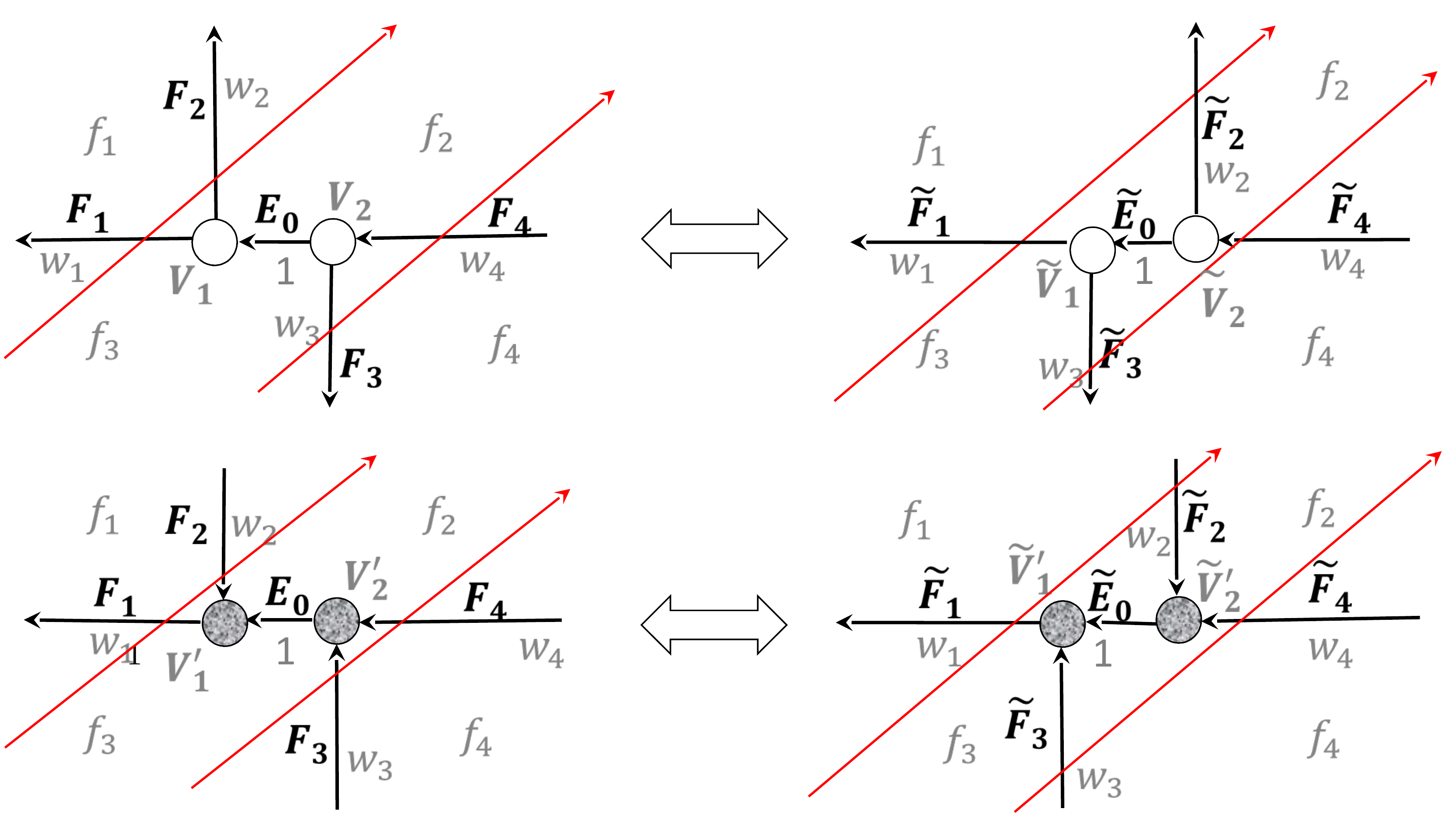}}
  \caption{\small{\sl The insertion/removal of an unicolored internal vertex is equivalent to a flip move of the unicolored vertices.}\label{fig:flipmove}}
\end{figure}

The contraction/uncontraction of an unicolored internal edge combined with the trivalency condition is equivalent to a flip of the unicolored vertices involved in the move. It is possible to express such transformation as a combination of elementary flip moves each involving a pair of unicolored vertices (see Figure \ref{fig:flipmove} for an example). In a flip move of a pair of unicolored vertices, we assume that the edge $e_0$ connecting this pair of vertices has unit weight and sufficiently small length before and after the move so that no gauge ray crosses $e_0$, all other edges at this pair of vertices do not change the ending point and their winding and intersection numbers. As a consequence, we have here additivity of winding numbers, which is generically not true:
$$
\mbox{wind} (e_i,e_0)+ \mbox{wind} (e_0,e_j) = \mbox{wind} (e_i,e_j),
$$
with $e_i$ -- any incoming vector, $e_j$ -- any outgoing vector.

The flip move acts untrivially on the edge vector on $e_0$ and as a change of sign at the edges $e_i$, $i \in 4$. As before, we use the tilde symbol for all quantities after the flip move.

\begin{lemma} Let the flip move act on a pair of trivalent equi-colored vertices and and mark the edges $e_i$, ${\tilde e}_i$, respectively before and after the move as in Figure \ref{fig:flipmove}[left]. Then ${\tilde F}_i = F_i$ after the flip move and 
\begin{equation}\label{eq:flip2}
{\tilde E}_0 = E_0,
\end{equation}
if the vertices are black.
\end{lemma}

The flip move at black vertices leaves the divisor invariant. The flip move at white vertices may transform trivial divisor points into untrivial divisor points and divisor points may move from one oval to another respecting the parity conditions 
settled in Theorems \ref{theo:exist} and \ref{theo:vac_div}.
In the following lemma we label $\Omega_l$, $l\in [4]$, the ovals involved in the flip move as in Figure \ref{fig:flip_move_poles}. 

\begin{lemma}\label{lemma:poles_move2}\textbf{The effect of the flip move (M2) at a pair of white vertices on the divisor}
Let ${\tilde {\mathcal N}}$ be obtained from $\mathcal N$ via a flip move move (M2) at a pair of trivalent white vertices. 
Let ${\mathcal D}= {\mathcal D}({\mathcal N})$, ${\tilde {\mathcal D}}= {\mathcal D}({\tilde {\mathcal N}})$, respectively be the dressed (vacuum) network divisor before and after such move. Then
\begin{enumerate}
\item ${\tilde g}=g$ and the number of ovals is invariant;
\item The number of divisor points is invariant in every oval except possibly at the ovals involved in the move.
In the ovals $\Omega_l$, $l\in [4]$ the parity of the number of divisor points before and after the move is invariant:
 ${\tilde \nu}_{l} -\nu_{l} = 0 \,\,(\!\!\!\!\mod 2)$, $l\in [4]$;
\item ${\tilde {\mathcal D}} = \left( {\mathcal D}\backslash \{ \gamma_1, \gamma_2 \} \right) \cup \{ {\tilde \gamma}_1, {\tilde \gamma}_2 \}$, where we use the same notations as in Figure \ref{fig:flip_move_poles}.
\end{enumerate} 
\end{lemma}

\begin{figure}
 \centering
	\includegraphics[width=0.45\textwidth]{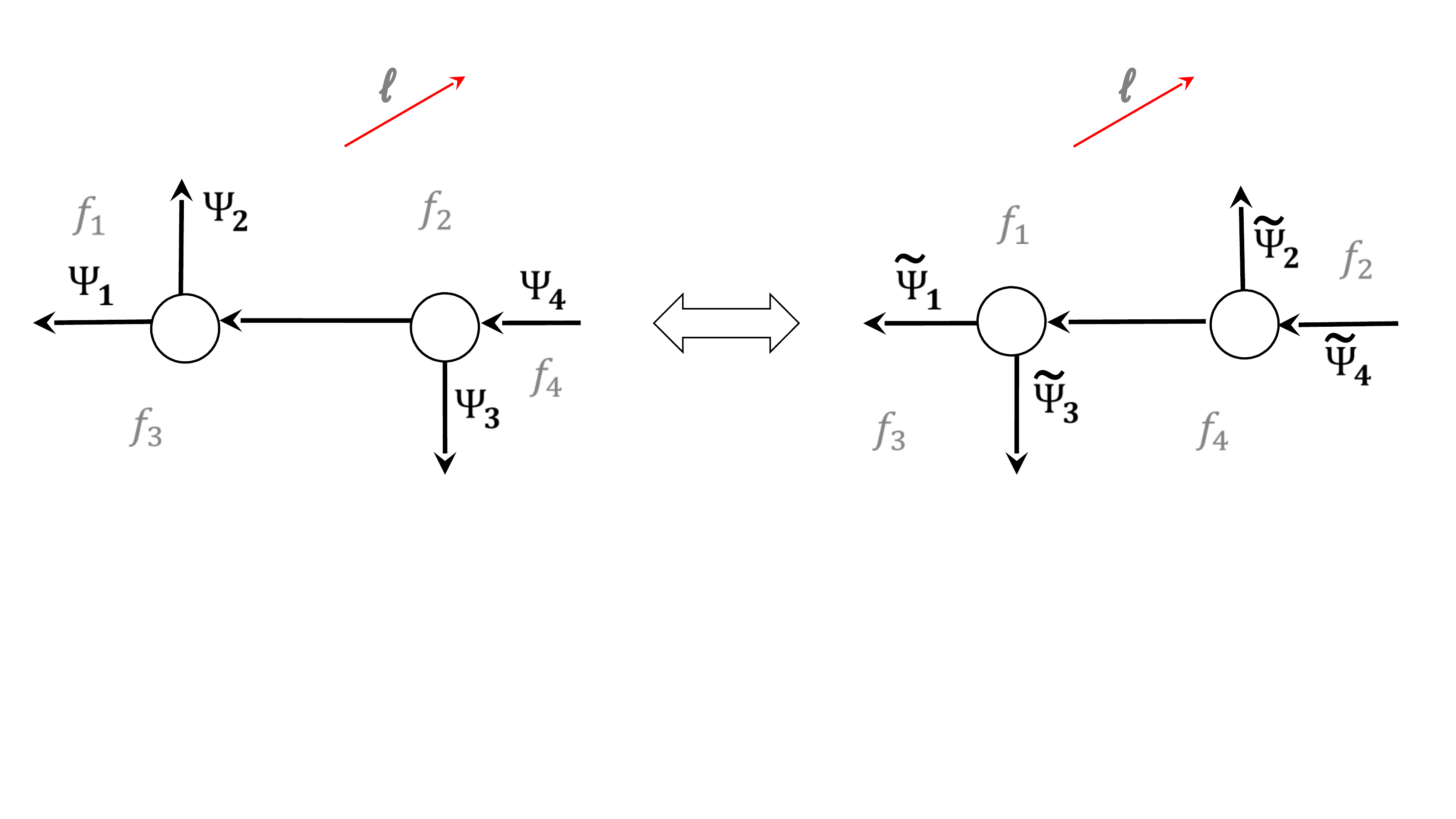}
	\hfill
	\includegraphics[width=0.47\textwidth]{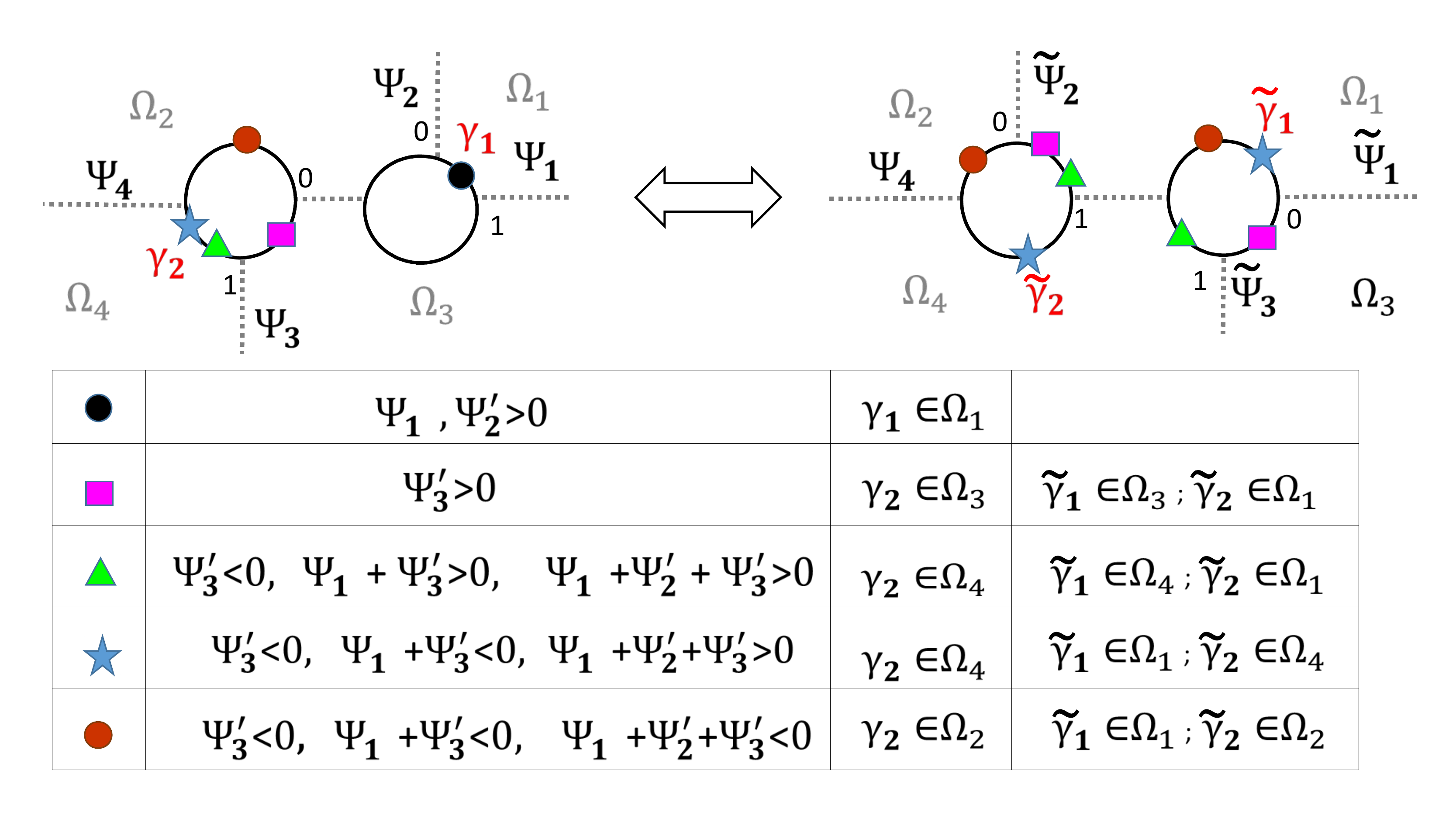}
	\includegraphics[width=0.47\textwidth]{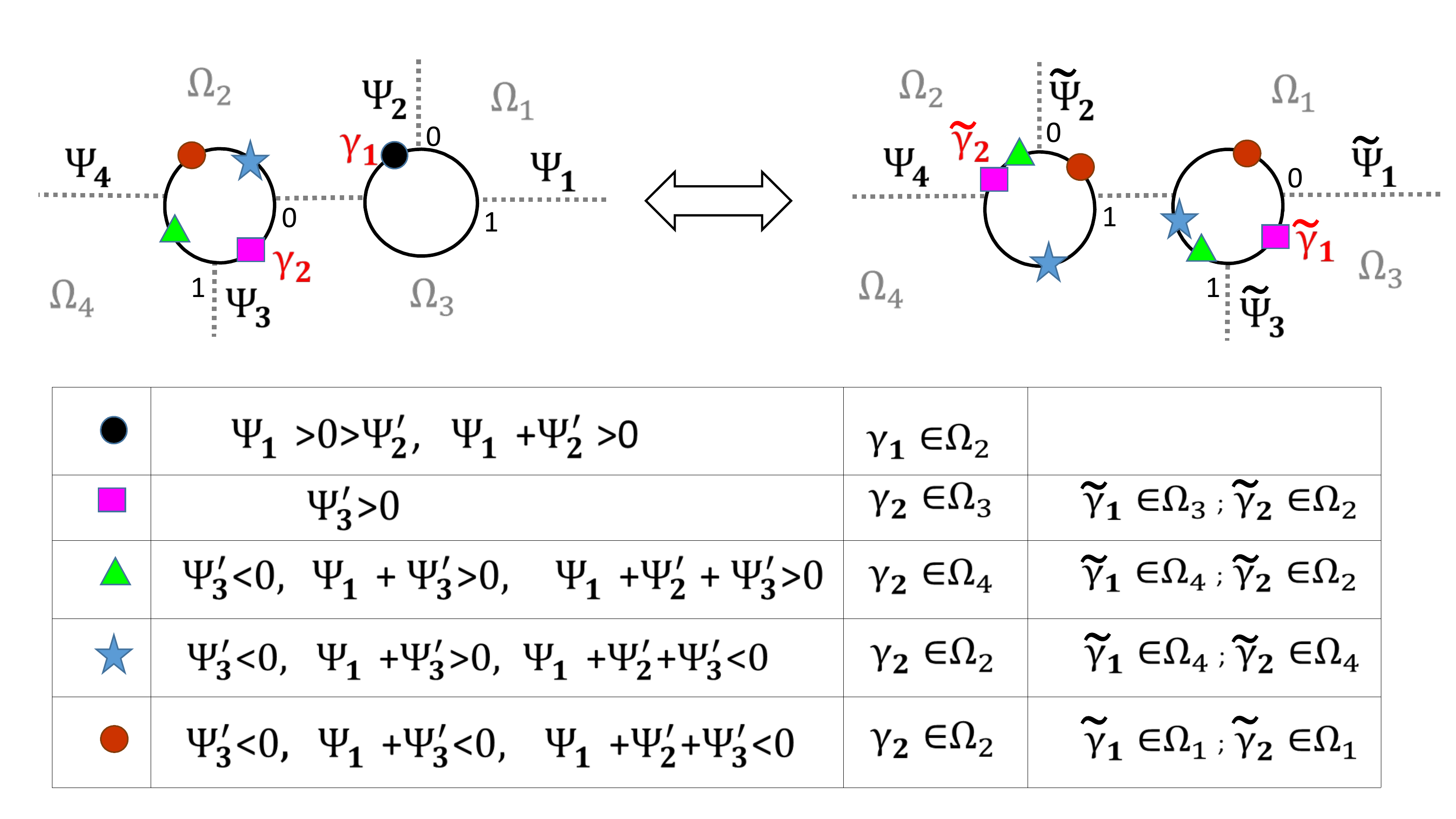}
	\hfill
	\includegraphics[width=0.47\textwidth]{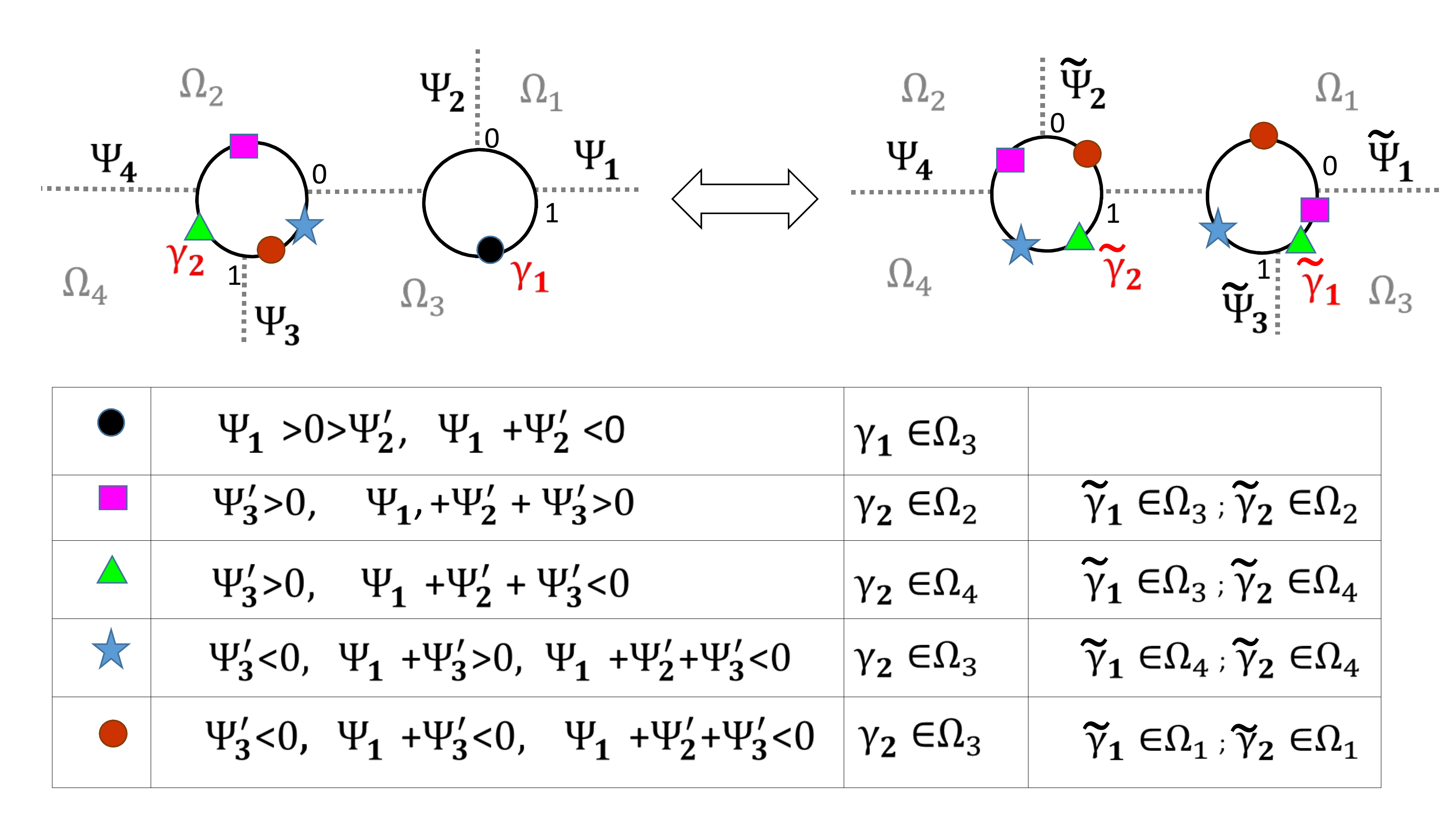}
  \caption{\small{\sl The effect of the flip move [top-left] on the possible configurations of (vacuum or dressed) divisor points [top-right],[bottom-left],[bottom-right].}\label{fig:flip_move_poles}}
\end{figure}

The proof is omitted. In Figure \ref{fig:flip_move_poles} we show the position of the divisor points before and after the flip move [top-left], in function of the relative signs of the value of the (vacuum or dressed) wave function at the double points of $\Gamma$. Faces $f_l$ correspond to ovals $\Omega_l$, the orientation of the edges in the graph at each vertex induces the local coordinates at each copy of $\mathbb{CP}^1$ in $\Gamma$.

\begin{corollary}\textbf{The effect of the flip move on the divisor}
Let the local coordinates on $\Gamma_i$, ${\tilde \Gamma_i}$, $i=1,2$, be as in Figure \ref{fig:flip_move_poles} and let $\Psi_1 = \Psi_1 (\vec t_0)$,
\[
\Psi^{\prime}_2 = (-1)^{\mbox{wind}(e_4,e_2)-\mbox{wind}(e_4,e_1)}\Psi_2(\vec t_0), \quad\quad
\Psi^{\prime}_3 = (-1)^{\mbox{wind}(e_4,e_3)-\mbox{wind}(e_4,e_1)}\Psi_3(\vec t_0),
\] 
where $\Psi_{j} (\vec t_0)$ is the value of the dressed (vacuum) edge wave function at the edges $e_j$, $j\in [3]$, in the initial configuration. Then
\[
\gamma_1 = \frac{\Psi^{\prime}_2}{\Psi_1+\Psi^{\prime}_2},\quad\quad \gamma_2 = \frac{\Psi_1+\Psi^{\prime}_2}{\Psi_1+\Psi^{\prime}_2+\Psi^{\prime}_3},\quad\quad{\tilde \gamma}_1 = \frac{\Psi_1}{\Psi_1+\Psi^{\prime}_3},\quad\quad {\tilde\gamma}_2 = \frac{\Psi^{\prime}_2}{\Psi_1+\Psi^{\prime}_2+\Psi^{\prime}_3},
\]
and the position of the divisor points in the ovals is shown in Figure \ref{fig:flip_move_poles}. 
\end{corollary}

The proof is straightforward using the definition of divisor coordinates.
We remark that, for the dressed divisor, not all combinations of signs are realizable at the finite ovals not containing Darboux sources or sinks. 
For instance in Figure \ref{fig:flip_move_poles}[bottom,left] one such forbidden combination is the star-shaped divisor which corresponds to the following choice of sign of the edge wave function 
$\Psi^{\prime}_2(\vec t_0), \Psi^{\prime}_3(\vec t_0)<0< \Psi_1 (\vec t_0)$, $\Psi_1(\vec t_0)+ \Psi^{\prime}_2(\vec t_0)<0<\Psi_1(\vec t_0)+ \Psi^{\prime}_3(\vec t_0)$,
and to divisor configurations $\gamma_1,\gamma_2 \in \Omega_2$ and ${\tilde \gamma}_1, {\tilde\gamma}_2 \in \Omega_4$, respectively before and after the flip move.

Finally, if $F_3 = c_{13}F_1$ for some $c_{13}\not =0$, whereas $F_2$, $F_1$ are linearly independent, then also $\Psi^{\prime}_3(\vec t)= \pm c_{13} \Psi_1(\vec t)$, for all $\vec t$. In such case $\gamma_l$, $l=1,2$, and ${\tilde \gamma}_2$ are untrivial divisor points, whereas ${\tilde \gamma}_1$ is a trivial divisor point, {\sl i.e.} the normalized wave function is constant on the corresponding copy of $\mathbb{CP}^1$.

\subsection{(M3) The middle edge insertion/removal}\label{sec:middle}

\begin{figure}
  \centering{\includegraphics[width=0.55\textwidth]{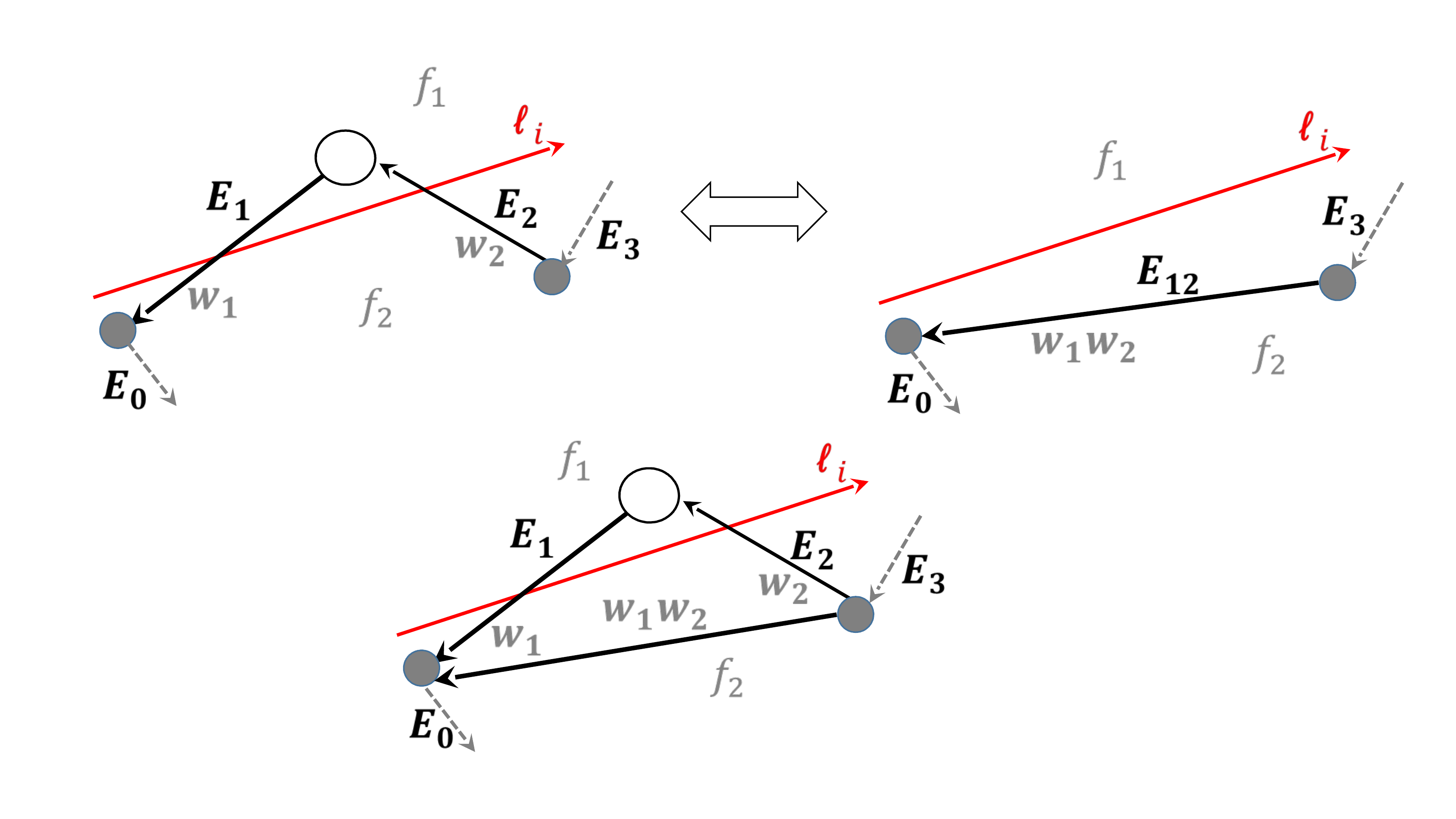}}
	\vspace{-.5 truecm}
  \caption{\small{\sl The middle edge insertion/removal.}\label{fig:middle}}
\end{figure}

The middle edge insertion/removal concerns bivalent vertices (see Figure \ref{fig:middle}) without changing the face configuration. i.e. the triangle formed by the edges $e_1$, $e_2$, $e_{12}$ does not contain other edges of the network. Then the relation between the vectors $E_2$ and $E_{12}$ is simply, 
\[
E_{12} =(-1)^{\mbox{wind}(e_3, e_2) -\mbox{wind}(e_3, e_{12})} E_2,
\]
since
\[
\mbox{int} (e_1) +\mbox{int} (e_2) = \mbox{int} (e_{12}), \quad (\!\!\!\!\!\!\mod 2), 
\]
\[
\mbox{wind}(e_3, e_2) +\mbox{wind}(e_2, e_1)+\mbox{wind}(e_1, e_0) =\mbox{wind}(e_3, e_{12}) +\mbox{wind}(e_{12}, e_0), \quad (\!\!\!\!\!\!\mod 2).
\]
This move does not change neither the number of ovals nor the divisor configuration.

\subsection{(R1) The parallel edge reduction}\label{sec:par_dip}

The parallel edge reduction consists of the removal of two trivalent vertices of different color connected by a pair of parallel edges (see Figure \ref{fig:parall_red_poles}[top]). If the parallel edge separates two distinct faces, the relation of the face weights before and after the reduction is \cite{Pos}
\[
{\tilde f}_1 = \frac{f_1}{1+(f_0)^{-1}}, \quad\quad {\tilde f}_2 = f_2 (1+f_0),
\]
otherwise ${\tilde f}_1 = {\tilde f}_2 = f_1 f_0$. In both case, for the choice of orientation in Figure \ref{fig:parall_red_poles}, the relations between the edge weights and the edge vectors respectively are 
\[
{\tilde w}_1 = w_1(w_2+w_3)w_4,
\]
\begin{equation}\label{eq:vectors_paredge}
\begin{array}{l}
\displaystyle {\tilde E}_{1} = E_1 = (-1)^{\mbox{int}(e_1)+\mbox{int}(e_2)}w_1(w_2+w_3) E_4,\\
E_2 = (-1)^{\mbox{int}(e_2)}w_2 E_4, \quad\quad E_3 = (-1)^{\mbox{int}(e_2)}w_3 E_4,
\end{array}
\end{equation}
since $\mbox{wind} (e_1,e_2) =\mbox{wind} (e_1,e_3) =\mbox{wind} (e_2,e_4) =\mbox{wind} (e_3,e_4) =0$, $\mbox{int}(e_1)+\mbox{int}(e_2)+\mbox{int}(e_4) =\mbox{int}({\tilde e}_1)$ and $\mbox{int}(e_2) =\mbox{int}(e_3)$.

\begin{figure}
  \centering{\includegraphics[width=0.55\textwidth]{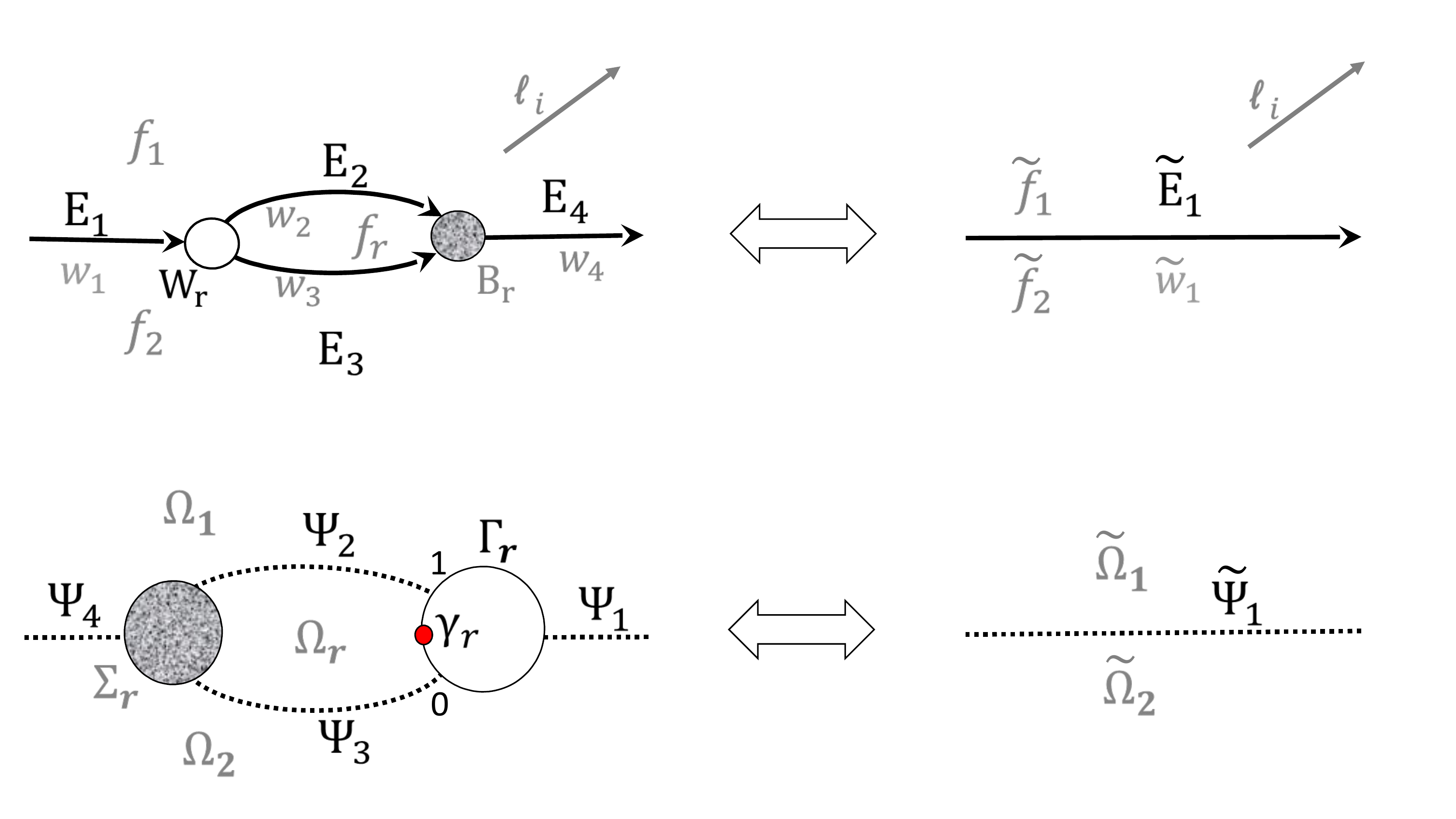}}
  \caption{\small{\sl The parallel edge reduction [top] eliminates an oval, diminishes by one the genus and eliminates a divisor point $\gamma_{\mbox{r}}$ [bottom].}\label{fig:parall_red_poles}}
\end{figure}

The parallel edge reduction eliminates an oval and a trivial divisor point (see also Figure~\ref{fig:parall_red_poles}): 

\begin{lemma}\label{lemma:poles_red1}\textbf{The effect of  (R1) on the divisor}
Let ${\tilde {\mathcal N}}^{\prime}$ be obtained from ${\mathcal N}^{\prime}$ via the parallel edge reduction (R1), and denote
${\tilde \Gamma}$ and $\Gamma$ the curve after and before such reduction. Then
\begin{enumerate}
\item The genus ${\tilde g}$ of ${\tilde \Gamma}$ is one less than that in $\Gamma$: ${\tilde g}=g-1$ ;
\item The oval $\Omega_{\mbox{r}}$ corresponding to the face $f_{\mbox{r}}$ and the components $\Gamma_{\mbox{r}},\Sigma_{\mbox{r}}$ corresponding to the white and black vertices $W_{\mbox{r}},B_{\mbox{r}}$, are removed by effect of the parallel edge reduction;
\item The trivial divisor point $\gamma_{\mbox{r}}\in \Gamma_{\mbox{r}}\cap \Omega_{\mbox{r}}$ is removed;
\item All other divisor points are not effected by the reduction.
\end{enumerate} 
\end{lemma}

In Figure \ref{fig:parall_red_poles} [bottom] we show the effect of the parallel edge reduction on the curve. Using the relations between the edge vectors in \ref{eq:vectors_paredge}, we get
\[
\Psi_{2} (\vec t)= (-1)^{\mbox{int}(e_2)} w_2 \Psi_4(\vec t), \quad \Psi_{3} (\vec t)= (-1)^{\mbox{int}(e_2)} w_3 \Psi_4(\vec t), \quad \Psi_{1} (\vec t)= (-1)^{\mbox{int}(e_1)}w_1 (\Psi_2 (\vec t)+\Psi_3(\vec t)),
\]
for all $\vec t$, so that the divisor point in $\Gamma^{(1)}_{\mbox{red}}$ is trivial
\[
\gamma_{\mbox{r}} = \frac{\Psi_{3} (\vec t_0)}{\Psi_{2} (\vec t_0)+\Psi_{3} (\vec t_0)} = \frac{w_{3} }{w_{2}+w_{3} }.
\]
On both $\Gamma^{(1)}_{\mbox{red}}$ and $\Sigma^{(1)}_{\mbox{red}}$ the normalized wave function is independent of the spectral parameter and takes the value $\frac{\Psi_{1} (\vec t)}{\Psi_{1} (\vec t_0)}\equiv \frac{\Psi_{4} (\vec t)}{\Psi_{4} (\vec t_0)}$. After the reduction, on the  edge ${\tilde e}_1$, the edge wave function ${\tilde \Psi}_1(\vec t)$  takes the value
\[
{\tilde \Psi}_{1} (\vec t) = \Psi_{1} (\vec t),
\]
so that, for any $\vec t$, the normalized wave function keeps the same value after the reduction at the double point corresponding to such edge in ${\tilde \Gamma}$.

\subsection{(R2) the dipole reduction}

The dipole reduction consists of the elimination of an isolated component consisting of two vertices joined by an edge $e$ (see Figure \ref{fig:dip_leaf}[left]). The transformation leaves invariant the weight of the face containing such component. Since the edge vector at $e$ is $E_e=0$, this transformation acts trivially on the vector system.

\begin{figure}
\includegraphics[width=0.48\textwidth]{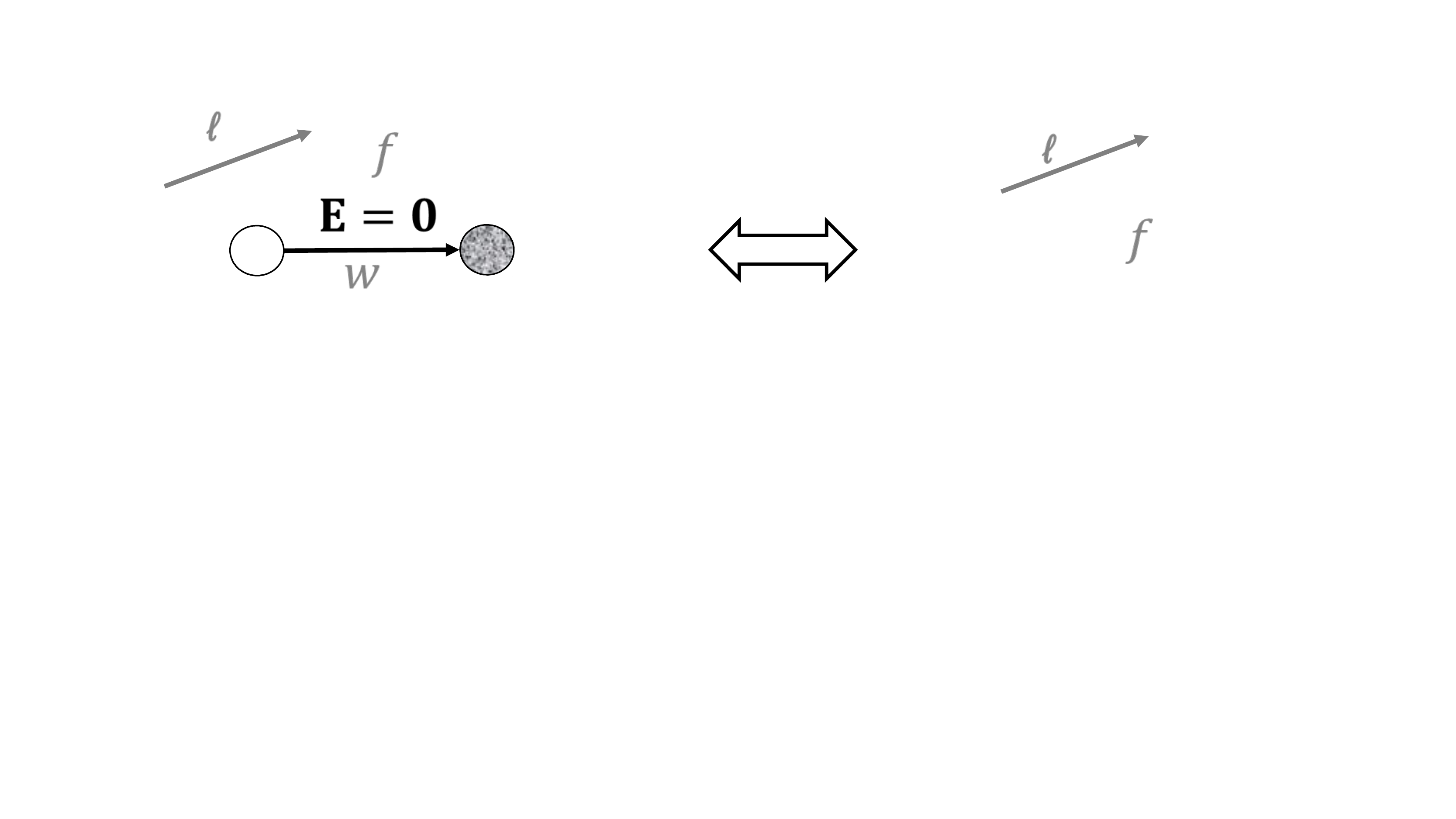}
\hfill
\includegraphics[width=0.48\textwidth]{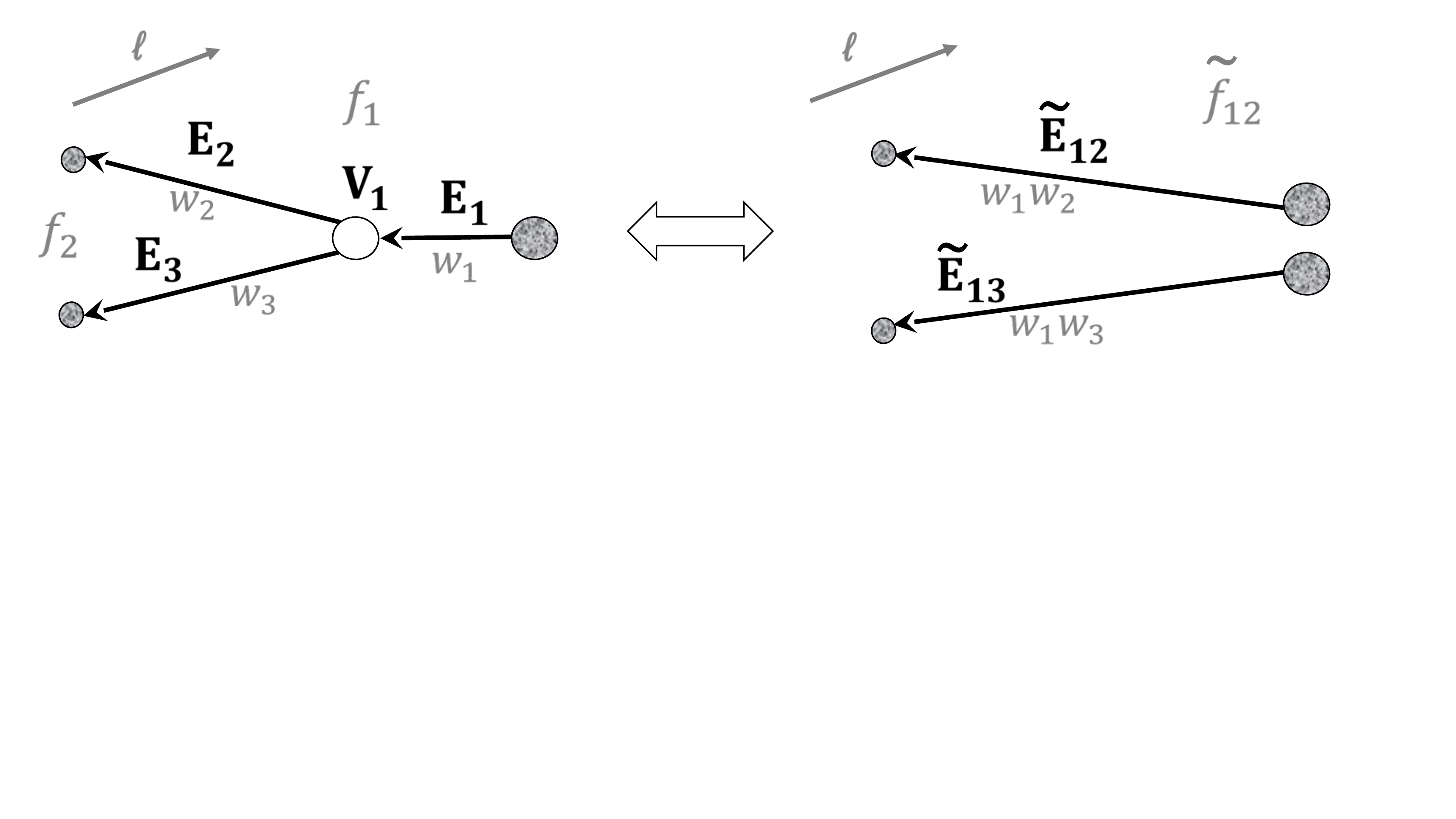}
\vspace{-2.5 truecm}
\caption{\small{\sl Left: the dipole reduction. Right: the leaf reduction.}}
\label{fig:dip_leaf}
\end{figure}

The dipole reduction (R2) corresponds to the removal of an edge carrying a null vector and leaves invariant the number of ovals and the position of the divisor points.

\subsection{(R3) The leaf reduction}\label{sec:leaf}

The leaf reduction occurs when a network contains a vertex $u$ incident to a single edge $e_1$ ending at a trivalent vertex (see Figure \ref{fig:dip_leaf}[right]): in this case it is possible to remove $u$ and $e_1$, disconnect $e_2$ and $e_3$, assign the color of $u$ at all newly created vertices of the edges $e_{12}$ and $e_{13}$. We assume that $e_1$ is short enough, and it does not intersect the gauge rays. If we have two faces of weights $f_1$ and $f_2$ in the initial configuration, then we merge them into a single face of weight ${\tilde f}_{12} =f_1f_2$; otherwise ${\tilde f}_{12}=f_1$ and the effect of the transformation is to create new isolated components. We also assume that the newly created vertices are close enough to $V_1$, therefore the windings are not affected. Then $E_1 = {\tilde E}_{12} + {\tilde E}_{13}$ and 
\[
{\tilde E}_{12} = w_{1} E_2, \quad\quad {\tilde E}_{13} =  w_{1} E_3.
\]
In the leaf reduction (R3) the only non-trivial case corresponds to the situation where the faces $f_1$, $f_2$ are distinct in the initial configuration. In this case
we eliminate one oval and one divisor point.

\begin{lemma}\label{lemma:poles_red3}\textbf{The effect of  (R3) on the divisor}
Let ${\tilde{\mathcal N}}$ is obtained from ${\mathcal N}$ via reduction (R3) and the faces $f_1$, $f_2$ be distinct. 
Let ${\mathcal D}= {\mathcal D}({\mathcal N})$ and ${\tilde {\mathcal D}}= {\mathcal D}({\tilde {\mathcal N}})$, respectively be the dressed (vacuum) network divisor before and after the reduction. Then
\begin{enumerate}
\item ${\tilde g}=g-1$ and the number of ovals diminishes by one;
\item ${\tilde {\mathcal D}} = {\mathcal D}\backslash \{ \gamma_1 \}$, where
$\gamma_1$ is the divisor point on the component associated to the white vertex $V_1$ involved in the reduction.
\end{enumerate} 
\end{lemma}

\section{Plane curves and divisors for soliton data in ${\mathcal S}_{34}^{\mbox{\tiny TNN}}\subset Gr^{\mbox{\tiny TNN}}(2,4)$}\label{sec:example}

\begin{figure}
\centering{
\includegraphics[width=0.45\textwidth]{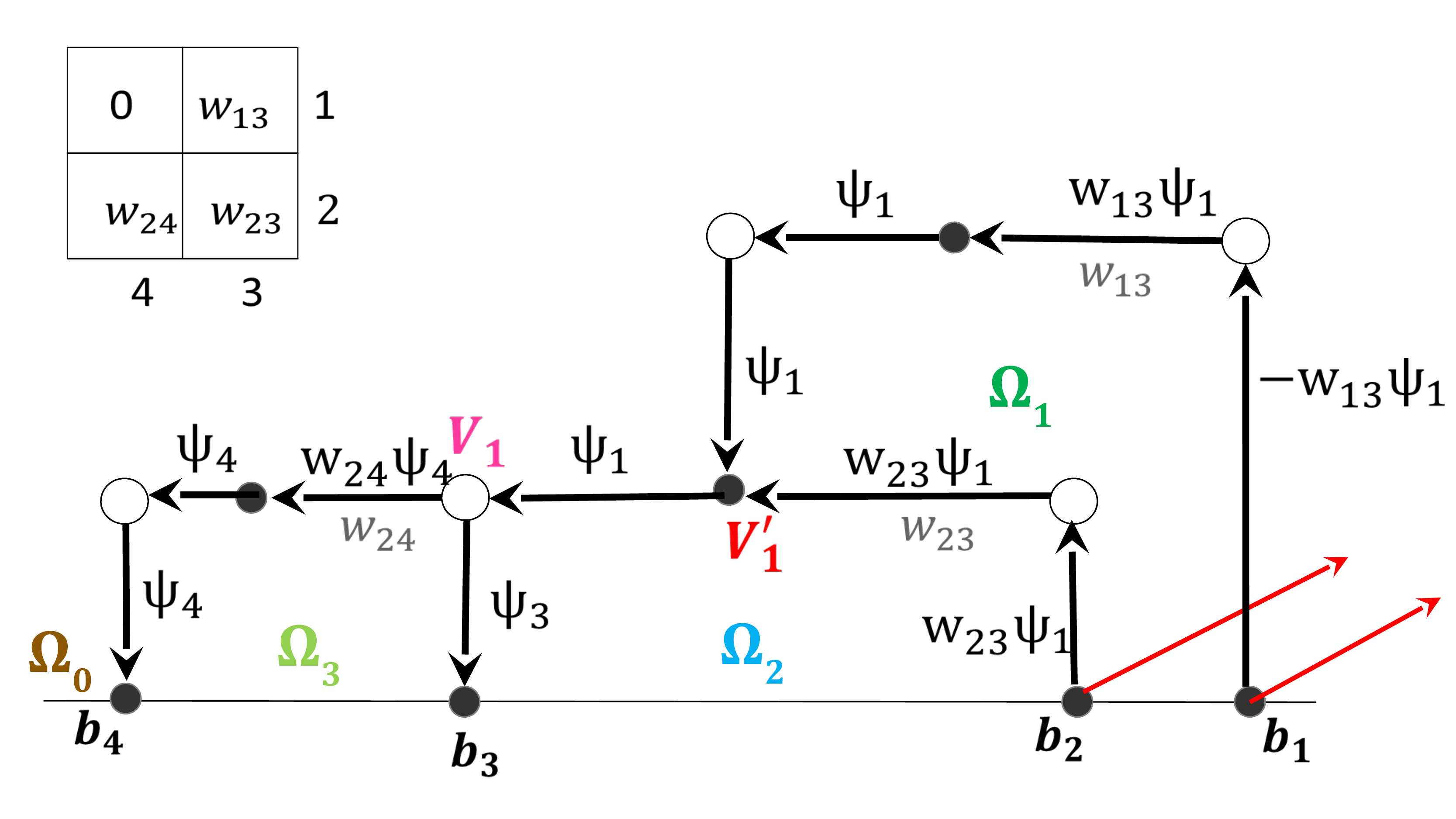}
\hfill
\includegraphics[width=0.46\textwidth]{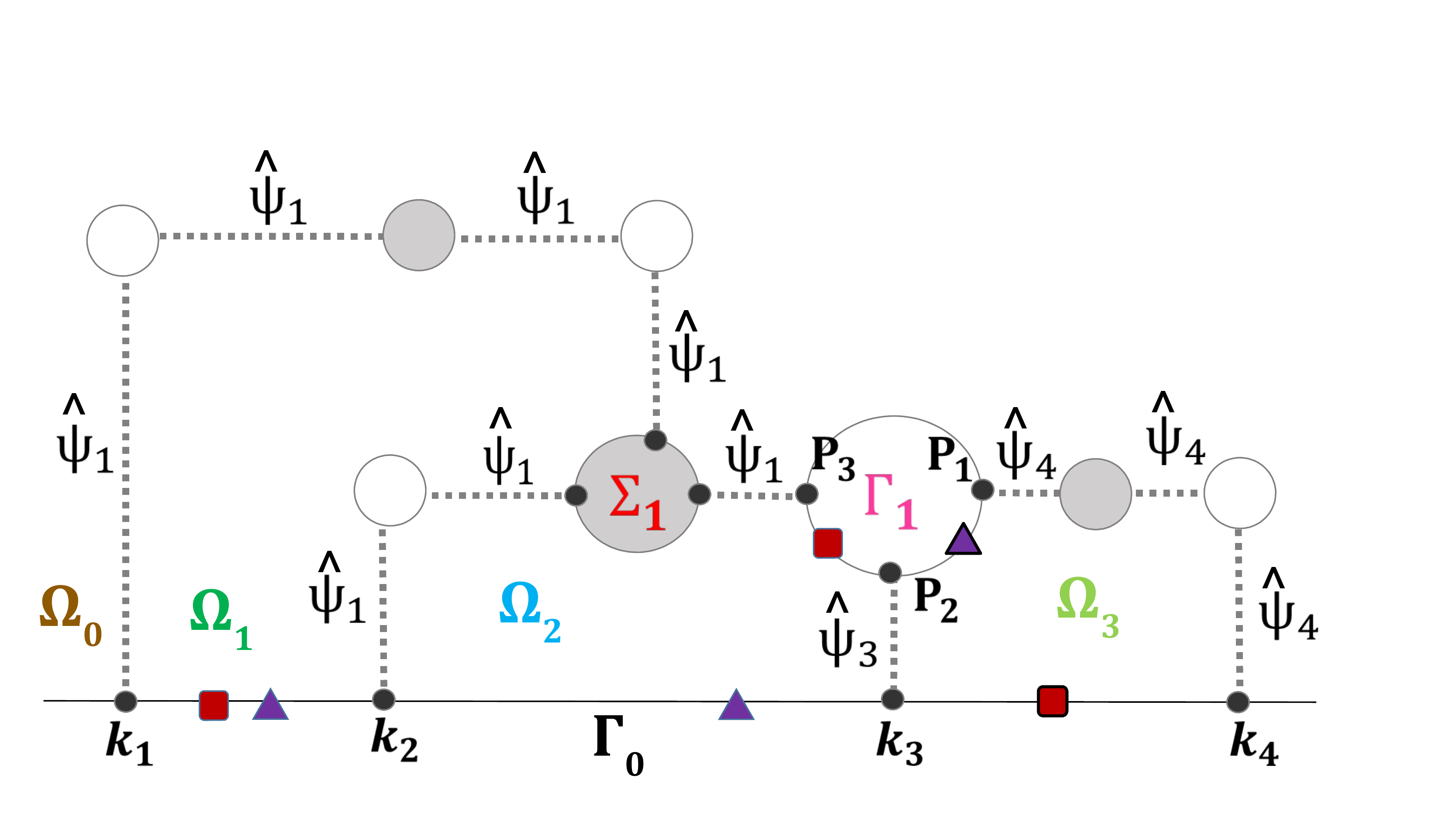}
\includegraphics[width=0.45\textwidth]{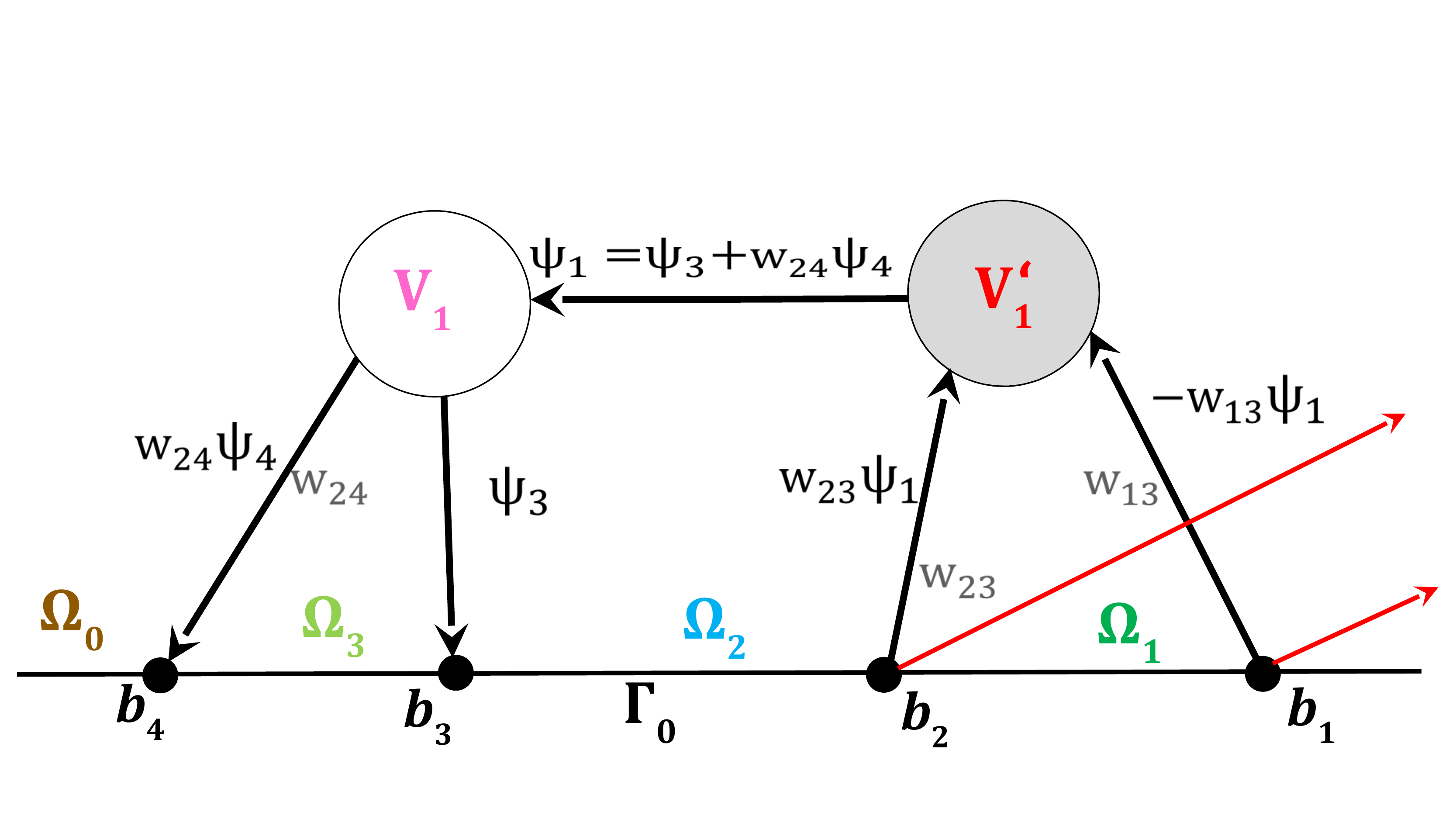}
\hfill
\includegraphics[width=0.46\textwidth]{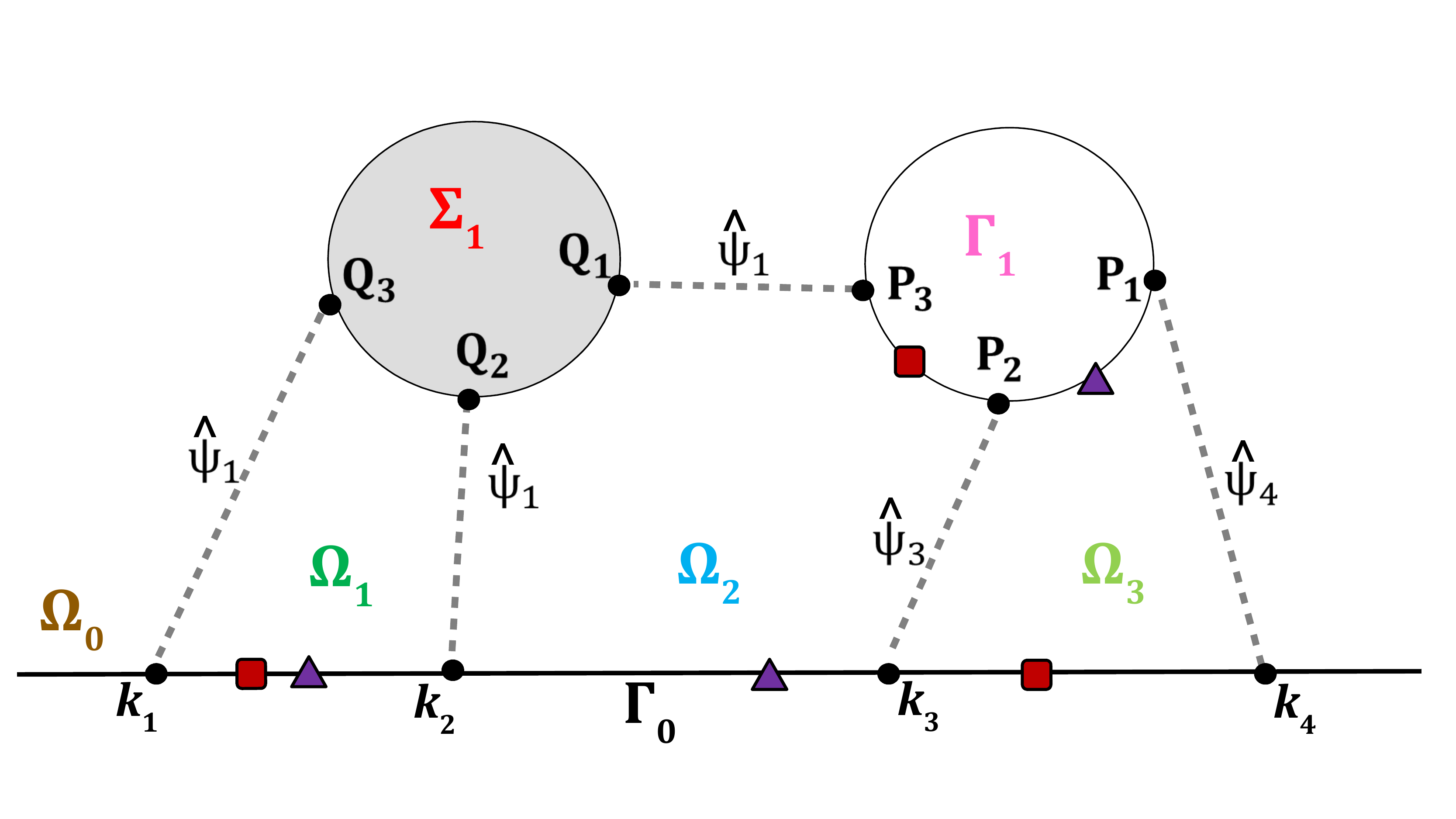}
}
\caption{\small{\sl Top: the Le--tableau, the Le-network ${\mathcal N}_T$ [left] and the topological model of the corresponding spectral curve $\Gamma_T$ [right] for soliton data in ${\mathcal S}_{34}^{\mbox{\tiny TNN}}\subset Gr^{\mbox{\tiny TNN}}(2,4)$. Bottom: the reduced Le-network ${\mathcal N}_{\mbox{\scriptsize red}}$ [left] and the topological model of the corresponding spectral curve $\Gamma_{T,\mbox{\scriptsize red}}$ [right] for the same cell. In all Figures  $\Psi_{1} (\vec t)=\Psi_{3} (\vec t)+w_{24}\Psi_{4} (\vec t)$. The configurations of the KP divisor (triangles/squares) are the same on both curves and depend only on the sign of $\Psi_3 (\vec t_0)$.
On the curves double points are represented as dotted segments and $\hat \psi_l \equiv \hat\psi_l (\vec t)$ is as in (\ref{eq:hat_psi}).}} 
\label{fig:Gr24_net}
\end{figure}

${\mathcal S}_{34}^{\mbox{\tiny TNN}}$ is the 3--dimensional positroid cell in $Gr^{\mbox{\tiny TNN}} (2,4)$ corresponding to the matroid
\[
{\mathcal M} = \{ \ 12 \ , \ 13 \ , \ 14 \ ,\ 23 \ , \ 24\ \},
\]
and its elements $[A]$ are equivalence classes of real $2\times 4$ matrices with all maximal minors positive, except $\Delta_{34}=0$. The three positive weights $w_{13},w_{23},w_{24}$ of the Le-tableau (see Figure \ref{fig:Gr24_net}[top,left]) parametrize ${\mathcal S}_{34}^{\mbox{\tiny TNN}}$ and correspond to the matrix in the reduced row echelon form (RREF),
\begin{equation}
\label{eq:RREF}
A = \left( \begin{array}{cccc}
1 & 0 & - w_{13} & -w_{13} w_{24}\\
0 & 1 & w_{23}   & w_{23}w_{24}
\end{array}
\right).
\end{equation}
The generators of the Darboux transformation ${\mathfrak D}=\partial_x^2 -{\mathfrak w}_1 (\vec t)\partial_x -{\mathfrak w}_2(\vec t)$ are
\begin{equation}\label{eq:gen_f2}
f^{(1)} (\vec t) = e^{\theta_1(\vec t)} -w_{13}e^{\theta_3(\vec t)}-w_{13} w_{24}e^{\theta_4(\vec t)}, \quad\quad
f^{(2)} (\vec t) = e^{\theta_2(\vec t)} +w_{23}e^{\theta_3(\vec t)}+w_{23}w_{24}e^{\theta_4(\vec t)}.
\end{equation}
In the following sections we construct a reducible rational curve $\Gamma_{T,\mbox{\scriptsize red}}$ and the divisor for soliton data $({\mathcal K}, [A])$ with $\mathcal K = \{ \kappa_1 < \kappa_2<\kappa_3<\kappa_4\}$ and $[A] \in {\mathcal S}_{34}^{\mbox{\tiny TNN}}$. We represent $\Gamma_{T,\mbox{\scriptsize red}}$ as a plane curve given by the intersection of a line and two quadrics (see (\ref{eq:lines}) and (\ref{eq:curveGr24})) and we verify that it is a rational degeneration of the genus 3 $\mathtt M$--curve $\Gamma_{\varepsilon}$ ($0<\varepsilon \ll 1$) in (\ref{eq:curveGr24_pert}). We then apply a parallel edge unreduction and a flip move to the reduced network and compute the transformed KP divisor on the transformed curves.

\subsection{Spectral curves for the reduced Le--network and their desingularizations}\label{sec:gamma_24}

We briefly illustrate the construction of a rational spectral curve $\Gamma_{T,\mbox{\scriptsize red}}$ for soliton data in ${\mathcal S}_{34}^{\mbox{\tiny TNN}}$. We start using the Le--graph as dual graph of the reducible rational curve and use Postnikov rules to assign the weights to construct the Le--network $\mathcal N_T$ for $[A]\in {\mathcal S}_{34}^{\mbox{\tiny TNN}}$  (Figure \ref{fig:Gr24_net}[top,left]). Then we use the middle edge insertion/removal move (M3) (see Section \ref{sec:middle}) to obtain an equivalent network $\mathcal N_{T,\mbox{\scriptsize red}}$ with only trivalent internal vertices (Figure \ref{fig:Gr24_net}[bottom,left]). In Figure \ref{fig:Gr24_net}[right], we show the topological model of the curves $\Gamma_{T}$ [top] and $\Gamma_{T,\mbox{\scriptsize red}}$ [bottom]. The elimination of bivalent vertices reduces the degree of the curve from 9 to 5. 

The reducible rational curve $\Gamma_{T,\mbox{\scriptsize red}}$ is obtained gluing three copies of $\mathbb{CP}^1$, $\Gamma_{T,\mbox{\scriptsize red}}=\Gamma_0\sqcup \Gamma_{1}\sqcup\Sigma_{1}$, and it may be represented as a plane curve given by the intersection of a line ($\Gamma_0$) and two quadrics ($\Gamma_1,\Sigma_1$). We plot both the topological model and the plane curve for this example in Figure \ref{fig:gr24_top_2}. 

\begin{figure}
  \centering
  {\includegraphics[width=0.42\textwidth]{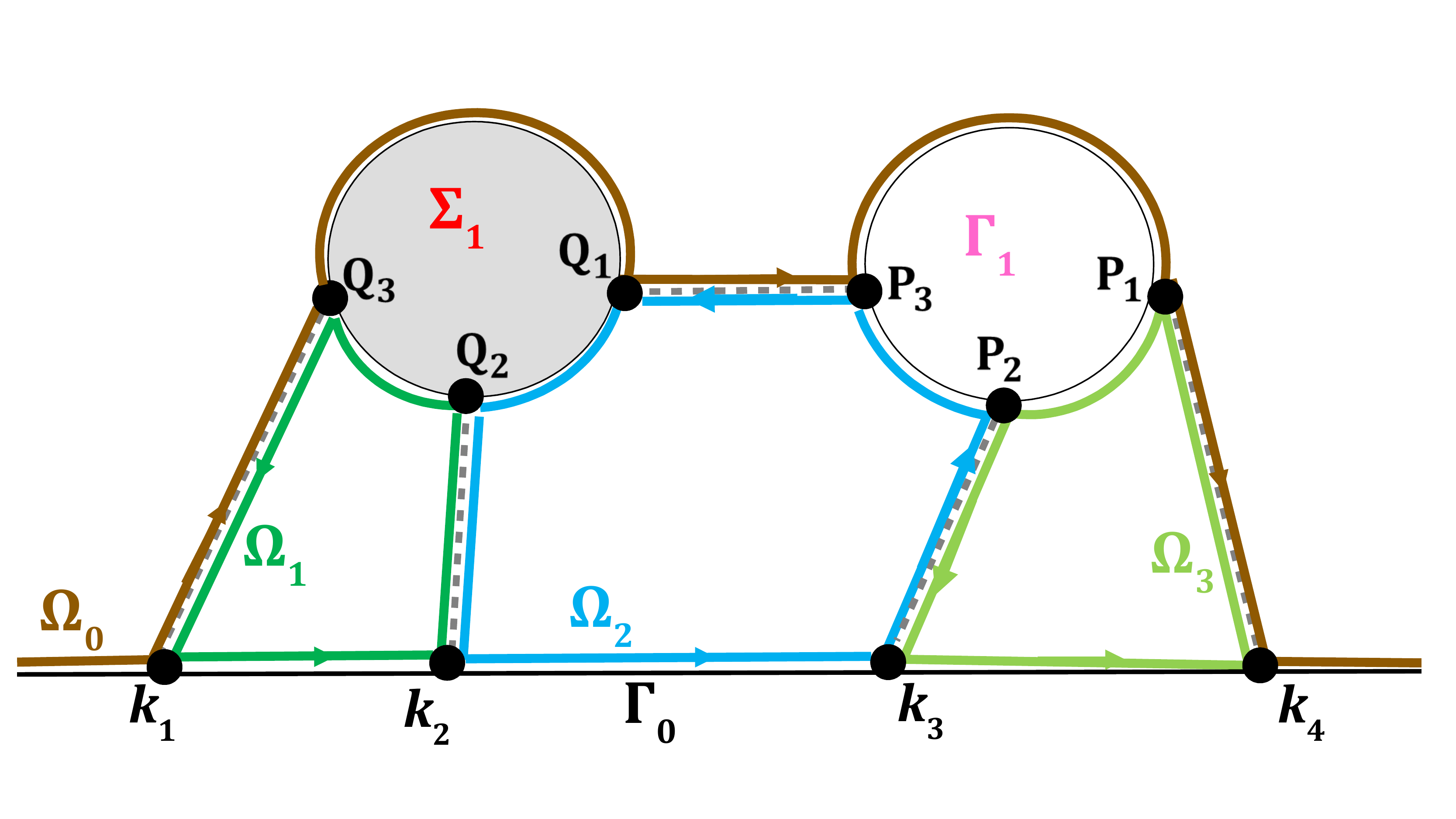}}
  \hspace{.6 truecm}
  {\includegraphics[,width=0.42\textwidth]{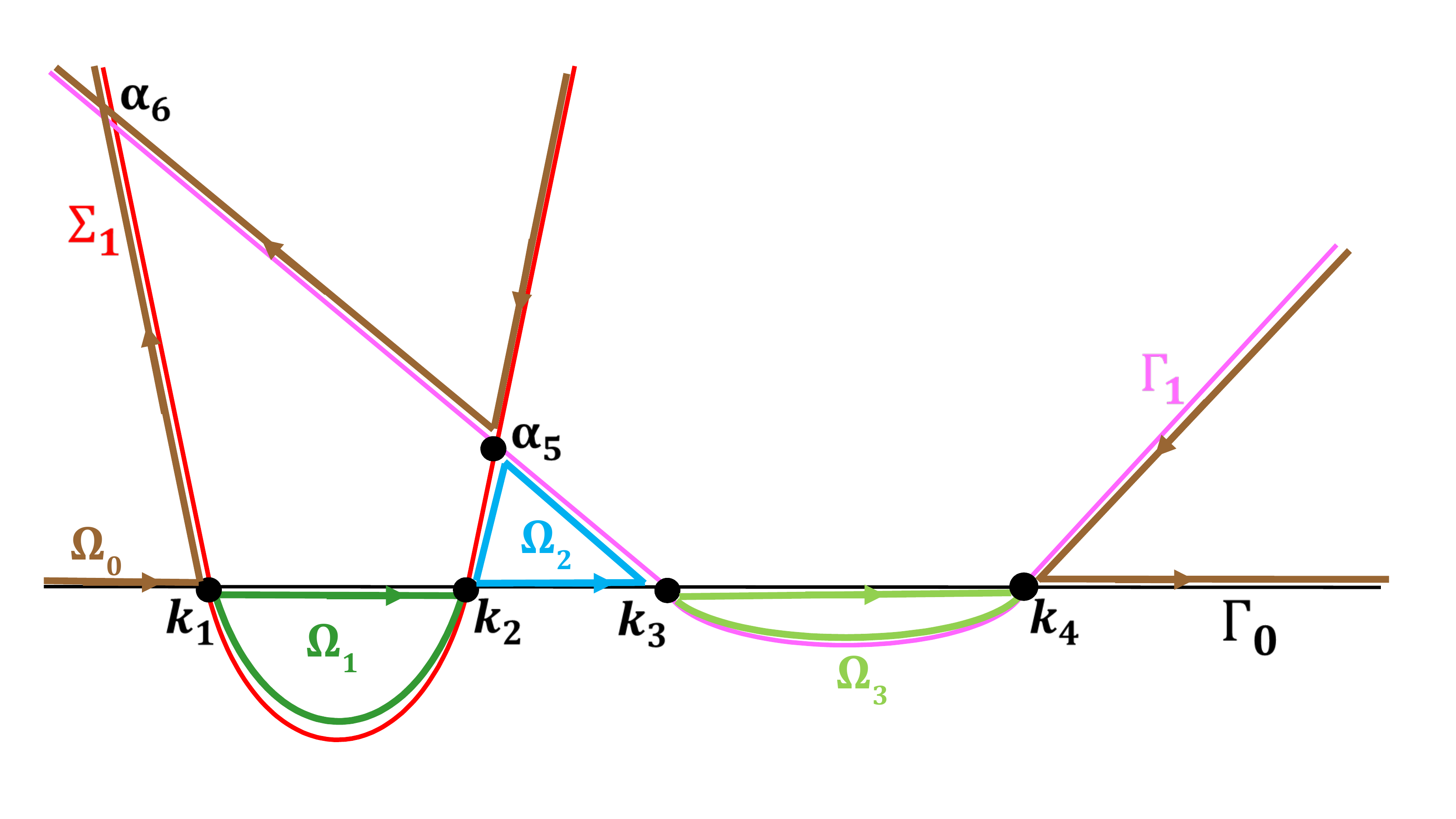}}
\caption{\footnotesize{\sl The topological model [left] of the spectral curve $\Gamma_{T,\mbox{\scriptsize red}}$ for soliton data in ${\mathcal S}_{34}^{\mbox{\tiny TNN}}$ is the partial normalization of the plane algebraic curve [right]. The ovals in the nodal plane curve are labeled as in the real part of its partial normalization. }}\label{fig:gr24_top_2}        
\end{figure}

To simplify its representation, we impose that $\Gamma_0$ is one of the coordinate axis in the $(\lambda,\mu)$--plane, say $\mu=0$, that $P_0\in \Gamma_0$ is the infinite point, that the quadrics $\Sigma_{1}$ and  $\Gamma_{1}$ are parabolas with two real finite intersection points $\alpha_5=(\lambda_5,\mu_5)$, $\alpha_6=(\lambda_6,\mu_6)$:
\begin{equation}\label{eq:lines}
\Gamma_0:  \mu=0, \ \ \Gamma_{1}: \mu- (\lambda - \kappa_3)(\lambda-\kappa_4) =0, \ \   \Sigma_{1}: \mu-c_1(\lambda -\kappa_1)(\lambda-\kappa_2)=0.
\end{equation}
In the following we also take $c_1>1$ and choose $\lambda(P_3)=\lambda(Q_1)=\lambda_5$. Then, by construction, $\lambda_5\in ]\kappa_2,\kappa_3[$ and $\lambda_6<\kappa_1$.
As usual we denote $\Omega_0$ the infinite oval, that is $P_0\in \Omega_0$, and $\Omega_j$, $j\in[3]$, the finite ovals (see Figure \ref{fig:gr24_top_2}). 
Since the singularity at infinity is completely resolved, the quadrics $\Sigma_{1}$ and $\Gamma_{1}$ do not intersect at infinity. The intersection point $\alpha_6$ does not correspond to any of the marked points of the topological model of $\Gamma$ (see Figure \ref{fig:gr24_top_2}). Such singularity is resolved in the partial normalization and therefore there are no extra conditions to be satisfied by the vacuum and the dressed wave functions at $\alpha_6$.

The relation between the coordinate $\lambda$ in the plane curve representation and the coordinate $\zeta$ introduced in Definition \ref{def:loccoor} may be easily worked out at each component of
$\Gamma_{T,\mbox{\scriptsize red}}$.
On $\Gamma_{1}$, we have 3 real ordered marked points $P_m$, $m\in [3]$, with $\zeta$--coordinates: $\zeta(P_1)=0<\zeta(P_2)=1<\zeta(P_3)=\infty$.
Comparing with (\ref{eq:lines}) we then easily conclude that
\[
\lambda = \frac{\lambda_5 (\kappa_4-\kappa_3) \zeta +(\kappa_3-\lambda_5)\kappa_4}{(\kappa_4-\kappa_3) \zeta +\kappa_3-\lambda_5}.
\]
Similarly, on $\Sigma_{1}$, we have 3 real ordered marked points $Q_m$, $m\in [3]$, with $\zeta$--coordinates: $\zeta(Q_1)=0<\zeta(Q_2)=1<\zeta(Q_3)=\infty$.
Comparing with (\ref{eq:lines}) we then easily conclude that
\[
\lambda = \frac{\kappa_1 (\lambda_5-\kappa_2) \zeta +(\kappa_2-\kappa_1)\lambda_5}{(\lambda_5-\kappa_2) \zeta +\kappa_2-\kappa_1}.
\]
Finally
$\Gamma_{T,\mbox{\scriptsize red}}$ is represented by the reducible plane curve $\Pi_0(\lambda,\mu)=0$, with
\begin{equation}
\label{eq:curveGr24}
\Pi_0(\lambda,\mu)=\mu\cdot\big(\mu- (\lambda - \kappa_3)(\lambda-\kappa_4)\big)\cdot\big(\mu-c_1 (\lambda - \kappa_1)(\lambda-\kappa_2)\big).
\end{equation}

Finally $\Gamma_{T,\mbox{\scriptsize red}}$ is a rational degeneration of a genus 3 $\mathtt M$--curve $\Gamma_{\varepsilon}$ ($0<\varepsilon \ll 1$):
\begin{equation}
\label{eq:curveGr24_pert}
\Gamma_{\varepsilon} \; : \quad\quad \Pi(\lambda, \mu;\varepsilon)=  \Pi_0(\lambda,\mu)-\varepsilon^2\left(\lambda -\lambda_6\right)^2=0.
\end{equation}

\begin{remark}
Of course, the plane curve representation for a given cell is not unique. For example, a rational spectral curve $\Gamma_{T,\mbox{\scriptsize red}}$ for soliton data in ${\mathcal S}_{34}^{\mbox{\tiny TNN}}$ can be also represented as the union of one quadric and two lines. 
\end{remark}

\subsection{The KP divisor on $\Gamma_{T,\mbox{\scriptsize red}}$}\label{sec:div_24}

Let us briefly illustrate the construction of the wave function and of the KP divisor on $\Gamma_{T,\mbox{\scriptsize red}}$. 
Since ${\mathcal N}_{T,{\mbox{\scriptsize red}}}$ is obtained from ${\mathcal N}_T$ via Move (M3) (Section \ref{sec:middle}),
the normalized KP edge wave function $\hat \psi (P, \vec t)$ is the same at the corresponding double points on $\Gamma_T$ and on $\Gamma_{T,\mbox{\scriptsize red}}$ (see Figure \ref{fig:Gr24_net} [right]). 

For this example, the KP wave function  may take only three possible values at the marked points:
\begin{equation}
\label{eq:hat_psi}
\hat \psi_l (\vec t) = \frac{{\mathfrak D} e^{\theta_l(\vec t)}}{{\mathfrak D} e^{\theta_l(\vec t_0)}}, \ \ l=1,3,4,
\end{equation}
where $\theta_l(\vec t) = \kappa_l x+  \kappa_l^2 y+ \kappa_l^3 t$. In Figure~\ref{fig:Gr24_net} [right] we show which double point carries which of the above values of $\hat \psi$. At each marked point the value $\hat \psi_l (\vec t)$ is independent on the choice of local coordinates on the components, {\sl i.e.} of the orientation in the network. 
The position of the KP divisor points of $\DKP$ is the same both on $\Gamma_T$ and $\Gamma_{T,\mbox{\scriptsize red}}$ and the local coordinate of each divisor point may be computed using Theorem \ref{theo:pos_div} and formula (\ref{eq:formula_div}). 

On $\Gamma=\Gamma_{T,\mbox{\scriptsize red}}$, the KP divisor $\DKP$ consists of the degree $k=2$ Sato divisor $(\gamma_{S,1} ,\gamma_{S,2} )=(\gamma_{S,1} (\vec t_0),\gamma_{S,2} (\vec t_0))$ defined in (\ref{eq:Satodiv}) and of $1$ simple pole $\gamma_{1}=\gamma_{1} (\vec t_0)$ belonging to the intersection of 
$\Gamma_{1}$ with the union of the finite ovals. In the local coordinates induced by the orientation of the network (see Definition \ref{def:loccoor}), we have
\begin{equation}\label{eq:ex_div1}
\zeta(\gamma_{S,1}) +\zeta(\gamma_{S,2}) = {\mathfrak w}_1 (\vec t_0), \quad  \zeta(\gamma_{S,1})\zeta(\gamma_{S,2}) = -{\mathfrak w}_2 (\vec t_0), \quad
\zeta(\gamma_{1}) = \frac{w_{24} {\mathfrak D} e^{\theta_4(\vec t_0)}}{{\mathfrak D} e^{\theta_3(\vec t_0)}+ w_{24} 
{\mathfrak D} e^{\theta_4(\vec t_0)}}.
\end{equation}
We remark that applying the Darboux dressing to (\ref{eq:gen_f2}), it is easy to verify that ${\mathfrak D} e^{\theta_1(\vec t)}$, ${\mathfrak D} e^{\theta_4(\vec t)}>0$ and ${\mathfrak D} e^{\theta_2(\vec t)} = - \frac{w_{23}}{w_{13}}{\mathfrak D} e^{\theta_1(\vec t)}$. Therefore, for generic soliton data $[A]\in {\mathcal S}_{34}^{\mbox{\tiny TNN}}$, the KP--II pole divisor configuration is one of the two shown in Figure \ref{fig:Gr24_net} [right]:
\begin{enumerate}
\item If ${\mathfrak D} e^{\theta_3(\vec t_0)}>0$, then $\gamma_{S,1} \in \Omega_1$, $\gamma_{S,2} \in \Omega_2$ and $\gamma_{1} \in \Omega_3$. One such configuration is illustrated by triangles in the Figure;
\item If ${\mathfrak D} e^{\theta_3(\vec t_0)}<0$, then $\gamma_{S,1} \in \Omega_1$, $\gamma_{S,2} \in \Omega_3$ and $\gamma_{1} \in \Omega_2$. One such configuration is illustrated by squares in the Figure.
\end{enumerate}
As expected, there is exactly one KP divisor point in each finite oval of $\Gamma$, where we use the counting rule established in \cite{AG1} for non--generic soliton data satisfying ${\mathfrak D} e^{\theta_3(\vec t_0)}=0$.

\subsection{The effect of Postnikov moves and reductions on the KP divisor}

\begin{figure}
\centering{
\includegraphics[width=0.48\textwidth]{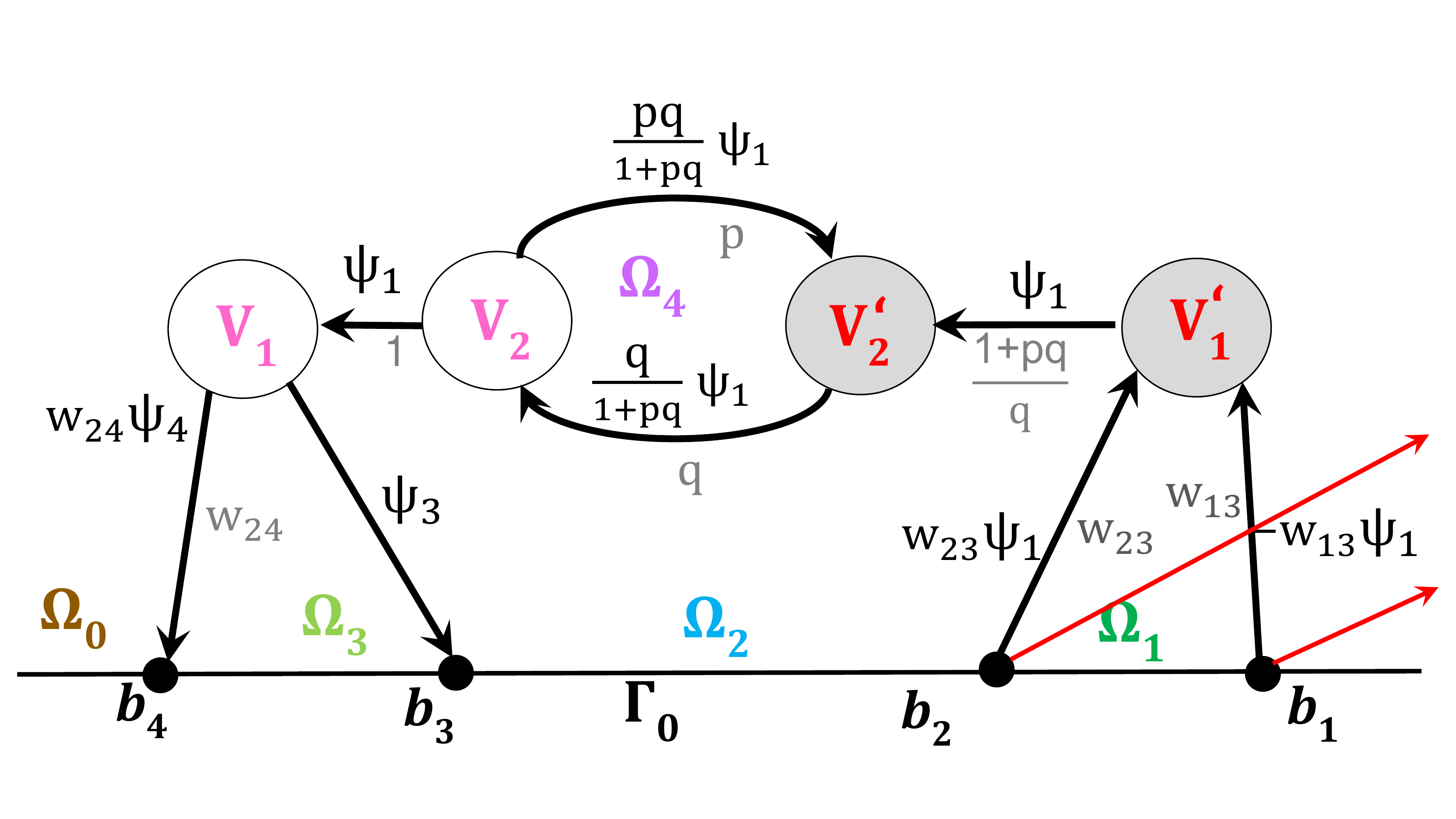}
\hfill
\includegraphics[width=0.49\textwidth]{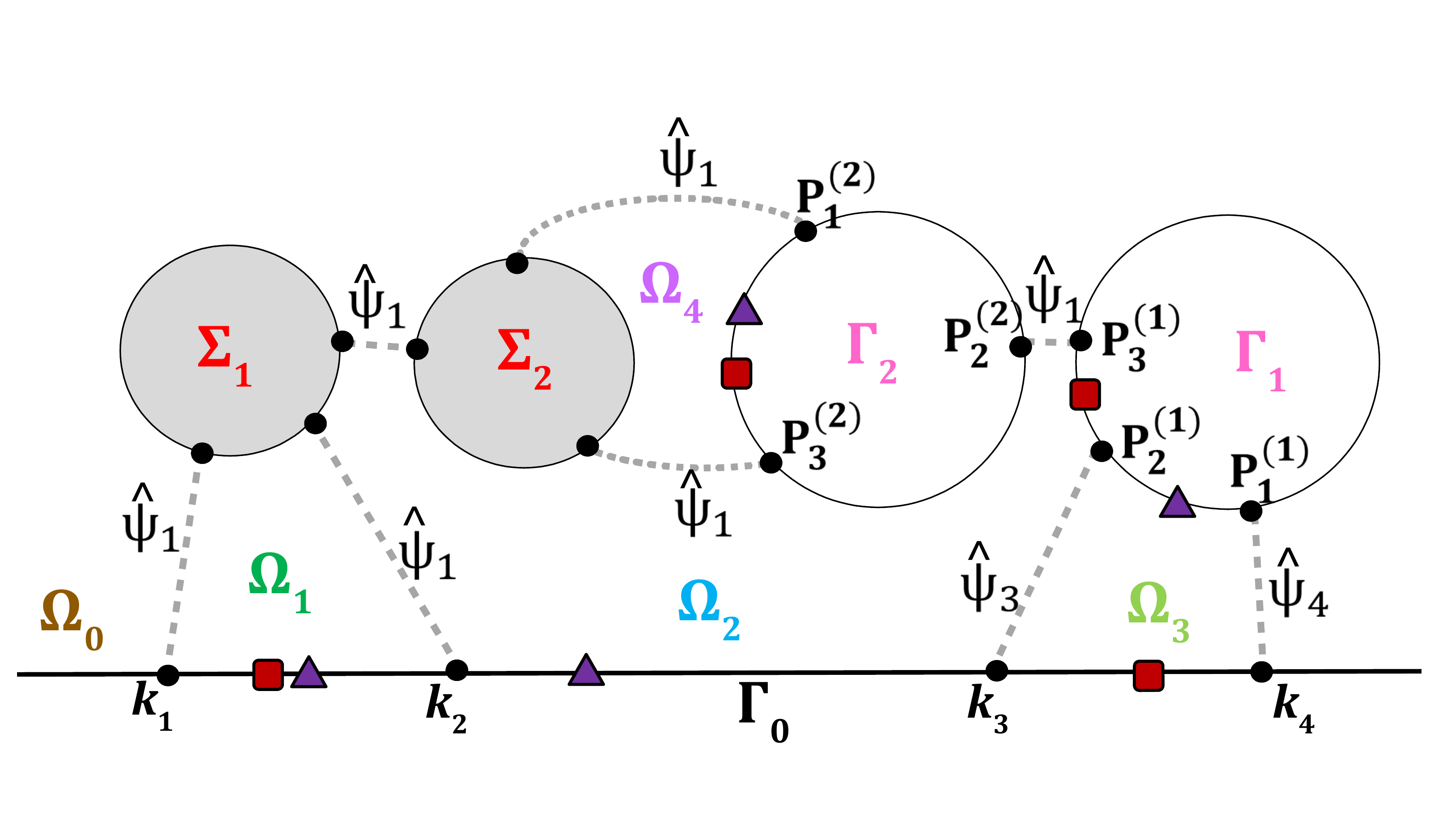}
\includegraphics[width=0.48\textwidth]{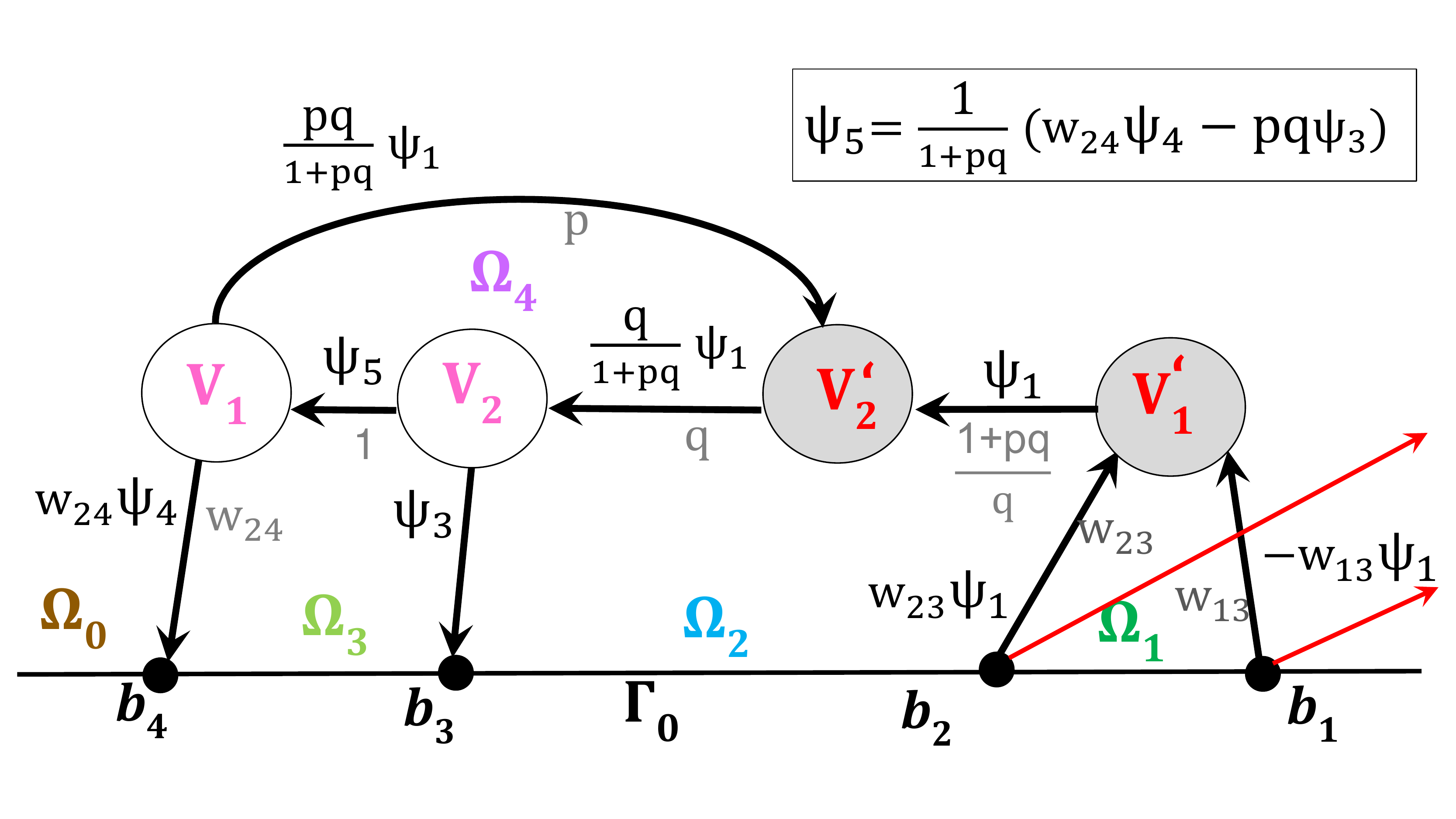}
\hfill
\includegraphics[width=0.49\textwidth]{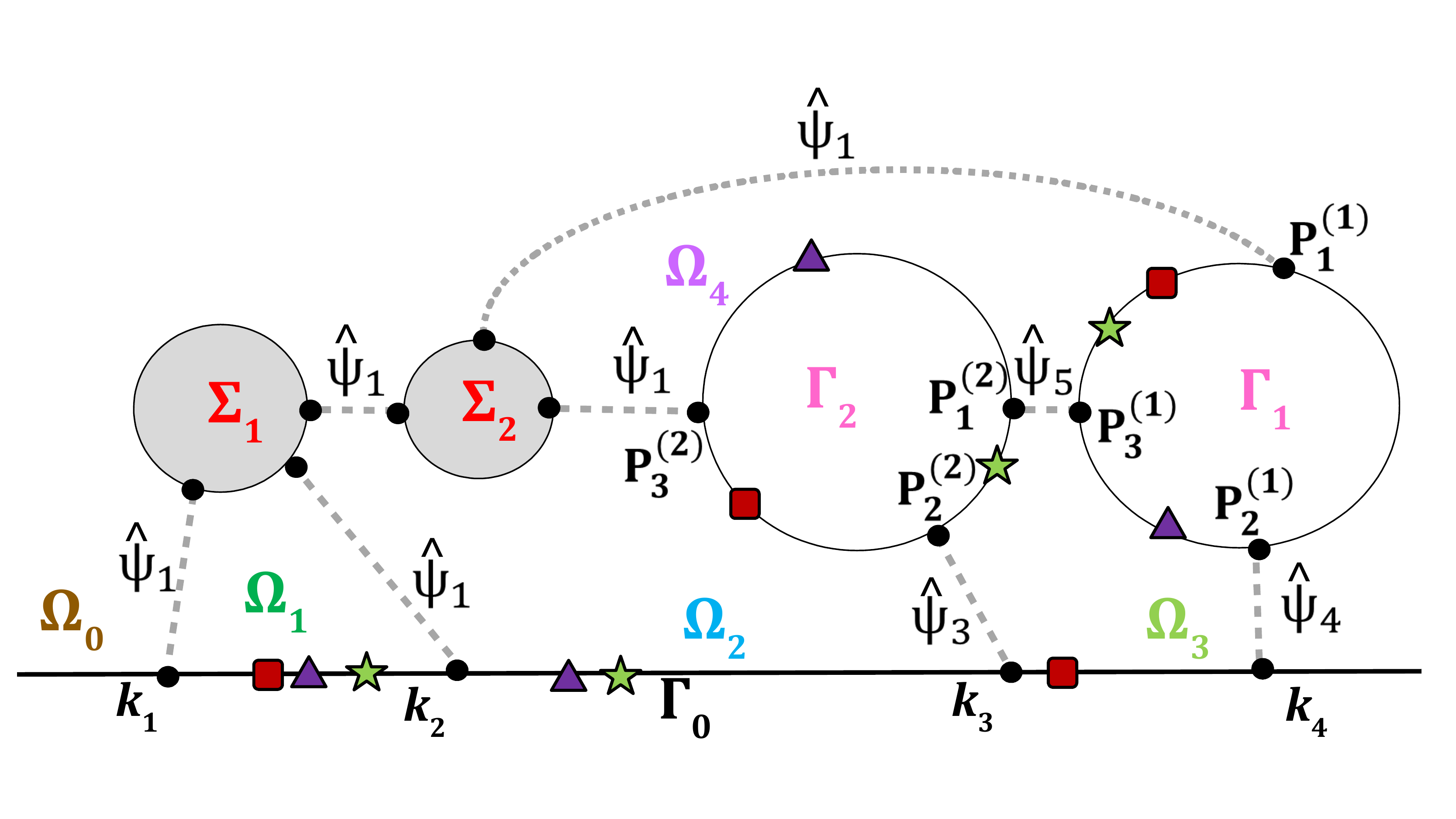}
}
\caption{\small{\sl Left: ${\mathcal N}_{\mbox{\scriptsize par}}$ [top] is obtained from the reduced Le-network ${\mathcal N}_{T,{\mbox{\scriptsize red}}}$ of Figure \ref{fig:Gr24_net} applying the parallel edge unreduction, whereas ${\mathcal N}_{\mbox{\scriptsize flip}}$ [bottom] is obtained from ${\mathcal N}_{\mbox{\scriptsize par}}$ applying a flip move.
Right: The partial normalization of the corresponding curves ${\Gamma}_{\mbox{\scriptsize par}}$ [top], ${\Gamma}_{\mbox{\scriptsize flip}}$ [bottom] and the possible divisor configurations.
}}
\label{fig:Gr24_moves}
\end{figure}

Next we show the effect of moves and reductions on the divisor position. 

We first apply a parallel edge unreduction to ${\mathcal N}_{T,{\mbox{\scriptsize red}}}$ (Figure \ref{fig:Gr24_net}) and obtain the network  ${\mathcal N}_{{\mbox{\scriptsize par}}}$ in Figure \ref{fig:Gr24_moves}[top,left]. We remark that we have a gauge freedom in assigning the weights to the edges involved in this transformation provided that $p,q>0$. The values on the unnormalized dressed wave function are shown in the Figure. The corresponding curve,  ${\Gamma}_{\mbox{\scriptsize par}}$, is presented in Figure \ref{fig:Gr24_moves}[top,right]. By construction the KP divisor $\DKP$ consists of the degree $k=2$ Sato divisor $(\gamma_{S,1} ,\gamma_{S,2} )=(\gamma_{S,1} (\vec t_0),\gamma_{S,2} (\vec t_0))$ computed in (\ref{eq:ex_div1}) and of the simple poles $\gamma_{i}=\gamma_{i} (\vec t_0)$ belonging to the intersection of 
$\Gamma_{i}$, $i=1,2$, with the union of the finite ovals. In the local coordinates induced by the orientation of the network (see Definition \ref{def:loccoor}), we have
\begin{equation}\label{eq:ex_div2}
\zeta(\gamma_{1}) = \frac{w_{24} {\mathfrak D} e^{\theta_4(\vec t_0)}}{{\mathfrak D} e^{\theta_3(\vec t_0)}+ w_{24} 
{\mathfrak D} e^{\theta_4(\vec t_0)}}, \quad\quad   \zeta(\gamma_{2}) = -pq.
\end{equation}
For generic soliton data $[A]\in {\mathcal S}_{34}^{\mbox{\tiny TNN}}$, the KP--II pole divisor configurations are shown in Figure \ref{fig:Gr24_moves} [top,right]:
\begin{enumerate}
\item $\gamma_{2} \in \Omega_4$ independently of the sign of ${\mathfrak D} e^{\theta_3(\vec t_0)}>0$;
\item If ${\mathfrak D} e^{\theta_3(\vec t_0)}>0$, then $\gamma_{S,1} \in \Omega_1$, $\gamma_{S,2} \in \Omega_2$ and $\gamma_{1} \in \Omega_3$. One such configuration is illustrated by triangles in the Figure;
\item If ${\mathfrak D} e^{\theta_3(\vec t_0)}<0$, then $\gamma_{S,1} \in \Omega_1$, $\gamma_{S,2} \in \Omega_3$ and $\gamma_{1} \in \Omega_2$. One such configuration is illustrated by squares in the Figure.
\end{enumerate}
Again, for any given $[A]\in {\mathcal S}_{34}^{\mbox{\tiny TNN}}$, there is exactly one KP divisor point in each finite oval, where again we use the counting rule established in \cite{AG1} for non--generic soliton data satisfying ${\mathfrak D} e^{\theta_3(\vec t_0)}=0$.

Next we apply a flip move to ${\mathcal N}_{\mbox{\scriptsize par}}$ and obtain the network  ${\mathcal N}_{{\mbox{\scriptsize flip}}}$ in Figure \ref{fig:Gr24_moves}[bottom,left]. Again the values on the unnormalized dressed wave function are shown in the Figure. The corresponding curve,  ${\Gamma}_{\mbox{\scriptsize flip}}$, is presented in Figure \ref{fig:Gr24_moves}[bottom,right]. By construction the KP divisor $\DKP$ consists of the degree $k=2$ Sato divisor $(\gamma_{S,1} ,\gamma_{S,2} )=(\gamma_{S,1} (\vec t_0),\gamma_{S,2} (\vec t_0))$ computed in (\ref{eq:ex_div1}) and of the simple poles $\tilde \gamma_{i}=\tilde \gamma_{i} (\vec t_0)$ belonging to the intersection of 
$\Gamma_{i}$, $i=1,2$, with the union of the finite ovals. In the local coordinates induced by the orientation of the network (see Definition \ref{def:loccoor}), we have
\begin{equation}\label{eq:ex_div3}
\zeta(\gamma_{1}) = \frac{pq \left({\mathfrak D} e^{\theta_3(\vec t_0)} + w_{24} {\mathfrak D} e^{\theta_4(\vec t_0)} \right)}{{pq\mathfrak D} e^{\theta_3(\vec t_0)}- w_{24} 
{\mathfrak D} e^{\theta_4(\vec t_0)}}, \quad\quad   \zeta(\gamma_{2}) = -\frac{pq{\mathfrak D} e^{\theta_3(\vec t_0)}- w_{24} 
{\mathfrak D} e^{\theta_4(\vec t_0)}}{{\mathfrak D} e^{\theta_3(\vec t_0)} + w_{24} {\mathfrak D} e^{\theta_4(\vec t_0)}}.
\end{equation}
For generic soliton data $[A]\in {\mathcal S}_{34}^{\mbox{\tiny TNN}}$, the KP--II pole divisor configuration is one of the three shown in Figure \ref{fig:Gr24_net} [right]:
\begin{enumerate}
\item If ${\mathfrak D} e^{\theta_3(\vec t_0)}>0$  and $pq{\mathfrak D} e^{\theta_3(\vec t_0)}> w_{24} 
{\mathfrak D} e^{\theta_4(\vec t_0)}$, then $\gamma_{S,1} \in \Omega_1$, $\gamma_{S,2} \in \Omega_2$,
 and $\tilde \gamma_{1} \in \Omega_3$ and $\tilde \gamma_{2} \in \Omega_4$. One such configuration is illustrated by triangles in the Figure;
\item If ${\mathfrak D} e^{\theta_3(\vec t_0)}>0$  and $pq{\mathfrak D} e^{\theta_3(\vec t_0)}< w_{24} 
{\mathfrak D} e^{\theta_4(\vec t_0)}$, then $\gamma_{S,1} \in \Omega_1$, $\gamma_{S,2} \in \Omega_2$,
 and $\tilde \gamma_{1} \in \Omega_4$ and $\tilde \gamma_{2} \in \Omega_3$. One such configuration is illustrated by stars in the Figure;
\item If ${\mathfrak D} e^{\theta_3(\vec t_0)}<0$, then $\gamma_{S,1} \in \Omega_1$, $\gamma_{S,2} \in \Omega_3$, $\tilde \gamma_{1} \in \Omega_4$ and $\tilde \gamma_{2} \in \Omega_2$. One such configuration is illustrated by squares in the Figure.
\end{enumerate}
Again there is exactly one KP divisor point in each finite oval, where we use the counting rule established in \cite{AG1} for non--generic soliton data satisfying ${\mathfrak D} e^{\theta_3(\vec t_0)}=0$.

\section{Effect of the square move on the KP divisor for soliton data in $Gr^{\mbox{\tiny TP}}(2,4)$}\label{sec:ex_Gr24top}
\begin{figure}
\centering{
\includegraphics[width=0.48\textwidth]{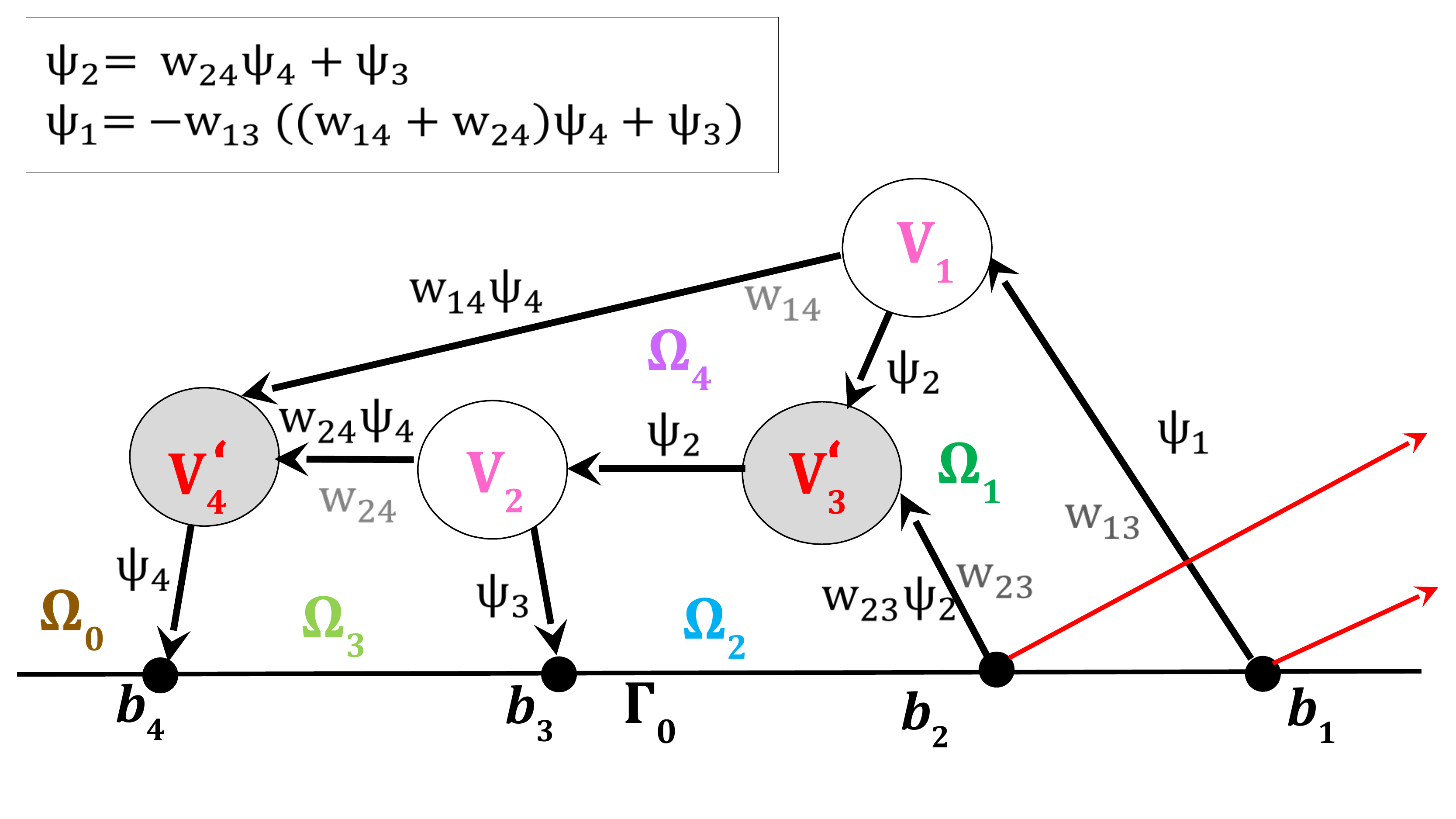}
\hfill
\includegraphics[width=0.49\textwidth]{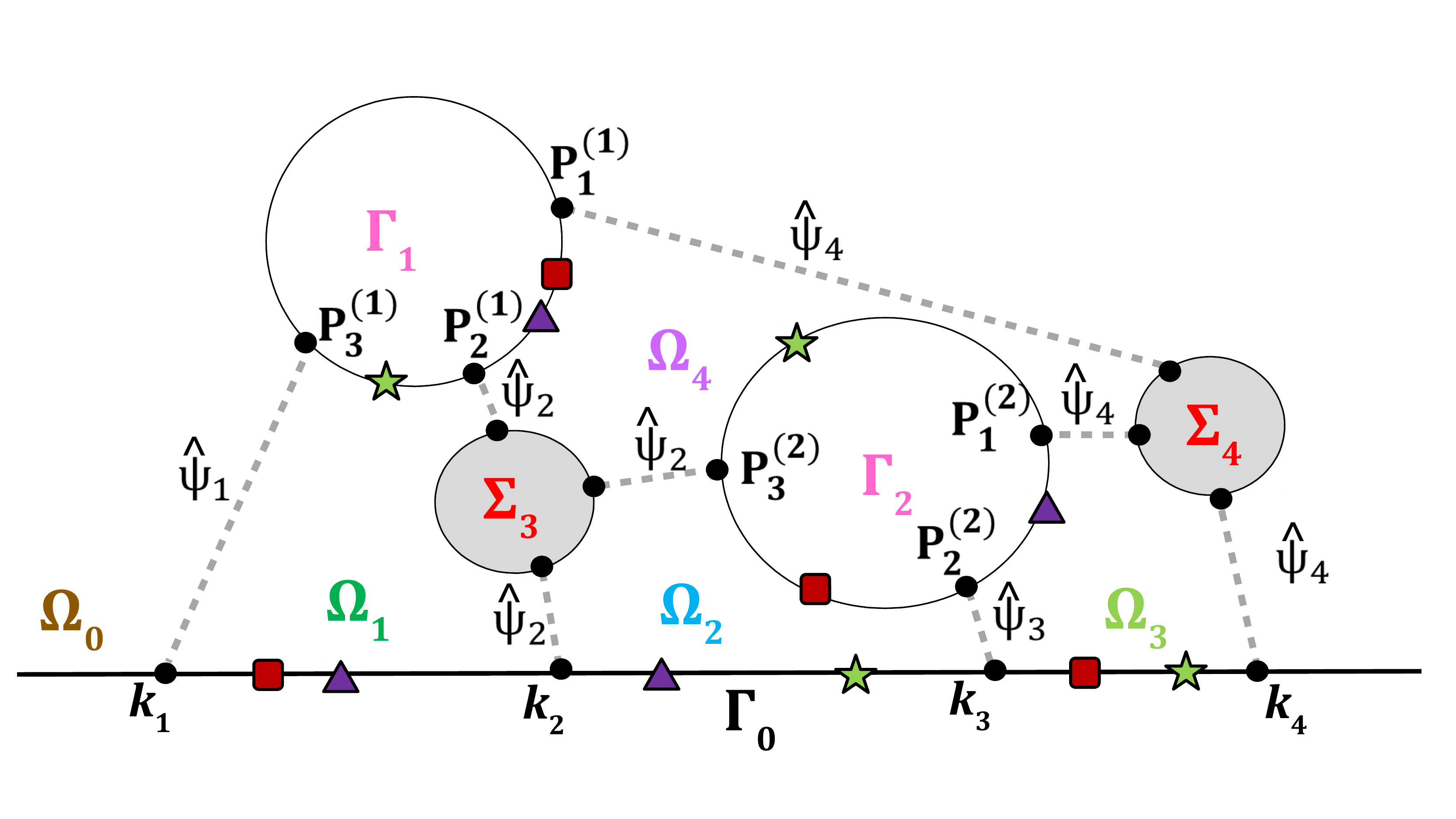}
\includegraphics[width=0.48\textwidth]{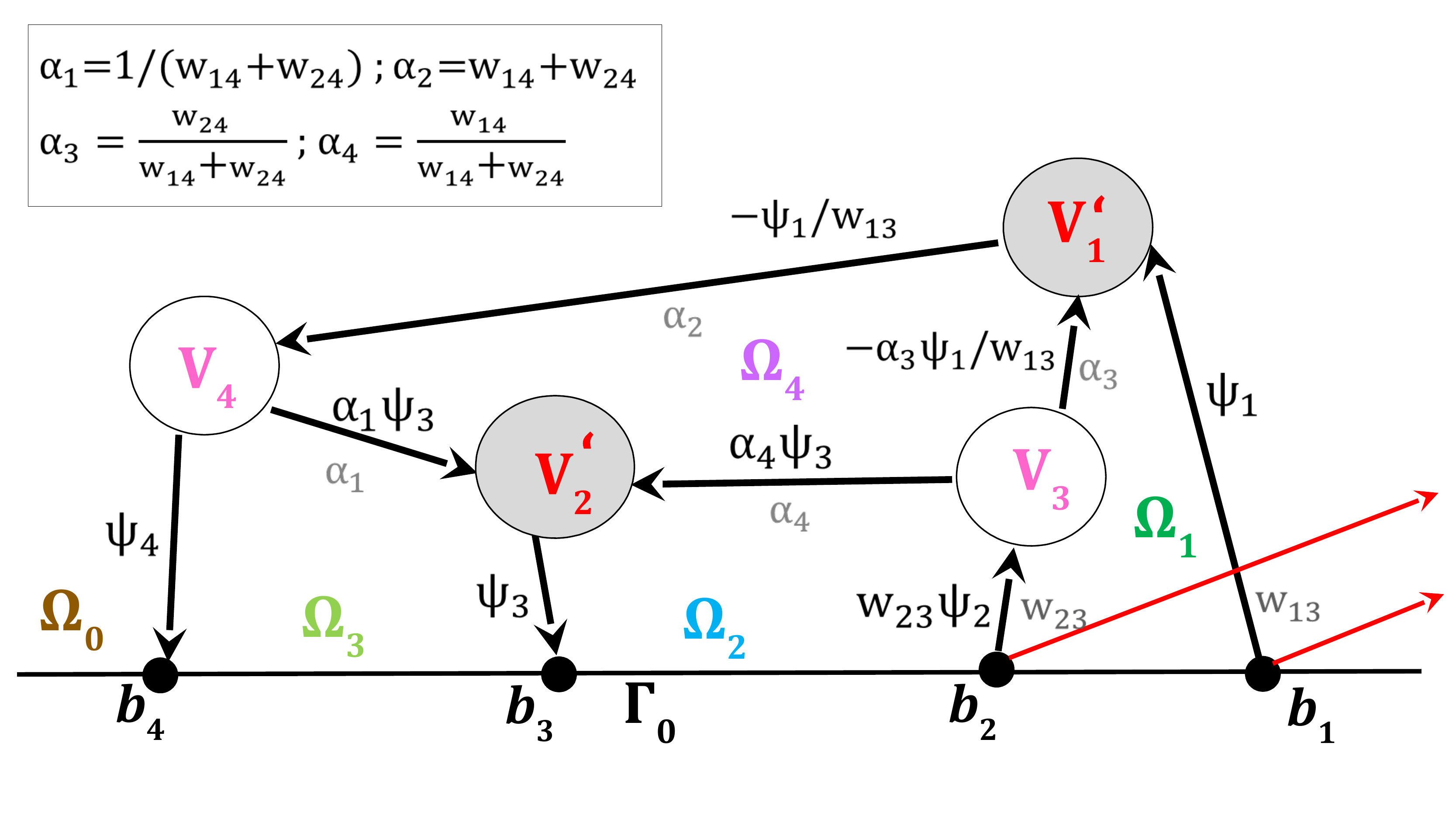}
\hfill
\includegraphics[width=0.49\textwidth]{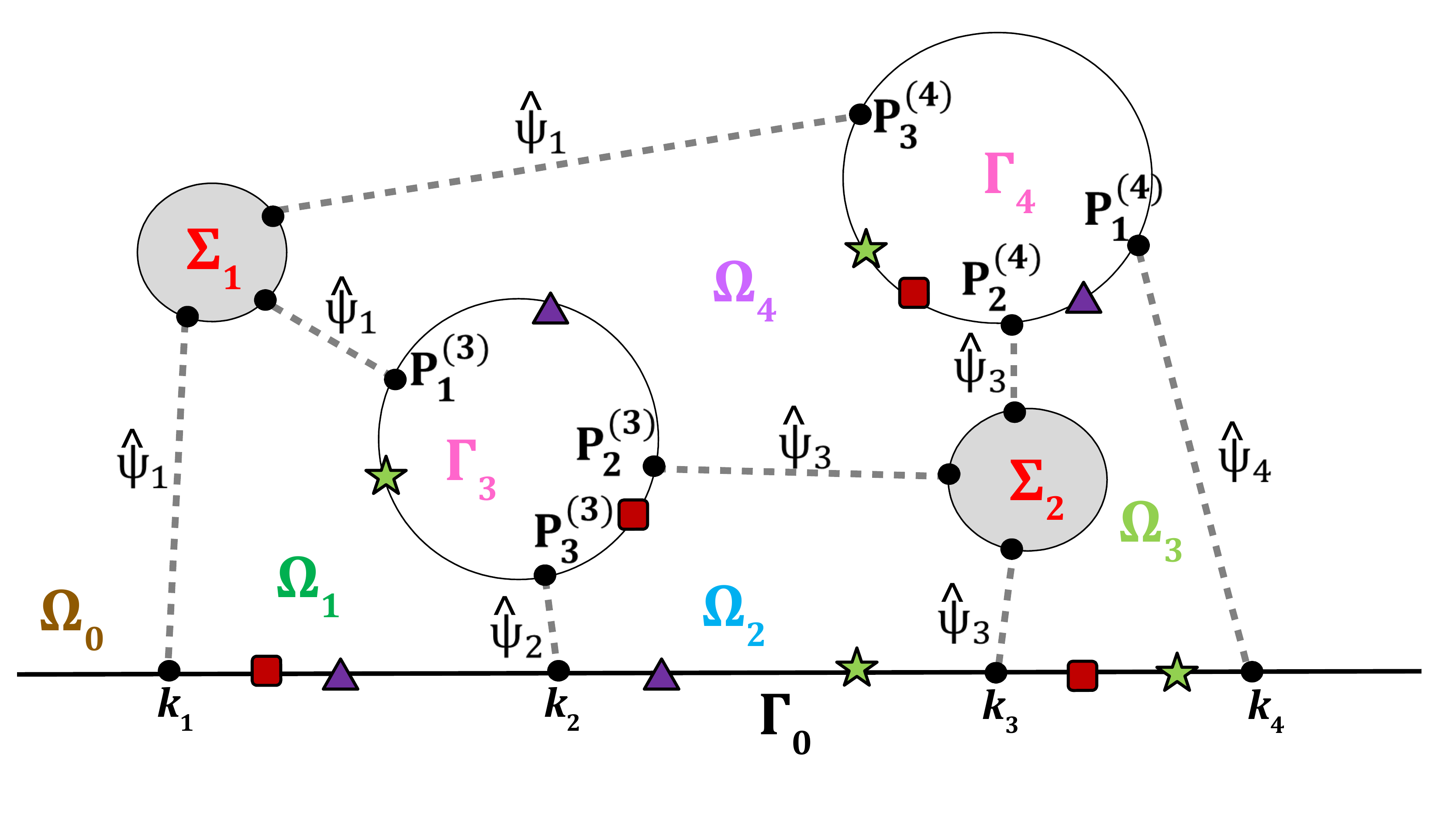}
}
\caption{\small{\sl Left: ${\mathcal N}_{\mbox{\scriptsize sq-mv}}$ [bottom] is obtained from the reduced Le-network ${\mathcal N}_{T,{\mbox{\scriptsize top}}}$ [top] applying a square move. Right: The partial normalization of the corresponding curves ${\Gamma}_{\mbox{\scriptsize top}}$ [top], ${\Gamma}_{\mbox{\scriptsize sq-mv}}$ [bottom] and the possible divisor configurations.}}
\label{fig:Gr24_square}
\end{figure}

The simplest network to which the square move is applicable is the reduced Le--network ${\mathcal N}_{T,{\mbox{\scriptsize top}}}$ associated to soliton data in $Gr^{\mbox{\tiny TP}} (2,4)$ and corresponds to the change of colour of all internal vertices. The latter transformation may be also interpreted as a self--dual transformation in $Gr^{\mbox{\tiny TP}} (2,4)$.

We present the reduced networks and the topological models of the curves before and after the square move for soliton data     in $Gr^{\mbox{\tiny TP}} (2,4)$ in Figure \ref{fig:Gr24_square}.
In \cite{A2}-\cite{AG3} we have already computed a plane curve representation and its desingularization, and discussed the divisor configurations on ${\mathcal N}_{T,{\mbox{\scriptsize top}}}$. We remark that duality transformation implies that we may use the same plane curve representation associated to ${\Gamma}_{\mbox{\scriptsize top}}$ in \cite{AG2} also for ${\Gamma}_{\mbox{\scriptsize sq-mv}}$, by conveniently relabeling $\mathbb{CP}^1$ components from $\Sigma_i,\Gamma_j$ to $\Gamma_i,\Sigma_j$ (compare the topological models of curves in Figure \ref{fig:Gr24_square}).

Here we just compute the KP divisor after the square move and refer to \cite{AG2} for more details on this example.
The values on the dressed edge wave functions are shown in Figure \ref{fig:Gr24_square}. 

By definition, the Sato divisor is not affected by the square move since the Darboux transformation is the same.
Therefore the degree $k=2$ Sato divisor $(\gamma_{S,1} ,\gamma_{S,2} )=(\gamma_{S,1} (\vec t_0),\gamma_{S,2} (\vec t_0))$ is obtained solving
\begin{equation}\label{eq:ex_div_Sato_top}
\zeta(\gamma_{S,1}) +\zeta(\gamma_{S,2}) = {\mathfrak w}_1 (\vec t_0), \quad  \zeta(\gamma_{S,1})\zeta(\gamma_{S,2}) = -{\mathfrak w}_2 (\vec t_0), 
\end{equation}
where the Darboux $\mathfrak D = \partial_x^2 -{\mathfrak w}_1 (\vec t) \partial_x -{\mathfrak w}_2 (\vec t)$ transformation is generated by the heat hierarchy solutions
\[
f^{(1)} (\vec t) = e^{\theta_1(\vec t)}-w_{13} e^{\theta_3(\vec t)}-w_{13}(w_{14}+w_{24}) e^{\theta_4(\vec t)},\quad\quad
f^{(2)} (\vec t) = e^{\theta_2(\vec t)}+w_{23} e^{\theta_3(\vec t)}+w_{23}w_{24} e^{\theta_4(\vec t)}.
\]
On $\Gamma_{\mbox{\scriptsize top}}$, $\DKP = (\gamma_{S,1} ,\gamma_{S,2},\gamma_1,\gamma_2 )$
where the simple poles $\gamma_{i}=\gamma_{i} (\vec t_0)$ belong to the intersection of 
$\Gamma_{i}$, $i=1,2$, with the union of the finite ovals. In the local coordinates induced by the orientation of ${\mathcal N}_{T,{\mbox{\scriptsize top}}}$, we have
\begin{equation}\label{eq:ex_div_red}
\zeta(\gamma_{1}) = \frac{w_{14} {\mathfrak D} e^{\theta_4(\vec t_0)}}{{\mathfrak D} e^{\theta_3(\vec t_0)}+(w_{14}+ w_{24})
{\mathfrak D} e^{\theta_4(\vec t_0)}}, \quad\quad   \zeta(\gamma_{2}) = \frac{w_{24}{\mathfrak D} e^{\theta_4(\vec t_0)}}{{\mathfrak D} e^{\theta_3(\vec t_0)}+ w_{24}
{\mathfrak D} e^{\theta_4(\vec t_0)}}.
\end{equation}
It is straightforward to verify that ${\mathfrak D} e^{\theta_1(\vec t)}, {\mathfrak D} e^{\theta_4(\vec t)}>0$ for all $\vec t$.
Therefore, as observed in \cite{A2,AG2}, there are three possible configurations of the KP--II pole divisor depending on the signs of 
${\mathfrak D} e^{\theta_2(\vec t_0)}$ and ${\mathfrak D} e^{\theta_3(\vec t_0)}$ (see also Figure \ref{fig:Gr24_square} [top,right]):
\begin{enumerate}
\item If ${\mathfrak D} e^{\theta_2(\vec t_0)}<0<{\mathfrak D} e^{\theta_3(\vec t_0)}$, then $\gamma_{S,1} \in \Omega_1$, $\gamma_{S,2} \in \Omega_2$, $\gamma_{1} \in \Omega_4$ and $\gamma_{2} \in \Omega_3$. One such configuration is illustrated by triangles in the Figure;
\item If ${\mathfrak D} e^{\theta_2(\vec t_0)},{\mathfrak D} e^{\theta_3(\vec t_0)}<0$, then $\gamma_{S,1} \in \Omega_1$, $\gamma_{S,2} \in \Omega_3$, $\gamma_{1} \in \Omega_4$ and $\gamma_{2} \in \Omega_2$. One such configuration is illustrated by squares in the Figure;
\item If ${\mathfrak D} e^{\theta_3(\vec t_0)}<0<{\mathfrak D} e^{\theta_2(\vec t_0)}$, then $\gamma_{S,1} \in \Omega_2$, $\gamma_{S,2} \in \Omega_3$, $\gamma_{1} \in \Omega_1$ and $\gamma_{2} \in \Omega_4$. One such configuration is illustrated by stars in the Figure.
\end{enumerate}

\begin{figure}
\centering{
\includegraphics[width=0.49\textwidth]{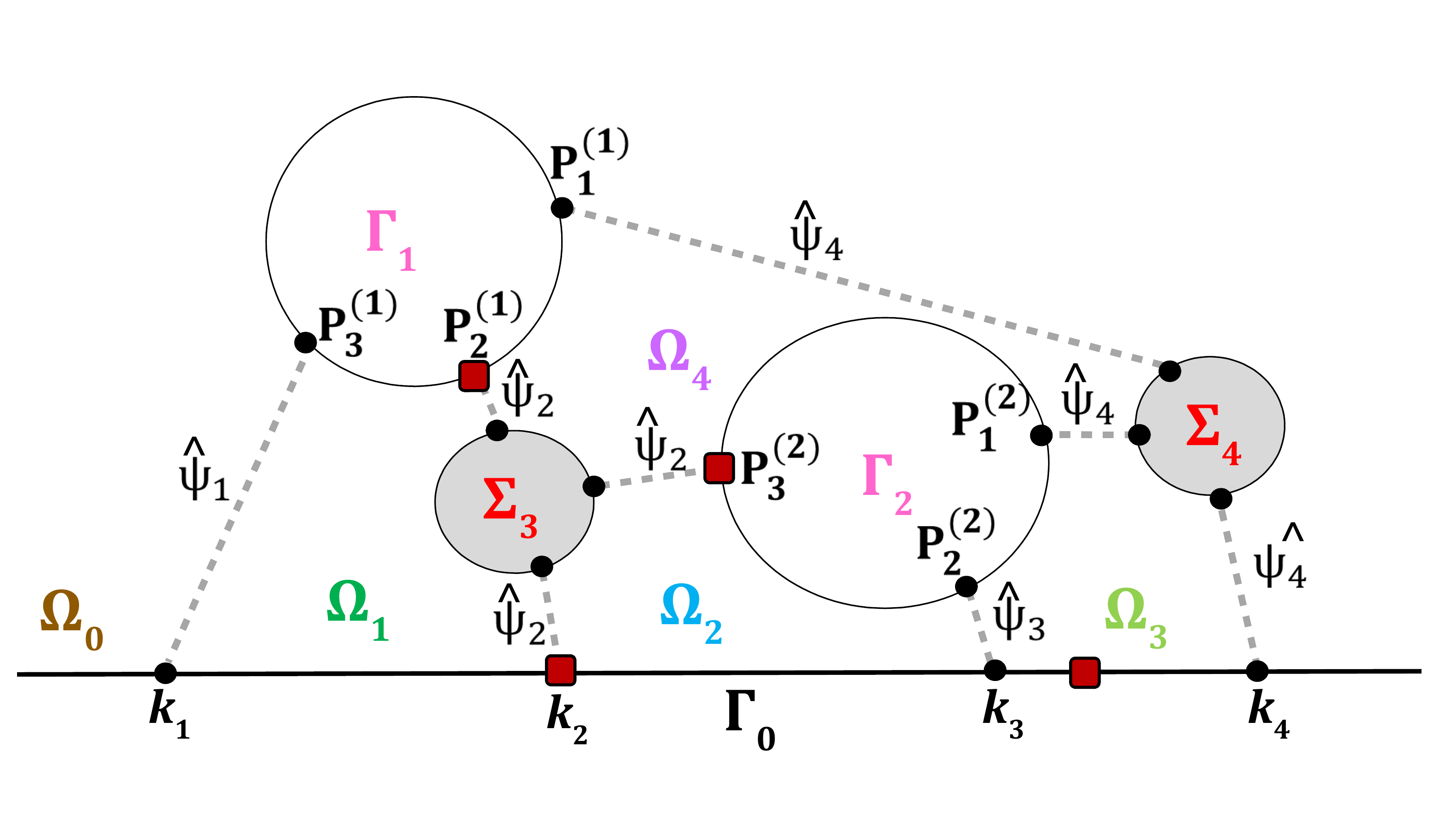}
\hfill
\includegraphics[width=0.49\textwidth]{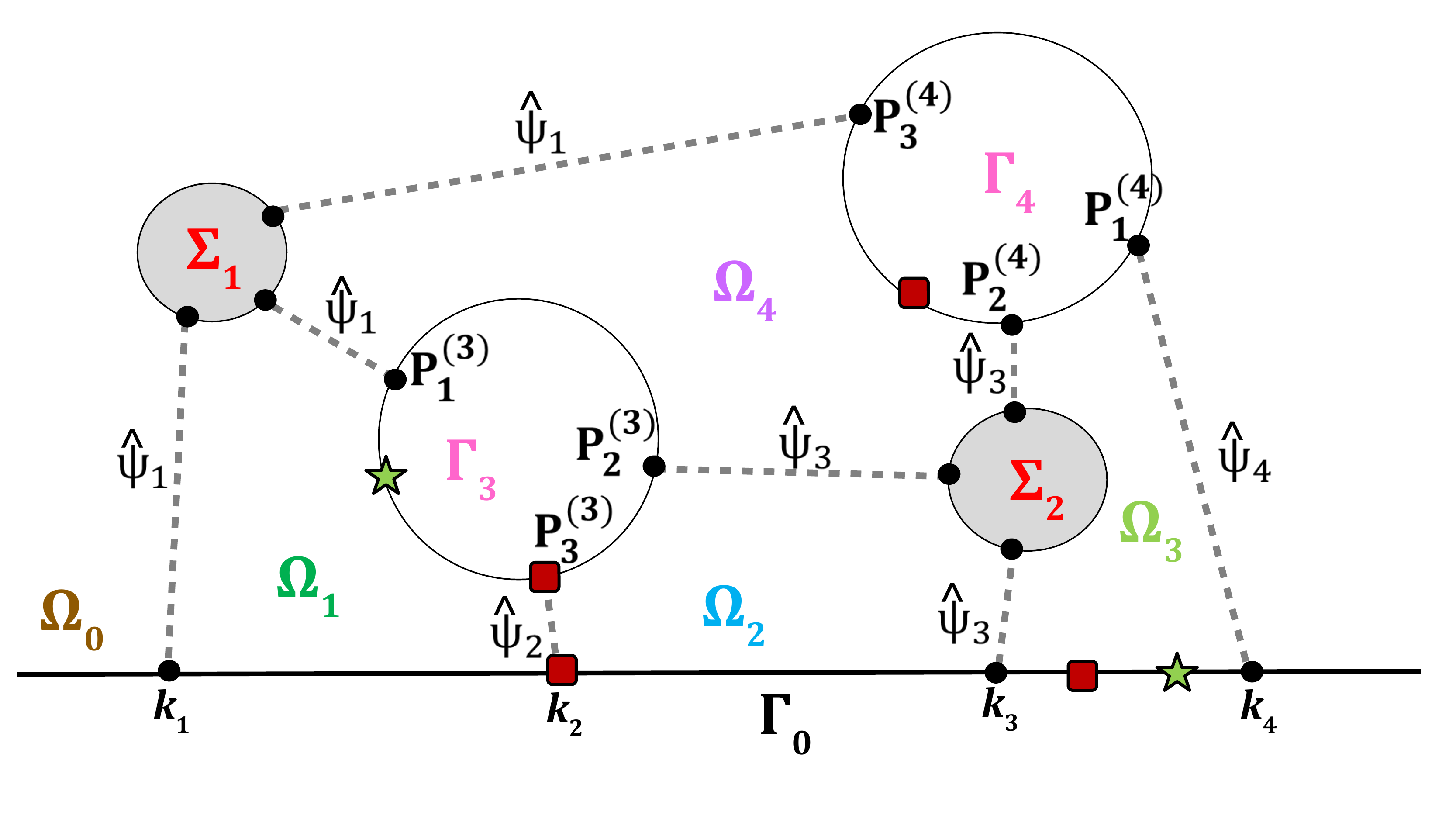}
}
\caption{\small{\sl Non--generic divisor configurations when ${\mathfrak D} e^{\theta_2(\vec t_0)} =0$ on ${\Gamma}_{\mbox{\scriptsize top}}$ [left], and ${\Gamma}_{\mbox{\scriptsize sq-mv}}$ [right].}}
\label{fig:Gr24_top_D2_null}
\end{figure}

\begin{figure}
\centering{
\includegraphics[width=0.49\textwidth]{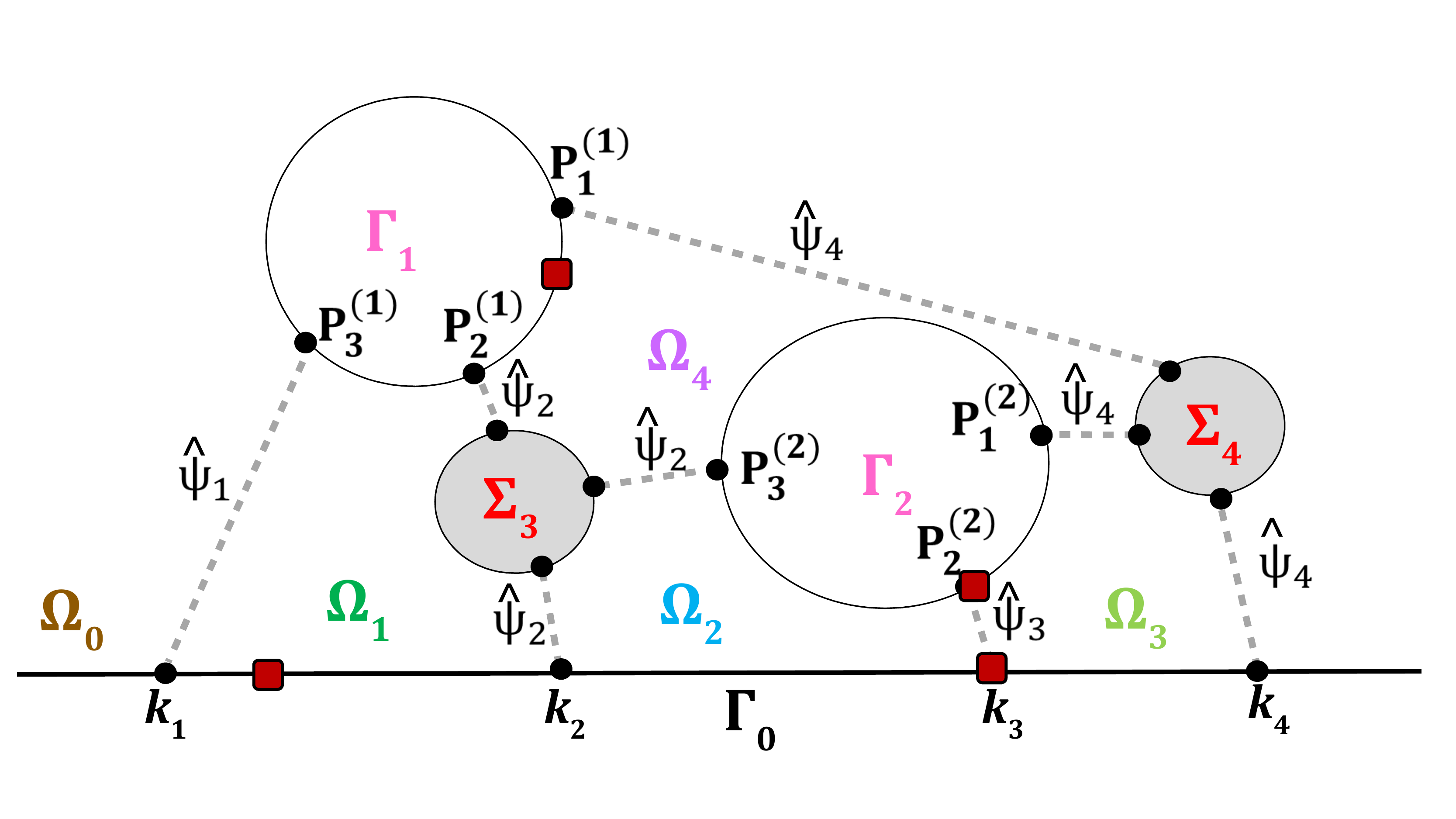}
\hfill
\includegraphics[width=0.49\textwidth]{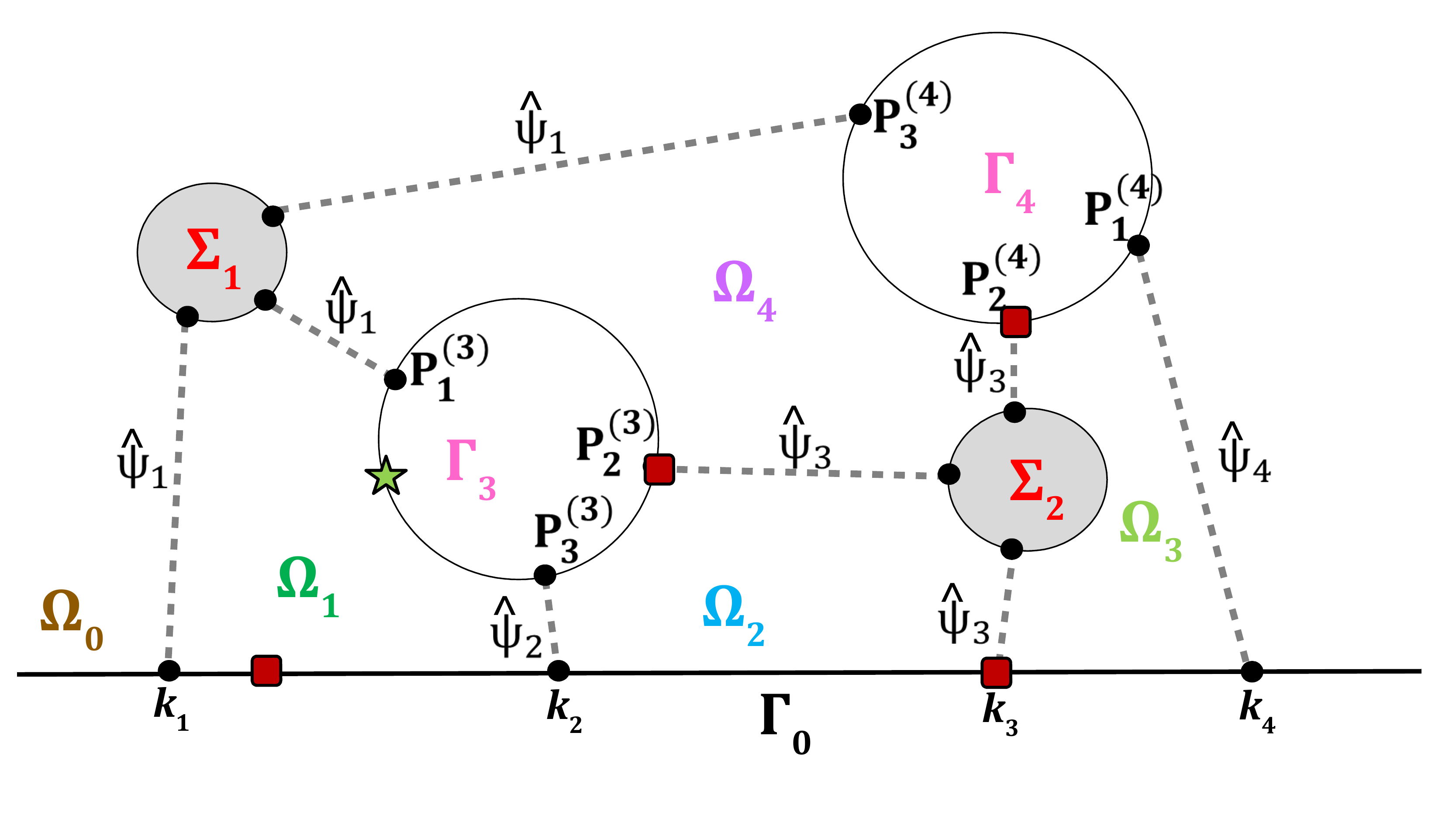}
}
\caption{\small{\sl Non--generic divisor configurations when ${\mathfrak D} e^{\theta_3(\vec t_0)} =0$ on ${\Gamma}_{\mbox{\scriptsize top}}$ [left], and ${\Gamma}_{\mbox{\scriptsize sq-mv}}$ [right].}}
\label{fig:Gr24_top_D3_null}
\end{figure}

On $\Gamma_{\mbox{\scriptsize sq-mv}}$, $\DKP = (\gamma_{S,1} ,\gamma_{S,2},\gamma_3,\gamma_4 )$
where the simple poles $\gamma_{i}=\gamma_{i} (\vec t_0)$ belong to the intersection of 
$\Gamma_{i}$, $i=3,4$, with the union of the finite ovals. In the local coordinates induced by the orientation of ${\mathcal N}_{\mbox{\scriptsize sq-mv}}$, we have
\begin{equation}\label{eq:ex_div_square}
\zeta(\gamma_{3}) = \frac{w_{24} \left( {\mathfrak D} e^{\theta_3(\vec t_0)} + (w_{14}+ w_{24})
{\mathfrak D} e^{\theta_4(\vec t_0)}\right)}{(w_{14}+ w_{24})\left({\mathfrak D} e^{\theta_3(\vec t_0)}+ w_{24}
{\mathfrak D} e^{\theta_4(\vec t_0)}\right)}, \quad\quad   \zeta(\gamma_{4}) = \frac{(w_{14}+w_{24}){\mathfrak D} e^{\theta_4(\vec t_0)}}{{\mathfrak D} e^{\theta_3(\vec t_0)}+ (w_{14}+w_{24})
{\mathfrak D} e^{\theta_4(\vec t_0)}}.
\end{equation}
The three possible configurations of the KP--II pole divisor after the square move are  then (see also Figure \ref{fig:Gr24_square} [bottom,right]):
\begin{enumerate}
\item If ${\mathfrak D} e^{\theta_2(\vec t_0)}<0<{\mathfrak D} e^{\theta_3(\vec t_0)}$, then $\gamma_{S,1} \in \Omega_1$, $\gamma_{S,2} \in \Omega_2$, $\gamma_{3} \in \Omega_4$ and $\gamma_{4} \in \Omega_3$. One such configuration is illustrated by triangles in the Figure;
\item If ${\mathfrak D} e^{\theta_2(\vec t_0)},{\mathfrak D} e^{\theta_3(\vec t_0)}<0$, then $\gamma_{S,1} \in \Omega_1$, $\gamma_{S,2} \in \Omega_3$, $\gamma_{3} \in \Omega_2$ and $\gamma_{4} \in \Omega_4$. One such configuration is illustrated by squares in the Figure;
\item If ${\mathfrak D} e^{\theta_3(\vec t_0)}<0<{\mathfrak D} e^{\theta_2(\vec t_0)}$, then $\gamma_{S,1} \in \Omega_2$, $\gamma_{S,2} \in \Omega_3$, $\gamma_{3} \in \Omega_1$ and $\gamma_{4} \in \Omega_4$. One such configuration is illustrated by stars in the Figure.
\end{enumerate}
We observe that the transformation rule of the divisor points is in agreement with the effect of the square move discussed in Section \ref{sec:square}.
As expected, for any given $[A]\in Gr^{\mbox{\tiny TP}}(2,4)$, there is exactly one KP divisor point in each finite oval, where we use the counting rule established in \cite{AG1} for non--generic soliton data. Non generic divisor configurations (squares) correspond either to
${\mathfrak D} e^{\theta_2(\vec t_0)}=0$ (Figure \ref{fig:Gr24_top_D2_null}) or to ${\mathfrak D} e^{\theta_3(\vec t_0)}=0$ (Figure \ref{fig:Gr24_top_D3_null}) since ${\mathfrak D} e^{\theta_2(\vec t)}
+w_{23}{\mathfrak D} e^{\theta_3(\vec t)}<0$, for all $\vec t$. We plan to discuss the exact definition of global parametrization for this case using resolution of singularities in a future publication.

\appendix

\section{Consistency of the system $\hat E_e$ at internal vertices}\label{app:orient}

In this Section we complete the proof of Lemmas~\ref{lemma:path} and \ref{lemma:cycle}.
Throughout this Appendix we use the same notations as in Section \ref{sec:orient}. In particular $\mathcal P_0$ is the simple path changing orientation and directed from the boundary source $b_{i_0}$ to the boundary sink $b_{j_0}$ in the initial orientation, and $\mathcal Q_0$ is the simple cycle changing orientation. 

We  start with a useful relation between the winding number and the indices $\epsilon_2(\cdot)$ and $s(\cdot,\cdot)$ at a pair of consecutive edges. 
	
\begin{lemma}
\label{lem:wind}
Let $a,b$ be a pair of consecutive edges at a vertex $V$ with $a$ incoming and $b$ outgoing at $V$. Then
\begin{equation}
\label{eq:wind}
\mbox{wind}(a,b) = \frac{\epsilon_2(b)-\epsilon_2(a) + s(a,b)[\epsilon_2(b)-\epsilon_2(a)]^2}{2},
\end{equation}
where, in case $a$ and $b$ are antiparallel the above formula is a limit as in Definition \ref{def:winding_pair}, see also Figure \ref{fig:antipar}.
\end{lemma}
\begin{proof}
First suppose $a$ and $b$ not parallel. Then formula (\ref{eq:wind}) holds since
\begin{enumerate} 
\item If $\epsilon_2(b)=\epsilon_2(a)$, then $\mbox{wind}(a,b)=0$;
\item If $\epsilon_2(a)=0=1-\epsilon_2(b)$, then $\mbox{wind}(a,b)=1$ if and only if $s(a,b)=1$, otherwise $\mbox{wind}(a,b)=0$;  
\item If $\epsilon_2(a)=1=1-\epsilon_2(b)$, then $\mbox{wind}(a,b)=-1$ if and only if $s(a,b)=-1$, otherwise $\mbox{wind}(a,b)=0$
\end{enumerate}
If $a$ and $b$ are parallel, then $\mbox{wind}(a,b)=0$, $\epsilon_2(b)=\epsilon_2(a)$, whereas $s(a,b)$ is not defined. Therefore both the left and the right-hand side are well-defined and equal to 0.
\end{proof}

\subsection{Proof of Lemmas~\ref{lemma:path} and \ref{lemma:cycle} at internal vertices not belonging to $\mathcal P_0$ or to $\mathcal Q_0$}

\begin{figure}
  \centering
  {\includegraphics[width=0.49\textwidth]{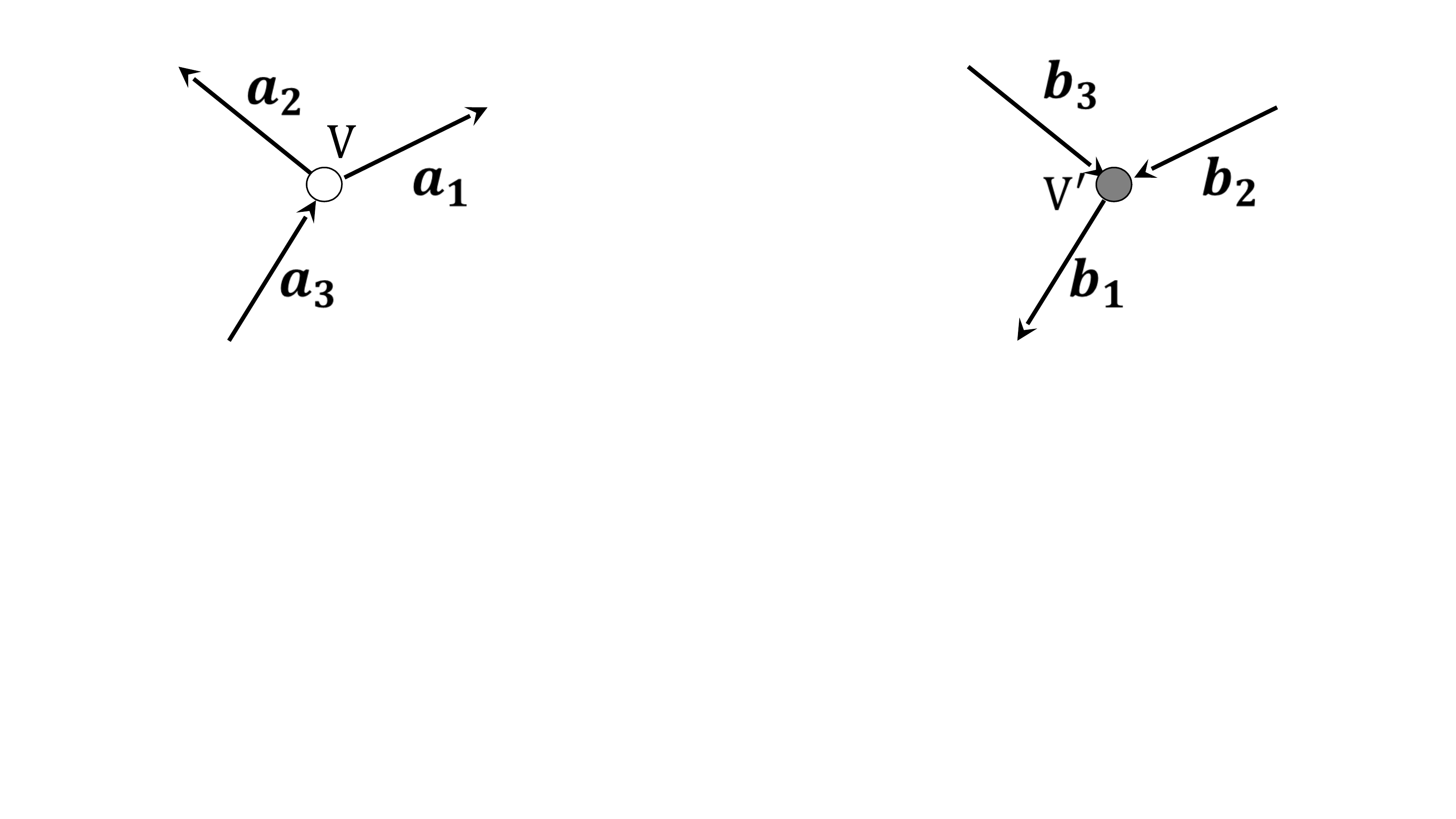}}
	\vspace{-2.3 truecm}
  \caption{Configurations at black [left] and white [right] vertices when the edges do not belong to $\mathcal P_0$ or to $\mathcal Q_0$.}
	\label{fig:config_triv1}
\end{figure}

In the following Lemma we provide equivalent relations to those in Lemmas~\ref{lemma:path} and \ref{lemma:cycle} at internal vertices not belonging to $\mathcal P_0$ or to $\mathcal Q_0$. Generic configurations are shown in Figure \ref{fig:config_triv1}. 

\begin{lemma}
\label{lemma:equiv_rel_1} 
Let $a_1,a_2,a_3$ be the edges at a trivalent white vertex not belonging to $\mathcal P_0$ or to $\mathcal Q_0$ with $a_3$ incoming edge (see Figure \ref{fig:config_triv1}[left]). Then the linear relations for the edge vectors ${\hat E}_{a_i}$, $i\in [3]$, as in (\ref{eq:hat_E_P}) and (\ref{eq:hat_E_Q}) are equivalent to 
\begin{equation}
\label{eq:equiv_white_1}
\mbox{int}(a_3) - \widehat{\mbox{int}}(a_3) = \epsilon(a_1) -\epsilon(a_3) = \epsilon(a_2) -\epsilon(a_3), \quad (\!\!\!\!\!\!\mod 2),
\end{equation}
where $\widehat{\mbox{int}}(a_3)$ denotes the number of intersection of gauge rays with $a_3$ after the change of orientation of the path.

Let $b_1,b_2,b_3$ be the edges at a trivalent black vertex not belonging to $\mathcal P_0$ or to $\mathcal Q_0$ with $b_1$ outgoing edge (see Figure \ref{fig:config_triv1}[right]). Then the linear relations of the vectors ${\hat E}_{b_i}$, $i\in [3]$, as in (\ref{eq:hat_E_P}) and (\ref{eq:hat_E_Q}) are equivalent to 
\begin{equation}
\label{eq:equiv_black_1}
\mbox{int}(b_i) - \widehat{\mbox{int}}(b_i) = \epsilon(b_1) -\epsilon(b_i)  \quad (\!\!\!\!\!\!\mod 2), \quad\quad i=2,3.
\end{equation}
where $\widehat{\mbox{int}}(b_i)$ denotes the number of intersection of gauge rays with $b_i$ after the change of orientation of the path.
\end{lemma}

\begin{proof}
The proof follows immediately substituting the linear relations at the black/white vertices. At the trivalent white vertex $V$  not belonging to the path ${\mathcal P}_0$ or cycle $\mathcal Q_0$, the linear relations before and after the change of orientation are
\[
{\tilde E}_{a_3} = w_{e_3} (-1)^{\mbox{int}(a_3)} \left( (-1)^{\mbox{wind}(a_3,a_2)}{\tilde E}_{a_2}  +  (-1)^{\mbox{wind}(a_3,a_1)}{\tilde E}_{a_1}\right), 
\]
\[
{\hat E}_{a_3} = w_{e_3} (-1)^{\widehat{\mbox{int}(a_3)}} \left( (-1)^{\mbox{wind}(a_3,a_2)}{\hat E}_{a_2}  +  (-1)^{\mbox{wind}(a_3,a_1)}{\hat E}_{a_1}\right).
\] 
Then, we immediately get (\ref{eq:equiv_white_1}) inserting (\ref{eq:hat_E_P}), ${\hat E}_{a_i} = (-1)^{\epsilon(a_i)} {\tilde E}_{a_i}$  in the above formulas. The proof in the other cases follows along similar lines.
\end{proof}

\smallskip

{\sl Proof of Lemma~\ref{lemma:path} at internal vertices not belonging to the directed path $\mathcal P_0$ changing of orientation: } 
First of all we remark that at each white internal vertex not belonging to $\mathcal P_0$, $\epsilon(a_1)=\epsilon(a_2)$ by definition (see (\ref{eq:eps_not_path})). To prove the Lemma we just need to show that the difference of the intersection index before and after the change of orientation at an edge ending at a given (white or black) vertex has the same parity as the difference of the $\epsilon$ index of the same edge and of an edge starting at the same vertex (compare (\ref{eq:equiv_white_1}) and (\ref{eq:equiv_black_1})). Therefore it is sufficient to prove the claim 
\[
\mbox{int}(a_3) - \widehat{\mbox{int}}(a_3) = \epsilon(a_1) -\epsilon(a_3) , \quad (\!\!\!\!\!\!\mod 2)
\]
at a white vertex. The above identity holds true since the incoming edge $a_3$ and the outgoing edge $a_1$ have the same index if and only if either none or both of the gauge rays $\mathfrak{l}_{i_0}$ and $\mathfrak{l}_{i_0}$ intersect $a_3$.$\quad\quad \square$

The proof of Lemma~\ref{lemma:cycle} at internal vertices not belonging to $\mathcal Q_0$ follows along similar lines.

\subsection{Proof of Lemmas~\ref{lemma:path} and \ref{lemma:cycle} at internal vertices belonging to $\mathcal P_0$ or to $\mathcal Q_0$}

\begin{figure}
  \centering
  {\includegraphics[width=0.49\textwidth]{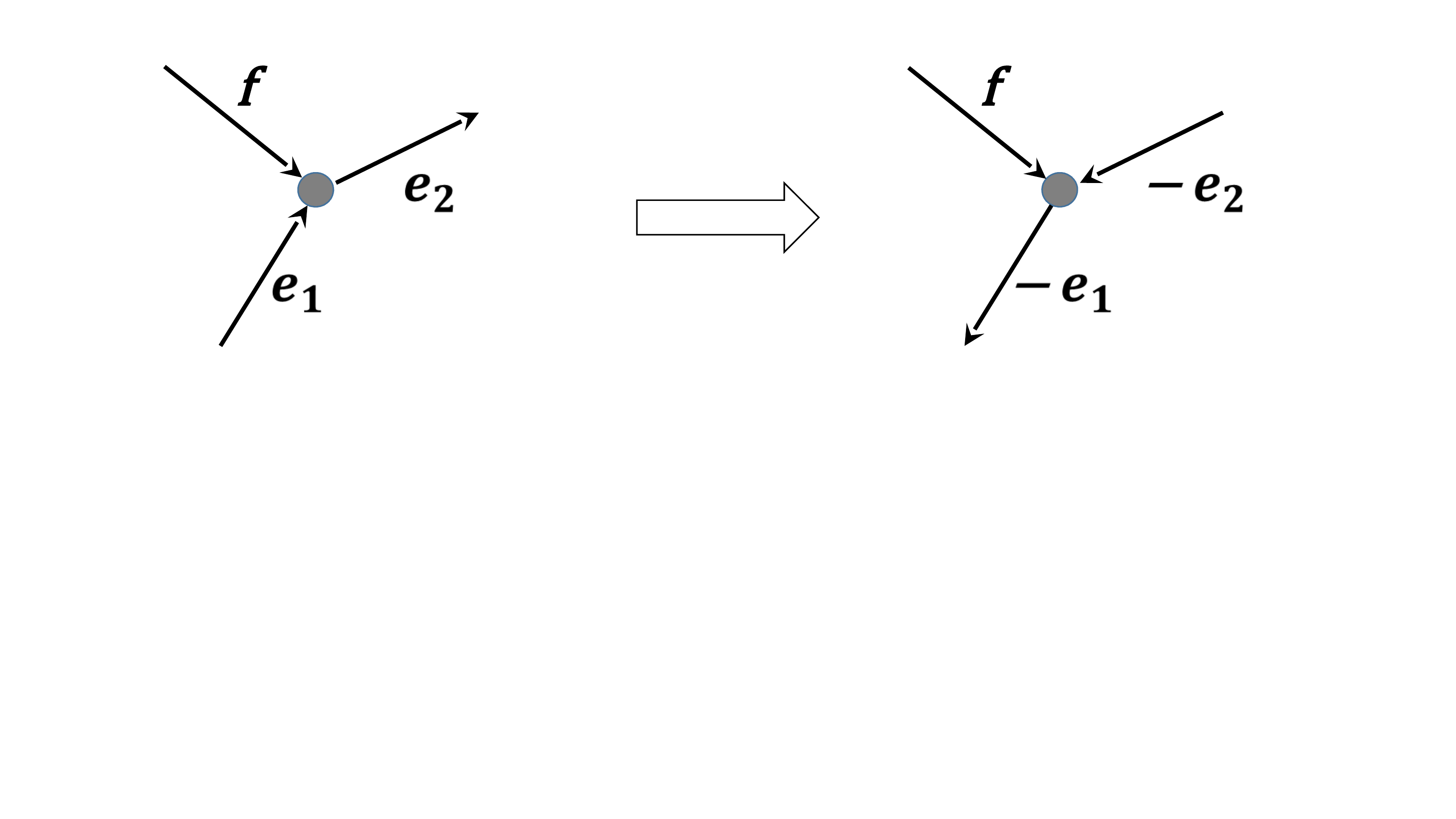}}
	\hfill
	{\includegraphics[width=0.49\textwidth]{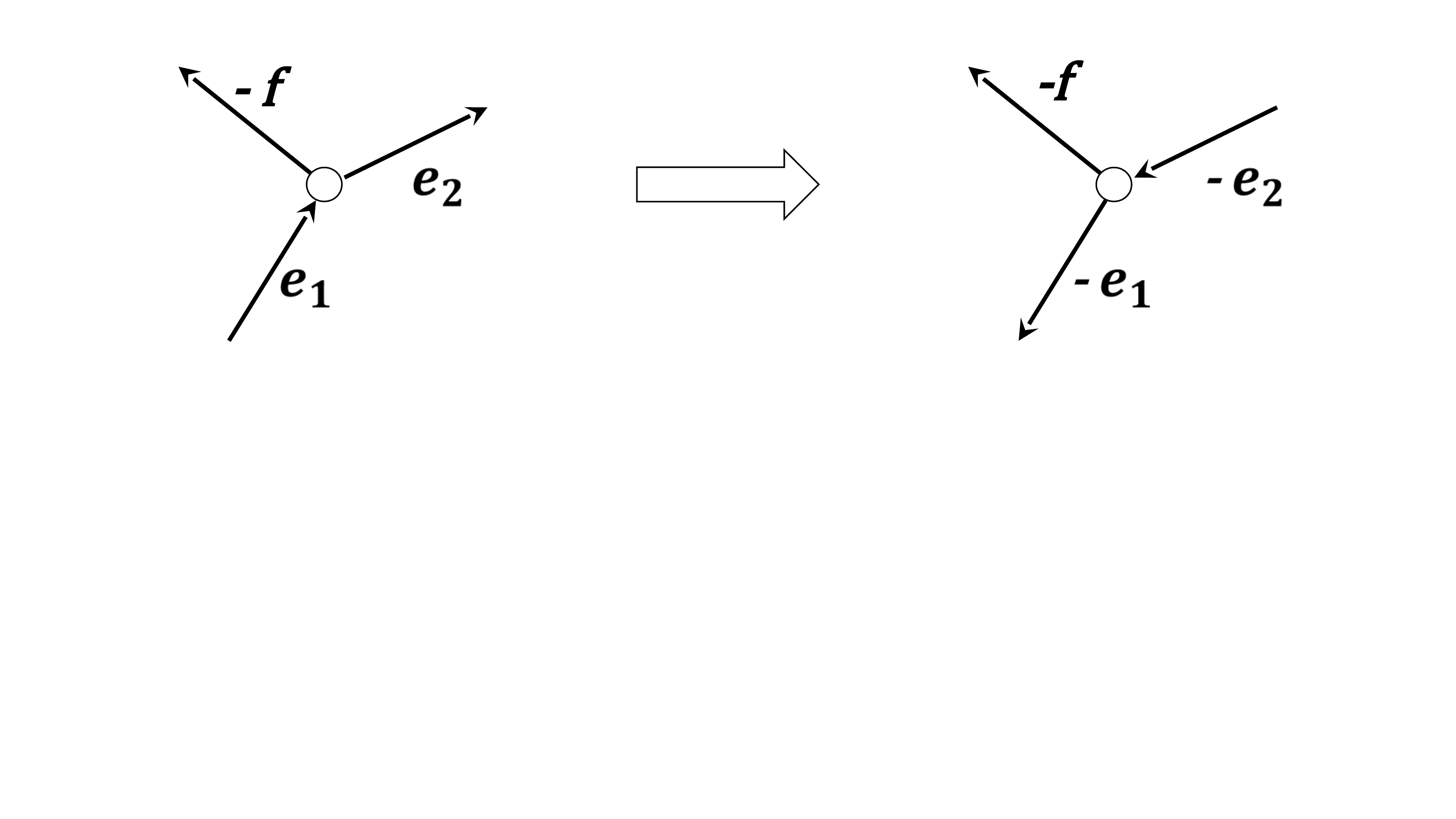}}
	\vspace{-2.3 truecm}
  \caption{Configurations at black [left] and white [right] vertices when $e_1,e_2$ belong to $\mathcal P_0$ or to $\mathcal Q_0$.}
	\label{fig:config_triv}
\end{figure}

In the following Lemma  we describe the relations between the winding number, the intersection number with gauge rays and the indices $\epsilon_i(e)$ introduced above before and after a change of orientation for a pair of consecutive edges $e_1$ and $e_2$ such that $e_1$ is incoming (resp. outgoing) and $e_2$ outgoing (resp. incoming) at the vertex $V$ in the initial (resp. final) orientation (see also Figure \ref{fig:config_triv}). $\mathfrak{l}$ is the gauge ray direction and $\hat\epsilon_3(e_2)$ denotes the intersection index at $e_2$ after the change of orientation. $\mbox{wind}(e_1,e_2)$ and $\mbox{wind}(-e_2,-e_1)$ respectively denote the winding number at the pair of edges before and after the change of orientation.

\begin{lemma}
\label{lem:0.1}
Let $e_1$ and $e_2$ be two consecutive edges at $V$ as above. 
Then 
\begin{equation}
\label{eq:wind12}
\mbox{wind}(e_1,e_2) + \mbox{wind}(-e_2,-e_1) = \epsilon_2(e_2) - \epsilon_2(e_1) =  \frac{s(e_1,\mathfrak{l})-s(e_2,\mathfrak{l}) }{2}
\end{equation}
\begin{equation}
\label{eq:eps12}
\epsilon_1(e_2) - \epsilon_1(e_1) = \hat\epsilon_3(e_2)-\epsilon_3(e_2) \ \  \quad
(\!\!\!\!\!\!\mod 2)
\end{equation}
\end{lemma}
\begin{proof}
If $s(e_1,\mathfrak{l})=s(\mathfrak{l},e_2)$, then either $\mbox{wind}(e_1,e_2)=\pm1$ and $\mbox{wind}(-e_2,-e_1)=0$, or $\mbox{wind}(e_1,e_2)=0$ and $\mbox{wind}(-e_2,-e_1)=\pm1$. 
If $s(e_1,\mathfrak{l})=-s(\mathfrak{l},e_2)$, then $\mbox{wind}(e_1,e_2)=\mbox{wind}(-e_2,-e_1)=0$.

If $\epsilon_1(e_2) = \epsilon_1(e_1)$, then either none or both of the gauge lines $\mathfrak{l}_{i_0}$, $\mathfrak{l}_{j_0}$ cross $e_2$, therefore the number of gauge lines intersecting $e_2$ is the same before and after changing the orientation, and both sides are equal to 0. If $\epsilon_1(e_2) \ne \epsilon_1(e_1)$, then exactly one of the gauge lines $\mathfrak{l}_{i_0}$, $\mathfrak{l}_{j_0}$ crosses $e_2$, therefore the number of gauge lines intersecting $e_2$ changes by 1 after changing the orientation, and both sides are equal to 1 modulo 2. 
\end{proof}

In the following Lemma we provide equivalent relations to those in Lemmas~\ref{lemma:path}, \ref{lemma:cycle} at internal vertices belonging to $\mathcal P_0$ or to $\mathcal Q_0$. 
In the following we denote $-e$ the edge with opposite orientation with respect to $e$. 

\begin{lemma}
\label{lemma:equiv_rel}
Let $e_1,e_2,f$ (respectively $e_1,e_2,-f$) be the edges at a trivalent black (resp. white) vertex with $e_1$ the incoming edge in the initial configuration and $-e_2$ the incoming edge in the final configuration (see Figure \ref{fig:config_triv}). 
Then 
\begin{enumerate}
\item At each internal black vertex belonging to $\mathcal P_0$ or $\mathcal Q_0$ the linear relations on the vectors $\hat E$ in the new orientation are equivalent to equations (\ref{eq:wind12}) and (\ref{eq:eps12}) and to
\begin{equation}
\label{eq:equiv_black}
W_b(e_1,e_2,f) = \widehat{\mbox{int}}(f) - \mbox{int}(f)-\epsilon(f) +\epsilon_1(e_1), \quad
(\!\!\!\!\!\!\mod 2),
\end{equation}
where
\begin{equation}\label{eq:W_b}
W_b(e_1,e_2,f)=\mbox{wind}(f,e_2)-\mbox{wind}(f,-e_1)-\mbox{wind}(e_1,e_2)-\epsilon_2(e_1),
\end{equation}
and $\widehat{\mbox{int}}(f)$ denotes the number of intersection of gauge rays with $f$ after the change of orientation of the path or cycle;
\item At each internal white vertex belonging to to $\mathcal P_0$ or $\mathcal Q_0$ the linear relations on the vectors $\hat E$ in the new orientation are equivalent to equations (\ref{eq:wind12}) and (\ref{eq:eps12}) and to
\begin{equation}
\label{eq:equiv_white}
 W_w(e_1,e_2,-f) = \epsilon(-f)-\epsilon_1(e_1),\quad
(\!\!\!\!\!\!\mod 2),
\end{equation}
where
\begin{equation}\label{eq:W_w}
W_w(e_1,e_2,-f)=\mbox{wind}(e_1,-f)+\mbox{wind}(-e_2,-e_1)-\mbox{wind}(-e_2,-f)+\epsilon_2(e_1)+1
\end{equation}
\end{enumerate}
\end{lemma}

The proof immediately follows substituting (\ref{eq:hat_E_P}) (respectively (\ref{eq:hat_E_Q})) in the linear relations at the black/white vertices before and after the change of orientation along $\mathcal P_0$ (respectively $\mathcal Q_0$).

To complete the proof of Lemmas~\ref{lemma:path}, \ref{lemma:cycle} in the case of internal vertices belonging to the path changing orientation, we preliminarily show that the left--hand side of (\ref{eq:equiv_black}) and of (\ref{eq:equiv_white}) does not depend
on the gauge ray direction in Lemma \ref{lem:indip} and provide relations between indices at black and white vertices with fixed edge configuration and opposite orientation at the edge not belonging to $\mathcal P_0$ or to $\mathcal Q_0$.

\begin{lemma}\label{lem:indip}
Let $e_1,e_2,\pm f$ be as in Lemma \ref{lemma:equiv_rel}. Then
\begin{enumerate}
\item The relation between $\epsilon(f)$ and $\epsilon(-f)$ is
\begin{equation}\label{eq:region_bw_rel}
\widehat{\mbox{int}}(f) - \mbox{int}(f) =\epsilon(f) -\epsilon(-f) \quad
(\!\!\!\!\!\!\mod 2);
\end{equation}
\item Both $W_w(e_1,e_2,-f)$ and $W_b(e_1,e_2,f)$ are independent of the gauge ray direction and satisfy
\begin{equation}\label{eq:wind_bw_rel}
\begin{array}{l}
W_b(e_1,e_2,f) = W_w(e_1,e_2,-f) -2 =\\
= \frac{s(f,e_2)\big[\epsilon_2(e_2)-\epsilon_2(f)\big]^2 + s(f,e_1)\big[1-\epsilon_2(e_1)-\epsilon_2(f)\big]^2 + s(e_2,e_1) \big[\epsilon_2(e_2)-\epsilon_2(e_1)\big]^2 -1}{2}.
\end{array}
\end{equation}
\end{enumerate}
\end{lemma}

\begin{proof}
The proof of (\ref{eq:region_bw_rel}) follows immediately taking into account that the initial and final vertices of $f$ lay in regions with the same index $\pm$ if and only if either none or both of the gauge rays $\mathfrak{l}_{i_0}$ and $\mathfrak{l}_{i_0}$ intersect $f$.

The identities in (\ref{eq:wind_bw_rel}) follow inserting (\ref{eq:wind12}) and (\ref{eq:wind}) in (\ref{eq:W_b}) and (\ref{eq:W_w}) and taking into account that, by definition,
\[
s(a,b) = -s(b,a) = s(-a,-b) = -s(-a,b).
\]
In order to verify that the right-hand side of (\ref{eq:wind_bw_rel}) does not depend on the gauge ray direction $\mathfrak{l}$, it is sufficient to verify that it does not change when the gauge ray direction $\mathfrak{l}$ crosses one of the vectors $f$, $-f$, $e_1$, $-e_1$, $e_2$, $-e_2$.
Let us denote
$A_1 = s(f,e_2)\bigg[\epsilon_2(e_2)-\epsilon_2(f)\bigg]^2$, $A_2 =s(f,e_1)\bigg[1-\epsilon_2(e_1)-\epsilon_2(f)\bigg]^2$ and
$A_3 =s(e_2,e_1) \bigg[\epsilon_2(e_2)-\epsilon_2(e_1)\bigg]^2$.
Then, for any fixed configuration of $e_1,e_2$ and $f$, the value $A_1+A_2+A_3$ is independent of $\mathfrak{l}$ as $\mathfrak{l}$ turn clockwise, since:
\begin{enumerate}
\item If $\mathfrak{l}\sim f$, then $\epsilon_2(f)$ increases from $0$ to $1$, $A_1$ decreases by one, $A_2$ increases by one and $A_3$ doesn't change since
\begin{itemize}
\item $s(f,e_2)=1$ implies $\epsilon_2(e_2)=1$, whereas $s(f,e_2)=-1$ implies $\epsilon_2(e_2)=0$;
\item $s(f,e_1)=1$ implies $\epsilon_2(e_1)=1$, whereas $s(f,e_1)=-1$ implies $\epsilon_2(e_1)=0$.
\end{itemize}
\item If $\mathfrak{l}\sim -f$, then $\epsilon_2(f)$ decreases from $1$ to $0$, $A_1$ decreases by one, $A_2$ increases by one and $A_3$ doesn't change since
\begin{itemize}
\item $s(f,e_2)=1$ implies $\epsilon_2(e_2)=0$, whereas $s(f,e_2)=-1$ implies $\epsilon_2(e_2)=1$;
\item $s(f,e_1)=1$ implies $\epsilon_2(e_1)=0$, whereas $s(f,e_1)=-1$ implies $\epsilon_2(e_1)=1$.
\end{itemize}
\item If $\mathfrak{l}\sim e_1$, then $\epsilon_2(e_1)$ increases from $0$ to $1$, $A_1$ doesn't change, $A_2$ decreases by one and
$A_3$ increases by one since
\begin{itemize}
\item $s(f,e_1)= 1$ implies $\epsilon_2(f)=0$, whereas $s(f,e_1)=-1$ implies $\epsilon_2(f)=1$;
\item $s(e_2,e_1)=1$ implies $\epsilon_2(e_2)=0$, whereas $s(e_2,e_1)=-1$ implies $\epsilon_2(e_2)=1$. 
\end{itemize}
\item If $\mathfrak{l}\sim -e_1$, then $\epsilon_2(e_1)$ decreases from $1$ to $0$, $A_1$ doesn't change, $A_2$ decreases by one and $A_3$ increases by one since
\begin{itemize}
\item $s(f,e_1)= 1$ implies $\epsilon_2(f)=1$, whereas $s(f,e_1)=-1$ implies $\epsilon_2(f)=0$;
\item $s(e_2,e_1)=1$ implies $\epsilon_2(e_2)=1$, whereas $s(e_2,e_1)=-1$ implies $\epsilon_2(e_2)=0$.
\end{itemize}
\item If $\mathfrak{l}\sim e_2$, then $\epsilon_2(e_2)$ increases from $0$ to $1$, $A_1$ increases by one, $A_2$ doesn't change and $A_3$ decreases by one since
\begin{itemize}
\item $s(f,e_2)= 1$ implies $\epsilon_2(f)=0$, whereas $s(f,e_2)=-1$ implies $\epsilon_2(f)=1$;
\item $s(e_2,e_1)=1$ implies $\epsilon_2(e_1)= 1$, whereas $s(e_2,e_1)=-1$ implies $\epsilon_2(e_1)= 0$. 
\end{itemize}
\item If $\mathfrak{l}\sim -e_2$, then $\epsilon_2(e_2)$ decreases from $1$ to $0$, $A_1$ increases by one, $A_2$ doesn't change and $A_3$ decreases by one since
\begin{itemize}
\item $s(f,e_2)= 1$ implies $\epsilon_2(f)=1$, whereas $s(f,e_2)=-1$ implies $\epsilon_2(f)=0$; 
\item $s(e_2,e_1)=1$ implies $\epsilon_2(e_1)= 0$, whereas $s(e_2,e_1)=-1$ implies $\epsilon_2(e_1)= 1$. 
\end{itemize}
\end{enumerate}
\end{proof}

We are now ready to complete the proof of Lemmas~\ref{lemma:path}, \ref{lemma:cycle} at all internal vertices belonging to $\mathcal P_0$ or to $\mathcal Q_0$. First of all, inserting (\ref{eq:region_bw_rel}) and (\ref{eq:wind_bw_rel}) into (\ref{eq:equiv_black}) and (\ref{eq:equiv_white}), it is evident that (\ref{eq:equiv_black}) holds if and only if (\ref{eq:equiv_white}) holds. Therefore it is sufficient to prove  
(\ref{eq:equiv_white}) for a fixed gauge ray direction, say $\mathfrak{l} \sim f$ and $\epsilon_2(e_1)=0$. Then there are two generic configurations: either $s(e_1,e_2)=-1$ (see Figure \ref{fig:vertex1}) or $s(e_1,e_2)=1$ (see Figure \ref{fig:vertex2}).

\begin{figure}[H]
  \centering
  {\includegraphics[width=0.95\textwidth]{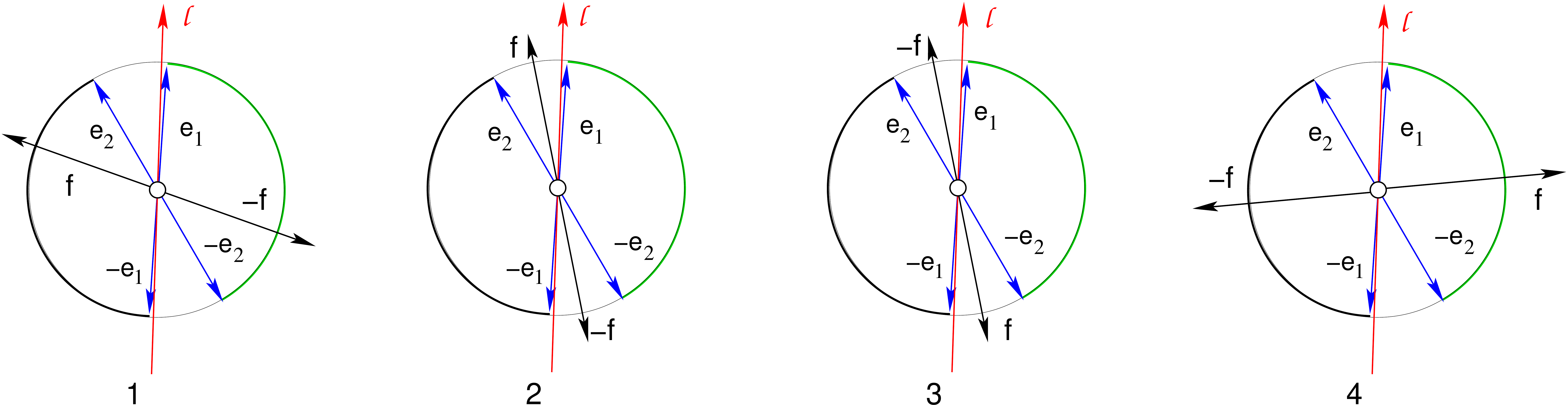}}
  \caption{The representation of sectors corresponding to the vertex configurations at a white vertex when $\mathfrak{l} \sim e_1$ and $s(e_2,e_1)=-1$.\label{fig:vertex1}}
\end{figure}

In the first case ($s(e_2,e_1)=-1$), $\epsilon_2(e_2)=1$, $\epsilon_2(e_1)=0$, and 
\[
W_w(e_1,e_2,-f)= \frac{\big[s(f,e_2) + s(f,e_1)\big]\big[1-\epsilon_2(f)\big]^2 -2}{2} \quad
(\!\!\!\!\!\!\mod 2).
\]
Therefore (\ref{eq:equiv_white}) holds in all configurations, since
\begin{enumerate}
\item If $f$ rotates clockwise from $-e_1$ to $e_1$ (configurations 1 and 2 in Figure \ref{fig:vertex1}), $\epsilon_2(f)=1$, $W_w(e_1,e_2,-f)=1\,\,
(\!\!\!\!\mod 2)$ and $\epsilon (-f) -\epsilon_1(e_1) = 1 \,\,
(\!\!\!\!\mod 2)$; 
\item If $f$ rotates from $e_1$ to $-e_2$ (configuration 3 in Figure \ref{fig:vertex1}), $\epsilon_2(f)=0$, $s(f,e_1)=1$, $s(f,e_2)=-1$,  $W_w(e_1,e_2,-f)=1\,\,(\!\!\!\!\mod 2)$ and $\epsilon (-f) -\epsilon_1(e_1) = 1 \,\,
(\!\!\!\!\mod 2)$;
\item If $f$ rotates clockwise from $-e_2$ to $e_1$ (configuration 4 in Figure \ref{fig:vertex1}), $\epsilon_2(f)=0$, $s(f,e_1)=1$, $s(f,e_2)=1$,  $W_w(e_1,e_2,-f)=0\,\,(\!\!\!\!\mod 2)$, $\epsilon (-f) -\epsilon_1(e_1) = 0$.
\end{enumerate}

\begin{figure}[H]
  \centering
  {\includegraphics[width=0.95\textwidth]{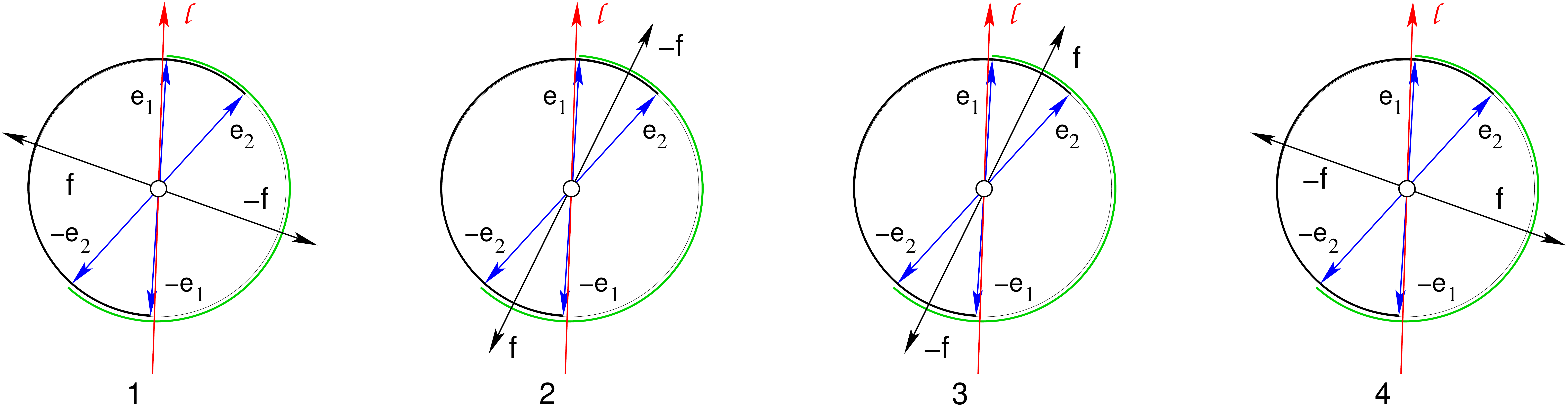}}
  \caption{The representation of sectors corresponding to the vertex configurations at a white vertex when $\mathfrak{l} \sim e_1$ and $s(e_2,e_1)=1$.\label{fig:vertex2}}
\end{figure}

In the second case ($s(e_2,e_1)=1$), $\epsilon_2(e_2)=0$, $\epsilon_2(e_1)=0$, and 
\[
W_w(e_1,e_2,-f)= \frac{ s(f,e_2) \big[\epsilon_2(f)\big]^2  + s(f,e_1)\big[1-\epsilon_2(f)\big]^2 -1}{2} \quad
(\!\!\!\!\!\!\mod 2).
\]
Therefore (\ref{eq:equiv_white}) holds in all configurations, since
\begin{enumerate}
\item If $f$ rotates clockwise from $-e_2$ to $e_1$ (configuration 1 and 2 in Figure \ref{fig:vertex2}), $\epsilon_2(f)=1$, $s(f,e_2)=-1$, $W_w(e_1,e_2,-f)=1\,\,
(\!\!\!\!\mod 2)$ and $\epsilon (-f) -\epsilon_1(e_1) = 1 \,\,
(\!\!\!\!\mod 2)$; 
\item If $f$ rotates from $-e_1$ to $-e_2$ (configuration 2 in Figure \ref{fig:vertex2}), $\epsilon_2(f)=0$, $s(f,e_1)=1$, $W_w(e_1,e_2,-f)=0\,\,
(\!\!\!\!\mod 2)$ and $\epsilon (-f) -\epsilon_1(e_1) = 0 \,\,
(\!\!\!\!\mod 2)$;
\item If $f$ rotates clockwise from $e_1$ to $-e_1$ (configurations 3 and 4 in Figure \ref{fig:vertex2}), $\epsilon_2(f)=0$, $s(f,e_1)=1$,  $W_w(e_1,e_2,-f)=0\,\,
(\!\!\!\!\mod 2)$ and $\epsilon (-f) -\epsilon_1(e_1) = 0$.
\end{enumerate}

Finally in the non generic configuration when $e_1$ and $e_2$ are parallel, (\ref{eq:equiv_white}) becomes
\[
\frac{1}{2} \left( s(f,e_1) + 3\right) = \epsilon(-f) - \epsilon (e_1) \quad (\!\!\!\!\!\!\mod 2).
\]
Then $s(f,e_1)=1$ if and only if $-f$ is to the left of $e_1$ and the above equality holds true. $\quad \square$

\bibliographystyle{alpha}

\end{document}